\documentclass[11pt,a4paper]{article}

\usepackage{fontspec}
\usepackage[titletoc,title,header]{appendix}

\usepackage{amsmath}
\usepackage{amssymb}
\usepackage{amsthm}
\usepackage{mathtools}
\usepackage{xparse}    

\usepackage[margin=1in]{geometry}

\usepackage{graphicx}
\usepackage{xcolor}

\usepackage{hyperref}
\hypersetup{
  colorlinks=true,
  linkcolor=blue,
  citecolor=blue,
  urlcolor=blue,
  pdfauthor={Yoshitsugu Sekine},
  pdftitle={The Buchholz Algebra from the Universal Resolvent Algebra}
}

\usepackage[numbers]{natbib}
\theoremstyle{plain}
\newtheorem{thm}{Theorem}[section]
\AtEndEnvironment{thm}{\qed}
\newtheorem{lem}[thm]{Lemma}
\AtEndEnvironment{lem}{\qed}
\newtheorem{prop}[thm]{Proposition}
\AtEndEnvironment{prop}{\qed}
\newtheorem{cor}[thm]{Corollary}
\AtEndEnvironment{cor}{\qed}

\AtEndEnvironment{conj}{\qed}

\AtEndEnvironment{fact}{\qed}

\theoremstyle{definition}
\newtheorem{defn}[thm]{Definition}
\AtEndEnvironment{defn}{\qed}
\newtheorem{ex}[thm]{Example}
\AtEndEnvironment{ex}{\qed}

\theoremstyle{remark}
\newtheorem{rem}[thm]{Remark}
\AtEndEnvironment{rem}{\qed}

\makeatletter
\providecommand*{\dashv}{\mathrel{\mathpalette\@Dashv\vDash}}
\newcommand*{\@dashv}[2]{\reflectbox{$\m@th#1#2$}}
\makeatother
\newcommand{\abscard}[1]{\abs{#1}} 
\newcommand{\abs}[1]{\left| #1 \right|} 
\NewDocumentCommand{\weaknorm}{O{\dbk} m}{#1{#2}} 
\newcommand{\norm}[1]{\left\Vert #1 \right\Vert} 

\newcommand{\bkt}[2]{\left\langle #1,\,#2 \right\rangle} 
\newcommand{\rbkt}[2]{\left( #1,\,#2 \right)} 
\newcommand{\isomto}{\mathrel{\rightarrowtail\kern-1.9ex\twoheadrightarrow}} 
\newcommand{\monoto}{\rightarrowtail} 
\newcommand{\cbk}[1]{\left\{ #1 \right\}} 
\newcommand{\dbk}[1]{\left\langle #1 \right\rangle} 
\newcommand{\rbk}[1]{\left( #1 \right)} 
\newcommand{\sqbk}[1]{\left[ #1 \right]} 
\newcommand{\fun}[2]{#1 \rbk{#2}} 
\newcommand{\sqfun}[2]{#1 \sqbk{#2}} 
\newcommand{\closedinterval}[2]{\sqbk{#1,\,#2}} 
\newcommand{\rightopeninterval}[2]{\left[#1, \, #2 \right)} 
\newcommand{\unitclosedinterval}{\closedinterval{0}{1}} 
\newcommand{\commutator}[2]{\sqbk{#1,\,#2}} 

\NewDocumentCommand{\imunit}{O{\mathsf{i}}}{#1} 
\NewDocumentCommand{\placeholder}{O{\bullet}}{#1} 
\NewDocumentCommand{\trace}{O{\operatorname{Tr}}}{#1} 
\newcommand{\Ker}{\operatorname{Ker}} 
\newcommand{\bigmiddleslash}[2]{\left. #1 \middle/ #2 \right.} 
\newcommand{\clos}[1]{\overline{#1}} 
\newcommand{\cmpconj}[1]{\overline{#1}} 
\newcommand{\diag}{\operatorname{diag}} 
\newcommand{\dom}{\operatorname{dom}} 
\newcommand{\dual}[1]{#1^{\ast}} 
\newcommand{\eqcsq}[1]{\sqbk{#1}} 
\newcommand{\idone}{1} 
\newcommand{\id}{\mathrm{id}} 
\newcommand{\napiernum}{\mathsf{e}} 
\newcommand{\net}[2]{\rbk{#1}_{#2}} 
\newcommand{\opchern}[1]{\operatorname{ch}} 
\newcommand{\opimag}{\operatorname{Im}} 
\newcommand{\setSymbolDownLeft}[2]{{\vphantom{#2}}_{#1}{#2}} 
\newcommand{\setSymbolUpLeft}[2]{{\vphantom{#2}}^{#1}{#2}} 
\newcommand{\supp}{\operatorname{supp}} 
\NewDocumentCommand{\agvariety}{O{\mathcal}}{#1} 
\NewDocumentCommand{\cmdrel}{O{\omega}}{#1} 
\NewDocumentCommand{\dfsp}{O{A}}{#1} 
\NewDocumentCommand{\eqcpointed}{O{\eqcsq} m}{#1{#2}_{\ast}} 
\NewDocumentCommand{\fnheaviside}{O{H}}{#1} 
\NewDocumentCommand{\fthol}{O{\mathcal{O}}}{#1} 
\NewDocumentCommand{\ftmero}{O{\mathcal{M}}}{#1} 
\NewDocumentCommand{\grcentralizer}{O{Z}}{#1} 
\NewDocumentCommand{\grmetform}{O{2} m m}{\grmet[#1] \! \rbkt{#2}{#3}} 
\NewDocumentCommand{\grmet}{O{2}}{\setSymbolDownLeft{#1}{g}} 
\NewDocumentCommand{\grnormalizer}{O{N}}{#1} 
\NewDocumentCommand{\gropasym}{O{A}}{#1} 
\NewDocumentCommand{\gropsym}{O{S}}{#1} 
\NewDocumentCommand{\grpermorderedpair}{O{\mathcal{P}}}{#1} 
\NewDocumentCommand{\grsym}{O{\mathfrak{S}} m}{#1_{#2}} 
\NewDocumentCommand{\gtbase}{O{\mathcal}}{#1} 
\NewDocumentCommand{\gtfilter}{O{\mathcal}}{#1} 
\NewDocumentCommand{\gtfmlclosed}{O{\mathcal}}{#1} 
\NewDocumentCommand{\gtfmlopen}{O{\mathcal}}{#1} 
\NewDocumentCommand{\gtopenball}{O{U}}{#1} 
\NewDocumentCommand{\gtopencover}{O{\mathcal}}{#1} 
\NewDocumentCommand{\gtopennbh}{O{\mathcal}}{#1} 
\NewDocumentCommand{\gtpreopencover}{O{\mathcal}}{#1} 
\NewDocumentCommand{\gtsubbase}{O{\mathcal}}{#1} 
\NewDocumentCommand{\gtvicinity}{O{\mathcal}}{#1} 
\NewDocumentCommand{\lasp}{O{\mathcal}}{#1} 
\NewDocumentCommand{\latprightrbk}{O{\top} m}{\rbk{#2}^{#1}} 
\NewDocumentCommand{\latpright}{O{\top} m}{#2^{#1}} 
\NewDocumentCommand{\latp}{O{t} m}{\setSymbolUpLeft{#1}{#2}} 
\NewDocumentCommand{\lpdistribution}{O{\mu} m}{#2_{\ast,#1}} 
\NewDocumentCommand{\lpmollifier}{O{\rho}}{#1} 
\NewDocumentCommand{\lpofpositive}{O{\chi}}{#1} 
\NewDocumentCommand{\manliederiv}{O{L}}{#1} 
\NewDocumentCommand{\mansmoothnbh}{O{\mathcal}}{#1} 
\NewDocumentCommand{\mblfmldsysgenerated}{O{d} m}{\fun{#1}{#2}} 
\NewDocumentCommand{\mblfmlgenerated}{O{\sigma} m}{\fun{#1}{#2}} 
\NewDocumentCommand{\oacorrfn}{O{\Gamma}}{#1} 
\NewDocumentCommand{\oagnsvector}{O{\Omega}}{#1} 
\NewDocumentCommand{\oaideal}{O{\mathcal}}{#1} 
\NewDocumentCommand{\oanumberoperator}{O{A}}{#1} 
\NewDocumentCommand{\oaposcone}{O{\mathcal{P}}}{#1} 
\NewDocumentCommand{\oapressure}{O{P}}{#1} 
\NewDocumentCommand{\oarepn}{O{\pi}}{#1} 
\NewDocumentCommand{\oaspkmsstate}{O{\txtkms}}{#1} 
\NewDocumentCommand{\oaspnormalstate}{O{N}}{#1} 
\NewDocumentCommand{\oasppurestate}{O{P}}{#1} 
\NewDocumentCommand{\oaspstate}{O{E}}{#1} 
\NewDocumentCommand{\oastatevector}{O{\Omega}}{#1} 
\NewDocumentCommand{\oastate}{O{\omega}}{#1} 
\NewDocumentCommand{\opdilation}{O{\delta}}{#1} 
\NewDocumentCommand{\opdmat}{O{\rho}}{#1} 
\NewDocumentCommand{\opfockan}{O{a}}{#1} 
\NewDocumentCommand{\opfockcran}{O{a}}{#1^{\#}} 
\NewDocumentCommand{\opfockcrdagger}{O{a}}{#1^{\dagger}} 
\NewDocumentCommand{\opfockcr}{O{a}}{#1^{\ast}} 
\NewDocumentCommand{\opfocknumber}{O{N}}{#1} 
\NewDocumentCommand{\opfocksegalconj}{O{\pi}}{#1} 
\NewDocumentCommand{\opfocksegal}{O{\phi}}{#1} 
\NewDocumentCommand{\opspecmeas}{O{E}}{#1} 
\NewDocumentCommand{\opspec}{O{} m}{\fun{\sigma_{#1}}{#2}} 
\NewDocumentCommand{\optransl}{O{\tau}}{#1} 
\NewDocumentCommand{\opvarspec}{O{} m}{\operatorname{spec}_{#1} #2} 
\NewDocumentCommand{\physaction}{O{\mathcal{A}}}{#1} 
\NewDocumentCommand{\physcharge}{O{e}}{\mathrm{#1}} 
\NewDocumentCommand{\physcplconst}{O{\mathsf{g}}}{#1} 
\NewDocumentCommand{\physelectrostaticcapasity}{O{\mathrm{Cap}}}{#1} 
\NewDocumentCommand{\physenergy}{O{E}}{#1} 
\NewDocumentCommand{\physgse}{O{E}}{#1_{0}} 
\NewDocumentCommand{\physham}{O{H}}{#1} 
\NewDocumentCommand{\physlagdensity}{O{\mathcal{L}}}{#1} 
\NewDocumentCommand{\physlag}{O{L}}{#1} 
\NewDocumentCommand{\physliouvilean}{O{L}}{#1} 
\NewDocumentCommand{\physmass}{O{m}}{#1} 
\NewDocumentCommand{\prbbrownmv}{O{B}}{#1} 
\NewDocumentCommand{\prbcharfun}{O{\chi}}{#1} 
\NewDocumentCommand{\prbdist}{O{\mathcal{P}}}{#1} 
\NewDocumentCommand{\prbgaussianmeasure}{O{\msrcal{N}}}{#1} 
\NewDocumentCommand{\prbnormaldist}{O{N}}{#1} 
\NewDocumentCommand{\prbpoissonprocess}{O{N}}{#1} 
\NewDocumentCommand{\prbprocess}{O{X}}{#1} 
\NewDocumentCommand{\prbqspace}{O{\mathcal{Q}}}{#1} 
\NewDocumentCommand{\prbspsample}{O{\Omega}}{#1} 
\NewDocumentCommand{\prbspstate}{O{\mathsf{S}}}{#1} 
\NewDocumentCommand{\psh}{O{\mathfrak}}{#1} 
\NewDocumentCommand{\qtquantumchannel}{O{\mathcal{L}}}{#1} 
\NewDocumentCommand{\repn}{O{\pi}}{#1} 
\NewDocumentCommand{\schattencls}{O{\mathbb{K}}}{#1} 
\NewDocumentCommand{\setfmlcylinder}{O{\mathcal{C}}}{#1} 
\NewDocumentCommand{\setfml}{O{\mathcal}}{#1} 
\NewDocumentCommand{\setindex}{O{\mathcal} m}{#1{#2}} 
\NewDocumentCommand{\setlattice}{O{\Gamma}}{#1} 
\NewDocumentCommand{\setspecial}{O{\mathcal} m}{#1{#2}} 
\NewDocumentCommand{\shdiffform}{O{\sheaf{A}}}{#1} 
\NewDocumentCommand{\sheaf}{O{\mathfrak}}{#1} 
\NewDocumentCommand{\smchemicalpotential}{O{\mu}}{#1} 
\NewDocumentCommand{\smenergydensity}{O{\varrho}}{#1} 
\NewDocumentCommand{\smfluctuationwithdmat}{O{\beta} m}{\smuncertaintywithdmat[#1]{#2}^2} 
\NewDocumentCommand{\sminvtemperature}{O{\beta}}{#1} 
\NewDocumentCommand{\smlocaldensityoperator}{O{\rho}}{#1} 
\NewDocumentCommand{\smmicrocanonicalstate}{O{\beta} m}{\physmean{#2}_{#1}} 
\NewDocumentCommand{\smnumberdensity}{O{\rho}}{#1} 
\NewDocumentCommand{\smooth}{O{\mathcal{E}}}{#1} 
\NewDocumentCommand{\smparticlenumber}{O{N}}{#1} 
\NewDocumentCommand{\smpressure}{O{p}}{#1} 
\NewDocumentCommand{\smspecificfreeenergy}{O{\bar{f}}}{#1} 
\NewDocumentCommand{\smthermalvac}{O{\beta}}{\Omega_{#1}} 
\NewDocumentCommand{\smuncertaintywithdmat}{O{\beta} m}{\rbk{\triangle #2}_{#1}} 
\NewDocumentCommand{\sphilbfrak}{O{\mathfrak}}{#1} 
\NewDocumentCommand{\sphilb}{O{\mathcal}}{#1} 
\NewDocumentCommand{\splowerhalf}{O{\mathbb{H}}}{#1_{\txtneg}} 
\NewDocumentCommand{\spupperhalf}{O{\mathbb{H}}}{#1_{\txtnonneg}} 
\NewDocumentCommand{\topmetric}{O{d}}{#1} 
\NewDocumentCommand{\vaoutnormal}{O{\widehat}}{#1} 
\newcommand{\category}[1]{\mathop{\mathsf{#1}}} 

\newcommand{\catpresheaf}[1]{\category{PSh}} 
\newcommand{\faadjrbk}[1]{\rbk{#1}^{\ast}} 
\newcommand{\faadj}[1]{#1^{\ast}} 
\newcommand{\fldcmp}{\fld{C}} 
\newcommand{\fldrealnn}{\fld{R}_{\geq 0}} 
\newcommand{\fldreal}{\fld{R}} 
\newcommand{\fld}[1]{\mathbb{#1}} 
\newcommand{\fndef}[1]{\boldsymbol{1}_{#1}} 
\newcommand{\fnrestr}[2]{\left. #1 \right|_{#2}} 
\newcommand{\gtclos}[1]{\overline{#1}} 
\newcommand{\larank}{\operatorname{rank}} 
\newcommand{\lpseq}{\ell} 
\newcommand{\lp}{L} 
\newcommand{\mblfmlfrak}[1]{\mathfrak{#1}} 
\newcommand{\monnat}{\mathbb{N}} 
\newcommand{\msrcal}[1]{\mathcal{#1}} 
\newcommand{\msrprb}{\mathrm{Pr}} 

\newcommand{\oacenter}{Z} 
\newcommand{\oacommutant}[1]{#1^{\prime}} 
\newcommand{\oacstar}{C^{\ast}} 
\newcommand{\oadoublecommutant}[1]{#1^{\prime \prime}} 
\newcommand{\oaresolventalgebra}{\oa{R}} 
\newcommand{\oaresolvent}{R} 
\newcommand{\oa}[1]{\mathcal{#1}} 
\newcommand{\opdmsr}[1]{\mathop{d #1}} 
\newcommand{\opspbddlin}[1]{\fun{\mathbb{B}}{#1}} 
\newcommand{\opspecint}[1]{\mathcal{E}} 
\newcommand{\physmean}[1]{\dbk{#1}} 
\newcommand{\physpoissonbkt}[1]{\cbk{#1}} 
\newcommand{\prbexp}{\mathbb{E}} 
\newcommand{\ringratint}{\mathbb{Z}} 
\newcommand{\semigrposint}{\monnat_1} 

\newcommand{\seq}[2]{\if\relax\detokenize{#1}\relax \rbk{#1} \else \rbk{#1}_{#2} \fi} 
\newcommand{\setisomorphism}[1]{\operatorname{Iso}} 
\newcommand{\setone}[1]{\cbk{#1}}
\newcommand{\setpowfin}[1]{\fun{\mathrm{Pow}_{\txtfin}}{#1}} 
\newcommand{\setquot}[2]{\bigmiddleslash{#1}{#2}} 
\newcommand{\set}[2]{\left\{#1 \, \middle| \, #2\right\}}
\newcommand{\spclosedlinspan}{\mathop{\clos{\splinspan}}} 
\newcommand{\spfockgamma}{\Gamma} 
\newcommand{\splinspan}{\operatorname{span}} 
\newcommand{\topdist}{\operatorname{dist}} 
\newcommand{\txtI}{\mathrm{I}} 
\newcommand{\txtbsn}{\mathrm{b}} 
\newcommand{\txtclassical}{\mathrm{cl}} 
\newcommand{\txtfin}{\mathrm{fin}} 
\newcommand{\txtfock}{\mathrm{F}} 
\newcommand{\txtkms}{\mathrm{KMS}} 
\newcommand{\txtloc}{\mathrm{loc}} 
\newcommand{\txtneg}{\mathrm{-}} 
\newcommand{\txtnonneg}{\mathrm{+}} 
\newcommand{\txtsym}{\mathrm{s}} 
\newcommand{\txtvac}{\mathrm{vac}} 

\title{The Buchholz Algebra from the Universal Resolvent Algebra\\\vspace{0.5em}{\large Particle Structure, Gibbs States, and Closed-Path Expansions without Fock Representations}}

\author{%
Yoshitsugu Sekine\\{\small\texttt{4429sekine@gmail.com}}%
}

\date{\today}

\begin{document}

\maketitle

\begin{abstract}
We construct the Buchholz algebra without a Fock representation, as the bounded inverse limit of intrinsic particle-cutoff quotients of the gauge-fixed resolvent algebra, with isometric particle-degree coordinates. Its finite matrix corners carry canonical traces, KMS states, and trace-norm product formulas, together with closed-path expansions that lay the groundwork for path-integral representations. The known Fock-space description and a proper quasilocal subalgebra arise as realizations.

\end{abstract}

\setcounter{tocdepth}{3}
\tableofcontents

\section{Introduction}\label{introduction}

The universal resolvent algebra separates the algebra of observables from any preferred Hilbert-space representation \cite[Definitions 3.1 and 3.4]{BuchholzGrundling002}. The enlarged algebra used for interacting bosons, however, is commonly described after the bosonic Fock representation has been selected; see, for example, the concrete Buchholz algebra in \cite[Definition 2.2]{DeuchertLampartLemm001}. This leaves a structural question: which properties of that concrete algebra can be extracted as representation-independent data and used to define the algebra before a representation is chosen?

The concrete Fock-space structure needed below is already known. The exact-particle algebra is identified with the finite sum of compact-on-\(m\)-particles blocks in \cite[Lemma 3.3]{DetlevBuchholz002}, the connecting map obtained by removing one identity factor is computed in \cite[Lemma 3.4]{DetlevBuchholz002}, and the bounded inverse-limit description is \cite[Theorem 3.5]{DetlevBuchholz002}. For particles on a countable graph, the corresponding statements are \cite[Lemma 2.3 and Proposition 2.4]{DeuchertLampartLemm001}; the failure of uniform quasilocal approximation is exhibited in \cite[Example 2.5]{DeuchertLampartLemm001}. These facts are not reproved here. Later work uses the same Fock-space algebra to study sectors and dynamics \cite{DetlevBuchholz003} and proper condensates \cite{DetlevBuchholz004,DetlevBuchholz005}. Our question is instead whether the cutoff system, its particle-degree coordinates, and its spatial embeddings can be recovered from representation-independent data in the universal resolvent algebra.

We answer this question by starting only with a separable infinite-dimensional complex Hilbert space \(\sphilbfrak{h}\) and the symplectic form given by the imaginary part of its inner product. Finite-dimensional resolvent subalgebras contain canonical minimal nonzero ideals. The particle-block condition is extracted from the properties verified in Proposition \ref{prop:paper-particle-block-independence}. It determines particle-cutoff ideals and quotients, whose bounded inverse limit is defined to be the abstract Buchholz algebra. The compression, coefficient-embedding, and overlap-product conditions, which are motivated by the finite-difference identity \eqref{eq:paper-newton-name} and extracted from the standard Fock representation, provide an additional coefficient realization: the particle-coefficient maps identify this inverse limit isometrically with the bounded particle-coefficient algebra. No Fock space or symmetric tensor power occurs in the resulting definition. The justification for adopting these conditions as the appropriate abstract conditions is that Proposition \ref{prop:paper-particle-block-independence}, Proposition \ref{prop:paper-coefficient-connecting-maps}, Theorem \ref{thm:paper-abstract-concrete-comparison}, and Theorem \ref{thm:paper-concrete-particle-coefficient-reconstruction-data} identify them with the established Fock system and supply the additional coefficient structure not contained in the inverse-limit description alone.

Path integrals are the method that we intend to use for formulating automorphism groups and equilibrium states of interacting bosons on this algebra. In the Fock representation, a discrete functional integral has already been used effectively for the Bose--Hubbard model: it gives volume-uniform local-number bounds from which equilibrium states on the concrete Buchholz algebra are obtained \cite{YoshitsuguSekine015}. Whether this mechanism can itself be generalized or abstracted is therefore an important task, and Section \ref{sec:intrinsic-finite-volume} is devoted to it. That section derives traces, compact heat kernels, trace-norm product formulas, and closed-path sums from finite-dimensional full matrix corners without a Fock representation and without imposing a graph. For the heat-kernel statements, polynomial corner growth and Hamiltonian superstability, a sector-uniform lower bound by a quadratic polynomial in the particle number with positive leading coefficient, give trace summability. For the split product formulas, the general analytic input is summability of the kinetic heat traces with the Gaussian particle-number weight supplied by potential superstability. This gives a majorant uniform over the approximating products even when the kinetic operators are unbounded. Polynomial corner growth, potential superstability, and a linear kinetic norm bound provide one finite-dimensional lattice criterion for this summability. Closed-path expansions require a separate absolute path-weight majorant when they are summed over all particle sectors. The arguments require these summability conclusions themselves. Superstability is a familiar sufficient condition, not a necessary one.

The comparison with the established bosonic Fock realization of the concrete Buchholz algebra takes the known exact-particle and inverse-limit formulas as input. Theorem \ref{thm:paper-abstract-concrete-comparison} identifies the abstract particle-cutoff quotient \(\oa{B}_n(\sphilbfrak{h})\) with the represented image of the gauge-fixed resolvent algebra specified in that theorem. Theorem \ref{thm:paper-concrete-particle-coefficient-reconstruction-data} identifies the degree-\(m\) particle coefficient with the algebra of compact operators on \(\sphilbfrak{h}^{\otimes_{\txtsym}m}\).

Spatial localization is additional structure rather than part of the representation-free construction. Definition \ref{def:paper-main-spatial-localization-data} formulates it as compatible injective embeddings of finite-mode block algebras into the full inverse limit, indexed by the finite subsets of a countable set \(\Gamma\). For \(\sphilbfrak{h}=\fun{\lpseq^{2}}{\Gamma}\), Proposition \ref{prop:concrete-spatial-localization-system} constructs such embeddings within the Fock realization and verifies their compatibility, their commutation for disjoint regions, and their density at every fixed particle cutoff. The construction uses only the standard orthonormal basis indexed by \(\Gamma\). Taking \(\Gamma\) to be the vertex set of a lattice graph prepares the application to lattice boson models, in particular to the Bose--Hubbard model treated on the concrete Buchholz algebra in \cite{YoshitsuguSekine015}. The norm-closed union of the local algebras is a quasilocal subalgebra of the full inverse limit. Theorem \ref{thm:paper-rank-one-tail-geometry} in the Appendix shows that this subalgebra is proper: for a unit vector \(f\) with infinite coordinate support, the element realized by the one-mode particle-number resolvent has distance \(1/2\) from the quasilocal subalgebra, although each of its fixed particle-cutoff coordinates is a norm limit of coordinates of local elements.

The paper is organized according to this logical boundary. Section \ref{sec:main-results} defines the objects needed to state the main results, including the additional index set used for spatial localization, and records their mathematical content before the technical constructions begin. Section \ref{sec:general-separable-theory} constructs the abstract algebra and its particle-coefficient coordinates for a general separable Hilbert space. Section \ref{sec:intrinsic-finite-volume} separates norm-level block results, matrix identities, bounded-functional evaluations, and trace-dependent equilibrium formulas before deriving closed-path sums from finite-dimensional corners. Sections \ref{sec:paper-bosonic-fock-representations}, \ref{sec:concrete-particle-coefficient-realization}, and \ref{sec:concrete-spatial-particle-cutoff-representations} then compare the abstract data with the bosonic Fock realization, taking the known exact-particle and inverse-limit formulas as input. The principal concrete identification is summarized in Theorem \ref{thm:paper-main-fock-realization}. The appendices collect the symmetry and Poisson refinements, Busby theory, the nonlocal part of the full abstract Buchholz algebra, and the remaining particle-cutoff representation theory.

\section{Main Results}\label{sec:main-results}

In this section, the principal results separate the representation-independent algebra, the intrinsic finite-mode analysis, spatial localization, and the concrete realizations in Sections \ref{sec:paper-bosonic-fock-representations}, \ref{sec:concrete-particle-coefficient-realization}, and \ref{sec:concrete-spatial-particle-cutoff-representations}. Only the data needed to state these results are defined here. Their complete axioms and constructions are given in the referenced definitions and theorems.

\subsection{Basic objects and notation}\label{sec:paper-main-basic-objects}

In this subsection, we introduce the Hilbert-space, resolvent-algebra, particle-block, quotient, and localization notation used to state the main results.

Let \(\sphilbfrak{h}\) be a separable infinite-dimensional complex Hilbert space with inner product linear in the second variable. For a self-adjoint operator \(A\) on a complex Hilbert space, write \(\msrprb_{A}\) for its projection-valued spectral measure on the Borel subsets of \(\fldreal\). Thus \(\fun{\msrprb_{A}}{E}\) is the spectral projection associated with a Borel set \(E\subset\fldreal\); the notation \(\fndef{E}\) denotes an ordinary indicator function of a set. The underlying real space carries the symplectic form \begin{equation}
\begin{aligned}
\sigma
\colon
\sphilbfrak{h}\times\sphilbfrak{h}
\to
\fldreal;\quad
\fun{\sigma}{f,g}
&=
\opimag\bkt{f}{g}_{\sphilbfrak{h}}.
\end{aligned}
\label{eq:abstract-symplectic-form}
\end{equation} Conjugate symmetry and complex linearity give \[\begin{aligned}
\fun{\sigma}{g,f}
&=
\opimag\bkt{g}{f}_{\sphilbfrak{h}}
=
\opimag\cmpconj{\bkt{f}{g}_{\sphilbfrak{h}}}
=
-\fun{\sigma}{f,g},
\\
\fun{\sigma}{f,\imunit f}
&=
\opimag\bkt{f}{\imunit f}_{\sphilbfrak{h}}
=
\norm{f}_{\sphilbfrak{h}}^{2}
>
0
&&\rbk{f\neq0}.
\end{aligned}\] Thus \(\sigma\) is a nondegenerate real bilinear antisymmetric form. For \(\lambda,\mu,\nu\in\fldreal\setminus\setone{0}\) and \(f,g\in\sphilbfrak{h}\), let the symbols \(\oaresolvent(\lambda,f)\) satisfy \begin{equation}
\begin{aligned}
\oaresolvent(\lambda,0)
&=
-\frac{\imunit}{\lambda}1,
\\
\faadj{\oaresolvent(\lambda,f)}
&=
\oaresolvent(-\lambda,f),
\\
\nu\oaresolvent(\nu\lambda,\nu f)
&=
\oaresolvent(\lambda,f),
\\
\oaresolvent(\lambda,f)-\oaresolvent(\mu,f)
&=
\imunit(\mu-\lambda)\oaresolvent(\lambda,f)\oaresolvent(\mu,f),
\\
\commutator{\oaresolvent(\lambda,f)}{\oaresolvent(\mu,g)}
&=
\imunit
\fun{\sigma}{f,g}
\oaresolvent(\lambda,f)
\oaresolvent(\mu,g)^2
\oaresolvent(\lambda,f).
\end{aligned}
\label{eq:paper-main-resolvent-relations}
\end{equation} For \(\lambda+\mu\neq0\), the additional relation is \begin{equation}
\begin{aligned}
\oaresolvent(\lambda,f)\oaresolvent(\mu,g)
&=
\oaresolvent(\lambda+\mu,f+g)
\rbk{
\oaresolvent(\lambda,f)
+
\oaresolvent(\mu,g)
+
\imunit\fun{\sigma}{f,g}
\oaresolvent(\lambda,f)^2
\oaresolvent(\mu,g)}.
\end{aligned}
\label{eq:paper-main-resolvent-product-relation}
\end{equation} The universal resolvent algebra \(\oaresolventalgebra(\sphilbfrak{h},\sigma)\) is the completion of the quotient of this universal unital \(\ast\)-algebra by the kernel of its enveloping \(\oacstar\)-seminorm \cite[Definitions 3.1 and 3.4]{BuchholzGrundling002}. For a nonzero complex subspace \(\sphilbfrak{k}\subset\sphilbfrak{h}\), the restriction of \(\sigma\) to \(\sphilbfrak{k}\times\sphilbfrak{k}\) is nondegenerate because \(\imunit\sphilbfrak{k}=\sphilbfrak{k}\) and \(\fun{\sigma}{f,\imunit f}=\norm{f}_{\sphilbfrak{h}}^{2}\). Write \(\oaresolventalgebra(\sphilbfrak{k})\) for the universal resolvent algebra of \(\sphilbfrak{k}\) with this restricted form; the symplectic form is omitted from the notation whenever it is the restriction of \eqref{eq:abstract-symplectic-form}. For \(\theta\in\fldreal\), multiplication by \(\napiernum^{\imunit\theta}\) preserves \(\sigma\) and is therefore a symplectic transformation. Its gauge action on the generators is \begin{equation}
\begin{aligned}
\gamma
\colon
\fldreal
\curvearrowright
\oaresolventalgebra(\sphilbfrak{h},\sigma);\quad
\fun{\gamma_{\theta}}{\oaresolvent(\lambda,f)}
&=
\oaresolvent(\lambda,\napiernum^{\imunit\theta}f).
\end{aligned}
\label{eq:paper-main-gauge-action}
\end{equation} The fixed-point algebra of this action is \begin{equation}
\begin{aligned}
\oaresolventalgebra(\sphilbfrak{h},\sigma)^{\gamma}
&=
\set{A\in\oaresolventalgebra(\sphilbfrak{h},\sigma)}
{\fun{\gamma_\theta}{A}=A\text{ for every }\theta\in\fldreal}.
\end{aligned}
\label{eq:paper-main-gauge-fixed-algebra}
\end{equation} The generator formula holds for \(\theta\in\fldreal\), \(\lambda\in\fldreal\setminus\setone{0}\), and \(f\in\sphilbfrak{h}\).

For every nonzero finite-dimensional complex subspace \(\sphilbfrak{f}\subset\sphilbfrak{h}\), define the intrinsic minimal ideal by \begin{equation}
\begin{aligned}
\oa{K}(\sphilbfrak{f})
&=
\bigcap
\set{\oaideal{I}\subset\oaresolventalgebra(\sphilbfrak{f})}
{\oaideal{I}\text{ is a nonzero closed two-sided ideal}}.
\end{aligned}
\label{eq:paper-main-intrinsic-minimal-ideal}
\end{equation} Its gauge-fixed part is \begin{equation}
\begin{aligned}
\oa{K}(\sphilbfrak{f})^{\gamma}
&=
\set{K\in\oa{K}(\sphilbfrak{f})}
{\fun{\gamma_{\theta}}{K}=K\text{ for every }\theta\in\fldreal}.
\end{aligned}
\label{eq:paper-main-gauge-fixed-minimal-ideal}
\end{equation} A finite-dimensional full matrix algebra means a unital \(\oacstar\)-algebra \(\oa{A}\) for which there are \(d\in\semigrposint\) and a unital \(\ast\)-isomorphism \(\oa{A}\to\fun{M_d}{\fldcmp}\). The gauge-fixed minimal ideals will be decomposed into ordered families of full matrix corners.

\begin{defn}[Intrinsic particle-block system]\label{def:intrinsic-particle-block-system}
An intrinsic particle-block system consists of projections
$e_{\sphilbfrak{f},m}\in\oa{K}(\sphilbfrak{f})^{\gamma}$, indexed by nonzero finite-dimensional complex subspaces
$\sphilbfrak{f}\subset\sphilbfrak{h}$ and $m\in\monnat$, such that distinct projections for a fixed $\sphilbfrak{f}$ are orthogonal and
every corner $e_{\sphilbfrak{f},m}\oa{K}(\sphilbfrak{f})^{\gamma}e_{\sphilbfrak{f},m}$ is a nonzero finite-dimensional full matrix algebra.
Every $K\in\oa{K}(\sphilbfrak{f})^{\gamma}$ is required to have a unique norm-convergent decomposition
\begin{equation}
\begin{aligned}
K
&=
\sum_{m\in\monnat}K_m,
&
K_m
&\in
e_{\sphilbfrak{f},m}\oa{K}(\sphilbfrak{f})^{\gamma}e_{\sphilbfrak{f},m},
&
\lim_{m\to\infty}\norm{K_m}
&=
0.
\end{aligned}
\label{eq:paper-main-intrinsic-particle-block-decomposition}
\end{equation}
\end{defn}

The order on \(m\in\monnat\) selects the initial segment \(0\leq m\leq n\) in the cutoff projection \eqref{eq:abstract-high-particle-projection}; the same index selects the corner \(e_{\sphilbfrak{f},m}\oa{K}(\sphilbfrak{f})^{\gamma}e_{\sphilbfrak{f},m}\) and the inductive limit \(\oa{C}_m(\sphilbfrak{h})\) in \eqref{eq:paper-main-degree-coefficient-algebra}. Proposition \ref{prop:paper-particle-block-independence} shows that, in every faithful regular irreducible representation \(\pi_{\sphilbfrak{f}}\), the projection \(e_{\sphilbfrak{f},m}\) is represented by \(\fun{\msrprb_{N_{\sphilbfrak{f}}}}{\setone{m}}\), where \(N_{\sphilbfrak{f}}\) is the self-adjoint generator of the normalized gauge-implementing unitary group. Under the Stone--von Neumann identification with the regular Fock representation, this spectral subspace is the exact \(m\)-particle subspace, as recorded in Theorem \ref{thm:paper-main-fock-realization}. We supplement this system with compression maps for inclusions of finite-dimensional subspaces. For \(\sphilbfrak{f}\subset \sphilbfrak{g}\), the canonical inclusion is \begin{equation}
\begin{aligned}
\iota_{\sphilbfrak{f},\sphilbfrak{g}}
\colon
\oaresolventalgebra(\sphilbfrak{f})
\to
\oaresolventalgebra(\sphilbfrak{g});\quad
\fun{\iota_{\sphilbfrak{f},\sphilbfrak{g}}}{\oaresolvent(\lambda,f)}
&=
\oaresolvent(\lambda,f).
\end{aligned}
\label{eq:paper-canonical-resolvent-inclusion}
\end{equation} The inclusion into the resolvent algebra over \(\sphilbfrak{h}\) is \begin{equation}
\begin{aligned}
\iota_{\sphilbfrak{f}}
\colon
\oaresolventalgebra(\sphilbfrak{f})
\to
\oaresolventalgebra(\sphilbfrak{h},\sigma);\quad
\fun{\iota_{\sphilbfrak{f}}}{\oaresolvent(\lambda,f)}
&=
\oaresolvent(\lambda,f).
\end{aligned}
\label{eq:paper-global-resolvent-inclusion}
\end{equation} The compressed corner map is \begin{equation}
\begin{aligned}
j_{\sphilbfrak{f},\sphilbfrak{g},m}
\colon
e_{\sphilbfrak{f},m}\oa{K}(\sphilbfrak{f})^{\gamma}e_{\sphilbfrak{f},m}
\to
e_{\sphilbfrak{g},m}\oa{K}(\sphilbfrak{g})^{\gamma}e_{\sphilbfrak{g},m};\quad
\fun{j_{\sphilbfrak{f},\sphilbfrak{g},m}}{K}
&=
e_{\sphilbfrak{g},m}\fun{\iota_{\sphilbfrak{f},\sphilbfrak{g}}}{K}e_{\sphilbfrak{g},m}.
\end{aligned}
\label{eq:paper-coefficient-connecting-map}
\end{equation} We assume that every \(j_{\sphilbfrak{f},\sphilbfrak{g},m}\) is an injective \(\ast\)-homomorphism and that \begin{equation}
\begin{aligned}
j_{\sphilbfrak{g},\sphilbfrak{k},m}\circ j_{\sphilbfrak{f},\sphilbfrak{g},m}
&=
j_{\sphilbfrak{f},\sphilbfrak{k},m}
\end{aligned}
\label{eq:paper-coefficient-connecting-compatibility}
\end{equation} for all \(m\in\monnat\) and all nonzero finite-dimensional complex subspaces \(\sphilbfrak{f}\subset \sphilbfrak{g}\subset \sphilbfrak{k}\) of \(\sphilbfrak{h}\).

For \(n\in\monnat\), the projection onto degrees greater than \(n\) is \begin{equation}
\begin{aligned}
Q_{\sphilbfrak{f},>n}
&=
1-
\sum_{m=0}^{n}e_{\sphilbfrak{f},m}.
\end{aligned}
\label{eq:abstract-high-particle-projection}
\end{equation} The corresponding ideal of the global gauge-fixed resolvent algebra is \begin{equation}
\begin{aligned}
\oaideal{I}_{>n}(\sphilbfrak{h})
&=
\gtclos{
\sum_{\substack{\sphilbfrak{f}\subset\sphilbfrak{h}\\0<\dim_{\fldcmp}\sphilbfrak{f}<\infty}}
\oaresolventalgebra(\sphilbfrak{h},\sigma)^{\gamma}
\fun{\iota_{\sphilbfrak{f}}}{Q_{\sphilbfrak{f},>n}}
\oaresolventalgebra(\sphilbfrak{h},\sigma)^{\gamma}
}^{\norm{\mathord{\cdot}}}.
\end{aligned}
\label{eq:abstract-high-particle-ideal}
\end{equation} The particle-cutoff quotient is

\begin{equation}
\begin{aligned}
\oa{B}_n(\sphilbfrak{h})
&=
\setquot{\oaresolventalgebra(\sphilbfrak{h},\sigma)^{\gamma}}
{\oaideal{I}_{>n}(\sphilbfrak{h})}.
\end{aligned}
\label{eq:paper-main-particle-cutoff-quotient}
\end{equation} For \(n\geq1\), the quotient map between consecutive cutoffs is \begin{equation}
\begin{aligned}
r_n
\colon
\oa{B}_n(\sphilbfrak{h})
\to
\oa{B}_{n-1}(\sphilbfrak{h});\quad
\fun{r_n}{A+\oaideal{I}_{>n}(\sphilbfrak{h})}
&=
A+\oaideal{I}_{>n-1}(\sphilbfrak{h}).
\end{aligned}
\label{eq:paper-main-cutoff-connecting-map}
\end{equation} The bounded inverse limit is \begin{equation}
\begin{aligned}
\fun{\oa{B}}{\sphilbfrak{h}}
&=
\set{\seq{A_n}{n\in\monnat}}
{A_n\in\oa{B}_n(\sphilbfrak{h})\rbk{n\in\monnat},
\fun{r_n}{A_n}=A_{n-1}\rbk{n\geq1},
\sup_{n\in\monnat}\norm{A_n}<\infty}.
\end{aligned}
\label{eq:paper-main-cutoff-inverse-limit}
\end{equation} The norm on \(\fun{\oa{B}}{\sphilbfrak{h}}\) is the supremum of the coordinate norms. For \(n\in\monnat\), define the coordinate map by \begin{equation}
\begin{aligned}
p_n
\colon
\fun{\oa{B}}{\sphilbfrak{h}}
\to
\oa{B}_n(\sphilbfrak{h});\quad
\fun{p_n}{\seq{A_m}{m\in\monnat}}
&=
A_n.
\end{aligned}
\label{eq:paper-main-coordinate-map}
\end{equation}

\begin{defn}[Essential ideal and essential extension]\label{def:paper-essential-ideal-extension}
Let $\oa{E}$ be a $\oacstar$-algebra.
A closed two-sided ideal $\oaideal{I}\subset\oa{E}$ is essential if
$\oaideal{I}\cap\oa{L}\neq0$ for every nonzero closed two-sided ideal
$\oa{L}\subset\oa{E}$.
Let $\oa{J}$ and $\oa{A}$ be $\oacstar$-algebras, and let
\begin{equation}
\begin{aligned}
0
\to
\oa{J}
\xrightarrow{\iota}
\oa{E}
\xrightarrow{p}
\oa{A}
\to
0
\end{aligned}
\label{eq:paper-main-essential-extension}
\end{equation}
be a sequence of $\ast$-homomorphisms in which $\iota$ is injective, $p$ is surjective, and
$\fun{\iota}{\oa{J}}=\Ker p$.
A sequence satisfying these three conditions is called short exact.
The sequence in \eqref{eq:paper-main-essential-extension} is an essential extension of $\oa{A}$ by $\oa{J}$
if the closed two-sided ideal $\fun{\iota}{\oa{J}}\subset\oa{E}$ is essential.
\end{defn}

\begin{defn}[Simple $\oacstar$-algebra and minimal projection]\label{def:paper-simple-algebra-minimal-projection}
A nonzero $\oacstar$-algebra $\oa{A}$ is simple if its only closed two-sided ideals are $0$ and $\oa{A}$.
For a $\oacstar$-algebra $\oa{A}$, a nonzero projection $p\in\oa{A}$ is minimal if
$$\begin{aligned}
p\oa{A}p
&=
\fldcmp p.
\end{aligned}$$
\end{defn}

The scalar model for the coefficient reconstruction is the Newton forward-difference expansion. Indeed, for a scalar sequence \(s\colon\monnat\to\fldcmp\), define \[\begin{aligned}
\fun{\Delta s}{k}
&=
\fun{s}{k+1}-\fun{s}{k}.
\end{aligned}\] Repeated binomial inversion gives \begin{equation}
\begin{aligned}
\fun{s}{n}
&=
\sum_{m=0}^{n}\binom{n}{m}\fun{\Delta^{m}s}{0},
&
\fun{\Delta^{m}s}{0}
&=
\sum_{j=0}^{m}(-1)^{m-j}\binom{m}{j}\fun{s}{j}.
\end{aligned}
\label{eq:paper-newton-name}
\end{equation} Throughout this paper, \(\sphilbfrak{f}\nearrow\sphilbfrak{h}\) means that \(\sphilbfrak{f}\) runs through the directed set of all nonzero finite-dimensional subspaces of \(\sphilbfrak{h}\), ordered by inclusion.

For fixed \(m\in\monnat\), the symbol \(\varinjlim\) in \eqref{eq:paper-main-degree-coefficient-algebra} denotes a \(\oacstar\)-algebra with compatible injective \(\ast\)-homomorphisms from the corner algebras in \eqref{eq:paper-main-degree-coefficient-algebra}. Their ranges have dense union, and every compatible family of \(\ast\)-homomorphisms from these corner algebras factors uniquely. For every \(m\in\monnat\), define the degree-\(m\) coefficient algebra by \begin{equation}
\begin{aligned}
\oa{C}_m(\sphilbfrak{h})
&=
\varinjlim_{\sphilbfrak{f}\nearrow\sphilbfrak{h}}
e_{\sphilbfrak{f},m}\oa{K}(\sphilbfrak{f})^{\gamma}e_{\sphilbfrak{f},m}.
\end{aligned}
\label{eq:paper-main-degree-coefficient-algebra}
\end{equation} For every nonzero finite-dimensional complex subspace \(\sphilbfrak{f}\subset\sphilbfrak{h}\) and every \(m\in\monnat\), let \begin{equation}
\begin{aligned}
j_{\sphilbfrak{f},m}
\colon
e_{\sphilbfrak{f},m}\oa{K}(\sphilbfrak{f})^{\gamma}e_{\sphilbfrak{f},m}
&\to
\oa{C}_m(\sphilbfrak{h}),
\end{aligned}
\label{eq:paper-main-degree-coefficient-canonical-injection}
\end{equation} be the canonical injection induced by the connecting maps in \eqref{eq:paper-coefficient-connecting-map}.

\begin{rem}[Simplicity of the degree coefficient algebras]\label{rem:paper-degree-coefficient-algebra-simplicity}
For every $m\in\monnat$, the corners in the inductive system
\eqref{eq:paper-main-degree-coefficient-algebra} are nonzero finite-dimensional full matrix algebras by
Definition \ref{def:intrinsic-particle-block-system} and are therefore simple.
The connecting maps in \eqref{eq:paper-coefficient-connecting-map} are injective.
The standard permanence of simplicity under inductive limits of simple $\oacstar$-algebras with injective
$\ast$-homomorphisms therefore shows that $\oa{C}_m(\sphilbfrak{h})$ is simple.
\end{rem}

\begin{defn}[Particle-coefficient embedding data]\label{def:paper-particle-coefficient-embedding-data}
The degree-zero coefficient algebra is required to satisfy
$\oa{C}_0(\sphilbfrak{h})=\fldcmp$.
Particle-coefficient embedding data consist of bounded linear maps
$$\begin{aligned}
E_{m,n}
\colon
\oa{C}_m(\sphilbfrak{h})
\to
\oa{B}_n(\sphilbfrak{h})
\end{aligned}$$
for $n\in\monnat$ and $0\leq m\leq n$.
For $n\in\monnat$, $0\leq m\leq n$, $a_m\in\oa{C}_m(\sphilbfrak{h})$, and
$\lambda\in\fldcmp$, the maps preserve adjoints and satisfy
\begin{equation}
\begin{aligned}
\oa{B}_n(\sphilbfrak{h})
&=
\bigoplus_{m=0}^{n}\fun{E_{m,n}}{\oa{C}_m(\sphilbfrak{h})},
\\ 
\fun{E_{m,n}}{\faadj{a_m}}
&=
\faadjrbk{\fun{E_{m,n}}{a_m}}
\quad\rbk{0\leq m\leq n},
\\ 
\fun{E_{0,n}}{\lambda}
&=
\lambda.
\end{aligned}
\label{eq:paper-abstract-particle-coefficient-embedding-identities}
\end{equation}
For $n\geq1$, the coefficient maps and the connecting maps satisfy
\begin{equation}
\begin{aligned}
r_n\circ E_{m,n}
&=
E_{m,n-1}
\quad\rbk{0\leq m<n},
\\ 
r_n\circ E_{n,n}
&=
0.
\end{aligned}
\label{eq:paper-abstract-particle-coefficient-connecting-identities}
\end{equation}
For $n\geq1$, the top-degree image
$$\begin{aligned}
\fun{E_{n,n}}{\oa{C}_n(\sphilbfrak{h})}
&=
\Ker r_n
\end{aligned}$$
is an essential ideal of $\oa{B}_n(\sphilbfrak{h})$.
\end{defn}

Injectivity of the coefficient maps is not imposed in Definition \ref{def:paper-particle-coefficient-embedding-data}. Proposition \ref{prop:paper-particle-coefficient-map-injectivity} derives injectivity from the reconstruction conditions in Definition \ref{def:paper-abstract-particle-coefficient-reconstruction-data}.

\begin{defn}[Abstract particle-coefficient reconstruction data]\label{def:paper-abstract-particle-coefficient-reconstruction-data}
Abstract particle-coefficient reconstruction data consist of the embedding data in Definition
\ref{def:paper-particle-coefficient-embedding-data} together with bounded bilinear maps
$$\begin{aligned}
\mu_{p,q}^{r}
\colon
\oa{C}_p(\sphilbfrak{h})\times\oa{C}_q(\sphilbfrak{h})
\to
\oa{C}_r(\sphilbfrak{h})
\end{aligned}$$
for $p,q\in\monnat$ and the output degrees
\begin{equation}
\max\setone{p,q}\leq r\leq p+q.
\label{eq:paper-abstract-particle-coefficient-output-degree-range}
\end{equation}
For $p,q,m\in\monnat$, $n\geq\max\setone{p,q}$,
$a_p\in\oa{C}_p(\sphilbfrak{h})$, $b_q\in\oa{C}_q(\sphilbfrak{h})$, and
$a_m,b_m\in\oa{C}_m(\sphilbfrak{h})$, the overlap products satisfy
\begin{equation}
\begin{aligned}
\fun{E_{p,n}}{a_p}\fun{E_{q,n}}{b_q}
&=
\sum_{r=\max\setone{p,q}}^{\min\setone{p+q,n}}
\fun{E_{r,n}}{\fun{\mu_{p,q}^{r}}{a_p,b_q}},
\\ 
\fun{\mu_{p,q}^{p+q}}{a_p,b_q}
&=
\fun{\mu_{q,p}^{p+q}}{b_q,a_p},
\\ 
\fun{\mu_{m,m}^{m}}{a_m,b_m}
&=
a_mb_m.
\end{aligned}
\label{eq:paper-abstract-particle-coefficient-overlap-products}
\end{equation}
\end{defn}

Write \begin{equation}
\begin{aligned}
\oa{C}(\sphilbfrak{h})
&=
\prod_{m\in\monnat}\oa{C}_m(\sphilbfrak{h}).
\end{aligned}
\label{eq:paper-main-particle-coefficient-space}
\end{equation} For \(a=\seq{a_m}{m\in\monnat}\in\oa{C}(\sphilbfrak{h})\), its reconstruction at cutoff \(n\) is \begin{equation}
\begin{aligned}
\Pi_n
\colon
\oa{C}(\sphilbfrak{h})
\to
\oa{B}_n(\sphilbfrak{h});\quad
\fun{\Pi_n}{a}
&=
\sum_{m=0}^{n}\fun{E_{m,n}}{a_m}.
\end{aligned}
\label{eq:paper-main-particle-coefficient-reconstruction}
\end{equation} For \(a,b\in\oa{C}(\sphilbfrak{h})\), the coefficient product and adjoint are \begin{equation}
\begin{aligned}
\star
\colon
\oa{C}(\sphilbfrak{h})\times\oa{C}(\sphilbfrak{h})
\to
\oa{C}(\sphilbfrak{h});\quad
\rbk{a\star b}_r
&=
\sum_{\substack{0\leq p,q\leq r\\p+q\geq r}}
\fun{\mu_{p,q}^{r}}{a_p,b_q},
\\
\faadj{\mathord{\cdot}}
\colon
\oa{C}(\sphilbfrak{h})
\to
\oa{C}(\sphilbfrak{h});\quad
\rbk{\faadj{a}}_r
&=
\faadj{a_r}
\end{aligned}
\label{eq:paper-main-particle-coefficient-operations}
\end{equation} for every \(r\in\monnat\). Thus the values \(\fun{s}{n}\) in \eqref{eq:paper-newton-name} are reconstructed triangularly from the uniquely determined coefficients \(\fun{\Delta^{m}s}{0}\). The coefficient embeddings \(E_{m,n}\) of Definition \ref{def:paper-particle-coefficient-embedding-data} and the maps \(\Pi_n\) in \eqref{eq:paper-main-particle-coefficient-reconstruction} are the operator-algebraic counterparts of multiplication by \(\binom{n}{m}\) and of the finite sum in \eqref{eq:paper-newton-name}, respectively. The index \(m\) is the particle degree, and \(\Pi_n\) reconstructs the cutoff quotient from the degree-\(m\) coefficients. We call \(a_m\) the degree-\(m\) particle coefficient. The embeddings \(E_{m,n}\) control linear reconstruction and compatibility with the particle cutoffs, while the maps \(\mu_{p,q}^{r}\) determine multiplication of the coefficients. The bounded particle-coefficient space is \begin{equation}
\begin{aligned}
\oa{C}_{\mathrm{b}}(\sphilbfrak{h})
&=
\set{a\in\oa{C}(\sphilbfrak{h})}
{\sup_{n\in\monnat}\norm{\fun{\Pi_n}{a}}<\infty}.
\end{aligned}
\label{eq:paper-main-bounded-particle-coefficients}
\end{equation} For \(n\in\monnat\) and \(0\leq k\leq n+1\), define the degree-tail subspaces by \begin{equation}
\begin{aligned}
\oa{J}_k^{(n)}
&=
\begin{cases}
\oa{B}_n(\sphilbfrak{h}),&k=0,\\
\Ker\rbk{r_k\circ r_{k+1}\circ\cdots\circ r_n},&1\leq k\leq n,\\
0,&k=n+1.
\end{cases}
\end{aligned}
\label{eq:paper-main-degree-tail-ideals}
\end{equation}

For the finite-mode results, fix a nonzero finite-dimensional complex subspace \(\sphilbfrak{k}\subset\sphilbfrak{h}\). Every \(A\in\oa{K}(\sphilbfrak{k})^{\gamma}\) has the unique norm-convergent decomposition \begin{equation}
\begin{aligned}
A
&=
\sum_{n\in\monnat}A_n,
&
A_n
&\in
e_{\sphilbfrak{k},n}\oa{K}(\sphilbfrak{k})^{\gamma}e_{\sphilbfrak{k},n},
&
\lim_{n\to\infty}\norm{A_n}
&=
0.
\end{aligned}
\label{eq:paper-main-finite-mode-block-algebra}
\end{equation} Every corner \(e_{\sphilbfrak{k},n}\oa{K}(\sphilbfrak{k})^{\gamma}e_{\sphilbfrak{k},n}\) is a nonzero finite-dimensional full matrix algebra. The standard trace on the \(n\)th corner is the unique linear functional \begin{equation}
\begin{aligned}
\trace_{n}
\colon
e_{\sphilbfrak{k},n}\oa{K}(\sphilbfrak{k})^{\gamma}e_{\sphilbfrak{k},n}
\to
\fldcmp;\quad
\sqfun{\trace_{n}}{p}
&=
1,
\end{aligned}
\label{eq:paper-standard-corner-trace}
\end{equation} such that the equality in \eqref{eq:paper-standard-corner-trace} holds for every minimal projection \(p\) in the \(n\)th corner. Thus \(\trace_{n}\) is the ordinary, unnormalized matrix trace; in particular, \(\sqfun{\trace_{n}}{e_{\sphilbfrak{k},n}}\) is the matrix size of the \(n\)th corner rather than \(1\). The extended positive trace is \begin{equation}
\begin{aligned}
\trace
\colon
\rbk{\oa{K}(\sphilbfrak{k})^{\gamma}}_{+}
\to
\closedinterval{0}{\infty};\quad
\sqfun{\trace}{A}
&=
\sum_{n\in\monnat}\sqfun{\trace_{n}}{A_n}
\end{aligned}
\label{eq:paper-intrinsic-standard-trace}
\end{equation} for every positive \(A=\sum_{n\in\monnat}A_n\), where the series is the supremum of its finite partial sums. A positive \(A\in\oa{K}(\sphilbfrak{k})^{\gamma}\) is trace summable when \(\sum_{n\in\monnat}\sqfun{\trace_{n}}{A_n}<\infty\). For \(T\in e_{\sphilbfrak{k},n}\oa{K}(\sphilbfrak{k})^{\gamma}e_{\sphilbfrak{k},n}\), the trace norm is \[\begin{aligned}
\norm{T}_{1,\trace_{n}}
&=
\sqfun{\trace_{n}}{\abs{T}}.
\end{aligned}\] A family \(\seq{A_{n}}{n\in\monnat}\) of self-adjoint elements \(A_{n}\in e_{\sphilbfrak{k},n}\oa{K}(\sphilbfrak{k})^{\gamma}e_{\sphilbfrak{k},n}\) is called superstable when there are constants \(a>0\) and \(b,c\geq0\) such that \begin{equation}
\begin{aligned}
A_{n}
&\geq
an^2-bn-c
&&\rbk{n\in\monnat}.
\end{aligned}
\label{eq:paper-main-superstability-bound}
\end{equation} In particular, this definition applies to the Hamiltonian family \(\seq{\physham_{n}}{n\in\monnat}\). The corners have polynomial dimension growth when there are constants \(C>0\) and \(s\geq0\) such that \begin{equation}
\begin{aligned}
\sqfun{\trace_{n}}{e_{\sphilbfrak{k},n}}
&\leq
C\rbk{1+n}^{s}
&&\rbk{n\in\monnat}.
\end{aligned}
\label{eq:paper-corner-dimension-growth}
\end{equation} For the split product formulas, let \(\seq{T_{n}}{n\in\monnat}\) and \(\seq{V_{n}}{n\in\monnat}\) be families of self-adjoint elements \(T_{n},V_{n}\in e_{\sphilbfrak{k},n}\oa{K}(\sphilbfrak{k})^{\gamma}e_{\sphilbfrak{k},n}\), called the kinetic and potential families. The potential family is superstable when there are constants \(a>0\) and \(b,c\geq0\) such that \begin{equation}
\begin{aligned}
V_{n}
&\geq
\rbk{an^2-bn-c}e_{\sphilbfrak{k},n}
&&\rbk{n\in\monnat}.
\end{aligned}
\label{eq:paper-main-corner-potential-superstability}
\end{equation} For these constants and \(\sminvtemperature>0\), the kinetic heat traces are Gaussian-summable when \begin{equation}
\begin{aligned}
\sum_{n\in\monnat}
\napiernum^{-\sminvtemperature\rbk{an^2-bn-c}}
\sqfun{\trace_{n}}{\napiernum^{-\sminvtemperature T_n}}
&<
\infty.
\end{aligned}
\label{eq:paper-main-kinetic-heat-trace-summability}
\end{equation} Every corner is a finite-dimensional full matrix algebra, so every kinetic heat kernel \(\napiernum^{-\sminvtemperature T_n}\) has finite standard trace. On infinite-dimensional sectors, the corresponding trace-class property is the kinetic Gibbs condition \eqref{eq:paper-kinetic-gibbs-condition}, which is an assumption. The kinetic family has a linear norm bound when there are constants \(b_T,c_T\geq0\) such that \begin{equation}
\begin{aligned}
\norm{T_{n}}
&\leq
b_Tn+c_T
&&\rbk{n\in\monnat}.
\end{aligned}
\label{eq:paper-main-linear-kinetic-bound}
\end{equation} Together with polynomial corner growth, Hamiltonian superstability is sufficient for heat-kernel trace summability. On a lattice, polynomial corner growth \eqref{eq:paper-corner-dimension-growth} and the linear kinetic norm bound \eqref{eq:paper-main-linear-kinetic-bound} imply \eqref{eq:paper-main-kinetic-heat-trace-summability} by Lemma \ref{lem:paper-summable-sector-majorant}; these conditions are sufficient and not necessary. A family \(\seq{\physham_{n}}{n\in\monnat}\) of self-adjoint elements \(\physham_{n}\in e_{\sphilbfrak{k},n}\oa{K}(\sphilbfrak{k})^{\gamma}e_{\sphilbfrak{k},n}\) has a compact heat kernel when, for every \(\sminvtemperature>0\), \begin{equation}
\begin{aligned}
\seq{\napiernum^{-\sminvtemperature \physham_{n}}}{n\in\monnat}
&\in
\oa{K}(\sphilbfrak{k})^{\gamma}.
\end{aligned}
\label{eq:paper-main-compact-heat-kernel}
\end{equation} Its compact heat-kernel family is trace summable at inverse temperature \(\sminvtemperature>0\) when \begin{equation}
\begin{aligned}
\sum_{n\in\monnat}
\sqfun{\trace_{n}}{\napiernum^{-\sminvtemperature \physham_{n}}}
&<
\infty.
\end{aligned}
\label{eq:paper-main-heat-kernel-trace-summability}
\end{equation}

For every nonzero finite-dimensional subspace \(\sphilbfrak{f}\subset\sphilbfrak{h}\), define the intrinsic finite-mode algebra \begin{equation}
\begin{aligned}
\oa{M}(\oa{K}(\sphilbfrak{f})^{\gamma})
&=
\set{\seq{K_m}{m\in\monnat}}
{K_m\in e_{\sphilbfrak{f},m}\oa{K}(\sphilbfrak{f})^{\gamma}e_{\sphilbfrak{f},m},
\sup_{m\in\monnat}\norm{K_m}<\infty},
\end{aligned}
\label{eq:paper-main-finite-mode-algebra}
\end{equation} with coordinatewise operations and the supremum norm.

Spatial support is additional structure and is not part of the data used to define the abstract inverse limit or the intrinsic finite-mode traces and closed-path sums. The abstract spatial theory uses only the inclusion maps and their stated algebraic properties. For the concrete bosonic system, Proposition \ref{prop:concrete-spatial-localization-system} constructs these data from the Fock realization.

For the spatial statements, let \(\Gamma\) be the countably infinite vertex set of a graph. The constructions below use only the vertex set \(\Gamma\).

Let \begin{equation}
\begin{aligned}
\setpowfin{\Gamma}
&=
\set{\Lambda\subset\Gamma}{\Lambda\text{ is finite}},
\end{aligned}
\label{eq:paper-main-finite-region-directed-set}
\end{equation} ordered by inclusion. We write \(\Lambda\nearrow\Gamma\) when \(\Lambda\) runs through the nonempty members of \(\setpowfin{\Gamma}\), directed by inclusion. For every \(\Lambda\subset\Gamma\), define the global and local one-particle spaces by \begin{equation}
\begin{aligned}
\sphilbfrak{h}
&=
\fun{\lpseq^{2}}{\Gamma},
&
\sphilbfrak{h}_{\Lambda}
&=
\fun{\lpseq^{2}}{\Lambda}
\subset
\sphilbfrak{h}
\end{aligned}
\label{eq:paper-main-local-one-particle-subspace}
\end{equation} by zero extension. Write \(\delta_x\) for the standard basis vector at \(x\in\Gamma\). For every \(\Lambda\in\setpowfin{\Gamma}\), define the finite-region one-particle projection by \begin{equation}
\begin{aligned}
P_{\Lambda}
\colon
\sphilbfrak{h}
\to
\sphilbfrak{h}_{\Lambda};\quad
\fun{P_{\Lambda}}{\xi}
&=
\sum_{x\in \Lambda}
\bkt{\delta_x}{\xi}_{\sphilbfrak{h}}\delta_x.
\end{aligned}
\label{eq:paper-finite-region-one-particle-projection}
\end{equation}

\begin{defn}[Spatial localization data]\label{def:paper-main-spatial-localization-data}
Let the finite regions and local one-particle subspaces be those defined in
\eqref{eq:paper-main-finite-region-directed-set}--\eqref{eq:paper-main-local-one-particle-subspace}.
For nonempty finite sets $\Lambda\subset \Lambda'\subset\Gamma$, spatial localization data consist of
injective unital $\ast$-homomorphisms
\begin{equation}
\begin{aligned}
i_{\Lambda,\Lambda'}
&\colon
\oa{M}(\oa{K}(\sphilbfrak{h}_{\Lambda})^{\gamma})
\to
\oa{M}(\oa{K}(\sphilbfrak{h}_{\Lambda'})^{\gamma}),
&
i_{\Lambda,\Gamma}
&\colon
\oa{M}(\oa{K}(\sphilbfrak{h}_{\Lambda})^{\gamma})
\to
\fun{\oa{B}}{\sphilbfrak{h}}.
\end{aligned}
\label{eq:paper-main-spatial-inclusion-types}
\end{equation}
These maps satisfy
\begin{equation}
\begin{aligned}
i_{\Lambda',\Lambda''}\circ i_{\Lambda,\Lambda'}
&=
i_{\Lambda,\Lambda''},
&
i_{\Lambda',\Gamma}\circ i_{\Lambda,\Lambda'}
&=
i_{\Lambda,\Gamma}
&&\rbk{\Lambda\subset \Lambda'\subset \Lambda''}.
\end{aligned}
\label{eq:paper-main-spatial-inclusion-compatibility}
\end{equation}
For disjoint nonempty $\Lambda,\Lambda'\in\setpowfin{\Gamma}$, their ranges satisfy
\begin{equation}
\begin{aligned}
\commutator{
\fun{i_{\Lambda,\Gamma}}{\oa{M}(\oa{K}(\sphilbfrak{h}_{\Lambda})^{\gamma})}}
{
\fun{i_{\Lambda',\Gamma}}{\oa{M}(\oa{K}(\sphilbfrak{h}_{\Lambda'})^{\gamma})}}
&=
0
&&\rbk{\Lambda\cap \Lambda'=\emptyset}.
\end{aligned}
\label{eq:paper-main-disjoint-local-ranges-commute}
\end{equation}
For every $n\in\monnat$, the local coordinate ranges satisfy
\begin{equation}
\begin{aligned}
\gtclos{
\bigcup_{\substack{\Lambda\in\setpowfin{\Gamma}\\\Lambda\neq\emptyset}}
\fun{p_n\circ i_{\Lambda,\Gamma}}{\oa{M}(\oa{K}(\sphilbfrak{h}_{\Lambda})^{\gamma})}
}^{\norm{\mathord{\cdot}}_{\oa{B}_n(\sphilbfrak{h})}}
&=
\oa{B}_n(\sphilbfrak{h}).
\end{aligned}
\label{eq:paper-main-local-coordinate-density}
\end{equation}
\end{defn}

The particle-coefficient embeddings must respect these local ranges.

\subsection{Algebraic main results}\label{algebraic-main-results}

In this subsection, we state the structure of the abstract Buchholz algebra, which is defined in \eqref{eq:paper-main-cutoff-inverse-limit} without a Fock representation.

\begin{thm}[Main algebraic result]\label{thm:paper-main-algebraic-result}
Assume the intrinsic particle-block system of Definition \ref{def:intrinsic-particle-block-system}, the connecting maps
\eqref{eq:paper-coefficient-connecting-map} satisfying \eqref{eq:paper-coefficient-connecting-compatibility},
and the reconstruction data of Definition \ref{def:paper-abstract-particle-coefficient-reconstruction-data}.
Then the following statements hold.
\begin{enumerate}
\item
The coordinatewise operations make $\fun{\oa{B}}{\sphilbfrak{h}}$ a unital $\oacstar$-algebra, the operations
\eqref{eq:paper-main-particle-coefficient-operations} make $\oa{C}_{\mathrm{b}}(\sphilbfrak{h})$ a unital
$\oacstar$-algebra, and
\begin{equation}
\begin{aligned}
\Upsilon
\colon
\oa{C}_{\mathrm{b}}(\sphilbfrak{h})
\to
\fun{\oa{B}}{\sphilbfrak{h}};\quad
\fun{\Upsilon}{a}
&=
\seq{\fun{\Pi_n}{a}}{n\in\monnat}
\end{aligned}
\label{eq:paper-main-particle-coefficient-isomorphism}
\end{equation}
is an isometric unital $\ast$-isomorphism.
\item
The closed two-sided ideals of $\oa{B}_0(\sphilbfrak{h})$ are $0$ and
$\oa{B}_0(\sphilbfrak{h})$.
For every $n\geq1$, the closed two-sided ideals of $\oa{B}_n(\sphilbfrak{h})$ are exactly
$$\begin{aligned}
0
&=
\oa{J}_{n+1}^{(n)}
\subset
\oa{J}_{n}^{(n)}
\subset
\cdots
\subset
\oa{J}_{1}^{(n)}
\subset
\oa{J}_{0}^{(n)}
=
\oa{B}_n(\sphilbfrak{h}),
\end{aligned}$$
and, for $0\leq k\leq n$,
$$\begin{aligned}
\setquot{\oa{J}_k^{(n)}}{\oa{J}_{k+1}^{(n)}}
&\simeq
\oa{C}_k(\sphilbfrak{h}).
\end{aligned}$$
\item
For every $n\geq1$, the ideal $\oa{J}_n^{(n)}$ is essential in $\oa{B}_n(\sphilbfrak{h})$, and
$$\begin{aligned}
0
\to
\oa{J}_n^{(n)}
\monoto
\oa{B}_n(\sphilbfrak{h})
\xrightarrow{r_n}
\oa{B}_{n-1}(\sphilbfrak{h})
\to
0
\end{aligned}$$
is an essential extension.
\end{enumerate}
\end{thm}

The first item is Proposition \ref{prop:bounded-inverse-limit-cstar} and Theorem \ref{thm:paper-particle-coefficient-buchholz-isomorphism}, the second is Theorem \ref{thm:paper-finite-stage-ideal-chain}, and the essentiality in the third is part of Definition \ref{def:paper-particle-coefficient-embedding-data}. Thus the particle-degree coefficients are coordinates on the abstract Buchholz algebra, in which products are computed by the overlap maps and the particle cutoffs are read from the degree filtration.

\subsection{Intrinsic trace and path-integral results}\label{intrinsic-trace-and-path-integral-results}

In this subsection, we state the equilibrium, product-formula, and closed-path results for the matrix corners of one finite-dimensional one-particle space; they use neither a Fock representation nor a graph.

\begin{thm}[Intrinsic Gibbs states]\label{thm:paper-main-intrinsic-gibbs-states}
Fix a nonzero finite-dimensional complex subspace $\sphilbfrak{k}\subset\sphilbfrak{h}$ and
$\sminvtemperature>0$, and let $\seq{\physham_{n}}{n\in\monnat}$ be a Hamiltonian family of self-adjoint corner
elements.
Assume the polynomial corner-dimension growth \eqref{eq:paper-corner-dimension-growth} and the superstability
\eqref{eq:paper-main-superstability-bound} of $\seq{\physham_{n}}{n\in\monnat}$.
Then the heat kernel $\seq{\napiernum^{-\sminvtemperature\physham_{n}}}{n\in\monnat}$ belongs to the gauge-fixed
minimal ideal $\oa{K}(\sphilbfrak{k})^{\gamma}$ and is trace summable, and the normalized trace functional
$$\begin{aligned}
A
\mapsto
\frac{
\sum_{n\in\monnat}
\sqfun{\trace_{n}}{A_n\napiernum^{-\sminvtemperature \physham_{n}}}
}{
\sum_{n\in\monnat}
\sqfun{\trace_{n}}{\napiernum^{-\sminvtemperature \physham_{n}}}
}
\end{aligned}$$
is a $\sminvtemperature$-KMS state on $\oa{K}(\sphilbfrak{k})^{\gamma}$ for the block dynamics generated by
$\seq{\physham_{n}}{n\in\monnat}$.
\end{thm}

The heat-kernel assertions are Proposition \ref{prop:paper-compact-heat-kernel} and Corollary \ref{cor:paper-trace-summable-heat-kernel}. The Gibbs state and the block dynamics are \eqref{eq:paper-intrinsic-gibbs-state} and \eqref{eq:paper-intrinsic-block-dynamics}, and the KMS property is Proposition \ref{prop:paper-intrinsic-gibbs-kms-state}.

\begin{thm}[Main intrinsic finite-mode product formula]\label{thm:paper-main-intrinsic-finite-mode-product-formula}
Fix a nonzero finite-dimensional complex subspace $\sphilbfrak{k}\subset\sphilbfrak{h}$ and
$\sminvtemperature>0$.
Let $\trace_n$ and $\trace$ be the traces in \eqref{eq:paper-standard-corner-trace} and
\eqref{eq:paper-intrinsic-standard-trace}, respectively.
Assume that the kinetic and potential families satisfy the potential superstability
\eqref{eq:paper-main-corner-potential-superstability} and the Gaussian summability
\eqref{eq:paper-main-kinetic-heat-trace-summability}.
Then the heat-kernel family $\seq{\napiernum^{-\sminvtemperature\rbk{T_n+V_n}}}{n\in\monnat}$ and, for every
$M\in\semigrposint$, the split-product family
$\seq{\rbk{\napiernum^{-\frac{\sminvtemperature}{M}V_n}\napiernum^{-\frac{\sminvtemperature}{M}T_n}}^{M}}{n\in\monnat}$
are trace summable, and
$$\begin{aligned}
\lim_{M\to\infty}
\sum_{n\in\monnat}
\norm{
\rbk{
\napiernum^{-\frac{\sminvtemperature}{M}V_n}
\napiernum^{-\frac{\sminvtemperature}{M}T_n}}^{M}
-
\napiernum^{-\sminvtemperature\rbk{T_n+V_n}}
}_{1,\trace_n}
&=
0.
\end{aligned}$$
\end{thm}

This theorem is Proposition \ref{prop:intrinsic-trace-norm-product-formula}. On a lattice, Lemma \ref{lem:paper-summable-sector-majorant} derives the Gaussian summability from \eqref{eq:paper-corner-dimension-growth} and \eqref{eq:paper-main-linear-kinetic-bound}. For infinite-dimensional sectors with possibly unbounded kinetic operators, Proposition \ref{prop:paper-gibbs-sector-product-formula} gives the same conclusion under the kinetic Gibbs condition and sectorwise trace-norm Trotter convergence.

\begin{thm}[Closed-path expansion]\label{thm:paper-main-closed-path-expansion}
Assume the hypotheses of Theorem \ref{thm:paper-main-intrinsic-finite-mode-product-formula}, and assume that each
$V_n$ is diagonal in a fixed decomposition of the corner unit into mutually orthogonal minimal projections.
Then, for every number $M$ of time steps, the trace of the split product is the particle-number cutoff limit of
a sum over closed paths of length $M$ in the labels of these projections.
The weight of a closed path is the product of the potential factors along the path and the matrix elements of
the kinetic factors between consecutive labels.
As $M\to\infty$, these finite-time partition functions converge to
$\sum_{n\in\monnat}\sqfun{\trace_{n}}{\napiernum^{-\sminvtemperature\rbk{T_n+V_n}}}$.
If the path weights are nonnegative, their normalization is a probability measure on the closed paths, and
this measure is reflection positive under a reflection-invariance and positivity condition on the weights.
\end{thm}

The open-path operator identity is Proposition \ref{prop:paper-matrix-open-path-expansion}, the closed-path expansion and the partition-function limit are Proposition \ref{prop:paper-abstract-path-integral} with the weights \eqref{eq:paper-closed-path-weight}, and the probabilistic statements are Proposition \ref{prop:path-weight-probability-criterion}. These closed-path expansions are the finite-mode starting point for path integrals on the abstract Buchholz algebra.

\subsection{Spatial localization}\label{spatial-localization}

In this subsection, we define the local and quasilocal subalgebras and state their properties for the lattice one-particle space \(\sphilbfrak{h}=\fun{\lpseq^{2}}{\Gamma}\).

For \(\Lambda\in\setpowfin{\Gamma}\), define the local algebras, using the local inclusions \(i_{\Lambda,\Gamma}\) in \eqref{eq:paper-main-spatial-inclusion-types} when \(\Lambda\neq\emptyset\), by \begin{equation}
\begin{aligned}
\oa{B}_{\Lambda}
&=
\begin{cases}
\fldcmp 1_{\fun{\oa{B}}{\sphilbfrak{h}}},&\Lambda=\emptyset,\\
\fun{i_{\Lambda,\Gamma}}{\oa{M}(\oa{K}(\sphilbfrak{h}_{\Lambda})^{\gamma})},
&\emptyset\neq\Lambda\in\setpowfin{\Gamma}.
\end{cases}
\end{aligned}
\label{eq:paper-main-local-buchholz-algebra}
\end{equation} The quasilocal subalgebra is \begin{equation}
\begin{aligned}
\oa{B}_{\mathrm{ql}}
&=
\gtclos{
\bigcup_{\Lambda\in\setpowfin{\Gamma}}
\oa{B}_{\Lambda}
}^{\norm{\mathord{\cdot}}_{\fun{\oa{B}}{\sphilbfrak{h}}}}.
\end{aligned}
\label{eq:paper-main-quasilocal-algebra}
\end{equation}

\begin{thm}[Spatial localization on a lattice]\label{thm:paper-main-spatial-localization}
In the Fock realization, spatial localization data in the sense of Definition
\ref{def:paper-main-spatial-localization-data} exist: the local algebra of a finite region consists of the
gauge-invariant operators acting on the modes in the region and trivially on the other modes, local algebras of
nested regions are compatible, those of disjoint regions commute, and their union is dense at every fixed
particle cutoff.
The quasilocal algebra $\oa{B}_{\mathrm{ql}}$ is a proper subalgebra of $\fun{\oa{B}}{\sphilbfrak{h}}$: for a
unit vector $f$ with infinitely many nonzero coordinates in the standard basis, the element realized by the
resolvent $\rbk{1+N_f}^{-1}$ of the number $N_f$ of particles in the mode $f$ has distance $1/2$ from
$\oa{B}_{\mathrm{ql}}$, although each of its fixed particle-cutoff coordinates is approximated by coordinates of
local elements.
\end{thm}

The existence and the three properties of the local algebras are Proposition \ref{prop:concrete-spatial-localization-system}, and their Fock images are Corollary \ref{cor:paper-fock-local-algebras}. The distance statement is Theorem \ref{thm:paper-rank-one-tail-geometry}, specialized to the standard basis of \(\fun{\lpseq^{2}}{\Gamma}\) in Appendix \ref{sec:paper-graph-spatial-specialization}, with the one-mode resolvent of Section \ref{sec:paper-one-mode-resolvent}.

\subsection{Fock-space realization of the abstract objects}\label{sec:paper-main-fock-realization}

In this subsection, we state how the regular bosonic Fock representation realizes the abstract objects; the known exact-particle and inverse-limit descriptions \cite[Lemmas 3.3 and 3.4 and Theorem 3.5]{DetlevBuchholz002} enter only as input.

\begin{thm}[Fock realization]\label{thm:paper-main-fock-realization}
In the regular bosonic Fock representation, the following statements hold.
\begin{enumerate}
\item
For every nonzero finite-dimensional complex subspace $\sphilbfrak{f}\subset\sphilbfrak{h}$, the minimal ideal
$\oa{K}(\sphilbfrak{f})$ is represented by the compact operators on the Fock space over $\sphilbfrak{f}$, and
$e_{\sphilbfrak{f},m}$ by the projection onto the $m$-particle space
$\sphilbfrak{f}^{\otimes_{\txtsym}m}$.
\item
The cutoff ideal $\oaideal{I}_{>n}(\sphilbfrak{h})$ is the kernel of the restriction to the $n$-particle space.
Hence $\oa{B}_n(\sphilbfrak{h})$ is the exact $n$-particle algebra of \cite[Lemma 3.3]{DetlevBuchholz002}, and
$\fun{\oa{B}}{\sphilbfrak{h}}$ is faithfully represented as the concrete Buchholz algebra of
\cite[Theorem 3.5]{DetlevBuchholz002}.
\item
The coefficient algebra $\oa{C}_m(\sphilbfrak{h})$ is the algebra of compact operators on
$\sphilbfrak{h}^{\otimes_{\txtsym}m}$.
The embedding $E_{m,n}$ places a degree-$m$ coefficient on every choice of $m$ of the $n$ tensor factors and
sums the results, and $\mu_{p,q}^{r}$ sums the products of two such placements whose union has $r$ factors.
\end{enumerate}
\end{thm}

The first item is Theorem \ref{thm:paper-fock-intrinsic-block-realization}, the second is Theorem \ref{thm:paper-abstract-concrete-comparison} and Corollary \ref{cor:paper-fock-cutoff-realization}, and the third is Theorem \ref{thm:paper-concrete-particle-coefficient-reconstruction-data} and Corollary \ref{cor:paper-fock-degree-tail-ideals}. The Fock objects are defined in Section \ref{sec:paper-standard-fock-objects}.

\section{Abstract Buchholz algebra and particle-degree coefficients}\label{sec:general-separable-theory}

In this section, the representation-independent construction starts from a separable infinite-dimensional complex Hilbert space and does not use a graph, a coordinate basis, or a Fock representation.

\subsection{Finite-dimensional resolvent subalgebras}\label{finite-dimensional-resolvent-subalgebras}

In this subsection, the finite-dimensional resolvent subalgebras used in the representation-independent construction are identified.

For every nonzero finite-dimensional complex subspace \(\sphilbfrak{f}\subset\sphilbfrak{h}\), the algebra \(\oaresolventalgebra(\sphilbfrak{f})\) is formed with the nondegenerate restriction of \eqref{eq:abstract-symplectic-form}, as fixed in Section \ref{sec:paper-main-basic-objects}. The universal inclusion \(\iota_{\sphilbfrak{f}}\) is defined in \eqref{eq:paper-global-resolvent-inclusion} and is injective by \cite[Proposition 2.6(i)]{DetlevBuchholz001}. We identify its domain with the subalgebra \(\oaresolventalgebra(\sphilbfrak{f})\).

\begin{thm}[Minimal compact ideal]\label{thm:abstract-minimal-compact-ideal}
For every nonzero finite-dimensional complex subspace $\sphilbfrak{f}\subset\sphilbfrak{h}$, the algebra
$\oaresolventalgebra(\sphilbfrak{f})$ is type I in the sense of
\cite[Theorem 3.8(ii)]{DetlevBuchholz001}.
In particular, it is postliminal: the image of every irreducible representation contains the compact operators
on its representation Hilbert space.
The algebra has the unique minimal nonzero closed two-sided ideal
$\oa{K}(\sphilbfrak{f})$ defined in \eqref{eq:paper-main-intrinsic-minimal-ideal}, and this ideal is essential.
For every faithful regular irreducible representation as in
\cite[Theorem 5.4(vi)]{BuchholzGrundling002},
$\pi_{\sphilbfrak{f}}\colon\oaresolventalgebra(\sphilbfrak{f})\to\opspbddlin{\sphilbfrak{H}_{\sphilbfrak{f}}}$, one has
\begin{equation}
\begin{aligned}
\fun{\pi_{\sphilbfrak{f}}}{\oa{K}(\sphilbfrak{f})}
&=
\fun{\schattencls}{\sphilbfrak{H}_{\sphilbfrak{f}}}.
\end{aligned}
\label{eq:paper-minimal-ideal-representation}
\end{equation}
\end{thm}

\begin{proof}
The type-I assertion is \cite[Theorem 3.8(ii)]{DetlevBuchholz001}.
The existence, uniqueness, essentiality, and representation formula for the ideal follow from
\cite[Theorem 5.4(ii), (v), and (vi)]{BuchholzGrundling002}.
The intrinsic intersection formula is \cite[Theorem 4.5(i)]{DetlevBuchholz001}.
\end{proof}

\begin{rem}[Characterization of compact-operator algebras]\label{rem:paper-compact-operator-algebra-characterization}
Let $\oa{A}$ be a nonzero $\oacstar$-algebra.
Then $\oa{A}$ is $\ast$-isomorphic to $\fun{\schattencls}{\sphilbfrak{H}}$ for some nonzero complex Hilbert
space $\sphilbfrak{H}$ if and only if $\oa{A}$ is simple and contains a minimal projection in the sense of
Definition \ref{def:paper-simple-algebra-minimal-projection}.
For the nontrivial implication, choose a minimal projection $p\in\oa{A}$ and give $\oa{A}p$ the inner product
determined by
$$\begin{aligned}
\bkt{ap}{bp}_{\oa{A}p}p
&=
p\faadj{a}bp.
\end{aligned}$$
Left multiplication defines $\pi\colon\oa{A}\to\opspbddlin{\oa{A}p}$, and
$\fun{\pi}{ap\faadj{b}}=P_{ap,bp}$.
Simplicity gives $\gtclos{\oa{A}p\oa{A}}^{\norm{\mathord{\cdot}}}=\oa{A}$, so
$\fun{\pi}{\oa{A}}=\fun{\schattencls}{\oa{A}p}$.
The kernel of $\pi$ is an ideal and $\fun{\pi}{p}\neq0$; hence simplicity makes $\pi$ faithful.
Conversely, compact-operator algebras are simple and every rank-one projection is minimal.
In the present construction, \eqref{eq:paper-main-intrinsic-minimal-ideal} defines
$\oa{K}(\sphilbfrak{f})$, while Theorem \ref{thm:abstract-minimal-compact-ideal} and
\eqref{eq:paper-minimal-ideal-representation} identify its image as a compact-operator algebra.
Thus simplicity and the existence of minimal projections are consequences of the representation formula,
not additional hypotheses on $\oa{K}(\sphilbfrak{f})$.
The later finite-particle construction separately assumes that each corner is a full matrix algebra and uses
its minimal projections to normalize the traces in \eqref{eq:paper-standard-corner-trace}.
\end{rem}

\subsection{Intrinsic particle blocks and particle-cutoff quotients}\label{intrinsic-particle-blocks-and-particle-cutoff-quotients}

In this subsection, the intrinsic particle-block system produces the particle-cutoff ideals, quotients, and connecting maps used in the inverse system.

The intrinsic particle-block system in Definition \ref{def:intrinsic-particle-block-system} provides the local projections used in the global cutoff ideals \eqref{eq:abstract-high-particle-ideal}. The particle blocks and compression maps are hypotheses of the abstract construction, whereas gauge invariance and the decreasing family of cutoff ideals are consequences.

\begin{rem}[Hypotheses on the particle blocks and compression maps]
The universal resolvent algebra and the intrinsic minimal ideals do not determine the particle-block
system in Definition \ref{def:intrinsic-particle-block-system}.
The particle-block system is additional structure on the gauge-fixed minimal ideals
\eqref{eq:paper-main-gauge-fixed-minimal-ideal}.
The injectivity of the compression maps \eqref{eq:paper-coefficient-connecting-map} and their compatibility
\eqref{eq:paper-coefficient-connecting-compatibility} are likewise additional assumptions.
The particle-cutoff quotients
$\oa{B}_n(\sphilbfrak{h})$ in \eqref{eq:paper-main-particle-cutoff-quotient} imply neither condition.
In a chosen concrete representation, one instead proves these conditions within that representation.
\end{rem}

\begin{prop}[Gauge invariance and decreasing particle-cutoff ideals]\label{prop:paper-gauge-invariance-decreasing-particle-cutoff-ideals}
For every $n\in\monnat$ and every nonzero finite-dimensional complex subspace
$\sphilbfrak{f}\subset\sphilbfrak{h}$, the projection $Q_{\sphilbfrak{f},>n}$ in
\eqref{eq:abstract-high-particle-projection} is gauge invariant.
The ideals in \eqref{eq:abstract-high-particle-ideal} satisfy
$$\begin{aligned}
\oaideal{I}_{>n}(\sphilbfrak{h})
&\subset
\oaideal{I}_{>n-1}(\sphilbfrak{h})
&&\rbk{n\geq1}.
\end{aligned}$$
Consequently, the map $r_n$ in \eqref{eq:paper-main-cutoff-connecting-map} is a well-defined surjective
$\ast$-homomorphism for every $n\geq1$.
\end{prop}

\begin{proof}
The universal inclusion \eqref{eq:paper-global-resolvent-inclusion} is gauge equivariant:
$$\begin{aligned}
\gamma_\theta\circ\iota_{\sphilbfrak{f}}
&=
\iota_{\sphilbfrak{f}}\circ\gamma_\theta.
\end{aligned}$$
Every $e_{\sphilbfrak{f},m}$ is gauge invariant by Definition
\ref{def:intrinsic-particle-block-system}.
Substitution into \eqref{eq:abstract-high-particle-projection} gives
$$\begin{aligned}
\fun{\gamma_{\theta}}{Q_{\sphilbfrak{f},>n}}
&=
1-
\sum_{m=0}^{n}\fun{\gamma_{\theta}}{e_{\sphilbfrak{f},m}}
=
1-
\sum_{m=0}^{n}e_{\sphilbfrak{f},m}
=
Q_{\sphilbfrak{f},>n}.
\end{aligned}$$
For $n\geq1$, orthogonality gives the local identities
$$\begin{aligned}
Q_{\sphilbfrak{f},>n-1}
=
e_{\sphilbfrak{f},n}+Q_{\sphilbfrak{f},>n},
\quad
Q_{\sphilbfrak{f},>n-1}Q_{\sphilbfrak{f},>n}
=
Q_{\sphilbfrak{f},>n}
=
Q_{\sphilbfrak{f},>n}Q_{\sphilbfrak{f},>n-1}.
\end{aligned}$$
For $A,C\in\oaresolventalgebra(\sphilbfrak{h},\sigma)^{\gamma}$, the generators of the cutoff ideal satisfy
$$\begin{aligned}
A\fun{\iota_{\sphilbfrak{f}}}{Q_{\sphilbfrak{f},>n}}C
&=
A\fun{\iota_{\sphilbfrak{f}}}{Q_{\sphilbfrak{f},>n-1}}
\rbk{\fun{\iota_{\sphilbfrak{f}}}{Q_{\sphilbfrak{f},>n}}C}
\in
\oaideal{I}_{>n-1}(\sphilbfrak{h}).
\end{aligned}$$
The algebraic sum over $\sphilbfrak{f}$ and its norm closure therefore satisfy the asserted ideal inclusion.
This inclusion makes $r_n$ well-defined.
Every $A+\oaideal{I}_{>n-1}(\sphilbfrak{h})$ is the image under $r_n$ of
$A+\oaideal{I}_{>n}(\sphilbfrak{h})$, which proves surjectivity.
The quotient operations show that $r_n$ is a $\ast$-homomorphism.
\end{proof}

The ideals \eqref{eq:abstract-high-particle-ideal} lie in the single universal fixed-point algebra \(\oaresolventalgebra(\sphilbfrak{h},\sigma)^{\gamma}\). Their definition does not require inclusions \(\fun{\iota_{\sphilbfrak{f},\sphilbfrak{g}}}{\oa{K}(\sphilbfrak{f})}
\subset\oa{K}(\sphilbfrak{g})\) for \(\sphilbfrak{f}\subset\sphilbfrak{g}\).

\begin{rem}[Exact-particle and particle-cutoff representations]
Let $\sphilbfrak{k}$ be a one-dimensional complex Hilbert space.
For every $n\in\monnat$, its exact
$n$-particle space is one-dimensional and
$$\begin{aligned}
\opspbddlin{\sphilbfrak{k}^{\otimes_{\txtsym}n}}
&\cong
\fldcmp.
\end{aligned}$$
The direct sum of the sectors from zero through $n$ instead satisfies
$$\begin{aligned}
\opspbddlin{\bigoplus_{j=0}^{n}\sphilbfrak{k}^{\otimes_{\txtsym}j}}
&\cong
\fun{M_{n+1}}{\fldcmp}.
\end{aligned}$$
Thus, for $n\geq1$, the algebra on the cutoff direct sum and the algebra represented on the exact
$n$-particle subspace are not isomorphic.
\end{rem}

\begin{defn}[Abstract Buchholz algebra]\label{def:paper-abstract-buchholz-algebra}
The bounded inverse limit $\fun{\oa{B}}{\sphilbfrak{h}}$ defined in
\eqref{eq:paper-main-cutoff-inverse-limit} is called the abstract Buchholz algebra associated with
$\sphilbfrak{h}$ and the intrinsic particle-block system.
Its norm is
$$\begin{aligned}
\norm{\seq{A_n}{n\in\monnat}}_{\fun{\oa{B}}{\sphilbfrak{h}}}
&=
\sup_{n\in\monnat}\norm{A_n}.
\end{aligned}$$
Products and adjoints are defined coordinatewise.
\end{defn}

\begin{prop}[The bounded inverse limit is a $\oacstar$-algebra]\label{prop:bounded-inverse-limit-cstar}
The algebra $\fun{\oa{B}}{\sphilbfrak{h}}$ in Definition
\ref{def:paper-abstract-buchholz-algebra} is a unital $\oacstar$-algebra.
The quotient-coordinate map is the unital $\ast$-homomorphism
$$\begin{aligned}
\rho
\colon
\oaresolventalgebra(\sphilbfrak{h},\sigma)^{\gamma}
\to
\fun{\oa{B}}{\sphilbfrak{h}};\quad
\fun{\rho}{A}
&=
\seq{A+\oaideal{I}_{>n}(\sphilbfrak{h})}{n\in\monnat}.
\end{aligned}$$
Its kernel is
\begin{equation}
\begin{aligned}
\Ker\rho
&=
\bigcap_{n\in\monnat}\oaideal{I}_{>n}(\sphilbfrak{h}).
\end{aligned}
\label{eq:paper-quotient-coordinate-kernel}
\end{equation}
\end{prop}

\begin{proof}
Let $A=\seq{A_n}{n\in\monnat}$ and $C=\seq{C_n}{n\in\monnat}$ belong to
$\fun{\oa{B}}{\sphilbfrak{h}}$.
Thus, $A_n$ and $C_n$ are equivalence classes modulo $\oaideal{I}_{>n}(\sphilbfrak{h})$, and
$\fun{r_n}{A_n}=A_{n-1}$ and $\fun{r_n}{C_n}=C_{n-1}$ for every $n\geq1$.
Multiplicativity and adjoint preservation of $r_n$ give
$$\begin{aligned}
\fun{r_n}{A_nC_n}
&=
\fun{r_n}{A_n}\fun{r_n}{C_n}
=
A_{n-1}C_{n-1},
&
\fun{r_n}{\faadj{A_n}}
&=
\faadj{\fun{r_n}{A_n}}
=
\faadj{A_{n-1}}.
\end{aligned}$$
Therefore, coordinatewise multiplication and adjoint preserve the defining relations in
\eqref{eq:paper-main-cutoff-inverse-limit}.
Linearity and unitality of $r_n$ show the same for linear combinations and the unit sequence.
The coordinatewise norm estimates are
$$\begin{aligned}
\sup_{n\in\monnat}\norm{A_nC_n}
&\leq
\rbk{\sup_{n\in\monnat}\norm{A_n}}
\rbk{\sup_{n\in\monnat}\norm{C_n}},
&
\sup_{n\in\monnat}\norm{\faadj{A_n}}
&=
\sup_{n\in\monnat}\norm{A_n}.
\end{aligned}$$

Suppose that $\seq{A_j}{j\in\semigrposint}$ is a sequence in
$\fun{\oa{B}}{\sphilbfrak{h}}$ such that $A_j=\seq{A_{j,n}}{n\in\monnat}$ for every
$j\in\semigrposint$ and converges to
$A=\seq{A_n}{n\in\monnat}$ in the supremum norm as $j\to\infty$.
Continuity of $r_n$ gives
$$\begin{aligned}
\fun{r_n}{A_n}
&=
\fun{r_n}{\lim_{j\to\infty}A_{j,n}}
=
\lim_{j\to\infty}\fun{r_n}{A_{j,n}}
=
\lim_{j\to\infty}A_{j,n-1}
=
A_{n-1}.
\end{aligned}$$
Hence the set in \eqref{eq:paper-main-cutoff-inverse-limit} is closed in the supremum norm.
The coordinatewise $\oacstar$-identities give
$$\begin{aligned}
\norm{\faadj{A}A}_{\fun{\oa{B}}{\sphilbfrak{h}}}
&=
\sup_{n\in\monnat}\norm{\faadj{A_n}A_n}
=
\sup_{n\in\monnat}\norm{A_n}^{2}
=
\norm{A}_{\fun{\oa{B}}{\sphilbfrak{h}}}^{2}.
\end{aligned}$$
This proves the first assertion.

For every $A\in\oaresolventalgebra(\sphilbfrak{h},\sigma)^{\gamma}$, quotient contractivity and the defining
relations in \eqref{eq:paper-main-cutoff-inverse-limit} are calculated by
$$\begin{aligned}
\norm{A+\oaideal{I}_{>n}(\sphilbfrak{h})}
&\leq
\norm{A},
&
\fun{r_n}{A+\oaideal{I}_{>n}(\sphilbfrak{h})}
&=
A+\oaideal{I}_{>n-1}(\sphilbfrak{h}).
\end{aligned}$$
The map $\rho$ is well-defined.
For $A,C\in\oaresolventalgebra(\sphilbfrak{h},\sigma)^{\gamma}$ and $n\in\monnat$, its quotient coordinates
satisfy
$$\begin{aligned}
\rbk{\fun{\rho}{AC}}_n
&=
AC+\oaideal{I}_{>n}(\sphilbfrak{h})
\\
&=
\rbk{A+\oaideal{I}_{>n}(\sphilbfrak{h})}
\rbk{C+\oaideal{I}_{>n}(\sphilbfrak{h})}
=
\rbk{\fun{\rho}{A}}_n\rbk{\fun{\rho}{C}}_n,
\\
\rbk{\fun{\rho}{\faadj{A}}}_n
&=
\faadj{A}+\oaideal{I}_{>n}(\sphilbfrak{h})
=
\faadjrbk{\rbk{\fun{\rho}{A}}_n},
\\
\rbk{\fun{\rho}{1}}_n
&=
1.
\end{aligned}$$
Thus $\rho$ is a unital $\ast$-homomorphism.
For $A\in\oaresolventalgebra(\sphilbfrak{h},\sigma)^{\gamma}$, the equality $\fun{\rho}{A}=0$ holds exactly
when $A+\oaideal{I}_{>n}(\sphilbfrak{h})=0$ for every $n\in\monnat$, which is equivalent to
$A\in\oaideal{I}_{>n}(\sphilbfrak{h})$ for every $n\in\monnat$.
This proves \eqref{eq:paper-quotient-coordinate-kernel}.
\end{proof}

\begin{rem}[Abstract faithfulness criterion and the Fock representation]\label{rem:paper-abstract-faithfulness-criterion}
For every $n\in\monnat$, let $\sphilbfrak{H}_n$ be a complex Hilbert space and let
$\pi_n\colon\oaresolventalgebra(\sphilbfrak{h},\sigma)^{\gamma}
\to\opspbddlin{\sphilbfrak{H}_n}$ be a $\ast$-representation satisfying
$\Ker\pi_n=\oaideal{I}_{>n}(\sphilbfrak{h})$.
Equation \eqref{eq:paper-quotient-coordinate-kernel} gives the abstract criterion
\begin{equation}
\begin{aligned}
\Ker\rho
&=
\bigcap_{n\in\monnat}\oaideal{I}_{>n}(\sphilbfrak{h})
=
\bigcap_{n\in\monnat}\Ker\pi_n
=
\Ker\rbk{\bigoplus_{n\in\monnat}\pi_n}.
\end{aligned}
\label{eq:paper-abstract-faithfulness-criterion}
\end{equation}
Thus $\rho$ is injective exactly when the direct-sum representation
$\bigoplus_{n\in\monnat}\pi_n$ is faithful.
This criterion contains no Fock-specific assumption.

The Fock realization supplies representations satisfying this criterion in two separate steps.
Theorem \ref{thm:paper-abstract-concrete-comparison} identifies the cutoff ideals with the kernels of the
exact-particle representations $\rho_n$, while the particle-number decomposition gives
$$\begin{aligned}
\fnrestr{\oarepn_{\txtfock}}
{\oaresolventalgebra(\sphilbfrak{h},\sigma)^{\gamma}}
&=
\bigoplus_{n\in\monnat}\rho_n.
\end{aligned}$$
The direct sum is faithful because the Fock representation is regular and every regular representation of the
resolvent algebra is faithful by \cite[Theorem 4.10]{BuchholzGrundling002}; the sector decomposition alone does
not imply faithfulness.
Corollary \ref{cor:paper-faithful-quotient-coordinate-embedding} records the resulting injectivity of $\rho$.
For the same system, the faithful Fock representation of the abstract Buchholz algebra constructed below in
Corollary \ref{cor:paper-fock-cutoff-realization} satisfies
$$\begin{aligned}
\fun{\oarepn_{\txtfock}}{\fun{\rho}{A}}
&=
\bigoplus_{n\in\monnat}\fun{\rho_n}{A}
=
\fun{\oarepn_{\txtfock}}{A}
&&\rbk{A\in\oaresolventalgebra(\sphilbfrak{h},\sigma)^{\gamma}}.
\end{aligned}$$
Thus $\rho$ identifies the universal fixed-point algebra with its original Fock image inside the represented
bounded inverse limit.
Corollary \ref{cor:paper-faithful-quotient-coordinate-embedding} does not assert that every bounded compatible
sequence in $\fun{\oa{B}}{\sphilbfrak{h}}$ is obtained from a single element of
$\oaresolventalgebra(\sphilbfrak{h},\sigma)^{\gamma}$.
\end{rem}

\subsection{Particle-degree coefficients and reconstruction}\label{particle-degree-coefficients-and-reconstruction}

In this subsection, we isolate the particle-degree coefficient spaces and formulate the data that reconstruct the bounded inverse limit from them.

Proposition \ref{prop:bounded-inverse-limit-cstar} supplies the ambient \(\oacstar\)-algebra for the coefficient calculation. The inverse-limit definition becomes computational after one isolates the coefficient data common to all particle-cutoff quotients \(\oa{B}_n(\sphilbfrak{h})\). For every \(m\in\monnat\), the degree-\(m\) coefficient algebra \(\oa{C}_m(\sphilbfrak{h})\) is defined in \eqref{eq:paper-main-degree-coefficient-algebra}. Definition \ref{def:paper-particle-coefficient-embedding-data} specifies the degree-zero identification, the embeddings \(E_{m,n}\), their direct-sum and adjoint properties, and their compatibility with the particle cutoffs. Definition \ref{def:paper-abstract-particle-coefficient-reconstruction-data} adds the overlap maps \(\mu_{p,q}^{r}\) that determine multiplication. The connecting maps in this inductive limit are given by \eqref{eq:paper-coefficient-connecting-map}.

Definition \ref{def:paper-essential-ideal-extension} and the sequence \eqref{eq:paper-main-essential-extension} apply to arbitrary \(\oacstar\)-algebras and arbitrary closed two-sided ideals in the particle-cutoff construction.

Neither definition includes a Fock representation. Their conditions abstract the linear decomposition, cutoff compatibility, and multiplication of the concrete particle-number blocks. Proposition \ref{prop:paper-particle-block-independence} verifies Definition \ref{def:intrinsic-particle-block-system}, Proposition \ref{prop:paper-coefficient-connecting-maps} verifies the compression compatibility \eqref{eq:paper-coefficient-connecting-compatibility}, and Theorem \ref{thm:paper-concrete-particle-coefficient-reconstruction-data} verifies Definitions \ref{def:paper-particle-coefficient-embedding-data} and \ref{def:paper-abstract-particle-coefficient-reconstruction-data}. Theorem \ref{thm:paper-abstract-concrete-comparison} then identifies the resulting particle-cutoff quotients and bounded inverse limit with their concrete Fock counterparts.

\begin{prop}[Injectivity of the particle-coefficient maps]\label{prop:paper-particle-coefficient-map-injectivity}
Assume the abstract particle-coefficient reconstruction data of Definition
\ref{def:paper-abstract-particle-coefficient-reconstruction-data}.
For every $n\in\monnat$ and $0\leq m\leq n$, the map $E_{m,n}$ is injective.
\end{prop}

\begin{proof}
The scalar identity in \eqref{eq:paper-abstract-particle-coefficient-embedding-identities} proves the assertion
for $m=0$.
Fix $m\geq1$.
The same identity shows that $\oa{B}_m(\sphilbfrak{h})$ contains the identity and is nonzero.
The top-degree condition in Definition \ref{def:paper-particle-coefficient-embedding-data} identifies
$\fun{E_{m,m}}{\oa{C}_m(\sphilbfrak{h})}$ with $\Ker r_m$ and requires this image to be an essential ideal of
$\oa{B}_m(\sphilbfrak{h})$.
The algebra $\oa{B}_m(\sphilbfrak{h})$ is a nonzero closed two-sided ideal of itself.
Definition \ref{def:paper-essential-ideal-extension} requires an essential ideal to have nonzero intersection
with every nonzero closed two-sided ideal of the ambient algebra.
Taking $\oa{B}_m(\sphilbfrak{h})$ itself as that ideal gives
$$\begin{aligned}
\fun{E_{m,m}}{\oa{C}_m(\sphilbfrak{h})}
=
\fun{E_{m,m}}{\oa{C}_m(\sphilbfrak{h})}\cap\oa{B}_m(\sphilbfrak{h})
\neq
0.
\end{aligned}$$
Therefore $E_{m,m}$ has nonzero range.
For $a_m,b_m\in\oa{C}_m(\sphilbfrak{h})$, the overlap-product identities
\eqref{eq:paper-abstract-particle-coefficient-overlap-products}, applied with $p=q=n=m$, give
$$\begin{aligned}
\fun{E_{m,m}}{a_m}\fun{E_{m,m}}{b_m}
=
\fun{E_{m,m}}{\fun{\mu_{m,m}^{m}}{a_m,b_m}}
=
\fun{E_{m,m}}{a_mb_m}.
\end{aligned}$$
Together with the adjoint identity in
\eqref{eq:paper-abstract-particle-coefficient-embedding-identities}, this proves that $E_{m,m}$ is a nonzero
$\ast$-homomorphism.
Its kernel is a proper closed two-sided ideal of $\oa{C}_m(\sphilbfrak{h})$.
Remark \ref{rem:paper-degree-coefficient-algebra-simplicity} records that this algebra is simple.
The definition of simplicity now gives $\Ker E_{m,m}=0$.

Let $n>m$ and suppose that $\fun{E_{m,n}}{a_m}=0$.
Repeated application of \eqref{eq:paper-abstract-particle-coefficient-connecting-identities} gives
$$\begin{aligned}
\fun{E_{m,m}}{a_m}
=
\fun{r_{m+1}\circ r_{m+2}\circ\cdots\circ r_n}{\fun{E_{m,n}}{a_m}}
=
0.
\end{aligned}$$
The injectivity of $E_{m,m}$ gives $a_m=0$.
Thus $E_{m,n}$ is injective for every $0\leq m\leq n$.
\end{proof}

Proposition \ref{prop:paper-particle-coefficient-map-injectivity} and the direct-sum identity in \eqref{eq:paper-abstract-particle-coefficient-embedding-identities} give uniqueness of the finite particle-degree coefficients.

\begin{lem}[Uniqueness of finite particle-degree coefficients]\label{lem:paper-finite-particle-coefficient-uniqueness}
Assume the abstract particle-coefficient reconstruction data of Definition
\ref{def:paper-abstract-particle-coefficient-reconstruction-data}.
For every $n\in\monnat$ and $A_n\in\oa{B}_n(\sphilbfrak{h})$, there are unique coefficients
$a_m\in\oa{C}_m(\sphilbfrak{h})$, $0\leq m\leq n$, such that
$$\begin{aligned}
A_n
&=
\sum_{m=0}^{n}\fun{E_{m,n}}{a_m}.
\end{aligned}$$
\end{lem}

\begin{proof}
Existence is the direct-sum identity in
\eqref{eq:paper-abstract-particle-coefficient-embedding-identities}.
If two coefficient families represent $A_n$, subtraction gives
$$\begin{aligned}
0
&=
\sum_{m=0}^{n}\fun{E_{m,n}}{a_m-b_m}.
\end{aligned}$$
The sum is direct, and Proposition \ref{prop:paper-particle-coefficient-map-injectivity} proves that each
$E_{m,n}$ is injective.
Consequently $a_m=b_m$ for every $0\leq m\leq n$.
\end{proof}

The full coefficient space and its operations are defined in \eqref{eq:paper-main-particle-coefficient-space} and \eqref{eq:paper-main-particle-coefficient-operations}. The evaluation \(\Pi_n\) in the particle-cutoff quotient is defined in \eqref{eq:paper-main-particle-coefficient-reconstruction}. The inequalities \(p,q\leq r\leq p+q\) make the sum defining the \(r\)th coefficient finite. The overlap maps enter through the multiplicativity of the evaluations \(\Pi_n\).

\begin{lem}[Evaluation identities]\label{lem:paper-particle-coefficient-evaluation-identities}
Let $a,b$ belong to the particle-coefficient space $\oa{C}(\sphilbfrak{h})$ defined in
\eqref{eq:paper-main-particle-coefficient-space}, and let $n\in\monnat$.
For the coefficient product and adjoint in \eqref{eq:paper-main-particle-coefficient-operations} and the
evaluation $\Pi_n$ in \eqref{eq:paper-main-particle-coefficient-reconstruction}, one has
\begin{equation}
\begin{aligned}
\fun{\Pi_n}{a\star b}
&=
\fun{\Pi_n}{a}\fun{\Pi_n}{b},
&
\fun{\Pi_n}{\faadj{a}}
&=
\faadj{\fun{\Pi_n}{a}}.
\end{aligned}
\label{eq:paper-particle-coefficient-evaluation-product}
\end{equation}
For $n\geq1$, one also has
\begin{equation}
\begin{aligned}
\fun{r_n\circ\Pi_n}{a}
&=
\fun{\Pi_{n-1}}{a}.
\end{aligned}
\label{eq:paper-particle-coefficient-evaluation-compatibility}
\end{equation}
\end{lem}

\begin{proof}
The overlap-product formula in \eqref{eq:paper-abstract-particle-coefficient-overlap-products} yields the
finite calculation
$$\begin{aligned}
&
\fun{\Pi_n}{a}\fun{\Pi_n}{b}
=
\sum_{p=0}^{n}\sum_{q=0}^{n}
\fun{E_{p,n}}{a_p}\fun{E_{q,n}}{b_q}
\\
&=
\sum_{p=0}^{n}\sum_{q=0}^{n}
\sum_{r=\max\setone{p,q}}^{\min\setone{p+q,n}}
\fun{E_{r,n}}{\fun{\mu_{p,q}^{r}}{a_p,b_q}}
=
\sum_{r=0}^{n}
\fun{E_{r,n}}{
\sum_{\substack{0\leq p,q\leq r\\p+q\geq r}}
\fun{\mu_{p,q}^{r}}{a_p,b_q}}
=
\fun{\Pi_n}{a\star b}.
\end{aligned}$$
The adjoint identity in \eqref{eq:paper-abstract-particle-coefficient-embedding-identities} gives
$$\begin{aligned}
\fun{\Pi_n}{\faadj{a}}
&=
\sum_{m=0}^{n}\fun{E_{m,n}}{\faadj{a_m}}
=
\faadjrbk{\sum_{m=0}^{n}\fun{E_{m,n}}{a_m}}
=
\faadj{\fun{\Pi_n}{a}}.
\end{aligned}$$
The connecting-map identities in \eqref{eq:paper-abstract-particle-coefficient-connecting-identities} give
$$\begin{aligned}
\fun{r_n\circ\Pi_n}{a}
&=
\sum_{m=0}^{n-1}\fun{r_n\circ E_{m,n}}{a_m}
=
\sum_{m=0}^{n-1}\fun{E_{m,n-1}}{a_m}
=
\fun{\Pi_{n-1}}{a}.
\end{aligned}$$
\end{proof}

\begin{prop}[Surjectivity of the finite coefficient evaluation]\label{prop:pi-surjectivity}
For every $n\in\monnat$, the map $\Pi_n$ defined in
\eqref{eq:paper-main-particle-coefficient-reconstruction} is a surjective unital $\ast$-homomorphism.
\end{prop}

\begin{proof}
Lemma \ref{lem:paper-particle-coefficient-evaluation-identities}, especially
\eqref{eq:paper-particle-coefficient-evaluation-product}, proves that $\Pi_n$ preserves products and adjoints.
If $e=\seq{e_m}{m\in\monnat}$ is given by $e_0=1$ and $e_m=0$ for $m\geq1$, then
$$\begin{aligned}
\fun{\Pi_n}{e}
&=
\fun{E_{0,n}}{1}
=
\idone_{\oa{B}_n(\sphilbfrak{h})},
\end{aligned}$$
so $\Pi_n$ is unital.

Fix $A\in\oa{B}_n(\sphilbfrak{h})$.
Lemma
\ref{lem:paper-finite-particle-coefficient-uniqueness} gives unique coefficients $a_m$, one for each integer
$m$ with $0\leq m\leq n$.
They satisfy
$$\begin{aligned}
a_m
&\in
\oa{C}_m(\sphilbfrak{h}),
\quad
A
=
\sum_{m=0}^{n}\fun{E_{m,n}}{a_m}.
\end{aligned}$$
Set $a_m=0$ for $m>n$.
Then
$a=\seq{a_m}{m\in\monnat}\in\oa{C}(\sphilbfrak{h})$ and
$$\begin{aligned}
\fun{\Pi_n}{a}
&=
\sum_{m=0}^{n}\fun{E_{m,n}}{a_m}
=
A.
\end{aligned}$$
Thus $\Pi_n$ is surjective.
\end{proof}

Proposition \ref{prop:pi-surjectivity} lifts every element of \(\oa{B}_n(\sphilbfrak{h})\) to a particle-coefficient sequence.

\begin{thm}[Particle-coefficient realization]\label{thm:paper-particle-coefficient-buchholz-isomorphism}
The space $\oa{C}_{\mathrm{b}}(\sphilbfrak{h})$ is defined in
\eqref{eq:paper-main-bounded-particle-coefficients} and has the norm
$$\begin{aligned}
\norm{a}_{\mathrm{b}}
&=
\sup_{n\in\monnat}\norm{\fun{\Pi_n}{a}}.
\end{aligned}$$
The space $\oa{C}_{\mathrm{b}}(\sphilbfrak{h})$ is a unital $\oacstar$-algebra.
The map $\Upsilon$ defined in \eqref{eq:paper-main-particle-coefficient-isomorphism} is an isometric unital
$\ast$-isomorphism.
\end{thm}

\begin{proof}
Lemma \ref{lem:paper-particle-coefficient-evaluation-identities}, in particular
\eqref{eq:paper-particle-coefficient-evaluation-product}, and the linearity of $\Pi_n$ give
$$\begin{aligned}
\norm{a\star b}_{\mathrm{b}}
&\leq
\norm{a}_{\mathrm{b}}\norm{b}_{\mathrm{b}},
&
\norm{\faadj{a}}_{\mathrm{b}}
&=
\norm{a}_{\mathrm{b}}.
\end{aligned}$$
For every $n\in\monnat$, the separating-coordinate calculations are
$$\begin{aligned}
\fun{\Pi_n}{\rbk{a\star b}\star c}
&=
\fun{\Pi_n}{a}\fun{\Pi_n}{b}\fun{\Pi_n}{c}
=
\fun{\Pi_n}{a\star\rbk{b\star c}},
\\
\fun{\Pi_n}{\faadjrbk{a\star b}}
&=
\faadjrbk{\fun{\Pi_n}{b}}\faadjrbk{\fun{\Pi_n}{a}}
=
\fun{\Pi_n}{\faadj{b}\star\faadj{a}}.
\end{aligned}$$
Finite-coefficient uniqueness applied to the $\Pi_n$-evaluation identity for
$\rbk{a\star b}\star c$ and $a\star\rbk{b\star c}$ gives, for every $r\in\monnat$,
$$\begin{aligned}
\rbk{\rbk{a\star b}\star c}_r
&=
\rbk{a\star\rbk{b\star c}}_r.
\end{aligned}$$
Thus the coefficient product is associative.
Finite-coefficient uniqueness applied to the $\Pi_n$-evaluation identity for
$\faadjrbk{a\star b}$ and $\faadj{b}\star\faadj{a}$ gives
$$\begin{aligned}
\faadjrbk{a\star b}
&=
\faadj{b}\star\faadj{a}.
\end{aligned}$$
For $r\in\monnat$ and $\alpha,\lambda\in\fldcmp$, the coefficientwise adjoint satisfies
$$\begin{aligned}
\rbk{\faadjrbk{\alpha a+\lambda b}}_r
&=
\faadjrbk{\alpha a_r+\lambda b_r}
=
\cmpconj{\alpha}\faadj{a_r}+\cmpconj{\lambda}\faadj{b_r},
&
\rbk{\faadjrbk{\faadj{a}}}_r
&=
\faadjrbk{\faadj{a_r}}
=
a_r.
\end{aligned}$$
Hence the adjoint is conjugate linear and involutive.
The bilinearity of every $\mu_{p,q}^{r}$ gives
$$\begin{aligned}
\rbk{\rbk{\alpha a+\lambda b}\star c}_r
&=
\sum_{\substack{0\leq p,q\leq r\\p+q\geq r}}
\fun{\mu_{p,q}^{r}}{\alpha a_p+\lambda b_p,c_q}
=
\alpha\rbk{a\star c}_r+\lambda\rbk{b\star c}_r,
\\ 
\rbk{a\star\rbk{\alpha b+\lambda c}}_r
&=
\sum_{\substack{0\leq p,q\leq r\\p+q\geq r}}
\fun{\mu_{p,q}^{r}}{a_p,\alpha b_q+\lambda c_q}
=
\alpha\rbk{a\star b}_r+\lambda\rbk{a\star c}_r.
\end{aligned}$$
Thus $\star$ is bilinear.
Define $e=\seq{e_m}{m\in\monnat}$ by $e_0=1$ and $e_m=0$ for $m\geq1$.
Since $\fun{E_{0,n}}{1}=1$, one has $\fun{\Pi_n}{e}=1$, and
$$\begin{aligned}
\fun{\Pi_r}{e\star a}
&=
\fun{\Pi_r}{e}\fun{\Pi_r}{a}
=
\fun{\Pi_r}{a},
&
\fun{\Pi_r}{a\star e}
&=
\fun{\Pi_r}{a}\fun{\Pi_r}{e}
=
\fun{\Pi_r}{a}.
\end{aligned}$$
Lemma \ref{lem:paper-finite-particle-coefficient-uniqueness} therefore gives
$$\begin{aligned}
\rbk{e\star a}_r
&=
a_r,
&
\rbk{a\star e}_r
&=
a_r.
\end{aligned}$$
Thus $e\star a=a=a\star e$.
Equation \eqref{eq:paper-particle-coefficient-evaluation-compatibility} proves that the map $\Upsilon$ defined
in \eqref{eq:paper-main-particle-coefficient-isomorphism} takes values in the bounded inverse limit.
Equation \eqref{eq:paper-particle-coefficient-evaluation-product} proves that $\Upsilon$ is a unital
$\ast$-homomorphism.
The $\oacstar$-identity is calculated before completeness as
$$\begin{aligned}
\norm{\faadj{a}\star a}_{\mathrm{b}}
&=
\sup_{n\in\monnat}
\norm{\faadjrbk{\fun{\Pi_n}{a}}\fun{\Pi_n}{a}}
=
\sup_{n\in\monnat}\norm{\fun{\Pi_n}{a}}^2
=
\norm{a}_{\mathrm{b}}^2.
\end{aligned}$$

If $\fun{\Upsilon}{a}=0$, then $\fun{\Pi_0}{a}=0$ gives $a_0=0$.
Assume $a_0=\cdots=a_{m-1}=0$.
The $m$th coordinate gives
$$\begin{aligned}
0
&=
\fun{\Pi_m}{a}
=
\fun{E_{m,m}}{a_m}.
\end{aligned}$$
Injectivity of $E_{m,m}$ gives $a_m=0$.
Induction proves that $\Upsilon$ is injective.

Let $A=\seq{A_n}{n\in\monnat}\in\fun{\oa{B}}{\sphilbfrak{h}}$.
Lemma \ref{lem:paper-finite-particle-coefficient-uniqueness} gives unique coefficients $a_{m,n}$ with
$$\begin{aligned}
A_n
&=
\sum_{m=0}^{n}\fun{E_{m,n}}{a_{m,n}}.
\end{aligned}$$
The relation $\fun{r_n}{A_n}=A_{n-1}$ and the connecting identities give two decompositions of $A_{n-1}$,
$$\begin{aligned}
A_{n-1}
&=
\fun{r_n}{A_n}
=
\sum_{m=0}^{n-1}\fun{E_{m,n-1}}{a_{m,n}},
\\
A_{n-1}
&=
\sum_{m=0}^{n-1}\fun{E_{m,n-1}}{a_{m,n-1}}.
\end{aligned}$$
Uniqueness gives $a_{m,n}=a_{m,n-1}$ for $0\leq m<n$.
Define $a_m=a_{m,n}$ for any $n\geq m$.
The definition is independent of $n$, and $\fun{\Pi_n}{a}=A_n$.
The boundedness of $A$ gives $a\in\oa{C}_{\mathrm{b}}(\sphilbfrak{h})$, which proves surjectivity.
The norm of $\fun{\Upsilon}{a}$ is
$$\begin{aligned}
\norm{\fun{\Upsilon}{a}}_{\fun{\oa{B}}{\sphilbfrak{h}}}
&=
\sup_{n\in\monnat}\norm{\fun{\Pi_n}{a}}
=
\norm{a}_{\mathrm{b}}.
\end{aligned}$$
The map is isometric, and completeness follows from completeness of
$\fun{\oa{B}}{\sphilbfrak{h}}$.
\end{proof}

\subsection{Surjectivity of the particle-cutoff coordinate maps}\label{surjectivity-of-the-particle-cutoff-coordinate-maps}

In this subsection, we prove that, for every \(N\in\monnat\), the coordinate map \(p_N\) defined in \eqref{eq:paper-main-coordinate-map} is surjective by constructing a bounded compatible preimage of each element of \(\oa{B}_N(\sphilbfrak{h})\).

\begin{cor}[Surjectivity of the particle-cutoff coordinate maps]\label{cor:paper-particle-cutoff-coordinate-surjectivity}
For every $N\in\monnat$, the coordinate map $p_N$ defined in
\eqref{eq:paper-main-coordinate-map} is surjective.
\end{cor}

\begin{proof}
Proposition \ref{prop:paper-gauge-invariance-decreasing-particle-cutoff-ideals} shows that every $r_n$ is a
surjective $\ast$-homomorphism.
Fix $A_N\in\oa{B}_N(\sphilbfrak{h})$ and $\varepsilon>0$.
For $0\leq n<N$, define
$$\begin{aligned}
A_n
&=
\fun{r_{n+1}\circ\cdots\circ r_N}{A_N}.
\end{aligned}$$
Norm nonincrease for $\ast$-homomorphisms gives
$\norm{A_n}
\leq
\norm{A_N}$
for $0\leq n<N$.
For $n>N$, the quotient norm formula is
$$\begin{aligned}
\norm{A_{n-1}}
&=
\inf\set{\norm{B}}
{B \in \oa{B}_n(\sphilbfrak{h}), \fun{r_n}{B} = A_{n-1}}.
\end{aligned}$$
It permits a recursive choice of $A_n\in\oa{B}_n(\sphilbfrak{h})$ satisfying
$$\begin{aligned}
\fun{r_n}{A_n}
&=
A_{n-1},
&
\norm{A_n}
&\leq
\norm{A_{n-1}}+2^{-n}\varepsilon.
\end{aligned}$$
The resulting family is compatible and satisfies
$$\begin{aligned}
\sup_{n\in\monnat}\norm{A_n}
&\leq
\norm{A_N}+\sum_{n=N+1}^{\infty}2^{-n}\varepsilon
\leq
\norm{A_N}+\varepsilon.
\end{aligned}$$
The boundedness and compatibility show that $A=\seq{A_n}{n\in\monnat}$ belongs to
$\fun{\oa{B}}{\sphilbfrak{h}}$ and
$\fun{p_N}{A}=A_N$.
\end{proof}

\subsection{Ideals of the particle-cutoff quotients}\label{sec:paper-finite-stage-ideals}

In this subsection, the connecting maps and the lower output-degree bound \eqref{eq:paper-abstract-particle-coefficient-output-degree-range} control the ideal structure of each particle-cutoff quotient. They identify the degree-tail subspaces \eqref{eq:paper-main-degree-tail-ideals} as closed two-sided ideals and yield a complete ideal chain whose successive quotients are the particle-degree coefficient algebras.

\begin{prop}[Closed degree-tail ideals]\label{prop:paper-degree-tail-ideals}
For every $n\in\monnat$ and $0\leq k\leq n+1$, the degree-tail subspace $\oa{J}_k^{(n)}$ defined in
\eqref{eq:paper-main-degree-tail-ideals} is a norm-closed two-sided ideal of $\oa{B}_n(\sphilbfrak{h})$ and satisfies
\begin{equation}
\begin{aligned}
\oa{J}_k^{(n)}
&=
\bigoplus_{m=k}^{n}\fun{E_{m,n}}{\oa{C}_m(\sphilbfrak{h})}.
\end{aligned}
\label{eq:paper-degree-tail-coefficient-decomposition}
\end{equation}
The direct sum is understood to be zero when $k=n+1$.
\end{prop}

\begin{proof}
For $k=0$ and $k=n+1$, both assertions follow from
\eqref{eq:paper-main-degree-tail-ideals} and the direct-sum identity in
\eqref{eq:paper-abstract-particle-coefficient-embedding-identities}.
Fix $1\leq k\leq n$.
If $x_j\in\oa{J}_k^{(n)}$ and $x_j\to x$ in $\oa{B}_n(\sphilbfrak{h})$, continuity of the connecting maps gives
$$\begin{aligned}
\fun{r_k\circ r_{k+1}\circ\cdots\circ r_n}{x}
&=
\lim_{j\to\infty}
\fun{r_k\circ r_{k+1}\circ\cdots\circ r_n}{x_j}
=
0.
\end{aligned}$$
Thus $\oa{J}_k^{(n)}$ is norm closed.
For $0\leq m\leq n$, the connecting identities in
\eqref{eq:paper-abstract-particle-coefficient-connecting-identities} give
$$\begin{aligned}
\fun{r_k\circ r_{k+1}\circ\cdots\circ r_n\circ E_{m,n}}{a_m}
&=
\begin{cases}
\fun{E_{m,k-1}}{a_m},&0\leq m<k,\\
0,&k\leq m\leq n.
\end{cases}
\end{aligned}$$
The composite restriction is
$$\begin{aligned}
\fun{r_k\circ r_{k+1}\circ\cdots\circ r_n}{
\sum_{m=0}^{n}\fun{E_{m,n}}{a_m}}
&=
\sum_{m=0}^{k-1}\fun{E_{m,k-1}}{a_m}.
\end{aligned}$$
Lemma \ref{lem:paper-finite-particle-coefficient-uniqueness} makes the right-hand side zero exactly when
$a_0=\cdots=a_{k-1}=0$, proving \eqref{eq:paper-degree-tail-coefficient-decomposition}.
For elements $x$ and $y$ of the forms
$$\begin{aligned}
x
&=
\sum_{p=k}^{n}\fun{E_{p,n}}{a_p},
&
y
&=
\sum_{q=0}^{n}\fun{E_{q,n}}{b_q},
\end{aligned}$$
the lower degree bound in \eqref{eq:paper-abstract-particle-coefficient-output-degree-range} implies that every
output degree $s$ in either product satisfies $s\geq\max\setone{p,q}\geq p\geq k$.
Consequently, the products are
$$\begin{aligned}
xy
&=
\sum_{p=k}^{n}\sum_{q=0}^{n}
\sum_{s=\max\setone{p,q}}^{\min\setone{p+q,n}}
\fun{E_{s,n}}{\fun{\mu_{p,q}^{s}}{a_p,b_q}}
\in
\oa{J}_k^{(n)},
\\
yx
&=
\sum_{q=0}^{n}\sum_{p=k}^{n}
\sum_{s=\max\setone{q,p}}^{\min\setone{q+p,n}}
\fun{E_{s,n}}{\fun{\mu_{q,p}^{s}}{b_q,a_p}}
\in
\oa{J}_k^{(n)}.
\end{aligned}$$
Thus $\oa{J}_k^{(n)}$ is a two-sided ideal.
\end{proof}

\begin{thm}[Ideal chain of a particle-cutoff quotient]\label{thm:paper-finite-stage-ideal-chain}
Under Definition \ref{def:paper-abstract-particle-coefficient-reconstruction-data}, the closed two-sided ideals of
$\oa{B}_n(\sphilbfrak{h})$ are exactly
\begin{equation}
\begin{aligned}
0
&=
\oa{J}_{n+1}^{(n)}
\subset
\oa{J}_{n}^{(n)}
\subset
\cdots
\subset
\oa{J}_{1}^{(n)}
\subset
\oa{J}_{0}^{(n)}
=
\oa{B}_n(\sphilbfrak{h}).
\end{aligned}
\label{eq:paper-finite-stage-ideal-chain}
\end{equation}
For $0\leq k\leq n$, the quotient is canonically isomorphic to the degree-$k$ coefficient algebra,
$$\begin{aligned}
\setquot{\oa{J}_k^{(n)}}{\oa{J}_{k+1}^{(n)}}
&\simeq
\oa{C}_k(\sphilbfrak{h}).
\end{aligned}$$
\end{thm}

\begin{proof}
Proposition \ref{prop:paper-degree-tail-ideals} shows that every $\oa{J}_k^{(n)}$ in
\eqref{eq:paper-finite-stage-ideal-chain} is a closed two-sided ideal.
Take two elements $a_k$ and $b_k$ in the degree-$k$ coefficient algebra.
Their degree-$k$ product is
$$\begin{aligned}
\fun{\mu_{k,k}^{k}}{a_k,b_k}
&=
a_kb_k.
\end{aligned}$$
All other output degrees exceed $k$, and the exact remainder calculation is
\begin{equation}
\begin{aligned}
\fun{E_{k,n}}{a_k}\fun{E_{k,n}}{b_k}
-
\fun{E_{k,n}}{a_kb_k}
&=
\sum_{r=k+1}^{\min\setone{2k,n}}
\fun{E_{r,n}}{\fun{\mu_{k,k}^{r}}{a_k,b_k}}
\in
\oa{J}_{k+1}^{(n)}.
\end{aligned}
\label{eq:paper-degree-k-product-remainder}
\end{equation}
The map
$$\begin{aligned}
\chi_{k,n}
\colon
\oa{C}_k(\sphilbfrak{h})
\to
\setquot{\oa{J}_k^{(n)}}{\oa{J}_{k+1}^{(n)}};\quad
\fun{\chi_{k,n}}{a_k}
&=
\fun{E_{k,n}}{a_k}+\oa{J}_{k+1}^{(n)}
\end{aligned}$$
is multiplicative by \eqref{eq:paper-degree-k-product-remainder} and preserves adjoints because $E_{k,n}$ does.
Multiplicativity and adjoint preservation are expressed by
$$\begin{aligned}
\fun{\chi_{k,n}}{a_kb_k}
&=
\fun{\chi_{k,n}}{a_k}\fun{\chi_{k,n}}{b_k},
&
\fun{\chi_{k,n}}{\faadj{a_k}}
&=
\faadjrbk{\fun{\chi_{k,n}}{a_k}}.
\end{aligned}$$
If $\fun{\chi_{k,n}}{a_k}=0$, the direct-sum decomposition
\eqref{eq:paper-degree-tail-coefficient-decomposition} gives
$$\begin{aligned}
\fun{E_{k,n}}{a_k}
&\in
\fun{E_{k,n}}{\oa{C}_k(\sphilbfrak{h})}
\cap
\oa{J}_{k+1}^{(n)}
=
0.
\end{aligned}$$
Injectivity of $E_{k,n}$ gives $a_k=0$.
By \eqref{eq:paper-degree-tail-coefficient-decomposition}, every
$x=\sum_{m=k}^{n}\fun{E_{m,n}}{a_m}\in\oa{J}_k^{(n)}$ satisfies
$$\begin{aligned}
x+\oa{J}_{k+1}^{(n)}
&=
\fun{E_{k,n}}{a_k}+\oa{J}_{k+1}^{(n)}
=
\fun{\chi_{k,n}}{a_k},
\end{aligned}$$
which proves surjectivity and the quotient formula.

The ideal $\oa{J}_n^{(n)}=\Ker r_n$ is simple and essential by
Definition \ref{def:paper-particle-coefficient-embedding-data}.
For $n=0$, simplicity gives the asserted two-element ideal chain.
Assume the result at degree $n-1$.
Let $\oa{L}\subset\oa{B}_n(\sphilbfrak{h})$ be a nonzero closed two-sided ideal.
Essentiality and simplicity give
$$\begin{aligned}
0
&\neq
\oa{L}\cap\oa{J}_n^{(n)}
\subset
\oa{J}_n^{(n)},
&
\oa{L}\cap\oa{J}_n^{(n)}
&=
\oa{J}_n^{(n)}.
\end{aligned}$$
The equality implies $\oa{J}_n^{(n)}\subset\oa{L}$.
Since $\Ker r_n\subset\oa{L}$, the quotient correspondence makes
$\fun{r_n}{\oa{L}}$ a closed ideal of $\oa{B}_{n-1}(\sphilbfrak{h})$.
Induction in $n$ makes this image one of the tails $\oa{J}_k^{(n-1)}$.
For $x=\sum_{m=0}^{n}\fun{E_{m,n}}{a_m}$, one has
$$\begin{aligned}
\fun{r_n}{x}
&=
\sum_{m=0}^{n-1}\fun{E_{m,n-1}}{a_m},
\\
\fun{r_n}{x}\in\oa{J}_k^{(n-1)}
&\iff
a_0=\cdots=a_{k-1}=0
\iff
x\in\oa{J}_k^{(n)}.
\end{aligned}$$
The inverse-image calculation is
$$\begin{aligned}
\fun{r_n^{-1}}{\oa{J}_k^{(n-1)}}
&=
\oa{J}_k^{(n)}.
\end{aligned}$$
The final inverse-image calculation is
$$\begin{aligned}
\fun{r_n^{-1}}{\fun{r_n}{\oa{L}}}
&=
\oa{L}+\Ker r_n
=
\oa{L}.
\end{aligned}$$
Consequently $\oa{L}$ is one of the ideals in \eqref{eq:paper-finite-stage-ideal-chain}.
\end{proof}

\section{Intrinsic finite-mode traces and closed paths}\label{sec:intrinsic-finite-volume}

In this section, the principal application is the construction of finite-volume equilibrium states from the standard trace and, after a dynamics has been specified, the identification of the resulting state as a KMS state. The arguments are nevertheless separated according to the structure that they actually use. This separation determines which conclusions remain available for the continuous systems treated by Buchholz and which conclusions remain valid for nonequilibrium states, including states used in return-to-equilibrium problems.

The first level uses only the norm and functional calculus of a block \(\oacstar\)-algebra. The second level uses the finite-dimensional full matrix structure and expands products with fixed minimal projections before applying any functional. This operator identity is the appropriate starting point for a matrix-based path expansion, because subsequent matrix coefficients can be evaluated in a chosen basis without first taking a trace. The third level applies an arbitrary bounded linear functional. Positivity and normalization enter only when that functional is a state. The fourth level uses the standard trace, its trace norm, and its cyclicity. Trace-class estimates, partition functions, and the passage from open matrix paths to closed paths belong only to this fourth level.

The word KMS has an additional dynamical meaning in this hierarchy. For \(\sminvtemperature>0\) and a strongly continuous one-parameter group \(\alpha\) of automorphisms of a \(\oacstar\)-algebra, an element \(A\) is entire analytic when the orbit map \(t\mapsto\fun{\alpha_t}{A}\) extends to an entire algebra-valued function on \(\fldcmp\). A state \(\oastate[\psi]\) is a \(\sminvtemperature\)-KMS state when \begin{equation}
\begin{aligned}
\fun{\oastate[\psi]}{AB}
&=
\fun{\oastate[\psi]}{B\fun{\alpha_{\imunit\sminvtemperature}}{A}}
\end{aligned}
\label{eq:paper-kms-boundary-condition}
\end{equation} for every entire analytic element \(A\) and every \(B\). None of the norm-level, matrix-level, or bounded-functional statements below requires this boundary condition. It is needed only when an evaluated functional is asserted to describe equilibrium for the chosen dynamics. The trace Gibbs functional constructed below supplies that specialization, whereas the preceding levels remain available for general states.

\subsection{\texorpdfstring{Norm-level statements in block \(\oacstar\)-algebras}{Norm-level statements in block \textbackslash oacstar-algebras}}\label{norm-level-statements-in-block-oacstar-algebras}

In this subsection, the norm-level statements require neither a trace nor a state. They isolate the conclusions that can be transported to another block \(\oacstar\)-algebra after its sectorwise functional calculus and norm estimates have been verified.

Fix a nonzero finite-dimensional complex subspace \(\sphilbfrak{k}\subset\sphilbfrak{h}\). For every \(n\in\monnat\), let \(\physham_{n}\in e_{\sphilbfrak{k},n}\oa{K}(\sphilbfrak{k})^{\gamma}e_{\sphilbfrak{k},n}\) be self-adjoint. The block family \(\seq{\physham_{n}}{n\in\monnat}\) is affiliated with \(\oa{K}(\sphilbfrak{k})^{\gamma}\) when its two resolvent families satisfy \[\begin{aligned}
\seq{\rbk{\physham_{n}+\imunit}^{-1}}{n\in\monnat},
\seq{\rbk{\physham_{n}-\imunit}^{-1}}{n\in\monnat}
&\in
\oa{K}(\sphilbfrak{k})^{\gamma}.
\end{aligned}\] This condition does not identify the generally unbounded direct sum with an element of the \(\oacstar\)-algebra. It places its two resolvents in the fixed-point compact ideal.

\begin{prop}[Affiliation and heat kernels in the compact ideal from superstability]\label{prop:paper-compact-heat-kernel}\label{prop:intrinsic-compact-gibbs-semigroup}
Assume that the Hamiltonian family is superstable in the sense of
\eqref{eq:paper-main-superstability-bound}.
Then it is affiliated with $\oa{K}(\sphilbfrak{k})^{\gamma}$, and, for every $\sminvtemperature>0$,
$$\begin{aligned}
\seq{\napiernum^{-\sminvtemperature \physham_{n}}}{n\in\monnat}
&\in
\oa{K}(\sphilbfrak{k})^{\gamma}.
\end{aligned}$$
\end{prop}

\begin{proof}
Choose constants $a>0$ and $b,c\geq0$ witnessing
\eqref{eq:paper-main-superstability-bound}.
For all sufficiently large $n$, the scalar $an^2-bn-c$ is nonnegative.
Continuous functional calculus gives
$$\begin{aligned}
\norm{\rbk{\physham_{n}\pm\imunit}^{-1}}
&\leq
\frac{1}{\sqrt{\rbk{an^2-bn-c}^{2}+1}}
\to
0
&&\rbk{n\to\infty},
\\
\norm{\napiernum^{-\sminvtemperature \physham_{n}}}
&\leq
\napiernum^{-\sminvtemperature\rbk{an^2-bn-c}}
\to
0
&&\rbk{n\to\infty}.
\end{aligned}$$
Both conclusions follow from the unique norm-convergent decomposition
\eqref{eq:paper-main-finite-mode-block-algebra}.
\end{proof}

The Hamiltonian family defines a one-parameter group of \(\ast\)-automorphisms on \(\oa{K}(\sphilbfrak{k})^{\gamma}\) by \begin{equation}
\begin{aligned}
\alpha_t
\colon
\oa{K}(\sphilbfrak{k})^{\gamma}
\to
\oa{K}(\sphilbfrak{k})^{\gamma};\quad
\fun{\alpha_t}{A}
&=
\seq{
\napiernum^{\imunit t\physham_n}
A_n
\napiernum^{-\imunit t\physham_n}
}{n\in\monnat}
&&\rbk{t\in\fldreal}.
\end{aligned}
\label{eq:paper-intrinsic-block-dynamics}
\end{equation} The group in \eqref{eq:paper-intrinsic-block-dynamics} is strongly continuous. Indeed, for \(A=\seq{A_n}{n\in\monnat}\in\oa{K}(\sphilbfrak{k})^{\gamma}\) and \(A_{\leq N}=\sum_{n=0}^{N}A_n\), the isometry of every \(\alpha_t\) in \eqref{eq:paper-intrinsic-block-dynamics} gives \[\begin{aligned}
\norm{\fun{\alpha_t}{A}-A}
&\leq
2\norm{A-A_{\leq N}}
+
\norm{\fun{\alpha_t}{A_{\leq N}}-A_{\leq N}}.
\end{aligned}\] For fixed \(N\), the last term tends to zero as \(t\to0\) because it is the maximum of finitely many continuous block orbits. Equation \eqref{eq:paper-main-finite-mode-block-algebra} gives \(\norm{A-A_{\leq N}}\to0\) as \(N\to\infty\), and the norm bound for \(\fun{\alpha_t}{A}-A\) proves strong continuity.

The following criterion separates sectorwise convergence from convergence in the block direct sum.

\begin{prop}[Norm convergence in a block direct sum]\label{prop:paper-block-direct-sum-norm-convergence}
For $n\in\monnat$, let $\oa{A}_n$ be a $\oacstar$-algebra, and let
$X_{M,n},X_n\in\oa{A}_n$ for $M\in\semigrposint$.
Assume that $\norm{X_{M,n}-X_n}\to0$ as $M\to\infty$ for every fixed $n$.
Assume also that there is a sequence $\seq{q_n}{n\in\monnat}$ in
$\closedinterval{0}{\infty}$ satisfying $q_n\to0$ and
$$\begin{aligned}
\norm{X_{M,n}}+\norm{X_n}
&\leq
q_n
&&\rbk{M\in\semigrposint,\ n\in\monnat}.
\end{aligned}$$
Then $\seq{X_{M,n}}{n\in\monnat}$ and $\seq{X_n}{n\in\monnat}$ belong to
$\bigoplus_{n\in\monnat}\oa{A}_n$, and
$$\begin{aligned}
\lim_{M\to\infty}
\sup_{n\in\monnat}\norm{X_{M,n}-X_n}
&=
0.
\end{aligned}$$
\end{prop}

\begin{proof}
Fix $\varepsilon>0$.
Choose $N\in\monnat$ such that $q_n<\varepsilon$ for every $n>N$.
The hypotheses give
$$\begin{aligned}
\sup_{n>N}\norm{X_{M,n}-X_n}
&\leq
\sup_{n>N}q_n
\leq
\varepsilon
\end{aligned}$$
for every $M\in\semigrposint$.
Sectorwise convergence on the finite set $\set{n\in\monnat}{0\leq n\leq N}$ gives an integer $M_0$ such that
$$\begin{aligned}
\max_{0\leq n\leq N}\norm{X_{M,n}-X_n}
&<
\varepsilon
\end{aligned}$$
for every $M\geq M_0$.
The two estimates prove the asserted supremum-norm convergence, and the bound by $q_n$ proves membership in
the direct sum.
\end{proof}

For every \(n\in\monnat\), let \(T_n\) and \(V_n\) be self-adjoint elements of the \(n\)th corner, set \(\physham_n=T_n+V_n\), and, for \(\sminvtemperature>0\) and \(M\in\semigrposint\), define \begin{equation}
\begin{aligned}
G_{\sminvtemperature,M,n}
&=
\rbk{
\napiernum^{-\frac{\sminvtemperature}{M}V_n}
\napiernum^{-\frac{\sminvtemperature}{M}T_n}
}^{M},
&
G_{\sminvtemperature,M}
&=
\seq{G_{\sminvtemperature,M,n}}{n\in\monnat}.
\end{aligned}
\label{eq:paper-trotter-product-family}
\end{equation}

\begin{prop}[Norm product formula for bounded $\oacstar$-algebra elements]\label{prop:paper-finite-corner-norm-product-formula}
Let $\oa{A}$ be a unital $\oacstar$-algebra, let $T,V\in\oa{A}$ be self-adjoint, and fix $\sminvtemperature>0$.
Then
$$\begin{aligned}
\lim_{M\to\infty}
\norm{
\rbk{
\napiernum^{-\frac{\sminvtemperature}{M}V}
\napiernum^{-\frac{\sminvtemperature}{M}T}
}^{M}
-
\napiernum^{-\sminvtemperature\rbk{T+V}}}
&=
0.
\end{aligned}$$
In particular, the products in \eqref{eq:paper-trotter-product-family} satisfy this conclusion on every fixed
particle block.
\end{prop}

\begin{proof}
Set $X=-\sminvtemperature V$ and $Y=-\sminvtemperature T$.
For $Q\in\oa{A}$ and $M\in\semigrposint$, the exponential series gives
$$\begin{aligned}
\norm{
\napiernum^{\frac{1}{M}Q}
-
1
-
\frac{1}{M}Q}
&\leq
\frac{\norm{Q}^{2}}{2M^{2}}
\napiernum^{\frac{1}{M}\norm{Q}}.
\end{aligned}$$
Define the exponential remainder by
$$\begin{aligned}
R_{M,Q}
&=
\napiernum^{\frac{1}{M}Q}
-
1
-
\frac{1}{M}Q.
\end{aligned}$$
The one-step difference has the exact expansion
$$\begin{aligned}
&\napiernum^{\frac{1}{M}X}
\napiernum^{\frac{1}{M}Y}
-
\napiernum^{\frac{1}{M}\rbk{X+Y}}
\\ 
&=
R_{M,X}
+
R_{M,Y}
-
R_{M,X+Y}
+
\frac{1}{M^{2}}XY
+
\frac{1}{M}XR_{M,Y}
+
\frac{1}{M}R_{M,X}Y
+
R_{M,X}R_{M,Y}.
\end{aligned}$$
For $M\geq1$, the remainder estimate applied to this expansion gives
$$\begin{aligned}
\norm{
\napiernum^{\frac{1}{M}X}
\napiernum^{\frac{1}{M}Y}
-
\napiernum^{\frac{1}{M}\rbk{X+Y}}}
&\leq
\frac{C_{X,Y}}{M^{2}},
\end{aligned}$$
where the $M$-independent constant is
$$\begin{aligned}
C_{X,Y}
&=
\frac{\norm{X}^{2}}{2}\napiernum^{\norm{X}}
+
\frac{\norm{Y}^{2}}{2}\napiernum^{\norm{Y}}
+
\frac{\norm{X+Y}^{2}}{2}\napiernum^{\norm{X+Y}}
+
\norm{X}\norm{Y}
\\
&\quad
+
\frac{\norm{X}\norm{Y}^{2}}{2}\napiernum^{\norm{Y}}
+
\frac{\norm{X}^{2}\norm{Y}}{2}\napiernum^{\norm{X}}
+
\frac{\norm{X}^{2}\norm{Y}^{2}}{4}
\napiernum^{\norm{X}+\norm{Y}}.
\end{aligned}$$
For
$A_M=\napiernum^{\frac{1}{M}X}\napiernum^{\frac{1}{M}Y}$ and
$B_M=\napiernum^{\frac{1}{M}\rbk{X+Y}}$, the telescopic identity gives
$$\begin{aligned}
A_M^{M}-B_M^{M}
&=
\sum_{j=0}^{M-1}
A_M^{M-1-j}
\rbk{A_M-B_M}
B_M^{j}.
\end{aligned}$$
The exponential norm bounds and the one-step estimate imply
$$\begin{aligned}
\norm{A_M^{M}-B_M^{M}}
&\leq
\frac{C_{X,Y}}{M}
\napiernum^{\norm{X}+\norm{Y}+\norm{X+Y}}
\to
0.
\end{aligned}$$
Substitution of $X=-\sminvtemperature V$ and $Y=-\sminvtemperature T$ proves the assertion.
\end{proof}

\begin{cor}[A finite-dimensional norm criterion over all particle blocks]\label{cor:paper-finite-mode-cstar-product-formula}
Assume the potential superstability \eqref{eq:paper-main-corner-potential-superstability} and the linear
kinetic norm bound \eqref{eq:paper-main-linear-kinetic-bound}.
Then $G_{\sminvtemperature,M}$ and $\seq{\napiernum^{-\sminvtemperature\physham_n}}{n\in\monnat}$ belong to
$\oa{K}(\sphilbfrak{k})^{\gamma}$ for every $M\in\semigrposint$, and
$$\begin{aligned}
\lim_{M\to\infty}
\norm{
G_{\sminvtemperature,M}
-
\seq{\napiernum^{-\sminvtemperature\physham_n}}{n\in\monnat}}
&=
0.
\end{aligned}$$
\end{cor}

\begin{proof}
Functional calculus and submultiplicativity give
$$\begin{aligned}
\norm{G_{\sminvtemperature,M,n}}
&\leq
\napiernum^{-\sminvtemperature\rbk{an^2-bn-c}+\sminvtemperature\rbk{b_Tn+c_T}},
\\
\norm{\napiernum^{-\sminvtemperature\physham_n}}
&\leq
\napiernum^{-\sminvtemperature\rbk{an^2-\rbk{b+b_T}n-\rbk{c+c_T}}}.
\end{aligned}$$
The sum of the two scalar bounds tends to zero as $n\to\infty$ and is independent of $M$.
Proposition \ref{prop:paper-finite-corner-norm-product-formula} supplies sectorwise convergence, and Proposition
\ref{prop:paper-block-direct-sum-norm-convergence} gives the conclusion.
\end{proof}

\subsection{Matrix identities before applying a functional}\label{matrix-identities-before-applying-a-functional}

In this subsection, the full matrix structure supplies finite algebraic expansions independently of every trace, state, or KMS condition. These identities are the level at which one can choose a basis and evaluate matrix entries in a subsequent path-integral construction.

For every \(n\in\monnat\), fix a finite set \(\Omega_n\) and mutually orthogonal minimal projections \[\begin{aligned}
\seq{p_{n,\omega}}{\omega\in\Omega_n}
&\subset
e_{\sphilbfrak{k},n}\oa{K}(\sphilbfrak{k})^{\gamma}e_{\sphilbfrak{k},n}
\end{aligned}\] whose sum is the corner unit: \begin{equation}
\begin{aligned}
\sum_{\omega\in\Omega_n}p_{n,\omega}
&=
e_{\sphilbfrak{k},n}.
\end{aligned}
\label{eq:paper-minimal-projection-resolution}
\end{equation} Products indexed by \(j\) below are ordered from left to right with increasing \(j\).

\begin{prop}[Open-path expansion in a matrix corner]\label{prop:paper-matrix-open-path-expansion}
Fix $n\in\monnat$ and $M\in\semigrposint$.
For elements $A_{0,n},\ldots,A_{M-1,n}$ of the $n$th corner, one has
\begin{equation}
\begin{aligned}
\prod_{j=0}^{M-1}A_{j,n}
&=
\sum_{\rbk{\omega_0,\ldots,\omega_M}\in\Omega_n^{M+1}}
p_{n,\omega_0}A_{0,n}p_{n,\omega_1}
A_{1,n}\cdots
p_{n,\omega_{M-1}}A_{M-1,n}p_{n,\omega_M}.
\end{aligned}
\label{eq:paper-matrix-open-path-expansion}
\end{equation}
\end{prop}

\begin{proof}
Inserting the resolution \eqref{eq:paper-minimal-projection-resolution} before the first factor, between every
pair of adjacent factors, and after the last factor gives
$$\begin{aligned}
\prod_{j=0}^{M-1}A_{j,n}
&=
\rbk{\sum_{\omega_0\in\Omega_n}p_{n,\omega_0}}
A_{0,n}
\rbk{\sum_{\omega_1\in\Omega_n}p_{n,\omega_1}}
\cdots
A_{M-1,n}
\rbk{\sum_{\omega_M\in\Omega_n}p_{n,\omega_M}}
\\
&=
\sum_{\rbk{\omega_0,\ldots,\omega_M}\in\Omega_n^{M+1}}
p_{n,\omega_0}A_{0,n}p_{n,\omega_1}
A_{1,n}\cdots
p_{n,\omega_{M-1}}A_{M-1,n}p_{n,\omega_M}.
\end{aligned}$$
This is \eqref{eq:paper-matrix-open-path-expansion}.
\end{proof}

Assume that \(V_n\) is diagonal in the fixed minimal-projection resolution, and define \begin{equation}
\begin{aligned}
v_n
\colon
\Omega_n
\to
\fldreal;\quad
V_np_{n,\omega}
&=
\fun{v_n}{\omega}p_{n,\omega}
&&\rbk{\omega\in\Omega_n}.
\end{aligned}
\label{eq:paper-minimal-projection-diagonalization}
\end{equation} Functional calculus and Proposition \ref{prop:paper-matrix-open-path-expansion} give the operator identity \begin{equation}
\begin{aligned}
G_{\sminvtemperature,M,n}
&=
\sum_{\rbk{\omega_0,\ldots,\omega_M}\in\Omega_n^{M+1}}
\prod_{j=0}^{M-1}
\napiernum^{-\frac{\sminvtemperature}{M}\fun{v_n}{\omega_j}}
p_{n,\omega_j}
\napiernum^{-\frac{\sminvtemperature}{M}T_n}
p_{n,\omega_{j+1}}.
\end{aligned}
\label{eq:paper-matrix-trotter-open-path-expansion}
\end{equation}

\subsection{Bounded linear functionals and states}\label{bounded-linear-functionals-and-states}

In this subsection, a bounded linear functional is applied to the preceding norm limits and finite matrix identities without being positive, tracial, or a KMS state. This level is therefore available for nonequilibrium states as well as for arbitrary bounded linear functionals.

\begin{cor}[Bounded-functional evaluation of matrix paths]\label{cor:paper-bounded-functional-path-expansion}
Fix $n\in\monnat$, $M\in\semigrposint$, and a bounded linear functional $\ell_n$ on the $n$th corner.
Under the diagonalization condition \eqref{eq:paper-minimal-projection-diagonalization}, one has
$$\begin{aligned}
\fun{\ell_n}{G_{\sminvtemperature,M,n}}
&=
\sum_{\rbk{\omega_0,\ldots,\omega_M}\in\Omega_n^{M+1}}
\fun{\ell_n}{
\prod_{j=0}^{M-1}
\napiernum^{-\frac{\sminvtemperature}{M}\fun{v_n}{\omega_j}}
p_{n,\omega_j}
\napiernum^{-\frac{\sminvtemperature}{M}T_n}
p_{n,\omega_{j+1}}}.
\end{aligned}$$
\end{cor}

\begin{proof}
The sum in \eqref{eq:paper-matrix-trotter-open-path-expansion} is finite, so linearity gives
$$\begin{aligned}
\fun{\ell_n}{G_{\sminvtemperature,M,n}}
&=
\sum_{\rbk{\omega_0,\ldots,\omega_M}\in\Omega_n^{M+1}}
\fun{\ell_n}{
\prod_{j=0}^{M-1}
\napiernum^{-\frac{\sminvtemperature}{M}\fun{v_n}{\omega_j}}
p_{n,\omega_j}
\napiernum^{-\frac{\sminvtemperature}{M}T_n}
p_{n,\omega_{j+1}}}.
\end{aligned}$$
This is the asserted functional identity.
\end{proof}

For a bounded linear functional \(\ell\) on \(\oa{K}(\sphilbfrak{k})^{\gamma}\), the norm convergence in Corollary \ref{cor:paper-finite-mode-cstar-product-formula} gives \begin{equation}
\begin{aligned}
\lim_{M\to\infty}
\fun{\ell}{G_{\sminvtemperature,M}}
&=
\fun{\ell}{\seq{\napiernum^{-\sminvtemperature\physham_n}}{n\in\monnat}}.
\end{aligned}
\label{eq:paper-bounded-functional-product-limit}
\end{equation} Indeed, the absolute value of the difference is bounded by \(\norm{\ell}\) times the norm difference in that corollary: \[\begin{aligned}
\abs{
\fun{\ell}{G_{\sminvtemperature,M}}
-
\fun{\ell}{\seq{\napiernum^{-\sminvtemperature\physham_n}}{n\in\monnat}}}
&\leq
\norm{\ell}
\norm{
G_{\sminvtemperature,M}
-
\seq{\napiernum^{-\sminvtemperature\physham_n}}{n\in\monnat}}
\to
0.
\end{aligned}\] If \(\ell=\oastate[\psi]\) is a state, this conclusion additionally records positivity and normalization of the functional, but statehood neither closes the paths in Corollary \ref{cor:paper-bounded-functional-path-expansion} nor makes its individual summands nonnegative. Those two conclusions require additional structure below.

\subsection{Trace-dependent equilibrium formulas}\label{trace-dependent-equilibrium-formulas}

In this subsection, the remaining results use the standard trace through at least one of its special properties: its normalization on minimal projections, its trace norm, its cyclicity, or the trace-class summation over all particle sectors. They specialize the preceding hierarchy to the equilibrium functional that is central to this paper.

\subsubsection{Standard traces on particle blocks}\label{standard-traces-on-particle-blocks}

In this subsubsection, the standard trace on each particle block extends to a positive trace on \(\oa{K}(\sphilbfrak{k})^{\gamma}\) and to a linear functional on its trace-summable ideal.

The block decomposition is \eqref{eq:paper-main-finite-mode-block-algebra}, and each corner \(e_{\sphilbfrak{k},n}\oa{K}(\sphilbfrak{k})^{\gamma}e_{\sphilbfrak{k},n}\) is a nonzero finite-dimensional full matrix algebra. The corner trace \(\trace_{n}\) and the extended positive trace \(\trace\) are those defined by \eqref{eq:paper-standard-corner-trace} and \eqref{eq:paper-intrinsic-standard-trace}, respectively. Equation \eqref{eq:paper-standard-corner-trace} defines the ordinary, unnormalized matrix trace on each corner uniquely. Throughout this subsubsection, assume the polynomial corner-dimension growth \eqref{eq:paper-corner-dimension-growth} with constants \(C>0\) and \(s\geq0\).

\begin{rem}[Origin of the dimension bound]
The bound \eqref{eq:paper-corner-dimension-growth} controls the particle-sector multiplicities in the
trace-summability estimate and depends on neither graph geometry nor a particular Hamiltonian.
It yields trace summability when combined with Hamiltonian superstability.
For $d\in\semigrposint$, the standard bosonic realization on a $d$-dimensional complex Hilbert
space $\sphilbfrak{k}$ has particle blocks satisfying
\begin{equation}
\begin{aligned}
\sqfun{\trace_{n}}{e_{\sphilbfrak{k},n}}
&=
\dim_{\fldcmp}\sphilbfrak{k}^{\otimes_{\txtsym}n}
=
\binom{n+d-1}{d-1}.
\end{aligned}
\label{eq:paper-corner-dimension-graph-independent-realization}
\end{equation}
If $d=1$, the binomial coefficient in
\eqref{eq:paper-corner-dimension-graph-independent-realization} equals $1$ for every $n\in\monnat$.
Assume $d\geq2$.
For $n\in\monnat$ and $1\leq j\leq d-1$, the factorwise bound and the factorial formula give
$$\begin{aligned}
n+j
&\leq
n+d
\leq
d(1+n),
\\ 
\binom{n+d-1}{d-1}
&=
\frac{\prod_{j=1}^{d-1}(n+j)}{(d-1)!}
\leq
\frac{\prod_{j=1}^{d-1}d(1+n)}{(d-1)!}
=
\frac{d^{d-1}}{(d-1)!}(1+n)^{d-1}.
\end{aligned}$$
The same inequality holds for $d=1$, because both sides equal $1$.
This proves \eqref{eq:paper-corner-dimension-growth} with
$s=d-1$ and
$C=\frac{d^{d-1}}{\rbk{d-1}!}$.
This calculation verifies polynomial corner growth for the standard bosonic blocks without introducing graph
data.
\end{rem}

For \(A=\sum_{n\in\monnat}A_n\in\rbk{\oa{K}(\sphilbfrak{k})^{\gamma}}_{+}\), the extended positive trace is the sum in \eqref{eq:paper-intrinsic-standard-trace}. The series denotes the supremum of its finite partial sums. For \(B\in e_{\sphilbfrak{k},n}\oa{K}(\sphilbfrak{k})^{\gamma}e_{\sphilbfrak{k},n}\), define its trace norm by \[\begin{aligned}
\norm{B}_{1,\trace_{n}}
&=
\sqfun{\trace_{n}}{\abs{B}}.
\end{aligned}\] If \(\sum_{n\in\monnat}\sqfun{\trace_{n}}{\abs{A_n}}
<
\infty\), define \[\begin{aligned}
\norm{A}_{1,\trace}
&=
\sum_{n\in\monnat}
\sqfun{\trace_{n}}{\abs{A_n}}
=
\sum_{n\in\monnat}
\norm{A_n}_{1,\trace_{n}}.
\end{aligned}\] For every block family \(A=\seq{A_n}{n\in\monnat}\) satisfying \(\sum_{n\in\monnat}\sqfun{\trace_{n}}{\abs{A_n}}<\infty\), define \[\begin{aligned}
\sqfun{\trace}{A}
&=
\sum_{n\in\monnat}\sqfun{\trace_{n}}{A_n}.
\end{aligned}\] The trace series is absolutely convergent. An element satisfying this summability condition is called trace summable, and the absolutely convergent trace series is the linear extension of the positive trace. Every sum \(\sum_{n\in\monnat}A_n\) asserted to belong to \(\oa{K}(\sphilbfrak{k})^{\gamma}\) is understood as a norm limit. The polynomial bound in \eqref{eq:paper-corner-dimension-growth} controls the multiplicity of each particle sector but does not by itself imply trace summability.

\subsubsection{Trace summability and the Gibbs state}\label{trace-summability-and-the-gibbs-state}

In this subsubsection, the heat-kernel family in Proposition \ref{prop:paper-compact-heat-kernel} becomes trace summable under the polynomial corner-dimension and Hamiltonian superstability bounds.

\begin{cor}[Trace-summable heat kernels]\label{cor:paper-trace-summable-heat-kernel}
Assume the polynomial corner-dimension bound \eqref{eq:paper-corner-dimension-growth} and the Hamiltonian
superstability bound \eqref{eq:paper-main-superstability-bound}.
Then the heat-kernel family satisfies \eqref{eq:paper-main-compact-heat-kernel} and
\eqref{eq:paper-main-heat-kernel-trace-summability} for every $\sminvtemperature>0$.
\end{cor}

\begin{proof}
Proposition \ref{prop:intrinsic-compact-gibbs-semigroup} gives membership in
$\oa{K}(\sphilbfrak{k})^{\gamma}$.
Choose constants $a>0$ and $b,c\geq0$ witnessing
\eqref{eq:paper-main-superstability-bound}.
The corner-dimension and superstability bounds give
$$\begin{aligned}
\norm{
\seq{\napiernum^{-\sminvtemperature \physham_{n}}}{n\in\monnat}
}_{1,\trace}
=
\sum_{n\in\monnat}
\sqfun{\trace_{n}}{\napiernum^{-\sminvtemperature \physham_{n}}}
\leq
\sum_{n\in\monnat}
C(1+n)^{s}
\napiernum^{-\sminvtemperature\rbk{an^2-bn-c}}
<
\infty.
\end{aligned}$$
This is \eqref{eq:paper-main-heat-kernel-trace-summability}.
\end{proof}

Specializing the bounded block product in \eqref{eq:paper-main-finite-mode-algebra} to \(\sphilbfrak{f}=\sphilbfrak{k}\) gives \begin{equation}
\begin{aligned}
\oa{M}(\oa{K}(\sphilbfrak{k})^{\gamma})
&=
\set{\seq{A_n}{n\in\monnat}\in
\prod_{n\in\monnat}e_{\sphilbfrak{k},n}\oa{K}(\sphilbfrak{k})^{\gamma}e_{\sphilbfrak{k},n}}
{\sup_{n\in\monnat}\norm{A_n}<\infty}.
\end{aligned}
\label{eq:paper-intrinsic-block-multiplier-algebra}
\end{equation} Its multiplication and adjoint are coordinatewise, and its norm is the supremum of the corner norms.

Fix \(\sminvtemperature>0\) and assume that the compact heat-kernel family is trace summable at inverse temperature \(\sminvtemperature\) in the sense of \eqref{eq:paper-main-heat-kernel-trace-summability}. Its trace sum is finite and strictly positive. The finite-mode Gibbs state on the bounded block product \eqref{eq:paper-intrinsic-block-multiplier-algebra} is \begin{equation}
\begin{aligned}
\oastate[\psi_{\sminvtemperature}]
\colon
\oa{M}(\oa{K}(\sphilbfrak{k})^{\gamma})
\to
\fldcmp;\quad
\fun{\oastate[\psi_{\sminvtemperature}]}{A}
&=
\frac{
\sum_{n\in\monnat}
\sqfun{\trace_{n}}{A_n\napiernum^{-\sminvtemperature \physham_{n}}}
}{
\sum_{n\in\monnat}
\sqfun{\trace_{n}}{\napiernum^{-\sminvtemperature \physham_{n}}}
}.
\end{aligned}
\label{eq:paper-intrinsic-gibbs-state}
\end{equation} For an arbitrary \(A=\seq{A_n}{n\in\monnat}\) in the bounded block product, the trace ideal estimate gives \[\begin{aligned}
\sum_{n\in\monnat}
\abs{\sqfun{\trace_n}{A_n\napiernum^{-\sminvtemperature\physham_n}}}
\leq
\sup_{n\in\monnat}\norm{A_n}
\sum_{n\in\monnat}
\sqfun{\trace_n}{\napiernum^{-\sminvtemperature\physham_n}}
=
\norm{A}
\sum_{n\in\monnat}
\sqfun{\trace_n}{\napiernum^{-\sminvtemperature\physham_n}}
<
\infty.
\end{aligned}\] Thus the numerator in \eqref{eq:paper-intrinsic-gibbs-state} is absolutely convergent for every \(A\). For \(A\geq0\), cyclicity gives \[\begin{aligned}
\sqfun{\trace_{n}}{A_n\napiernum^{-\sminvtemperature \physham_{n}}}
&=
\sqfun{\trace_{n}}{
\napiernum^{-\frac{\sminvtemperature}{2}\physham_{n}}
A_n
\napiernum^{-\frac{\sminvtemperature}{2}\physham_{n}}}
\geq
0.
\end{aligned}\] For \(A\geq0\), positivity and the trace ideal estimate give \[\begin{aligned}
0
&\leq
\sqfun{\trace_n}{A_n\napiernum^{-\sminvtemperature\physham_n}}
\leq
\norm{A_n}
\sqfun{\trace_n}{\napiernum^{-\sminvtemperature\physham_n}}
\leq
\norm{A}
\sqfun{\trace_n}{\napiernum^{-\sminvtemperature\physham_n}}.
\end{aligned}\] Thus the numerator in \eqref{eq:paper-intrinsic-gibbs-state} is at most the denominator multiplied by \(\norm{A}\). The numerator equals the denominator at \(A=1\). These calculations prove that \eqref{eq:paper-intrinsic-gibbs-state} defines a state.

\begin{prop}[KMS specialization of the trace Gibbs state]\label{prop:paper-intrinsic-gibbs-kms-state}
Assume the hypotheses of Corollary \ref{cor:paper-trace-summable-heat-kernel} for a fixed $\sminvtemperature>0$.
Then the restriction of the state $\oastate[\psi_{\sminvtemperature}]$ in
\eqref{eq:paper-intrinsic-gibbs-state} to $\oa{K}(\sphilbfrak{k})^{\gamma}$ is a $\sminvtemperature$-KMS state for the
dynamics in \eqref{eq:paper-intrinsic-block-dynamics}.
\end{prop}

\begin{proof}
For $N\in\monnat$, set
$E_N=\sum_{n=0}^{N}e_{\sphilbfrak{k},n}$.
For every $A=\seq{A_n}{n\in\monnat}\in\oa{K}(\sphilbfrak{k})^{\gamma}$,
\eqref{eq:paper-main-finite-mode-block-algebra} gives
$$\begin{aligned}
\norm{E_NA-A}
&=
\norm{AE_N-A}
=
\sup_{n>N}\norm{A_n}
\to
0.
\end{aligned}$$
Thus $\seq{E_N}{N\in\monnat}$ is a contractive approximate identity, and
$$\begin{aligned}
\fun{\oastate[\psi_{\sminvtemperature}]}{E_N}
=
\frac{
\sum_{n=0}^{N}\sqfun{\trace_n}{\napiernum^{-\sminvtemperature\physham_n}}
}{
\sum_{n\in\monnat}\sqfun{\trace_n}{\napiernum^{-\sminvtemperature\physham_n}}
}
\to
1,
\quad
\norm{
\fnrestr{\oastate[\psi_{\sminvtemperature}]}{\oa{K}(\sphilbfrak{k})^{\gamma}}
}
=
\sup_{N\in\monnat}\fun{\oastate[\psi_{\sminvtemperature}]}{E_N}
=
1.
\end{aligned}$$
Let $A$ be an entire analytic element for the dynamics in
\eqref{eq:paper-intrinsic-block-dynamics}, and let
$B\in\oa{K}(\sphilbfrak{k})^{\gamma}$.
For $N\in\monnat$, define the finite block truncation by
$$\begin{aligned}
A_{\leq N}
&=
\sum_{n=0}^{N}A_n.
\end{aligned}$$
The element $A_{\leq N}$ is entire analytic, and cyclicity of each matrix trace gives
$$\begin{aligned}
\fun{\oastate[\psi_{\sminvtemperature}]}{B\fun{\alpha_{\imunit\sminvtemperature}}{A_{\leq N}}}
&=
\frac{
\sum_{n=0}^{N}
\sqfun{\trace_n}{
B_n
\napiernum^{-\sminvtemperature\physham_n}
A_n
\napiernum^{\sminvtemperature\physham_n}
\napiernum^{-\sminvtemperature\physham_n}}}
{
\sum_{n\in\monnat}
\sqfun{\trace_n}{\napiernum^{-\sminvtemperature\physham_n}}}
\\
&=
\frac{
\sum_{n=0}^{N}
\sqfun{\trace_n}{A_nB_n\napiernum^{-\sminvtemperature\physham_n}}}
{
\sum_{n\in\monnat}
\sqfun{\trace_n}{\napiernum^{-\sminvtemperature\physham_n}}}
=
\fun{\oastate[\psi_{\sminvtemperature}]}{A_{\leq N}B}.
\end{aligned}$$
The dynamics \eqref{eq:paper-intrinsic-block-dynamics} fixes every block projection, so uniqueness of the
entire extension gives
$$\begin{aligned}
A_{\leq N}
&\to
A,
&
\fun{\alpha_{\imunit\sminvtemperature}}{A_{\leq N}}
&=
\rbk{\fun{\alpha_{\imunit\sminvtemperature}}{A}}_{\leq N}
\to
\fun{\alpha_{\imunit\sminvtemperature}}{A}
\end{aligned}$$
in norm by \eqref{eq:paper-main-finite-mode-block-algebra}.
Continuity of $\oastate[\psi_{\sminvtemperature}]$ and passage to the limit $N\to\infty$ give the $\sminvtemperature$-KMS boundary
identity \eqref{eq:paper-kms-boundary-condition} for $A$ and $B$.
\end{proof}

For a complex Hilbert space \(\sphilbfrak{H}\) and \(T\in\opspbddlin{\sphilbfrak{H}}_{+}\), define the canonical positive trace by \[\begin{aligned}
\sqfun{\trace_{\sphilbfrak{H}}}{T}
&=
\sup
\set{
\sum_{\xi\in F}\bkt{\xi}{T\xi}_{\sphilbfrak{H}}}
{F\subset\sphilbfrak{H}\text{ is a finite orthonormal set}}
\in
\closedinterval{0}{\infty}.
\end{aligned}\] For a complex Hilbert space \(\sphilbfrak{H}\) and \(1\leq p<\infty\), define the Schatten class and its norm by \[\begin{aligned}
\fun{\schattencls^{p}}{\sphilbfrak{H}}
=
\set{T\in\opspbddlin{\sphilbfrak{H}}}
{\sqfun{\trace_{\sphilbfrak{H}}}{\abs{T}^p}<\infty},
\quad
\norm{T}_{p,\trace_{\sphilbfrak{H}}}
=
\norm{T}_{p}
=
\sqfun{\trace_{\sphilbfrak{H}}}{\abs{T}^p}^{\frac{1}{p}}.
\end{aligned}\] For \(p=\infty\), set \(\fun{\schattencls^{\infty}}{\sphilbfrak{H}}=\opspbddlin{\sphilbfrak{H}}\) and \(\norm{T}_{\infty,\trace_{\sphilbfrak{H}}}=\norm{T}\). The Schatten--Hölder inequality is \cite[Theorem 2.8]{BarrySimon005}.

\subsubsection{Trace-norm product formula}\label{trace-norm-product-formula}

In this subsubsection, the trace-norm argument first establishes a majorant for arbitrary Hilbert-space sectors and then specializes the resulting all-sector criterion to the intrinsic finite-mode blocks.

For each \(n\in\monnat\), let \(\sphilbfrak{H}_n\) be a complex Hilbert space, and let \(T_n\) and \(V_n\) be self-adjoint operators on \(\sphilbfrak{H}_n\). Fix \(\sminvtemperature>0\). Assume that constants \(a>0\) and \(b,c\geq0\) satisfy the potential-superstability bound \begin{equation}
\begin{aligned}
V_n
&\geq
\rbk{an^2-bn-c}1_{\sphilbfrak{H}_n}
&&\rbk{n\in\monnat}
\end{aligned}
\label{eq:paper-general-potential-superstability}
\end{equation} in the quadratic-form sense. Assume also that \begin{equation}
\begin{aligned}
\napiernum^{-\sminvtemperature T_n}
&\in
\fun{\schattencls^{1}}{\sphilbfrak{H}_n}.
\end{aligned}
\label{eq:paper-kinetic-gibbs-condition}
\end{equation} For \(M\in\semigrposint\), define \begin{equation}
\begin{aligned}
F_{\sminvtemperature,M,n}
&=
\rbk{
\napiernum^{-\frac{\sminvtemperature}{M}V_n}
\napiernum^{-\frac{\sminvtemperature}{M}T_n}
}^{M}.
\end{aligned}
\label{eq:paper-general-sector-product}
\end{equation}

\begin{lem}[Gibbs-sector majorant]\label{lem:paper-gibbs-sector-majorant}
For every $n\in\monnat$ and $M\in\semigrposint$, one has
$F_{\sminvtemperature,M,n}\in\fun{\schattencls^{1}}{\sphilbfrak{H}_n}$ and
\begin{equation}
\begin{aligned}
\norm{F_{\sminvtemperature,M,n}}_{1,\trace_{\sphilbfrak{H}_n}}
&\leq
\napiernum^{-\sminvtemperature\rbk{an^2-bn-c}}
\sqfun{\trace_{\sphilbfrak{H}_n}}{\napiernum^{-\sminvtemperature T_n}}.
\end{aligned}
\label{eq:paper-general-sector-majorant}
\end{equation}
In particular, the majorant is summable in the particle number if
\begin{equation}
\begin{aligned}
\sum_{n\in\monnat}
\napiernum^{-\sminvtemperature\rbk{an^2-bn-c}}
\sqfun{\trace_{\sphilbfrak{H}_n}}{\napiernum^{-\sminvtemperature T_n}}
&<
\infty.
\end{aligned}
\label{eq:paper-general-sector-majorant-summability}
\end{equation}
\end{lem}

\begin{proof}
Functional calculus and \eqref{eq:paper-general-potential-superstability} give
$$\begin{aligned}
\norm{\napiernum^{-\frac{\sminvtemperature}{M}V_n}}
&\leq
\napiernum^{-\frac{\sminvtemperature}{M}\rbk{an^2-bn-c}}.
\end{aligned}$$
The kinetic Gibbs condition \eqref{eq:paper-kinetic-gibbs-condition} gives
$$\begin{aligned}
\napiernum^{-\frac{\sminvtemperature}{M}T_n}
&\in
\fun{\schattencls^{M}}{\sphilbfrak{H}_n},
&
\norm{\napiernum^{-\frac{\sminvtemperature}{M}T_n}}_{M,\trace_{\sphilbfrak{H}_n}}^{M}
&=
\sqfun{\trace_{\sphilbfrak{H}_n}}{\napiernum^{-\sminvtemperature T_n}}.
\end{aligned}$$
Apply the Schatten--Hölder inequality to the $M$ potential factors in the operator norm and the $M$ kinetic
factors in the Schatten $M$-norm.
The resulting estimate is
$$\begin{aligned}
\norm{F_{\sminvtemperature,M,n}}_{1,\trace_{\sphilbfrak{H}_n}}
\leq
\norm{\napiernum^{-\frac{\sminvtemperature}{M}V_n}}^{M}
\norm{\napiernum^{-\frac{\sminvtemperature}{M}T_n}}_{M,\trace_{\sphilbfrak{H}_n}}^{M}
\leq
\napiernum^{-\sminvtemperature\rbk{an^2-bn-c}}
\sqfun{\trace_{\sphilbfrak{H}_n}}{\napiernum^{-\sminvtemperature T_n}},
\end{aligned}$$
which is \eqref{eq:paper-general-sector-majorant}.
\end{proof}

If \(\sphilbfrak{H}_n\) is finite dimensional, the kinetic Gibbs condition \eqref{eq:paper-kinetic-gibbs-condition} is automatic for every lower-semibounded \(T_n\). For an infinite-dimensional continuous-system sector, it replaces an unavailable bound on \(\norm{T_n}\) by the spectral heat-trace estimate that is actually needed in the product argument.

For every \(n\in\monnat\), assume that the quadratic-form domains of \(T_n\) and \(V_n\) have dense intersection and that the sum of their closed quadratic forms on this intersection is closed and semibounded. Denote by \(\physham_n\) the self-adjoint operator associated with this form sum.

\begin{prop}[All-sector Gibbs product criterion]\label{prop:paper-gibbs-sector-product-formula}
Assume \eqref{eq:paper-general-sector-majorant-summability}.
Assume that, for every $n\in\monnat$,
\begin{equation}
\begin{aligned}
\lim_{M\to\infty}
\norm{
F_{\sminvtemperature,M,n}
-
\napiernum^{-\sminvtemperature\physham_n}
}_{1,\trace_{\sphilbfrak{H}_n}}
&=
0.
\end{aligned}
\label{eq:paper-sectorwise-gibbs-product-limit}
\end{equation}
Then the direct sums of the products in \eqref{eq:paper-general-sector-product} and of the heat kernels are
trace class on $\bigoplus_{n\in\monnat}\sphilbfrak{H}_n$, and
\begin{equation}
\begin{aligned}
\lim_{M\to\infty}
\norm{
\bigoplus_{n\in\monnat}F_{\sminvtemperature,M,n}
-
\bigoplus_{n\in\monnat}\napiernum^{-\sminvtemperature\physham_n}
}_{1,\trace_{\bigoplus_{n\in\monnat}\sphilbfrak{H}_n}}
&=
0.
\end{aligned}
\label{eq:paper-all-sector-gibbs-product-limit}
\end{equation}
\end{prop}

\begin{proof}
The sectorwise trace-norm limit \eqref{eq:paper-sectorwise-gibbs-product-limit} and
\eqref{eq:paper-general-sector-majorant} give
$$\begin{aligned}
\norm{\napiernum^{-\sminvtemperature\physham_n}}_{1,\trace_{\sphilbfrak{H}_n}}
&\leq
\napiernum^{-\sminvtemperature\rbk{an^2-bn-c}}
\sqfun{\trace_{\sphilbfrak{H}_n}}{\napiernum^{-\sminvtemperature T_n}}.
\end{aligned}$$
The sector bounds and \eqref{eq:paper-general-sector-majorant-summability} give
$$\begin{aligned}
\sum_{n\in\monnat}
\norm{F_{\sminvtemperature,M,n}}_{1,\trace_{\sphilbfrak{H}_n}}
&\leq
\sum_{n\in\monnat}
\napiernum^{-\sminvtemperature\rbk{an^2-bn-c}}
\sqfun{\trace_{\sphilbfrak{H}_n}}{\napiernum^{-\sminvtemperature T_n}}
<
\infty,
\\
\sum_{n\in\monnat}
\norm{\napiernum^{-\sminvtemperature\physham_n}}_{1,\trace_{\sphilbfrak{H}_n}}
&\leq
\sum_{n\in\monnat}
\napiernum^{-\sminvtemperature\rbk{an^2-bn-c}}
\sqfun{\trace_{\sphilbfrak{H}_n}}{\napiernum^{-\sminvtemperature T_n}}
<
\infty.
\end{aligned}$$
These two estimates prove the trace-class assertions for the direct sums.
Consequently,
$$\begin{aligned}
\norm{
F_{\sminvtemperature,M,n}
-
\napiernum^{-\sminvtemperature\physham_n}
}_{1,\trace_{\sphilbfrak{H}_n}}
&\leq
2\napiernum^{-\sminvtemperature\rbk{an^2-bn-c}}
\sqfun{\trace_{\sphilbfrak{H}_n}}{\napiernum^{-\sminvtemperature T_n}}.
\end{aligned}$$
The right-hand side is summable by
\eqref{eq:paper-general-sector-majorant-summability} and independent of $M$.
Dominated convergence for the series over $n$ proves
\eqref{eq:paper-all-sector-gibbs-product-limit}.
\end{proof}

For unbounded generators, \eqref{eq:paper-sectorwise-gibbs-product-limit} is a separate sectorwise analytic input and does not follow from lower semiboundedness alone. After scalar shifts to nonnegative generators, the Trotter--Kato theory for Gibbs semigroups supplies such trace-norm convergence under its own hypotheses. The fixed-\(\sminvtemperature\) kinetic Gibbs condition \eqref{eq:paper-kinetic-gibbs-condition} alone is not asserted to imply \eqref{eq:paper-sectorwise-gibbs-product-limit}. Lemma \ref{lem:paper-gibbs-sector-majorant} then performs the additional lift from the individual sectors to the particle-number direct sum.

We now return to the intrinsic finite-mode full matrix corners. For each \(n\in\monnat\), let \(T_{n}\) and \(V_{n}\) be the kinetic and potential elements of \(e_{\sphilbfrak{k},n}\oa{K}(\sphilbfrak{k})^{\gamma}e_{\sphilbfrak{k},n}\) fixed before \eqref{eq:paper-main-corner-potential-superstability}, and set \(\physham_{n}=T_{n}+V_{n}\). Every corner is a finite-dimensional full matrix algebra. Hence the kinetic Gibbs condition \eqref{eq:paper-kinetic-gibbs-condition} holds in every particle block, and the sectorwise trace-norm limit follows from the finite-dimensional product formula. Fix a unital \(\ast\)-isomorphism \(\Phi_n\) from the \(n\)th corner onto \(\opspbddlin{\fldcmp^{d_n}}\). It maps minimal projections to rank-one projections, and the uniqueness in \eqref{eq:paper-standard-corner-trace} gives \(\trace_n=\trace_{\fldcmp^{d_n}}\circ\Phi_n\). Consequently, the Schatten--Hölder inequality \cite[Theorem 2.8]{BarrySimon005} holds for the norms \(\norm{A}_{p,\trace_n}=\sqfun{\trace_n}{\abs{A}^{p}}^{1/p}\) on each corner.

\begin{prop}[Product formula over all particle blocks]\label{prop:intrinsic-trace-norm-product-formula}
Fix $\sminvtemperature>0$, and assume the potential superstability \eqref{eq:paper-main-corner-potential-superstability}
and the Gaussian summability \eqref{eq:paper-main-kinetic-heat-trace-summability}.
Then $G_{\sminvtemperature,M}$ is trace summable for every $M\in\semigrposint$, the heat-kernel family
$\seq{\napiernum^{-\sminvtemperature \physham_{n}}}{n\in\monnat}$ is trace summable, and
\begin{equation}
\begin{aligned}
\lim_{M\to\infty}
\norm{
G_{\sminvtemperature,M}
-
\seq{\napiernum^{-\sminvtemperature \physham_{n}}}{n\in\monnat}
}_{1,\trace}
&=
0.
\end{aligned}
\label{eq:paper-trace-norm-product-limit}
\end{equation}
\end{prop}

\begin{proof}
Proposition \ref{prop:paper-finite-corner-norm-product-formula} and finite dimensionality give, for every fixed
$n\in\monnat$,
$$\begin{aligned}
\norm{
G_{\sminvtemperature,M,n}
-
\napiernum^{-\sminvtemperature \physham_{n}}
}_{1,\trace_{n}}
&\leq
\sqfun{\trace_{n}}{e_{\sphilbfrak{k},n}}
\norm{
G_{\sminvtemperature,M,n}
-
\napiernum^{-\sminvtemperature \physham_{n}}
}
\to
0
&&\rbk{M\to\infty}.
\end{aligned}$$
Functional calculus and \eqref{eq:paper-main-corner-potential-superstability} give the potential factor bound
$$\begin{aligned}
\norm{\napiernum^{-\frac{\sminvtemperature}{M}V_n}}
&\leq
\napiernum^{-\frac{\sminvtemperature}{M}\rbk{an^2-bn-c}}.
\end{aligned}$$
Since $\napiernum^{-\frac{\sminvtemperature}{M}T_n}$ is positive, its Schatten $M$-norm is
$$\begin{aligned}
\norm{\napiernum^{-\frac{\sminvtemperature}{M}T_n}}_{M,\trace_n}^{M}
&=
\sqfun{\trace_n}{\rbk{\napiernum^{-\frac{\sminvtemperature}{M}T_n}}^{M}}
=
\sqfun{\trace_n}{\napiernum^{-\sminvtemperature T_n}}.
\end{aligned}$$
Apply the Schatten--Hölder inequality to the $M$ potential factors in the operator norm and the $M$ kinetic
factors in the Schatten $M$-norm, whose reciprocal exponents add up to $M\cdot\frac{1}{M}=1$.
The resulting sector majorant is
\begin{equation}
\begin{aligned}
\norm{G_{\sminvtemperature,M,n}}_{1,\trace_{n}}
\leq
\norm{\napiernum^{-\frac{\sminvtemperature}{M}V_n}}^{M}
\norm{\napiernum^{-\frac{\sminvtemperature}{M}T_n}}_{M,\trace_n}^{M}
\leq
\napiernum^{-\sminvtemperature\rbk{an^2-bn-c}}
\sqfun{\trace_{n}}{\napiernum^{-\sminvtemperature T_n}}.
\end{aligned}
\label{eq:paper-intrinsic-sector-majorant}
\end{equation}
The reverse triangle inequality for $\norm{\mathord{\cdot}}_{1,\trace_n}$ and the sectorwise limit give
$$\begin{aligned}
\norm{\napiernum^{-\sminvtemperature \physham_{n}}}_{1,\trace_{n}}
&=
\lim_{M\to\infty}\norm{G_{\sminvtemperature,M,n}}_{1,\trace_{n}}
\leq
\napiernum^{-\sminvtemperature\rbk{an^2-bn-c}}
\sqfun{\trace_{n}}{\napiernum^{-\sminvtemperature T_n}}.
\end{aligned}$$
Summing \eqref{eq:paper-intrinsic-sector-majorant} and the last bound over $n$ and using
\eqref{eq:paper-main-kinetic-heat-trace-summability} proves the two trace-summability assertions.
The triangle inequality gives, for every $M\in\semigrposint$ and $n\in\monnat$,
$$\begin{aligned}
\norm{
G_{\sminvtemperature,M,n}
-
\napiernum^{-\sminvtemperature \physham_{n}}
}_{1,\trace_{n}}
&\leq
\norm{G_{\sminvtemperature,M,n}}_{1,\trace_{n}}
+
\norm{\napiernum^{-\sminvtemperature \physham_{n}}}_{1,\trace_{n}}
\leq
2\napiernum^{-\sminvtemperature\rbk{an^2-bn-c}}
\sqfun{\trace_{n}}{\napiernum^{-\sminvtemperature T_n}}.
\end{aligned}$$
The right-hand side is summable in $n$ by \eqref{eq:paper-main-kinetic-heat-trace-summability} and is
independent of $M$.
Sectorwise convergence and dominated convergence for the series over $n$ give
$$\begin{aligned}
&\lim_{M\to\infty}
\norm{
G_{\sminvtemperature,M}
-
\seq{\napiernum^{-\sminvtemperature \physham_{n}}}{n\in\monnat}
}_{1,\trace}
=
\lim_{M\to\infty}
\sum_{n\in\monnat}
\norm{
G_{\sminvtemperature,M,n}
-
\napiernum^{-\sminvtemperature \physham_{n}}
}_{1,\trace_{n}}
\\
&=
\sum_{n\in\monnat}
\lim_{M\to\infty}
\norm{
G_{\sminvtemperature,M,n}
-
\napiernum^{-\sminvtemperature \physham_{n}}
}_{1,\trace_{n}}
=
0.
\end{aligned}$$
This is \eqref{eq:paper-trace-norm-product-limit}.
\end{proof}

On a lattice, the Gaussian summability \eqref{eq:paper-main-kinetic-heat-trace-summability} holds under explicit bounds on the corners.

\begin{lem}[Lattice criterion for the Gaussian summability]\label{lem:paper-summable-sector-majorant}
Assume the polynomial corner-dimension growth \eqref{eq:paper-corner-dimension-growth}, the potential
superstability \eqref{eq:paper-main-corner-potential-superstability}, and the linear kinetic norm bound
\eqref{eq:paper-main-linear-kinetic-bound}.
Then \eqref{eq:paper-main-kinetic-heat-trace-summability} holds for every $\sminvtemperature>0$, with the bound
$$\begin{aligned}
\sum_{n\in\monnat}
\napiernum^{-\sminvtemperature\rbk{an^2-bn-c}}
\sqfun{\trace_{n}}{\napiernum^{-\sminvtemperature T_n}}
&\leq
\sum_{n\in\monnat}
\sqfun{\trace_{n}}{e_{\sphilbfrak{k},n}}
\napiernum^{-\sminvtemperature\rbk{an^2-\rbk{b+b_T}n-\rbk{c+c_T}}}
&<
\infty.
\end{aligned}$$
\end{lem}

\begin{proof}
Functional calculus and \eqref{eq:paper-main-linear-kinetic-bound} give
$$\begin{aligned}
\sqfun{\trace_n}{\napiernum^{-\sminvtemperature T_n}}
\leq
\sqfun{\trace_n}{e_{\sphilbfrak{k},n}}
\napiernum^{\sminvtemperature\norm{T_n}}
\leq
\sqfun{\trace_n}{e_{\sphilbfrak{k},n}}
\napiernum^{\sminvtemperature\rbk{b_Tn+c_T}}.
\end{aligned}$$
Equation \eqref{eq:paper-corner-dimension-growth} gives
$$\begin{aligned}
\sum_{n\in\monnat}
\sqfun{\trace_{n}}{e_{\sphilbfrak{k},n}}
\napiernum^{-\sminvtemperature\rbk{an^2-\rbk{b+b_T}n-\rbk{c+c_T}}}
\leq
C
\sum_{n\in\monnat}
\rbk{1+n}^{s}
\napiernum^{-\sminvtemperature\rbk{an^2-\rbk{b+b_T}n-\rbk{c+c_T}}}
<
\infty.
\end{aligned}$$
\end{proof}

The product used on each finite-dimensional corner \(e_{\sphilbfrak{k},n}\oa{K}(\sphilbfrak{k})^{\gamma}e_{\sphilbfrak{k},n}\) is the finite-dimensional case of the Lie--Trotter formula \cite[Theorem VIII.31]{ReedSimon003}. It does not use a product formula for approximating unbounded generators or for sums of quadratic forms. The additional conclusion for the sequence indexed by all \(n\in\monnat\) is exactly the trace-norm lift \eqref{eq:paper-trace-norm-product-limit}, which follows from the summable majorant over \(n\).

\subsubsection{Trace closure and the abstract closed-path sum}\label{trace-closure-and-the-abstract-closed-path-sum}

In this subsubsection, the trace closes the open matrix paths from \eqref{eq:paper-matrix-open-path-expansion} by cyclicity and sums their values over all particle sectors. The resulting partition function requires trace summability in addition to the finite matrix identities.

\begin{prop}[Trace resolution by fixed minimal projections]\label{prop:intrinsic-path-projection-resolution}
The standard traces of the fixed minimal projections satisfy
$$\begin{aligned}
\sqfun{\trace_{n}}{p_{n,\omega}}
&=
1,
&
\sqfun{\trace_{n}}{e_{\sphilbfrak{k},n}}
&=
\abscard{\Omega_{n}}.
\end{aligned}$$
Let $A=\seq{A_n}{n\in\monnat}\in\oa{K}(\sphilbfrak{k})^{\gamma}$.
Assume that $A_n\in e_{\sphilbfrak{k},n}\oa{K}(\sphilbfrak{k})^{\gamma}e_{\sphilbfrak{k},n}$ for every
$n\in\monnat$ and that
$$\begin{aligned}
\sum_{n\in\monnat}\sqfun{\trace_{n}}{\abs{A_n}}
&<
\infty.
\end{aligned}$$
The extended trace has the diagonal resolution
\begin{equation}
\begin{aligned}
\sqfun{\trace}{A}
&=
\sum_{n\in\monnat}
\sum_{\omega\in\Omega_{n}}
\sqfun{\trace_{n}}{
p_{n,\omega}A_np_{n,\omega}}.
\end{aligned}
\label{eq:paper-intrinsic-trace-resolution}
\end{equation}
\end{prop}

\begin{proof}
Equation \eqref{eq:paper-standard-corner-trace} gives
$$\begin{aligned}
\sqfun{\trace_{n}}{p_{n,\omega}}
&=
1.
\end{aligned}$$
Summing this identity over $\omega\in\Omega_{n}$ and using
\eqref{eq:paper-minimal-projection-resolution} gives
$\sqfun{\trace_{n}}{e_{\sphilbfrak{k},n}}=\abscard{\Omega_{n}}$.
Insert \eqref{eq:paper-minimal-projection-resolution} on both sides of $A_n$ and use cyclicity:
$$\begin{aligned}
\sqfun{\trace_{n}}{A_n}
&=
\sum_{\omega\in\Omega_{n}}
\sqfun{\trace_{n}}{
p_{n,\omega}A_np_{n,\omega}}.
\end{aligned}$$
Summing these finite identities over $n$ is justified by trace summability and gives
\eqref{eq:paper-intrinsic-trace-resolution}.
\end{proof}

Under its hypotheses, Proposition \ref{prop:intrinsic-trace-norm-product-formula} shows that Proposition \ref{prop:intrinsic-path-projection-resolution}, through \eqref{eq:paper-intrinsic-trace-resolution}, applies both to the heat-kernel family in \eqref{eq:paper-main-compact-heat-kernel} and to the product family \(G_{\sminvtemperature,M}\) defined in \eqref{eq:paper-trotter-product-family}.

Fix \(\sminvtemperature>0\) for the remainder of this subsection, and assume the diagonalization condition \eqref{eq:paper-minimal-projection-diagonalization} for every \(n\in\monnat\). For \(n\in\monnat\), \(M\in\semigrposint\), and \(\boldsymbol{\omega}=\seq{\omega_j}{0\leq j<M}\in\Omega_n^M\), define the periodic extension by \(\omega_{j+M}=\omega_j\) for every \(j\in\monnat\). The extended sequence is called the based closed path represented by \(\boldsymbol{\omega}\); in particular, \(\omega_M=\omega_0\). For every \(M\in\semigrposint\), define the closed-path set by \begin{equation}
\begin{aligned}
\prbspsample_{M,\txtclassical}
&=
\coprod_{n\in\monnat}\Omega_{n}^{M}.
\end{aligned}
\label{eq:paper-closed-path-space}
\end{equation} For \(M\in\semigrposint\) and \(\boldsymbol{\omega}=\seq{\omega_j}{0\leq j<M}\in\Omega_{n}^{M}\), define the path weight by \begin{equation}
\begin{aligned}
W_{M}
\colon
\prbspsample_{M,\txtclassical}
\to
\fldcmp;\quad
\fun{W_{M}}{n,\boldsymbol{\omega}}
&=
\sqfun{\trace_{n}}{
\rbk{\prod_{j=0}^{M-1}
\napiernum^{-\frac{\sminvtemperature}{M}\fun{v_{n}}{\omega_j}}
p_{n,\omega_j}
\napiernum^{-\frac{\sminvtemperature}{M}T_{n}}}
p_{n,\omega_0}}.
\end{aligned}
\label{eq:paper-closed-path-weight}
\end{equation}

\begin{prop}[Intrinsic discrete closed-path expansion]\label{prop:paper-abstract-path-integral}\label{prop:abstract-projection-path-integral}
Assume that \eqref{eq:paper-minimal-projection-diagonalization} holds for every $n\in\monnat$ and
$\omega\in\Omega_{n}$.
Assume also the hypotheses of Proposition \ref{prop:intrinsic-trace-norm-product-formula}.
Let the closed-path space and its weight be those in \eqref{eq:paper-closed-path-space} and
\eqref{eq:paper-closed-path-weight}.
For every $M\in\semigrposint$, one has
\begin{equation}
\begin{aligned}
\sqfun{\trace}{G_{\sminvtemperature,M}}
&=
\lim_{N\to\infty}
\sum_{n=0}^{N}
\sum_{\boldsymbol{\omega}\in\Omega_{n}^{M}}
\fun{W_{M}}{n,\boldsymbol{\omega}}.
\end{aligned}
\label{eq:paper-intrinsic-path-expansion}
\end{equation}
The partition functions satisfy
$$\begin{aligned}
\lim_{M\to\infty}
\sqfun{\trace}{G_{\sminvtemperature,M}}
&=
\sum_{n\in\monnat}
\sqfun{\trace_{n}}{\napiernum^{-\sminvtemperature \physham_{n}}}.
\end{aligned}$$
\end{prop}

\begin{proof}
Fix $n$, and apply $\trace_n$ to the open-path operator identity
\eqref{eq:paper-matrix-trotter-open-path-expansion}.
For every corner element $A$ and every $\omega_0,\omega_M\in\Omega_n$, cyclicity and orthogonality give
$$\begin{aligned}
\sqfun{\trace_n}{p_{n,\omega_0}Ap_{n,\omega_M}}
&=
\sqfun{\trace_n}{p_{n,\omega_M}p_{n,\omega_0}A}
=
0
&&\rbk{\omega_M\neq\omega_0}.
\end{aligned}$$
Thus only the terms with $\omega_M=\omega_0$ remain, and each remaining term equals
$\fun{W_M}{n,\boldsymbol{\omega}}$ by \eqref{eq:paper-closed-path-weight}.
For fixed $n$, all sums over $\Omega_{n}$ are finite.
The trace-norm estimate
$$\begin{aligned}
\sum_{n\in\monnat}
\abs{\sqfun{\trace_{n}}{G_{\sminvtemperature,M,n}}}
&\leq
\sum_{n\in\monnat}
\norm{G_{\sminvtemperature,M,n}}_{1,\trace_{n}}
<
\infty
\end{aligned}$$
permits passage to the limit $N\to\infty$ in the partial sums over $0\leq n\leq N$ in
\eqref{eq:paper-intrinsic-path-expansion}.
This proves \eqref{eq:paper-intrinsic-path-expansion}.
The trace-norm continuity of $\trace$ and
\eqref{eq:paper-trace-norm-product-limit} prove the partition-function limit.
\end{proof}

Proposition \ref{prop:abstract-projection-path-integral} is the trace specialization of the algebraic open-path identity \eqref{eq:paper-matrix-trotter-open-path-expansion}. It is not a probability statement. Self-adjointness of \(T_{n}\) alone does not imply that the individual matrix-element weights are nonnegative or real. Assume first that every path weight in \eqref{eq:paper-closed-path-weight} is nonnegative: \begin{equation}
\begin{aligned}
\fun{W_{M}}{\omega}
&\in
\closedinterval{0}{\infty}
&&\rbk{\omega\in\prbspsample_{M,\txtclassical}}.
\end{aligned}
\label{eq:paper-finite-time-path-weight-nonnegativity}
\end{equation} Assume also that the partition function is strictly positive: \begin{equation}
\begin{aligned}
Z_{M}
&=
\sqfun{\trace}{G_{\sminvtemperature,M}}
>
0.
\end{aligned}
\label{eq:paper-finite-time-path-partition-function}
\end{equation} For \(M\in\semigrposint\), the full power-set \(\sigma\)-algebra on the closed-path space is \begin{equation}
\begin{aligned}
\mblfmlfrak{S}_{M}
&=
2^{\prbspsample_{M,\txtclassical}}.
\end{aligned}
\label{eq:paper-closed-path-sigma-algebra}
\end{equation} Then the probability measure is \begin{equation}
\begin{aligned}
\mu_{M}
\colon
\mblfmlfrak{S}_{M}
\to
\unitclosedinterval;\quad
\fun{\mu_{M}}{E}
&=
\frac{1}{Z_{M}}
\sum_{\omega\in E}\fun{W_{M}}{\omega}.
\end{aligned}
\label{eq:paper-normalized-finite-time-path-measure}
\end{equation} The conditions \eqref{eq:paper-finite-time-path-weight-nonnegativity} and \eqref{eq:paper-finite-time-path-partition-function} are additional model-dependent hypotheses. Under these two conditions, \eqref{eq:paper-intrinsic-path-expansion} gives the normalization of \eqref{eq:paper-normalized-finite-time-path-measure}. Reflection positivity likewise requires a specified time-reversal involution, a specified sub-\(\sigma\)-algebra, and the corresponding quadratic inequality for all bounded functions measurable with respect to that sub-\(\sigma\)-algebra.

\begin{defn}[Reflection positivity of a probability measure]\label{def:reflection-positive-probability-measure}
Let $\rbk{\prbspsample,\mblfmlfrak{S}}$ be a measurable space, let
$\mblfmlfrak{S}_{+}\subset\mblfmlfrak{S}$ be a sub-$\sigma$-algebra, and let
$\vartheta\colon\prbspsample\to\prbspsample$ be a measurable map satisfying
$\vartheta\circ\vartheta=\id_{\prbspsample}$.
A probability measure $\mu$ on $\rbk{\prbspsample,\mblfmlfrak{S}}$ is reflection positive with respect to
$\rbk{\mblfmlfrak{S}_{+},\vartheta}$ when the notation
$\sqfun{\prbexp_{\mu}}{G}=\int_{\prbspsample}G\opdmsr{\mu}$ is used for every bounded
$\mblfmlfrak{S}$-measurable function $G\colon\prbspsample\to\fldcmp$ and
\begin{equation}
\begin{aligned}
\sqfun{\prbexp_{\mu}}{
\cmpconj{F\circ\vartheta}F}
&\geq
0
\end{aligned}
\label{eq:reflection-positive-probability-measure}
\end{equation}
for every bounded $\mblfmlfrak{S}_{+}$-measurable function $F\colon\prbspsample\to\fldcmp$.
\end{defn}

For \(M\in\semigrposint\), define time reversal on the closed-path space \eqref{eq:paper-closed-path-space} by \begin{equation}
\begin{aligned}
\vartheta_{M}
\colon
\prbspsample_{M,\txtclassical}
\to
\prbspsample_{M,\txtclassical};\quad
\fun{\vartheta_{M}}{n,\rbk{\omega_0,\omega_1,\ldots,\omega_{M-1}}}
&=
\begin{cases}
\rbk{n,\rbk{\omega_0}},&M=1,\\
\rbk{n,\rbk{\omega_0,\omega_{M-1},\ldots,\omega_1}},&M\geq2.
\end{cases}
\end{aligned}
\label{eq:paper-closed-path-time-reversal}
\end{equation}

\begin{prop}[Criterion for a probability measure and reflection positivity]\label{prop:path-weight-probability-criterion}
Fix $M\in\semigrposint$.
Assume that the path weight $W_M$ in \eqref{eq:paper-closed-path-weight} satisfies
$\fun{W_{M}}{\omega}\geq0$ for every
$\omega\in\prbspsample_{M,\txtclassical}$ and that
$Z_{M}=\sqfun{\trace}{G_{\sminvtemperature,M}}>0$.
Then the set function $\mu_{M}$ in \eqref{eq:paper-normalized-finite-time-path-measure} is a probability measure.
For the time-reversal map $\vartheta_M$ in \eqref{eq:paper-closed-path-time-reversal}, fix a sub-$\sigma$-algebra
$\mblfmlfrak{S}_{M,+}\subset\mblfmlfrak{S}_{M}$, where $\mblfmlfrak{S}_{M}$ is defined in
\eqref{eq:paper-closed-path-sigma-algebra}.
Assume that
\begin{equation}
\begin{aligned}
\fun{W_{M}}{\fun{\vartheta_{M}}{\omega}}
&=
\fun{W_{M}}{\omega}.
\end{aligned}
\label{eq:path-weight-reflection-invariance}
\end{equation}
Assume also that
\begin{equation}
\begin{aligned}
\sum_{\omega\in\prbspsample_{M,\txtclassical}}
\cmpconj{\fun{F\circ\vartheta_{M}}{\omega}}
\fun{F}{\omega}
\fun{W_{M}}{\omega}
&\geq
0
\end{aligned}
\label{eq:path-weight-reflection-input}
\end{equation}
for every bounded $\mblfmlfrak{S}_{M,+}$-measurable function
$F
\colon
\prbspsample_{M,\txtclassical}
\to
\fldcmp$.
Then $\mu_{M}$ is invariant under $\vartheta_{M}$ and is reflection positive with respect to
$\rbk{\mblfmlfrak{S}_{M,+},\vartheta_{M}}$ in the sense of Definition
\ref{def:reflection-positive-probability-measure}.
\end{prop}

\begin{proof}
The measurable space $\prbspsample_{M,\txtclassical}$ carries the full power-set $\sigma$-algebra
\eqref{eq:paper-closed-path-sigma-algebra}, so
$\vartheta_{M}$ in \eqref{eq:paper-closed-path-time-reversal} is measurable.
Direct substitution in that equation gives
$$\begin{aligned}
\vartheta_{M}\circ\vartheta_{M}
&=
\id_{\prbspsample_{M,\txtclassical}}.
\end{aligned}$$
Nonnegativity makes the partial sums over particle numbers monotone.
The cutoff limit in \eqref{eq:paper-intrinsic-path-expansion}, together with the definition of
$Z_{M}$ in \eqref{eq:paper-finite-time-path-partition-function}, gives
$$\begin{aligned}
\sum_{\omega\in\prbspsample_{M,\txtclassical}}
\fun{W_{M}}{\omega}
&=
Z_{M}.
\end{aligned}$$
Thus $\mu_{M}$ is a nonnegative countably additive measure of total mass $1$.
The identity in \eqref{eq:path-weight-reflection-invariance} and involutivity give
$$\begin{aligned}
\fun{\mu_{M}}{\fun{\vartheta_{M}}{E}}
&=
\frac{1}{Z_{M}}
\sum_{\omega\in E}
\fun{W_{M}}{\fun{\vartheta_{M}}{\omega}}
=
\fun{\mu_{M}}{E}.
\end{aligned}$$
Dividing \eqref{eq:path-weight-reflection-input} by $Z_{M}$ proves reflection
positivity.
\end{proof}

Consequently, any probabilistic application of this algebraic formula must verify both nonnegativity and \eqref{eq:path-weight-reflection-input} for the heat kernel under consideration.

\section{Comparison with the standard Fock particle blocks}\label{sec:paper-bosonic-fock-representations}

In this section, we compare the intrinsic particle blocks and particle-cutoff quotients with the standard bosonic Fock realization.

\subsection{Standard Fock objects}\label{sec:paper-standard-fock-objects}

In this subsection, we define the bosonic Fock spaces, the regular Fock representation, and the tensor-factor operations used in the comparison, and we record the known exact-particle formulas that enter as input.

For a closed complex subspace \(\sphilbfrak{k}\subset\sphilbfrak{h}\), including \(\sphilbfrak{k}=\sphilbfrak{h}\), set \begin{equation}
\begin{aligned}
\fun{\spfockgamma_{\txtbsn}}{\sphilbfrak{k}}
&=
\bigoplus_{n\in\monnat}\sphilbfrak{k}^{\otimes_{\txtsym}n},
&
\fun{\spfockgamma_{\txtbsn,\txtfin}}{\sphilbfrak{k}}
&=
\bigcup_{N\in\monnat}\bigoplus_{n=0}^{N}
\sphilbfrak{k}^{\otimes_{\txtsym}n}.
\end{aligned}
\label{eq:paper-main-fock-space-definition}
\end{equation} Identify \(\sphilbfrak{k}^{\otimes_{\txtsym}0}\) with \(\fldcmp\) and write \(\Psi_{\txtfock,\txtvac}=1\in\fldcmp\subset\fun{\spfockgamma_{\txtbsn}}{\sphilbfrak{k}}\) for the vacuum vector; these canonical zero-sector identifications use the same symbol for every \(\sphilbfrak{k}\). For \(f\in\sphilbfrak{k}\), define on the finite-particle subspace \begin{equation}
\begin{aligned}
\fnrestr{\fun{a_{\txtfock,\sphilbfrak{k}}}{f}}
{\sphilbfrak{k}^{\otimes_{\txtsym}n}}
&=
\sqrt{n}\,
\fnrestr{
\rbk{\bkt{f}{\mathord{\cdot}}_{\sphilbfrak{k}}\otimes
1_{\sphilbfrak{k}^{\otimes\rbk{n-1}}}}}
{\sphilbfrak{k}^{\otimes_{\txtsym}n}}
&&\rbk{n\geq1},
&
\fun{a_{\txtfock,\sphilbfrak{k}}}{f}\Psi_{\txtfock,\txtvac}
&=0,
\end{aligned}
\label{eq:paper-main-fock-annihilation-definition}
\end{equation} and let \(\fun{\phi_{\txtfock,\sphilbfrak{k}}}{f}\) be the self-adjoint closure of \(2^{-1/2}\rbk{\faadj{\fun{a_{\txtfock,\sphilbfrak{k}}}{f}}+
\fun{a_{\txtfock,\sphilbfrak{k}}}{f}}\). Theorem X.41(a) of \cite{ReedSimon002} gives essential self-adjointness. The regular Fock representation is \begin{equation}
\begin{aligned}
\oarepn_{\txtfock,\sphilbfrak{k}}
\colon
\oaresolventalgebra(\sphilbfrak{k})
&\to
\opspbddlin{\fun{\spfockgamma_{\txtbsn}}{\sphilbfrak{k}}};
&
\fun{\oarepn_{\txtfock,\sphilbfrak{k}}}{\oaresolvent(\lambda,f)}
&=
\rbk{\imunit\lambda-\fun{\phi_{\txtfock,\sphilbfrak{k}}}{f}}^{-1}
\end{aligned}
\label{eq:paper-main-regular-fock-representation}
\end{equation} for \(\lambda\in\fldreal\setminus\setone{0}\) and \(f\in\sphilbfrak{k}\). This is the representation supplied by the regular representation--resolvent correspondence \cite[Corollary 4.4]{BuchholzGrundling002}. For \(\sphilbfrak{k}=\sphilbfrak{h}\), write \(a_{\txtfock}=a_{\txtfock,\sphilbfrak{h}}\), \(\phi_{\txtfock}=\phi_{\txtfock,\sphilbfrak{h}}\), and \(\oarepn_{\txtfock}=\oarepn_{\txtfock,\sphilbfrak{h}}\). For \(n\in\monnat\), define the exact-particle representation by \begin{equation}
\begin{aligned}
\rho_n
\colon
\oaresolventalgebra(\sphilbfrak{h},\sigma)^{\gamma}
&\to
\opspbddlin{\sphilbfrak{h}^{\otimes_{\txtsym}n}};
&
\fun{\rho_n}{A}
&=
\fnrestr{\fun{\oarepn_{\txtfock}}{A}}
{\sphilbfrak{h}^{\otimes_{\txtsym}n}}.
\end{aligned}
\label{eq:paper-main-exact-particle-representation}
\end{equation}

For a complex Hilbert space \(\sphilbfrak{H}\), write \(\fun{\schattencls}{\sphilbfrak{H}}\) for the norm closure in \(\opspbddlin{\sphilbfrak{H}}\) of the operators with finite-dimensional range.

For a nonzero finite-dimensional complex subspace \(\sphilbfrak{f}\subset\sphilbfrak{h}\) and \(m\in\monnat\), define \begin{equation}
\begin{aligned}
\pi_{\sphilbfrak{f},m}
\colon
e_{\sphilbfrak{f},m}\oa{K}(\sphilbfrak{f})^{\gamma}e_{\sphilbfrak{f},m}
&\to
\fun{\schattencls}{\sphilbfrak{f}^{\otimes_{\txtsym}m}};
&
\fun{\pi_{\sphilbfrak{f},m}}{K}
&=
\fnrestr{\fun{\oarepn_{\txtfock,\sphilbfrak{f}}}{K}}
{\sphilbfrak{f}^{\otimes_{\txtsym}m}}.
\end{aligned}
\label{eq:paper-main-finite-mode-block-representation}
\end{equation}

For a complex Hilbert space \(\sphilbfrak{H}\) and \(\xi,\eta\in\sphilbfrak{H}\), use \begin{equation}
\begin{aligned}
P_{\xi,\eta}
&\in
\opspbddlin{\sphilbfrak{H}},
&
\fun{P_{\xi,\eta}}{\zeta}
&=
\bkt{\eta}{\zeta}_{\sphilbfrak{H}}\xi
&&\rbk{\zeta\in\sphilbfrak{H}}.
\end{aligned}
\label{eq:paper-main-rank-one-operator}
\end{equation} For \(r\in\monnat\), let \begin{equation}
\begin{aligned}
P_r
\colon
\sphilbfrak{h}^{\otimes r}
&\to
\sphilbfrak{h}^{\otimes_{\txtsym}r};
&
\fun{P_r}{\xi_1\otimes\cdots\otimes\xi_r}
&=
\frac{1}{r!}\sum_{\pi\in\grsym{r}}
\xi_{\pi^{-1}(1)}\otimes\cdots\otimes\xi_{\pi^{-1}(r)}.
\end{aligned}
\label{eq:paper-main-symmetrizing-projection}
\end{equation} Let \(0\leq m\leq n\) and \(K\in\fun{\schattencls}{\sphilbfrak{h}^{\otimes_{\txtsym}m}}\). Let \(\widetilde{K}\in\opspbddlin{\sphilbfrak{h}^{\otimes m}}\) agree with \(K\) on \(\sphilbfrak{h}^{\otimes_{\txtsym}m}\) and vanish on its orthogonal complement. If \(S=\setone{s_1<\cdots<s_m}\subset\setone{1,\ldots,n}\) has complement \(\setone{t_1<\cdots<t_{n-m}}\), define \begin{equation}
\begin{aligned}
U_S
&\colon
\sphilbfrak{h}^{\otimes n}
\to
\sphilbfrak{h}^{\otimes m}\otimes\sphilbfrak{h}^{\otimes\rbk{n-m}},
\\
\fun{U_S}{\xi_1\otimes\cdots\otimes\xi_n}
&=
\xi_{s_1}\otimes\cdots\otimes\xi_{s_m}\otimes
\xi_{t_1}\otimes\cdots\otimes\xi_{t_{n-m}},
\\
K^{[S]}
&=
\faadj{U_S}\rbk{\widetilde{K}\otimes
1_{\sphilbfrak{h}^{\otimes\rbk{n-m}}}}U_S.
\end{aligned}
\label{eq:paper-main-tensor-factor-insertion}
\end{equation}

For \(0\leq m\leq n\), define \begin{equation}
\begin{aligned}
\widetilde{E}_{m,n}
\colon
\fun{\schattencls}{\sphilbfrak{h}^{\otimes_{\txtsym}m}}
&\to
\opspbddlin{\sphilbfrak{h}^{\otimes_{\txtsym}n}};
&
\fun{\widetilde{E}_{m,n}}{K}
&=
\fnrestr{
\sum_{\substack{S\subset\setone{1,\ldots,n}\\\abscard{S}=m}}K^{[S]}}
{\sphilbfrak{h}^{\otimes_{\txtsym}n}}.
\end{aligned}
\label{eq:paper-main-subset-sum-embedding}
\end{equation} Let \(p,q\in\monnat\), \(K\in\fun{\schattencls}{\sphilbfrak{h}^{\otimes_{\txtsym}p}}\), and \(L\in\fun{\schattencls}{\sphilbfrak{h}^{\otimes_{\txtsym}q}}\). For \(\max\setone{p,q}\leq r\leq p+q\), define \begin{equation}
\begin{aligned}
\widetilde{\mu}_{p,q}^{r}
&\colon
\fun{\schattencls}{\sphilbfrak{h}^{\otimes_{\txtsym}p}}
\times
\fun{\schattencls}{\sphilbfrak{h}^{\otimes_{\txtsym}q}}
\to
\fun{\schattencls}{\sphilbfrak{h}^{\otimes_{\txtsym}r}};
\quad
\fun{\widetilde{\mu}_{p,q}^{r}}{K,L}
&=
P_r
\sum_{\substack{A,B\subset\setone{1,\ldots,r}\\
\abscard{A}=p,\ \abscard{B}=q\\
A\cup B=\setone{1,\ldots,r}}}
K^{[A]}L^{[B]}P_r.
\end{aligned}
\label{eq:paper-main-subset-union-overlap-map}
\end{equation} Here \(K^{[A]}\) and \(L^{[B]}\) use \eqref{eq:paper-main-tensor-factor-insertion} with \(n=r\).

For a closed complex subspace \(\sphilbfrak{f}\subset\sphilbfrak{h}\), \(n\in\monnat\), \(0\leq k\leq n\), \(\xi\in\sphilbfrak{f}^{\otimes_{\txtsym}k}\), and \(\eta\in\rbk{\sphilbfrak{f}^{\perp}}^{\otimes_{\txtsym}\rbk{n-k}}\), define \begin{equation}
\begin{aligned}
W_{\sphilbfrak{f},n}
\colon
\sphilbfrak{h}^{\otimes_{\txtsym}n}
&\to
\bigoplus_{k=0}^{n}\sphilbfrak{f}^{\otimes_{\txtsym}k}\otimes
\rbk{\sphilbfrak{f}^{\perp}}^{\otimes_{\txtsym}\rbk{n-k}},
&
\fun{\faadj{W_{\sphilbfrak{f},n}}}{\xi\otimes\eta}
&=
\binom{n}{k}^{1/2}P_n\rbk{\xi\otimes\eta},
\end{aligned}
\label{eq:paper-main-particle-fock-decomposition-unitary}
\end{equation} and \begin{equation}
\begin{aligned}
U_{\sphilbfrak{f}}
\colon
\fun{\spfockgamma_{\txtbsn}}{\sphilbfrak{h}}
&\to
\fun{\spfockgamma_{\txtbsn}}{\sphilbfrak{f}}\otimes
\fun{\spfockgamma_{\txtbsn}}{\sphilbfrak{f}^{\perp}};
&
\fun{U_{\sphilbfrak{f}}}{\bigoplus_{n\in\monnat}\xi_n}
&=
\bigoplus_{n\in\monnat}\fun{W_{\sphilbfrak{f},n}}{\xi_n}.
\end{aligned}
\label{eq:paper-main-fock-decomposition-unitary}
\end{equation}

For \(n\geq1\), the exact-particle representations define the connecting map \begin{equation}
\begin{aligned}
\kappa_n
\colon
\fun{\rho_n}{\oaresolventalgebra(\sphilbfrak{h},\sigma)^{\gamma}}
\to
\fun{\rho_{n-1}}{\oaresolventalgebra(\sphilbfrak{h},\sigma)^{\gamma}};\quad
\fun{\kappa_n\circ\rho_n}{A}
&=
\fun{\rho_{n-1}}{A},
\end{aligned}
\label{eq:paper-main-fock-connecting-map}
\end{equation} which is well defined because \(\Ker\rho_n\subset\Ker\rho_{n-1}\) \cite[Lemma 3.4]{DetlevBuchholz002}. For \(n\in\monnat\) and \(0\leq m\leq n\), the concrete exact-particle algebra on \(\sphilbfrak{h}^{\otimes_{\txtsym}n}\) and, when \(n\geq1\), its connecting map satisfy the following formulas for \(K\in\fun{\schattencls}{\sphilbfrak{h}^{\otimes_{\txtsym}m}}\): \begin{equation}
\begin{aligned}
\fun{\rho_n}{\oaresolventalgebra(\sphilbfrak{h},\sigma)^\gamma}
&=
\sum_{r=0}^{n}
\fun{\widetilde{E}_{r,n}}{
\fun{\schattencls}{\sphilbfrak{h}^{\otimes_{\txtsym}r}}},
&
\fun{\kappa_n\circ\widetilde{E}_{m,n}}{K}
&=
\begin{cases}
\fun{\widetilde{E}_{m,n-1}}{K},&m<n,\\
0,&m=n.
\end{cases}
\end{aligned}
\label{eq:paper-literature-fock-input}
\end{equation} These are \cite[Lemmas 3.3 and 3.4]{DetlevBuchholz002}; in the normalization used here the binomial factors are absorbed into \(\widetilde{E}_{m,n}\). The compatible bounded sequences form the concrete Buchholz algebra by \cite[Theorem 3.5]{DetlevBuchholz002}. The same statements for \(\fun{\lpseq^{2}}{\Gamma}\) are \cite[Lemma 2.3 and Proposition 2.4]{DeuchertLampartLemm001}; the formulation for a general separable infinite-dimensional one-particle space is obtained by transport with a one-particle unitary. What remains to prove is the identification of this known concrete system with the intrinsic projections, ideals, and quotient maps constructed above.

\subsection{Standard particle blocks}\label{standard-particle-blocks}

In this subsection, the finite-dimensional Stone--von Neumann representation identifies the intrinsic particle blocks with the standard exact-particle Fock blocks.

Let \(\sphilbfrak{f}\subset\sphilbfrak{h}\) be a nonzero finite-dimensional complex subspace. Choose a faithful regular irreducible representation \(\pi_{\sphilbfrak{f}}\colon\oaresolventalgebra(\sphilbfrak{f})\to\opspbddlin{\sphilbfrak{H}_{\sphilbfrak{f}}}\). Finite-dimensional Stone--von Neumann uniqueness \cite[Corollary 4.4]{BuchholzGrundling002} identifies this representation with the regular Fock representation \(\oarepn_{\txtfock,\sphilbfrak{f}}\). Choose a unitary \[\begin{aligned}
S_{\sphilbfrak{f}}
\colon
\sphilbfrak{H}_{\sphilbfrak{f}}
\to
\fun{\spfockgamma_{\txtbsn}}{\sphilbfrak{f}};\quad
S_{\sphilbfrak{f}}\fun{\pi_{\sphilbfrak{f}}}{A}
&=
\fun{\oarepn_{\txtfock,\sphilbfrak{f}}}{A}S_{\sphilbfrak{f}}
\end{aligned}\] for \(A\in\oaresolventalgebra(\sphilbfrak{f})\). The second-quantization unitary is introduced only in this comparison representation, \[\begin{aligned}
\fun{\Gamma_{\sphilbfrak{f}}}{\theta}
\colon
\fun{\spfockgamma_{\txtbsn}}{\sphilbfrak{f}}
\to
\fun{\spfockgamma_{\txtbsn}}{\sphilbfrak{f}};\quad
\fnrestr{\fun{\Gamma_{\sphilbfrak{f}}}{\theta}}{\sphilbfrak{f}^{\otimes_{\txtsym}m}}
&=
\napiernum^{\imunit m\theta}.
\end{aligned}\] The transported unitary group and its action on the representation are \[\begin{aligned}
\fun{V_{\sphilbfrak{f}}}{\theta}
&=
\faadj{S_{\sphilbfrak{f}}}\fun{\Gamma_{\sphilbfrak{f}}}{\theta}S_{\sphilbfrak{f}},
&
\fun{V_{\sphilbfrak{f}}}{\theta}\fun{\pi_{\sphilbfrak{f}}}{A}\faadj{\fun{V_{\sphilbfrak{f}}}{\theta}}
&=
\fun{\pi_{\sphilbfrak{f}}\circ\gamma_{\theta}}{A}
&&\rbk{\theta\in\fldreal,\ A\in\oaresolventalgebra(\sphilbfrak{f})}.
\end{aligned}\] The strongly continuous unitary group \(\theta\mapsto \fun{V_{\sphilbfrak{f}}}{\theta}\) has, by Stone's theorem \cite[Theorem VIII.4]{ReedSimon003}, a unique self-adjoint generator \[\begin{aligned}
N_{\sphilbfrak{f}}
\colon
\dom N_{\sphilbfrak{f}}\subset\sphilbfrak{H}_{\sphilbfrak{f}}
\to
\sphilbfrak{H}_{\sphilbfrak{f}},
\end{aligned}\] which satisfies \(\fun{V_{\sphilbfrak{f}}}{\theta}=\napiernum^{\imunit\theta N_{\sphilbfrak{f}}}\) for every \(\theta\in\fldreal\). Normalize this generator by \(\min\opvarspec{N_{\sphilbfrak{f}}}=0\). For each \(m\in\monnat\), the spectral projection \(\fun{\msrprb_{N_{\sphilbfrak{f}}}}{\setone{m}}\) has finite rank and belongs to \(\fun{\pi_{\sphilbfrak{f}}}{\oa{K}(\sphilbfrak{f})}=\fun{\schattencls}{\sphilbfrak{H}_{\sphilbfrak{f}}}\). Faithfulness gives a unique \(e_{\sphilbfrak{f},m}\in\oa{K}(\sphilbfrak{f})\) with \(\fun{\pi_{\sphilbfrak{f}}}{e_{\sphilbfrak{f},m}}
=\fun{\msrprb_{N_{\sphilbfrak{f}}}}{\setone{m}}\).

\begin{prop}[Representation independence of the blocks]\label{prop:paper-particle-block-independence}
The projection $e_{\sphilbfrak{f},m}$ is independent of the faithful regular irreducible representation used in its
construction.
The family $\seq{e_{\sphilbfrak{f},m}}{m\in\monnat}$ satisfies Definition
\ref{def:intrinsic-particle-block-system}.
\end{prop}

\begin{proof}
Let $\pi'_{\sphilbfrak{f}}$ be another faithful regular irreducible representation, let
$\fun{V'_{\sphilbfrak{f}}}{\theta}=\napiernum^{\imunit\theta N'_{\sphilbfrak{f}}}$ be its normalized gauge implementation, and choose a unitary
$W\colon\sphilbfrak{H}_{\sphilbfrak{f}}\to\sphilbfrak{H}'_{\sphilbfrak{f}}$ with
$$\begin{aligned}
W\fun{\pi_{\sphilbfrak{f}}}{A}\faadj{W}
&=
\fun{\pi'_{\sphilbfrak{f}}}{A}
\end{aligned}$$
for every $A\in\oaresolventalgebra(\sphilbfrak{f})$.
The two gauge implementations differ, after transport by $W$, by a strongly continuous scalar character.
Consequently Stone's theorem \cite[Theorem VIII.4]{ReedSimon003} gives a real $c$ such that
$$\begin{aligned}
\fun{V'_{\sphilbfrak{f}}}{\theta}
&=
\napiernum^{\imunit c\theta}W\fun{V_{\sphilbfrak{f}}}{\theta}\faadj{W},
&
N'_{\sphilbfrak{f}}
&=
WN_{\sphilbfrak{f}}\faadj{W}+c.
\end{aligned}$$
Both generators have minimum spectral value $0$, and therefore $c=0$.
Their spectral projections satisfy
$$\begin{aligned}
\fun{\msrprb_{N'_{\sphilbfrak{f}}}}{\setone{m}}
&=
W\fun{\msrprb_{N_{\sphilbfrak{f}}}}{\setone{m}}\faadj{W}.
\end{aligned}$$
The inverse image of the projection is consequently independent of the representation.

In the Fock realization, $N_{\sphilbfrak{f}}$ is the particle-number operator.
Its spectral projections are
therefore mutually orthogonal and finite rank, and functional calculus gives
$e_{\sphilbfrak{f},m}\in\oa{K}(\sphilbfrak{f})^\gamma$.
Since
$\fun{\pi_{\sphilbfrak{f}}}{\oa{K}(\sphilbfrak{f})}$ is the compact-operator algebra, its gauge-fixed part is
$$\begin{aligned}
\fun{\pi_{\sphilbfrak{f}}}{\oa{K}(\sphilbfrak{f})^\gamma}
&=
\set{\bigoplus_{m\in\monnat}T_m}
{T_m\in\opspbddlin{\fun{\msrprb_{N_{\sphilbfrak{f}}}}{\setone{m}}\sphilbfrak{H}_{\sphilbfrak{f}}},
\lim_{m\to\infty}\norm{T_m}=0}.
\end{aligned}$$
This standard spectral decomposition proves orthogonality, the full-matrix corner property, and the unique
norm-convergent decomposition.
The appearance of the same blocks in the exact-particle algebra is the
known decomposition \eqref{eq:paper-literature-fock-input}; no further Fock-space calculation is needed.
\end{proof}

\begin{prop}[Verification of the coefficient connecting maps]\label{prop:paper-coefficient-connecting-maps}
Let $m\in\monnat$, and let
$$\begin{aligned}
0
&\neq
\sphilbfrak{f}
\subset
\sphilbfrak{g}
\subset
\sphilbfrak{k}
\subset
\sphilbfrak{h}
\end{aligned}$$
be finite-dimensional complex subspaces.
The compression map $j_{\sphilbfrak{f},\sphilbfrak{g},m}$ in
\eqref{eq:paper-coefficient-connecting-map} is an injective $\ast$-homomorphism and satisfies
$$\begin{aligned}
j_{\sphilbfrak{g},\sphilbfrak{k},m}\circ j_{\sphilbfrak{f},\sphilbfrak{g},m}
&=
j_{\sphilbfrak{f},\sphilbfrak{k},m}.
\end{aligned}$$
\end{prop}

\begin{proof}
Define the symmetric-tensor inclusion by
$$\begin{aligned}
I_{\sphilbfrak{f},\sphilbfrak{g},m}
\colon
\sphilbfrak{f}^{\otimes_{\txtsym}m}
\to
\sphilbfrak{g}^{\otimes_{\txtsym}m};\quad
\fun{I_{\sphilbfrak{f},\sphilbfrak{g},m}}{\xi}
&=
\xi.
\end{aligned}$$
The decomposition of the $m$-particle space induced by $\sphilbfrak{g}=\sphilbfrak{f}\oplus\rbk{\sphilbfrak{g}\ominus \sphilbfrak{f}}$ is
$$\begin{aligned}
\sphilbfrak{g}^{\otimes_{\txtsym}m}
&\simeq
\bigoplus_{r=0}^{m}
\sphilbfrak{f}^{\otimes_{\txtsym}r}\otimes
\rbk{\sphilbfrak{g}\ominus \sphilbfrak{f}}^{\otimes_{\txtsym}\rbk{m-r}}.
\end{aligned}$$
For $K\in e_{\sphilbfrak{f},m}\oa{K}(\sphilbfrak{f})^{\gamma}e_{\sphilbfrak{f},m}$, compression by $e_{\sphilbfrak{g},m}$ selects the $r=m$ summand and gives
$$\begin{aligned}
\fnrestr{
S_{\sphilbfrak{g}}\fun{\pi_{\sphilbfrak{g}}\circ j_{\sphilbfrak{f},\sphilbfrak{g},m}}{K}\faadj{S_{\sphilbfrak{g}}}}
{\sphilbfrak{g}^{\otimes_{\txtsym}m}}
&=
I_{\sphilbfrak{f},\sphilbfrak{g},m}
\fnrestr{S_{\sphilbfrak{f}}\fun{\pi_{\sphilbfrak{f}}}{K}\faadj{S_{\sphilbfrak{f}}}}
{\sphilbfrak{f}^{\otimes_{\txtsym}m}}
\faadjrbk{I_{\sphilbfrak{f},\sphilbfrak{g},m}}.
\end{aligned}$$
The right-hand side is zero extension from $\sphilbfrak{f}^{\otimes_{\txtsym}m}$ to
$\sphilbfrak{g}^{\otimes_{\txtsym}m}$.
For $K,L\in e_{\sphilbfrak{f},m}\oa{K}(\sphilbfrak{f})^{\gamma}e_{\sphilbfrak{f},m}$, zero extension gives
$$\begin{aligned}
\fun{j_{\sphilbfrak{f},\sphilbfrak{g},m}}{KL}
&=
\fun{j_{\sphilbfrak{f},\sphilbfrak{g},m}}{K}\fun{j_{\sphilbfrak{f},\sphilbfrak{g},m}}{L},
&
\fun{j_{\sphilbfrak{f},\sphilbfrak{g},m}}{\faadj{K}}
&=
\faadjrbk{\fun{j_{\sphilbfrak{f},\sphilbfrak{g},m}}{K}},
\\
\norm{\fun{j_{\sphilbfrak{f},\sphilbfrak{g},m}}{K}}
&=
\norm{K}.
\end{aligned}$$
Faithfulness transfers these identities back to the resolvent algebra.
Zero extension through $\sphilbfrak{g}^{\otimes_{\txtsym}m}$ equals direct zero extension to
$\sphilbfrak{k}^{\otimes_{\txtsym}m}$, which gives
$j_{\sphilbfrak{g},\sphilbfrak{k},m}\circ j_{\sphilbfrak{f},\sphilbfrak{g},m}=j_{\sphilbfrak{f},\sphilbfrak{k},m}$.
\end{proof}

The finite-mode block representation is the map \(\pi_{\sphilbfrak{f},m}\) defined in \eqref{eq:paper-main-finite-mode-block-representation}. The unitary \(S_{\sphilbfrak{f}}\) satisfying \(S_{\sphilbfrak{f}}\fun{\pi_{\sphilbfrak{f}}}{A}
=\fun{\oarepn_{\txtfock,\sphilbfrak{f}}}{A}S_{\sphilbfrak{f}}\) for every \(A\in\oaresolventalgebra(\sphilbfrak{f})\) shows that this representation is a \(\ast\)-isomorphism.

\begin{thm}[Fock realization of the intrinsic particle blocks]\label{thm:paper-fock-intrinsic-block-realization}
Let $0\neq\sphilbfrak{f}\subset\sphilbfrak{h}$ be finite-dimensional.
For $m,n\in\monnat$, the Fock images of $e_{\sphilbfrak{f},m}$ and
$Q_{\sphilbfrak{f},>n}$ are the orthogonal projections onto
$\sphilbfrak{f}^{\otimes_{\txtsym}m}$ and
$\bigoplus_{r>n}\sphilbfrak{f}^{\otimes_{\txtsym}r}$, respectively.
The intrinsic minimal ideal and its gauge-fixed part satisfy
\begin{equation}
\begin{aligned}
\fun{\oarepn_{\txtfock,\sphilbfrak{f}}}{\oa{K}(\sphilbfrak{f})}
&=
\fun{\schattencls}{\fun{\spfockgamma_{\txtbsn}}{\sphilbfrak{f}}},
\\ 
\fun{\oarepn_{\txtfock,\sphilbfrak{f}}}{\oa{K}(\sphilbfrak{f})^{\gamma}}
&=
\set{\bigoplus_{m\in\monnat}K_m}
{K_m\in\opspbddlin{\sphilbfrak{f}^{\otimes_{\txtsym}m}},
\lim_{m\to\infty}\norm{K_m}=0}.
\end{aligned}
\label{eq:paper-main-fock-intrinsic-block-identification}
\end{equation}
The block isomorphism in \eqref{eq:paper-main-finite-mode-block-representation} identifies
$e_{\sphilbfrak{f},m}\oa{K}(\sphilbfrak{f})^{\gamma}e_{\sphilbfrak{f},m}$ with
$\fun{\schattencls}{\sphilbfrak{f}^{\otimes_{\txtsym}m}}$.
\end{thm}

\begin{proof}
The regular Fock representation $\oarepn_{\txtfock,\sphilbfrak{f}}$ is faithful, regular, and irreducible
\cite[Corollary 4.4]{BuchholzGrundling002}.
Equation \eqref{eq:paper-minimal-ideal-representation} in Theorem \ref{thm:abstract-minimal-compact-ideal}
therefore gives the first identity in \eqref{eq:paper-main-fock-intrinsic-block-identification}.
Proposition \ref{prop:paper-particle-block-independence} identifies the Fock image of $e_{\sphilbfrak{f},m}$
with the projection onto $\sphilbfrak{f}^{\otimes_{\txtsym}m}$ and gives the gauge-fixed block decomposition,
which is the second identity.
Substitution of these projections into \eqref{eq:abstract-high-particle-projection} gives the Fock image of
$Q_{\sphilbfrak{f},>n}$.
The block identification is the $\ast$-isomorphism $\pi_{\sphilbfrak{f},m}$ in the preceding paragraph.
\end{proof}

Proposition \ref{prop:paper-particle-block-independence} resolves the representation choice but does not by itself fix the additive normalization of the generator \(N_{\sphilbfrak{f}}\). The condition \(\min\opvarspec{N_{\sphilbfrak{f}}}=0\) and the use of \(\fun{\msrprb_{N_{\sphilbfrak{f}}}}{\setone{m}}\) to construct \(e_{\sphilbfrak{f},m}\) fix the vacuum normalization and the nonnegative spectral labeling. The abstract Buchholz algebra in Definition \ref{def:paper-abstract-buchholz-algebra} is therefore canonical relative to these conditions and the projections \(e_{\sphilbfrak{f},m}\) of Proposition \ref{prop:paper-particle-block-independence}.

\subsection{Comparison with the concrete exact-particle algebras}\label{comparison-with-the-concrete-exact-particle-algebras}

In this subsection, we identify the abstract particle-cutoff quotients and connecting maps with their concrete exact-particle Fock counterparts.

Use the global regular Fock representation in \eqref{eq:paper-main-regular-fock-representation} and its exact-particle restriction \(\rho_n\) in \eqref{eq:paper-main-exact-particle-representation}.

\begin{thm}[Comparison of abstract and concrete particle-cutoff quotients]\label{thm:paper-abstract-concrete-comparison}
For the projections $e_{\sphilbfrak{f},m}$ constructed in Proposition
\ref{prop:paper-particle-block-independence} as inverse images of
$\fun{\msrprb_{N_{\sphilbfrak{f}}}}{\setone{m}}$ and every $n\in\monnat$, the cutoff ideal and the induced map satisfy
\begin{equation}
\begin{aligned}
\oaideal{I}_{>n}(\sphilbfrak{h})
&=
\Ker\rho_n,
\\ 
\varpi_n
\colon
\oa{B}_n(\sphilbfrak{h})
&\to
\fun{\rho_n}{\oaresolventalgebra(\sphilbfrak{h},\sigma)^{\gamma}};
\quad
\fun{\varpi_n}{A+\oaideal{I}_{>n}(\sphilbfrak{h})}
=
\fun{\rho_n}{A}.
\end{aligned}
\label{eq:paper-main-fock-cutoff-identification}
\end{equation}
The map $\varpi_n$ is an isometric $\ast$-isomorphism onto the concrete exact-particle Buchholz algebra
$\fun{\rho_n}{\oaresolventalgebra(\sphilbfrak{h},\sigma)^\gamma}$.
For $n\geq1$, the connecting map \eqref{eq:paper-main-fock-connecting-map} satisfies
$\kappa_n=\varpi_{n-1}\circ r_n\circ\varpi_n^{-1}$, and the bounded inverse limits are isometrically
$\ast$-isomorphic.
\end{thm}

\begin{proof}
The concrete range of $\rho_n$, the deletion map $\kappa_n$, and the bounded inverse limit are already
\eqref{eq:paper-literature-fock-input} and
\cite[Theorem 3.5]{DetlevBuchholz002}.
Thus only the equality between the new intrinsic cutoff ideal and the
kernel of the known exact-particle representation requires proof.

The comparison representation sends $e_{\sphilbfrak{f},m}$ to the local $m$-particle projection.
Every generator $Q_{\sphilbfrak{f},>n}$ of $\oaideal{I}_{>n}(\sphilbfrak{h})$ is annihilated in the exact
$n$-particle representation, and
$$\begin{aligned}
\oaideal{I}_{>n}(\sphilbfrak{h})
&\subset
\Ker\rho_n.
\end{aligned}$$

For the reverse inclusion, take $A\in\Ker\rho_n$ and $\varepsilon>0$.
The norm-density statement following \cite[Lemma 3.1]{DetlevBuchholz002}
gives a finite-dimensional $\sphilbfrak{f}\subset\sphilbfrak{h}$ and
$C\in\oaresolventalgebra(\sphilbfrak{f})^\gamma$ such that $\norm{A-C}<\varepsilon$.
Define
$$\begin{aligned}
C_{\leq n}
&=
\rbk{\sum_{m=0}^{n}e_{\sphilbfrak{f},m}}C.
\end{aligned}$$
The decomposition
$$\begin{aligned}
\sphilbfrak{h}^{\otimes_{\txtsym}n}
&\simeq
\bigoplus_{m=0}^{n}
\sphilbfrak{f}^{\otimes_{\txtsym}m}
\otimes
\rbk{\sphilbfrak{f}^\perp}^{\otimes_{\txtsym}\rbk{n-m}}
\end{aligned}$$
has no zero summand because $\sphilbfrak{f}$ is finite-dimensional and $\sphilbfrak{h}$ is infinite-dimensional.
The block norms therefore satisfy
$$\begin{aligned}
\norm{C_{\leq n}}
=
\max_{0\leq m\leq n}
\norm{\fnrestr{\fun{\oarepn_{\txtfock}}{C}}{\sphilbfrak{f}^{\otimes_{\txtsym}m}}}
=
\norm{\fun{\rho_n}{C}}
=
\norm{\fun{\rho_n}{C-A}}
\leq
\norm{C-A}
<
\varepsilon.
\end{aligned}$$
The complementary part satisfies
$$\begin{aligned}
C-C_{\leq n}
&=
Q_{\sphilbfrak{f},>n}C
\in
\oaideal{I}_{>n}(\sphilbfrak{h}).
\end{aligned}$$
The two estimates give
$$\begin{aligned}
\fun{\topdist}{A,\oaideal{I}_{>n}(\sphilbfrak{h})}
\leq
\norm{A-\rbk{C-C_{\leq n}}}
\leq
\norm{A-C}+\norm{C_{\leq n}}
<
2\varepsilon.
\end{aligned}$$
Closedness of the ideal and arbitrariness of $\varepsilon$ give
$A\in\oaideal{I}_{>n}(\sphilbfrak{h})$.
The equality $\Ker\rho_n=\oaideal{I}_{>n}(\sphilbfrak{h})$ makes the map $\varpi_n$ in
\eqref{eq:paper-main-fock-cutoff-identification} a well-defined injective
$\ast$-homomorphism from $\oa{B}_n(\sphilbfrak{h})$ onto the corresponding concrete exact-$n$-particle algebra.
It is therefore bijective.
Both $\varpi_n$ in \eqref{eq:paper-main-fock-cutoff-identification} and its
algebraic inverse are $\ast$-homomorphisms and hence
contractive, so
$$\begin{aligned}
\norm{\fun{\varpi_n}{A_n}}
&\leq
\norm{A_n}
\leq
\norm{\fun{\varpi_n}{A_n}}
\end{aligned}$$
for every $A_n\in\oa{B}_n(\sphilbfrak{h})$.
Under these quotient isomorphisms, $r_n$ becomes the already known $\kappa_n$ in
\eqref{eq:paper-literature-fock-input}.
The bounded inverse limits and their supremum norms therefore agree by
\cite[Theorem 3.5]{DetlevBuchholz002}.
\end{proof}

\begin{cor}[Faithful quotient-coordinate embedding]\label{cor:paper-faithful-quotient-coordinate-embedding}
For the projections $e_{\sphilbfrak{f},m}$ constructed in Proposition
\ref{prop:paper-particle-block-independence}, the quotient-coordinate map $\rho$ in Proposition
\ref{prop:bounded-inverse-limit-cstar} is an injective unital $\ast$-homomorphism.
It is therefore an isometric $\ast$-embedding.
\end{cor}

\begin{proof}
Theorem \ref{thm:paper-abstract-concrete-comparison} identifies
$\oaideal{I}_{>n}(\sphilbfrak{h})$ with $\Ker\rho_n$ for every $n\in\monnat$.
The restricted Fock representation is the faithful direct sum of these exact-particle representations, as
explained in Remark \ref{rem:paper-abstract-faithfulness-criterion}.
Equation \eqref{eq:paper-abstract-faithfulness-criterion} gives $\Ker\rho=0$.
An injective $\ast$-homomorphism between $\oacstar$-algebras is isometric.
\end{proof}

\begin{cor}[Fock realization of the abstract Buchholz algebra]\label{cor:paper-fock-cutoff-realization}
Denote also by $\oarepn_{\txtfock}$ the map
$$\begin{aligned}
\oarepn_{\txtfock}
&\colon
\fun{\oa{B}}{\sphilbfrak{h}}
\to
\opspbddlin{\fun{\spfockgamma_{\txtbsn}}{\sphilbfrak{h}}};
&
\fun{\oarepn_{\txtfock}}{A}
&=
\bigoplus_{n\in\monnat}\fun{\varpi_n\circ p_n}{A}.
\end{aligned}$$
It is a faithful representation of $\fun{\oa{B}}{\sphilbfrak{h}}$, and its image is
\begin{equation}
\begin{aligned}
\fun{\oarepn_{\txtfock}}{\fun{\oa{B}}{\sphilbfrak{h}}}
&=
\set{\bigoplus_{n\in\monnat}T_n}
{\parbox{23em}{
$T_n\in\fun{\varpi_n}{\oa{B}_n(\sphilbfrak{h})}\rbk{n\in\monnat}$,
$\fun{\kappa_n}{T_n}=T_{n-1}\rbk{n\in\semigrposint}$,
$\sup_{n\in\monnat}\norm{T_n}<\infty$.}}.
\end{aligned}
\label{eq:paper-main-fock-bounded-inverse-limit-identification}
\end{equation}
For $A\in\fun{\oa{B}}{\sphilbfrak{h}}$ and $n\in\monnat$, its $n$-particle restriction is
\begin{equation}
\begin{aligned}
\fnrestr{\fun{\oarepn_{\txtfock}}{A}}{\sphilbfrak{h}^{\otimes_{\txtsym}n}}
&=
\fun{\varpi_n\circ p_n}{A}.
\end{aligned}
\label{eq:paper-main-fock-coordinate-identification}
\end{equation}
\end{cor}

\begin{proof}
The supremum bound in \eqref{eq:paper-main-cutoff-inverse-limit} makes the direct sum bounded.
Theorem \ref{thm:paper-abstract-concrete-comparison} shows that every $\varpi_n$ is an isometric
$\ast$-isomorphism, which makes the representation faithful.
The same theorem gives $\kappa_n\circ\varpi_n=\varpi_{n-1}\circ r_n$.
A bounded family $\seq{A_n}{n\in\monnat}$ with $A_n\in\oa{B}_n(\sphilbfrak{h})$ therefore satisfies
$\fun{r_n}{A_n}=A_{n-1}$ for every $n\geq1$ exactly when $T_n=\fun{\varpi_n}{A_n}$ satisfies
$\fun{\kappa_n}{T_n}=T_{n-1}$, which proves \eqref{eq:paper-main-fock-bounded-inverse-limit-identification}.
The restriction formula \eqref{eq:paper-main-fock-coordinate-identification} is the definition of the direct sum.
\end{proof}

\section{Particle coefficients in the standard Fock realization}\label{sec:concrete-particle-coefficient-realization}

In this section, the Fock representation identifies the abstract particle-degree coefficients and the finite-mode occupation projections. The same realization supplies the concrete examples used to analyze the quasilocal subalgebra and its quotient.

The regular Fock representation of the resolvent algebra is defined in \eqref{eq:paper-main-regular-fock-representation}; its faithful extension to the abstract Buchholz algebra is defined in Corollary \ref{cor:paper-fock-cutoff-realization} by the coordinate formula \eqref{eq:paper-main-fock-coordinate-identification}.

\subsection{Particle-degree coefficients in the comparison representation}\label{particle-degree-coefficients-in-the-comparison-representation}

In this subsection, subset-sum embeddings realize the abstract particle-degree coefficients in the concrete Fock representation.

Use the tensor-factor insertion \(K^{[S]}\) and subset-sum map \(\widetilde{E}_{m,n}\) defined in \eqref{eq:paper-main-tensor-factor-insertion} and \eqref{eq:paper-main-subset-sum-embedding}.

\begin{ex}[Action of the tensor-factor permutation]
For $n=5$, $m=2$, and $S=\setone{2,5}$, the complement is $\setone{1,3,4}$ and
$$\begin{aligned}
\fun{U_{\setone{2,5}}}{\xi_1\otimes\xi_2\otimes\xi_3\otimes\xi_4\otimes\xi_5}
&=
\xi_2\otimes\xi_5\otimes\xi_1\otimes\xi_3\otimes\xi_4.
\end{aligned}$$
The operator acting on the second and fifth tensor factors is consequently
$$\begin{aligned}
K^{[\setone{2,5}]}
&=
\faadj{U_{\setone{2,5}}}
\rbk{\widetilde{K}\otimes1_{\sphilbfrak{h}^{\otimes3}}}
U_{\setone{2,5}}.
\end{aligned}$$
\end{ex}

Let \(p,q,r\in\monnat\) satisfy \(\max\setone{p,q}\leq r\leq p+q\), and let \(K\in\fun{\schattencls}{\sphilbfrak{h}^{\otimes_{\txtsym}p}}\) and \(L\in\fun{\schattencls}{\sphilbfrak{h}^{\otimes_{\txtsym}q}}\). Use the symmetrizing projection and subset-union map defined in \eqref{eq:paper-main-symmetrizing-projection} and \eqref{eq:paper-main-subset-union-overlap-map}.

\begin{prop}[Overlap product]\label{prop:paper-overlap-product}
For $n\geq\max\setone{p,q}$,
\begin{equation}
\begin{aligned}
\fun{\widetilde{E}_{p,n}}{K}\fun{\widetilde{E}_{q,n}}{L}
&=
\sum_{r=\max\setone{p,q}}^{\min\setone{n,p+q}}
\fun{\widetilde{E}_{r,n}\circ\widetilde{\mu}_{p,q}^{r}}{K,L}.
\end{aligned}
\label{eq:paper-exact-overlap-product}
\end{equation}
The overlap map also satisfies the norm bound
\begin{equation}
\begin{aligned}
\norm{\fun{\widetilde{\mu}_{p,q}^{r}}{K,L}}
&\leq
\frac{r!}{\rbk{r-p}!\rbk{r-q}!\rbk{p+q-r}!}
\norm{K}\norm{L}.
\end{aligned}
\label{eq:paper-overlap-product-bound}
\end{equation}
\end{prop}

\begin{proof}
Multiplying the two subset sums produces terms $K^{[A]}L^{[B]}$.
Their union size is
$$\begin{aligned}
\abscard{A\cup B}
&=
p+q-\abscard{A\cap B}.
\end{aligned}$$
Grouping first by $r=\abscard{A\cup B}$ and then by the union
$U=A\cup B\subset\setone{1,\ldots,n}$ gives exactly the right-hand side of
\eqref{eq:paper-exact-overlap-product}.

We verify compactness before completing the regrouping.
Fix vectors
$$\begin{aligned}
\xi,\eta
&\in
\sphilbfrak{h}^{\otimes\abscard{A\setminus B}}\otimes
\sphilbfrak{h}^{\otimes\abscard{A\cap B}},
&
\zeta,\theta
&\in
\sphilbfrak{h}^{\otimes\abscard{A\cap B}}\otimes
\sphilbfrak{h}^{\otimes\abscard{B\setminus A}}.
\end{aligned}$$
We first prove compactness for the rank-one inputs $P_{\xi,\eta}$ and $P_{\zeta,\theta}$.
For orthonormal bases $\seq{e_i}{i\in I}$, $\seq{h_a}{a\in I_0}$, and
$\seq{g_j}{j\in J}$ of $\sphilbfrak{h}^{\otimes\abscard{A\setminus B}}$,
$\sphilbfrak{h}^{\otimes\abscard{A\cap B}}$, and
$\sphilbfrak{h}^{\otimes\abscard{B\setminus A}}$, respectively, define
$$\begin{aligned}
R_{\eta,\zeta}
\colon
\sphilbfrak{h}^{\otimes\abscard{A\setminus B}}
\to
\sphilbfrak{h}^{\otimes\abscard{B\setminus A}};\quad
\bkt{g_j}{\fun{R_{\eta,\zeta}}{e_i}}
&=
\sum_{a\in I_0}
\cmpconj{\bkt{e_i\otimes h_a}{\eta}}
\bkt{h_a\otimes g_j}{\zeta}.
\end{aligned}$$
The Cauchy--Schwarz inequality gives the Hilbert--Schmidt estimate
$$\begin{aligned}
\norm{R_{\eta,\zeta}}_2^2
&=
\sum_{i\in I}\sum_{j\in J}
\abs{
\sum_{a\in I_0}
\cmpconj{\bkt{e_i\otimes h_a}{\eta}}
\bkt{h_a\otimes g_j}{\zeta}}^2
\\
&\leq
\sum_{i\in I}\sum_{j\in J}
\rbk{
\sum_{a\in I_0}
\abs{\bkt{e_i\otimes h_a}{\eta}}^2}
\rbk{
\sum_{a\in I_0}
\abs{\bkt{h_a\otimes g_j}{\zeta}}^2}
\\
&=
\norm{\eta}^2\norm{\zeta}^2.
\end{aligned}$$
Define the bounded operators $Y_{\theta}$ and $V_{\xi}$ by
$$\begin{aligned}
Y_{\theta}
&\colon
\sphilbfrak{h}^{\otimes\abscard{A\setminus B}}
\otimes\sphilbfrak{h}^{\otimes\abscard{A\cap B}}
\otimes\sphilbfrak{h}^{\otimes\abscard{B\setminus A}}
\to
\sphilbfrak{h}^{\otimes\abscard{A\setminus B}};
&
\fun{Y_{\theta}}{\alpha\otimes\sminvtemperature}
&=
\bkt{\theta}{\sminvtemperature}\alpha,
\\
V_{\xi}
&\colon
\sphilbfrak{h}^{\otimes\abscard{B\setminus A}}
\to
\sphilbfrak{h}^{\otimes\abscard{A\setminus B}}
\otimes\sphilbfrak{h}^{\otimes\abscard{A\cap B}}
\otimes\sphilbfrak{h}^{\otimes\abscard{B\setminus A}};
&
\fun{V_{\xi}}{\delta}
&=
\xi\otimes\delta,
\end{aligned}$$
where $\alpha\in\sphilbfrak{h}^{\otimes\abscard{A\setminus B}}$,
$\sminvtemperature\in\sphilbfrak{h}^{\otimes\abscard{A\cap B}}
\otimes\sphilbfrak{h}^{\otimes\abscard{B\setminus A}}$, and
$\delta\in\sphilbfrak{h}^{\otimes\abscard{B\setminus A}}$.
After the tensor factors are permuted into the order
$A\setminus B,A\cap B,B\setminus A$, direct evaluation on elementary tensors gives
$$\begin{aligned}
P_{\xi,\eta}^{[A]}P_{\zeta,\theta}^{[B]}
&=
V_{\xi}R_{\eta,\zeta}Y_{\theta},
&
\norm{P_{\xi,\eta}^{[A]}P_{\zeta,\theta}^{[B]}}_2
&\leq
\norm{\xi}\norm{\eta}\norm{\zeta}\norm{\theta}.
\end{aligned}$$
The rank-one product $P_{\xi,\eta}^{[A]}P_{\zeta,\theta}^{[B]}$ is Hilbert--Schmidt.
We verify directly that $R_{\eta,\zeta}$ is compact, meaning that the closure of the image of the unit ball is
compact.
Enumerate an orthonormal basis of
$\sphilbfrak{h}^{\otimes\abscard{A\setminus B}}$ as $\seq{e_i}{i\in\semigrposint}$ and let
$$\begin{aligned}
P_N
\colon
\sphilbfrak{h}^{\otimes\abscard{A\setminus B}}
\to
\sphilbfrak{h}^{\otimes\abscard{A\setminus B}};\quad
\fun{P_N}{x}
&=
\sum_{i=1}^{N}\bkt{e_i}{x}e_i.
\end{aligned}$$
For every unit vector $x$,
$$\begin{aligned}
&\norm{\fun{R_{\eta,\zeta}-R_{\eta,\zeta}P_N}{x}}
=
\norm{\sum_{i>N}\bkt{e_i}{x}\fun{R_{\eta,\zeta}}{e_i}}
\\ 
&\leq
\rbk{\sum_{i>N}\abs{\bkt{e_i}{x}}^2}^{1/2}
\rbk{\sum_{i>N}\norm{\fun{R_{\eta,\zeta}}{e_i}}^2}^{1/2}
\leq
\rbk{\sum_{i>N}\norm{\fun{R_{\eta,\zeta}}{e_i}}^2}^{1/2}.
\end{aligned}$$
Consequently $\norm{R_{\eta,\zeta}-R_{\eta,\zeta}P_N}\to0$.
Each $R_{\eta,\zeta}P_N$ has finite-dimensional range.
Given $\varepsilon>0$, choose $N$ with
$\norm{R_{\eta,\zeta}-R_{\eta,\zeta}P_N}<\varepsilon/2$.
The image under
$R_{\eta,\zeta}P_N$ of the unit ball has a finite
$\varepsilon/2$-net because it is bounded in a finite-dimensional space, and that net is an
$\varepsilon$-net for the image under $R_{\eta,\zeta}$ of the unit ball.
This image is therefore
relatively compact.
Moreover, $V_{\xi}R_{\eta,\zeta}P_NY_{\theta}$ has finite-dimensional range and
$$\begin{aligned}
\norm{
V_{\xi}R_{\eta,\zeta}Y_{\theta}
-
V_{\xi}R_{\eta,\zeta}P_NY_{\theta}}
&\leq
\norm{V_{\xi}}
\norm{R_{\eta,\zeta}-R_{\eta,\zeta}P_N}
\norm{Y_{\theta}}
\to
0.
\end{aligned}$$
Thus $P_{\xi,\eta}^{[A]}P_{\zeta,\theta}^{[B]}$ is compact.
Choose finite-rank approximants $K_j\to K$ and $L_j\to L$.
Each $K_j^{[A]}L_j^{[B]}$ is a finite sum of such compact rank-one products.
The approximation error satisfies
$$\begin{aligned}
\norm{K^{[A]}L^{[B]}-K_j^{[A]}L_j^{[B]}}
&\leq
\norm{K-K_j}\norm{L}
+
\norm{K_j}\norm{L-L_j}
\to
0.
\end{aligned}$$
It follows that every product $K^{[A]}L^{[B]}$ with $A\cup B=U$ is compact because $U$ contains all $r$
output factors.

For each fixed $U\subset\setone{1,\ldots,n}$ with $\abscard{U}=r$, the internal sum over
$A\cup B=U$ is the insertion of $\fun{\widetilde{\mu}_{p,q}^{r}}{K,L}$ into the factors indexed by $U$.
Consequently the complete regrouping is
$$\begin{aligned}
&\fun{\widetilde{E}_{p,n}}{K}\fun{\widetilde{E}_{q,n}}{L}
=
\sum_{\substack{A,B\subset\setone{1,\ldots,n}\\
\abscard{A}=p\\
\abscard{B}=q}}
K^{[A]}L^{[B]}
\\
&=
\sum_{r=\max\setone{p,q}}^{\min\setone{n,p+q}}
\sum_{\substack{U\subset\setone{1,\ldots,n}\\\abscard{U}=r}}
\rbk{\fun{\widetilde{\mu}_{p,q}^{r}}{K,L}}^{[U]}
=
\sum_{r=\max\setone{p,q}}^{\min\setone{n,p+q}}
\fun{\widetilde{E}_{r,n}\circ\widetilde{\mu}_{p,q}^{r}}{K,L}.
\end{aligned}$$

For a fixed union with $r$ elements, the three disjoint sets
$A\setminus B$, $B\setminus A$, and $A\cap B$ have cardinalities
$r-q$, $r-p$, and $p+q-r$.
The number of ordered pairs $(A,B)$ is
$$\begin{aligned}
\frac{r!}{\rbk{r-p}!\rbk{r-q}!\rbk{p+q-r}!}.
\end{aligned}$$
Each summand has norm at most $\norm{K}\norm{L}$.
The triangle inequality proves \eqref{eq:paper-overlap-product-bound}.
\end{proof}

\begin{thm}[Concrete realization of the particle-coefficient reconstruction data]\label{thm:paper-concrete-particle-coefficient-reconstruction-data}
For every $m\in\monnat$ and every nonzero finite-dimensional $\sphilbfrak{f}\subset\sphilbfrak{h}$, let
$I_{\sphilbfrak{f},m}\colon\sphilbfrak{f}^{\otimes_{\txtsym}m}\to\sphilbfrak{h}^{\otimes_{\txtsym}m}$ be the
canonical inclusion, and let $j_{\sphilbfrak{f},m}$ be the canonical map in
\eqref{eq:paper-main-degree-coefficient-canonical-injection}.
There are a unique $\ast$-isomorphism $\Theta_m$ and unique maps $E_{m,n}$ and $\mu_{p,q}^{r}$ such that
\begin{equation}
\begin{aligned}
\Theta_m
\colon
\oa{C}_m(\sphilbfrak{h})
&\to
\fun{\schattencls}{\sphilbfrak{h}^{\otimes_{\txtsym}m}},
\\
\fun{\Theta_m\circ j_{\sphilbfrak{f},m}}{K}
&=
I_{\sphilbfrak{f},m}
\fun{\pi_{\sphilbfrak{f},m}}{K}
\faadjrbk{I_{\sphilbfrak{f},m}},
\\
\varpi_n\circ E_{m,n}
&=
\widetilde{E}_{m,n}\circ\Theta_m
&&\rbk{0\leq m\leq n},
\\
\Theta_r\circ\mu_{p,q}^{r}
&=
\widetilde{\mu}_{p,q}^{r}\circ\rbk{\Theta_p\times\Theta_q}
&&\rbk{\max\setone{p,q}\leq r\leq p+q},
\end{aligned}
\label{eq:paper-main-fock-coefficient-identification}
\end{equation}
where $K\in e_{\sphilbfrak{f},m}\oa{K}(\sphilbfrak{f})^\gamma e_{\sphilbfrak{f},m}$ and
$\widetilde{E}_{m,n}$ and $\widetilde{\mu}_{p,q}^{r}$ are defined in
\eqref{eq:paper-main-subset-sum-embedding} and
\eqref{eq:paper-main-subset-union-overlap-map}, respectively.
The abstract overlap map is
\begin{equation}
\begin{aligned}
\mu_{p,q}^{r}
\colon
\oa{C}_p(\sphilbfrak{h})\times\oa{C}_q(\sphilbfrak{h})
\to
\oa{C}_r(\sphilbfrak{h});\quad
\mu_{p,q}^{r}
&=
\Theta_r^{-1}\circ\widetilde{\mu}_{p,q}^{r}\circ
\rbk{\Theta_p\times\Theta_q}.
\end{aligned}
\label{eq:paper-abstract-overlap-map-concrete-realization}
\end{equation}
The maps $E_{m,n}$ satisfy Definition \ref{def:paper-particle-coefficient-embedding-data}, and together with
$\mu_{p,q}^{r}$ they satisfy Definition \ref{def:paper-abstract-particle-coefficient-reconstruction-data}.
\end{thm}

\begin{proof}
Proposition \ref{prop:paper-coefficient-connecting-maps} identifies the connecting maps in the inductive limit
with zero extension.
The defining formula for $\Theta_m$ is therefore independent of the representative and defines an isometric
$\ast$-homomorphism.
Its range contains
$$\begin{aligned}
\bigcup_{\substack{\sphilbfrak{f}\subset\sphilbfrak{h}\\
\dim_{\fldcmp}\sphilbfrak{f}<\infty}}
\fun{\schattencls}{\sphilbfrak{f}^{\otimes_{\txtsym}m}},
\end{aligned}$$
whose norm closure is
$\fun{\schattencls}{\sphilbfrak{h}^{\otimes_{\txtsym}m}}$.
Thus $\Theta_m$ is a $\ast$-isomorphism.

The standard exact-particle decomposition is \cite[Lemma 3.3]{DetlevBuchholz002}; its graph formulation is
\cite[Lemma 2.3 and equation (12)]{DeuchertLampartLemm001}.
In our notation it is
\begin{equation}
\begin{aligned}
\fun{\rho_n}{\oaresolventalgebra(\sphilbfrak{h},\sigma)^{\gamma}}
&=
\bigoplus_{m=0}^{n}
\fun{\widetilde{E}_{m,n}}{
\fun{\schattencls}{\sphilbfrak{h}^{\otimes_{\txtsym}m}}}.
\end{aligned}
\label{eq:paper-concrete-particle-coefficient-decomposition}
\end{equation}
The compatible restriction computed in \cite[Lemma 3.4]{DetlevBuchholz002} and used in
\cite[Proposition 2.4]{DeuchertLampartLemm001} is the map $\kappa_n$ defined in
\eqref{eq:paper-main-fock-connecting-map}.
It acts on the summands by
\begin{equation}
\begin{aligned}
\kappa_n\circ\widetilde{E}_{m,n}
&=
\widetilde{E}_{m,n-1}
&&\rbk{0\leq m<n},
\\
\kappa_n\circ\widetilde{E}_{n,n}
&=
0.
\end{aligned}
\label{eq:paper-concrete-deletion-actions}
\end{equation}
Transport by $\varpi_n^{-1}$ in \eqref{eq:paper-main-fock-cutoff-identification} and the maps $\Theta_m$
yields the required direct-sum decomposition for
$\oa{B}_n(\sphilbfrak{h})$, the deletion identities for $E_{m,n}$, and the equality
$$\begin{aligned}
\Ker r_n
&=
\fun{E_{n,n}}{\oa{C}_n(\sphilbfrak{h})}.
\end{aligned}$$
The degree-zero and adjoint conditions follow directly from the subset-sum definition of
$\widetilde{E}_{m,n}$.

It remains to verify essentiality of the top-degree kernel.
Under the map $\varpi_n$ in \eqref{eq:paper-main-fock-cutoff-identification}, this kernel is
$\fun{\schattencls}{\sphilbfrak{h}^{\otimes_{\txtsym}n}}$.
If $0\neq A$ belongs to an ideal of the concrete exact-particle algebra, choose $\xi$ with $A\xi\neq0$.
The rank-one compact operator
$K=P_{A\xi,A\xi}$ satisfies $KA\xi=\norm{A\xi}^{2}A\xi\neq0$.
Thus every nonzero ideal meets the compact operators, and transport by $\varpi_n$ in
\eqref{eq:paper-main-fock-cutoff-identification} proves that $\Ker r_n$ is essential.
This proves every condition in Definition \ref{def:paper-particle-coefficient-embedding-data}.

Proposition \ref{prop:paper-overlap-product} and the map $\varpi_n$ in
\eqref{eq:paper-main-fock-cutoff-identification} give
\begin{equation}
\begin{aligned}
\fun{\varpi_n}{\fun{E_{p,n}}{a_p}\fun{E_{q,n}}{b_q}}
&=
\sum_{r=\max\setone{p,q}}^{\min\setone{p+q,n}}
\fun{\widetilde{E}_{r,n}\circ\widetilde{\mu}_{p,q}^{r}}
{\fun{\Theta_p}{a_p},\fun{\Theta_q}{b_q}}
\\
&=
\fun{\varpi_n}{
\sum_{r=\max\setone{p,q}}^{\min\setone{p+q,n}}
\fun{E_{r,n}}{\fun{\mu_{p,q}^{r}}{a_p,b_q}}}.
\end{aligned}
\label{eq:paper-transported-overlap-product}
\end{equation}
Injectivity of $\varpi_n$ from \eqref{eq:paper-main-fock-cutoff-identification} in
\eqref{eq:paper-transported-overlap-product} gives the first identity in
\eqref{eq:paper-abstract-particle-coefficient-overlap-products}.

For $r=p+q$, every pair $A,B$ in \eqref{eq:paper-main-subset-union-overlap-map} satisfies
$A\cap B=\emptyset$.
Interchanging $A$ and $B$ and using commutativity of operators acting on disjoint tensor factors gives
$$\begin{aligned}
\fun{\widetilde{\mu}_{p,q}^{p+q}}{K,L}
&=
\fun{\widetilde{\mu}_{q,p}^{p+q}}{L,K}.
\end{aligned}$$
Equation \eqref{eq:paper-abstract-overlap-map-concrete-realization} gives the second identity in
\eqref{eq:paper-abstract-particle-coefficient-overlap-products}.

For $p=q=r=m$, the only pair in \eqref{eq:paper-main-subset-union-overlap-map} is
$A=B=\setone{1,\ldots,m}$, and hence
$$\begin{aligned}
\fun{\widetilde{\mu}_{m,m}^{m}}{K,L}
&=
KL.
\end{aligned}$$
Equation \eqref{eq:paper-abstract-overlap-map-concrete-realization} gives the third identity in
\eqref{eq:paper-abstract-particle-coefficient-overlap-products}.
These three identities prove the remaining conditions in Definition
\ref{def:paper-abstract-particle-coefficient-reconstruction-data}.
\end{proof}

\begin{cor}[Degree-tail ideals in the Fock realization]\label{cor:paper-fock-degree-tail-ideals}
For $0\leq k\leq n$, the map $\varpi_n$ in \eqref{eq:paper-main-fock-cutoff-identification} identifies the
degree-tail ideal by
\begin{equation}
\begin{aligned}
\fun{\varpi_n}{\oa{J}_k^{(n)}}
&=
\bigoplus_{m=k}^{n}
\fun{\widetilde{E}_{m,n}}{
\fun{\schattencls}{\sphilbfrak{h}^{\otimes_{\txtsym}m}}}.
\end{aligned}
\label{eq:paper-main-fock-degree-tail-ideal-identification}
\end{equation}
\end{cor}

\begin{proof}
Proposition \ref{prop:paper-degree-tail-ideals} gives
$\oa{J}_k^{(n)}=\bigoplus_{m=k}^{n}\fun{E_{m,n}}{\oa{C}_m(\sphilbfrak{h})}$.
The injective map $\varpi_n$ preserves this direct sum, and the relation
$\varpi_n\circ E_{m,n}=\widetilde{E}_{m,n}\circ\Theta_m$ in \eqref{eq:paper-main-fock-coefficient-identification},
together with the surjectivity of $\Theta_m$, identifies each summand with
$\fun{\widetilde{E}_{m,n}}{\fun{\schattencls}{\sphilbfrak{h}^{\otimes_{\txtsym}m}}}$.
\end{proof}

\begin{lem}[Uniqueness of the concrete exact-particle coefficient decomposition]\label{lem:paper-concrete-finite-particle-coefficient-uniqueness}
For $0\leq m\leq n$, the linear map $\widetilde{E}_{m,n}$ is injective.
Moreover, every
$A\in\fun{\rho_n}{\oaresolventalgebra(\sphilbfrak{h},\sigma)^{\gamma}}$
has unique coefficients
$$\begin{aligned}
a_m
&\in
\fun{\schattencls}{\sphilbfrak{h}^{\otimes_{\txtsym}m}}
&&\rbk{0\leq m\leq n}
\end{aligned}$$
such that
$$\begin{aligned}
A
&=
\sum_{m=0}^{n}\fun{\widetilde{E}_{m,n}}{a_m}.
\end{aligned}$$
\end{lem}

\begin{proof}
The direct-sum decomposition in \cite[Lemma 3.3]{DetlevBuchholz002}, equivalently
\cite[Lemma 2.3 and equation (12)]{DeuchertLampartLemm001}, gives a unique component in each image of
$\widetilde{E}_{m,n}$.
Moreover, \eqref{eq:paper-concrete-deletion-actions} reduces $\widetilde{E}_{m,n}$ to
$\fun{\widetilde{E}_{m,m}}{K}=K$, so every $\widetilde{E}_{m,n}$ is injective.
The component decomposition therefore has unique coefficients.
\end{proof}

\begin{cor}[Associativity of overlap contractions]\label{cor:paper-overlap-associativity}
Let
$K\in\fun{\schattencls}{\sphilbfrak{h}^{\otimes_{\txtsym}p}}$,
$L\in\fun{\schattencls}{\sphilbfrak{h}^{\otimes_{\txtsym}q}}$, and
$M\in\fun{\schattencls}{\sphilbfrak{h}^{\otimes_{\txtsym}t}}$.
For every output degree $r\in\monnat$,
\begin{equation}
\begin{aligned}
\sum_{\substack{\max\setone{p,q}\leq s\leq p+q\\
\max\setone{s,t}\leq r\leq s+t}}
\fun{\widetilde{\mu}_{s,t}^{r}}{
\fun{\widetilde{\mu}_{p,q}^{s}}{K,L},M}
=
\sum_{\substack{\max\setone{q,t}\leq s\leq q+t\\
\max\setone{p,s}\leq r\leq p+s}}
\fun{\widetilde{\mu}_{p,s}^{r}}{
K,\fun{\widetilde{\mu}_{q,t}^{s}}{L,M}}.
\end{aligned}
\label{eq:paper-overlap-associativity}
\end{equation}
Each sum is finite and runs only over degrees $s$ for which both nested overlap maps in each summand of
\eqref{eq:paper-overlap-associativity} are defined.
\end{cor}

\begin{proof}
Fix $n\geq p+q+t$.
For the coefficient on the left of
\eqref{eq:paper-overlap-associativity}, Proposition
\ref{prop:paper-overlap-product} applied twice gives
$$\begin{aligned}
&\sum_{r=0}^{p+q+t}\fun{\widetilde{E}_{r,n}}{
\sum_{\substack{\max\setone{p,q}\leq s\leq p+q\\
\max\setone{s,t}\leq r\leq s+t}}
\fun{\widetilde{\mu}_{s,t}^{r}}{
\fun{\widetilde{\mu}_{p,q}^{s}}{K,L},M}}
\\
&=
\sum_{s=\max\setone{p,q}}^{p+q}
\fun{\widetilde{E}_{s,n}}{\fun{\widetilde{\mu}_{p,q}^{s}}{K,L}}
\fun{\widetilde{E}_{t,n}}{M}
\\
&=
\rbk{
\fun{\widetilde{E}_{p,n}}{K}
\fun{\widetilde{E}_{q,n}}{L}}
\fun{\widetilde{E}_{t,n}}{M}.
\end{aligned}$$
For the coefficient on the right, the same calculation in the other order gives
$$\begin{aligned}
&\sum_{r=0}^{p+q+t}\fun{\widetilde{E}_{r,n}}{
\sum_{\substack{\max\setone{q,t}\leq s\leq q+t\\
\max\setone{p,s}\leq r\leq p+s}}
\fun{\widetilde{\mu}_{p,s}^{r}}{
K,\fun{\widetilde{\mu}_{q,t}^{s}}{L,M}}}
\\
&=
\sum_{s=\max\setone{q,t}}^{q+t}
\fun{\widetilde{E}_{p,n}}{K}
\fun{\widetilde{E}_{s,n}}{\fun{\widetilde{\mu}_{q,t}^{s}}{L,M}}
\\
&=
\fun{\widetilde{E}_{p,n}}{K}
\rbk{
\fun{\widetilde{E}_{q,n}}{L}
\fun{\widetilde{E}_{t,n}}{M}}.
\end{aligned}$$
Associativity of operator multiplication equates the two parenthesizations of the product of
$\fun{\widetilde{E}_{p,n}}{K}$, $\fun{\widetilde{E}_{q,n}}{L}$, and $\fun{\widetilde{E}_{t,n}}{M}$.
Lemma \ref{lem:paper-concrete-finite-particle-coefficient-uniqueness} then equates their unique degree-$r$
coefficients and proves \eqref{eq:paper-overlap-associativity}.
\end{proof}

For every nonzero finite-dimensional complex subspace \(\sphilbfrak{f}\subset\sphilbfrak{h}\), define \[\begin{aligned}
\opspbddlin{\fun{\spfockgamma_{\txtbsn}}{\sphilbfrak{f}}}^{\gamma}
&=
\set{\bigoplus_{m\in\monnat}T_m}
{T_m\in\opspbddlin{\sphilbfrak{f}^{\otimes_{\txtsym}m}},\
\sup_{m\in\monnat}\norm{T_m}<\infty}.
\end{aligned}\] The block map is \begin{equation}
\begin{aligned}
\diag_{\sphilbfrak{f}}
\colon
\oa{M}(\oa{K}(\sphilbfrak{f})^\gamma)
\to
\opspbddlin{\fun{\spfockgamma_{\txtbsn}}{\sphilbfrak{f}}}^\gamma;
\quad
\fun{\diag_{\sphilbfrak{f}}}{\seq{K_m}{m\in\monnat}}
&=
\bigoplus_{m\in\monnat}\fun{\pi_{\sphilbfrak{f},m}}{K_m}.
\end{aligned}
\label{eq:paper-main-finite-mode-product-isomorphism}
\end{equation}

\begin{prop}[Block-product representation of the finite-mode algebra]\label{prop:finite-mode-product}
Let $\sphilbfrak{f}\subset\sphilbfrak{h}$ satisfy $1\leq\dim_{\fldcmp}\sphilbfrak{f}<\infty$.
Use the block representations $\pi_{\sphilbfrak{f},k}$ defined by
\eqref{eq:paper-main-finite-mode-block-representation}.
Then the map $\diag_{\sphilbfrak{f}}$ in
\eqref{eq:paper-main-finite-mode-product-isomorphism} is an isometric $\ast$-isomorphism.
In particular,
the centers of $\oa{M}(\oa{K}(\sphilbfrak{f})^{\gamma})$ and
$\opspbddlin{\fun{\spfockgamma_{\txtbsn}}{\sphilbfrak{f}}}^{\gamma}$ are isomorphic to
$\fun{\lpseq^\infty}{\monnat}$.
\end{prop}

\begin{proof}
Proposition \ref{prop:paper-particle-block-independence} and Theorem
\ref{thm:abstract-minimal-compact-ideal} show, for every $k\in\monnat$, that
$\pi_{\sphilbfrak{f},k}$ is a $\ast$-isomorphism onto the full operator algebra of the finite-dimensional
$k$-particle space.
Taking their bounded product gives the map $\diag_{\sphilbfrak{f}}$ in
\eqref{eq:paper-main-finite-mode-product-isomorphism}, and
$$\begin{aligned}
\norm{\fun{\diag_{\sphilbfrak{f}}}{K}}
&=
\sup_{k\in\monnat}\norm{\fun{\pi_{\sphilbfrak{f},k}}{K_k}}
=
\sup_{k\in\monnat}\norm{K_k}
=
\norm{K}_{\oa{M}(\oa{K}(\sphilbfrak{f})^{\gamma})}.
\end{aligned}$$
Surjectivity follows coordinatewise from $K_k=\fun{\pi_{\sphilbfrak{f},k}^{-1}}{A_k}$.
Finally, the center of
each full matrix block is scalar, so the center of the bounded product is $\fun{\lpseq^\infty}{\monnat}$.
\end{proof}

This proposition concerns the algebra associated with the finite-dimensional subspace \(\sphilbfrak{f}\).

\section{Spatial localization in the standard Fock realization}\label{sec:concrete-spatial-particle-cutoff-representations}

In this section, the concrete Fock representation embeds finite-mode gauge-invariant operators and specializes these embeddings to inclusions of finite graph regions.

\subsection{Comparison representation of the bounded inverse limit}\label{sec:concrete-fock-comparison}

In this subsection, the bounded inverse limit is represented by compatible operators on the Fock particle sectors, providing the comparison map used by the finite-mode embeddings.

The realization of the bounded inverse limit by compatible particle-sector operators is precisely \cite[Theorem 3.5]{DetlevBuchholz002}; its graph formulation is \cite[Proposition 2.4]{DeuchertLampartLemm001}. We therefore only record the comparison map needed below. Use the isomorphisms \(\varpi_n\) in \eqref{eq:paper-main-fock-cutoff-identification} and the connecting maps \(\kappa_n\) in \eqref{eq:paper-main-fock-connecting-map}. Define \[\begin{aligned}
\oa{D}_{\txtfock}
&=
\set{T\in\opspbddlin{\fun{\spfockgamma_{\txtbsn}}{\sphilbfrak{h}}}}
{\parbox{15em}
{$T\sphilbfrak{h}^{\otimes_{\txtsym}n}\subset\sphilbfrak{h}^{\otimes_{\txtsym}n}$,\\
$\fnrestr{T}{\sphilbfrak{h}^{\otimes_{\txtsym}n}}
\in\fun{\varpi_n}{\oa{B}_n(\sphilbfrak{h})}\rbk{n\in\monnat}$,\\
$\fun{\kappa_n}{\fnrestr{T}{\sphilbfrak{h}^{\otimes_{\txtsym}n}}}
=\fnrestr{T}{\sphilbfrak{h}^{\otimes_{\txtsym}\rbk{n-1}}}\rbk{n\in\semigrposint}$}}.
\end{aligned}\] Define \[\begin{aligned}
\jmath
\colon
\oa{D}_{\txtfock}
\to
\fun{\oa{B}}{\sphilbfrak{h}};\quad
\fun{\jmath}{T}
&=
\seq{\fun{\varpi_n^{-1}}{\fnrestr{T}{\sphilbfrak{h}^{\otimes_{\txtsym}n}}}}{n\in\monnat}.
\end{aligned}\] The comparison representation satisfies \[\begin{aligned}
\fun{\oarepn_{\txtfock}\circ\jmath}{T}
&=T
&&\rbk{T\in\oa{D}_{\txtfock}}.
\end{aligned}\] Thus \(\jmath\) is the inverse of the established Fock realization, transported through the new intrinsic quotient isomorphisms \(\varpi_n\) in \eqref{eq:paper-main-fock-cutoff-identification}; no independent inverse-limit verification is required here.

\subsection{Finite-mode embeddings and particle-degree coefficients}\label{sec:paper-concrete-finite-mode-embeddings}

In this subsection, the Fock decomposition associated with a finite-dimensional one-particle subspace embeds its gauge-invariant operators and identifies their particle-degree coefficients in the global system.

Fix a nonzero finite-dimensional subspace \(\sphilbfrak{f}\subset\sphilbfrak{h}\) and a bounded block sequence \(D=\seq{D_k}{k\in\monnat}\in\opspbddlin{\fun{\spfockgamma_{\txtbsn}}{\sphilbfrak{f}}}^{\gamma}\). Let \(P_n\colon\sphilbfrak{h}^{\otimes n}\to\sphilbfrak{h}^{\otimes_{\txtsym}n}\) be the specialization \(r=n\) of \eqref{eq:paper-main-symmetrizing-projection}. Use the particle-space and Fock-space decomposition unitaries \(W_{\sphilbfrak{f},n}\) and \(U_{\sphilbfrak{f}}\) defined in \eqref{eq:paper-main-particle-fock-decomposition-unitary} and \eqref{eq:paper-main-fock-decomposition-unitary}. For \(0\leq j\leq m\), define \begin{equation}
\begin{aligned}
E_{\sphilbfrak{f},j,m}
\colon
\fun{\schattencls}{\sphilbfrak{f}^{\otimes_{\txtsym}j}}
\to
\opspbddlin{\sphilbfrak{f}^{\otimes_{\txtsym}m}};\quad
\fun{E_{\sphilbfrak{f},j,m}}{K}
&=
\fnrestr{
\sum_{\substack{S\subset\setone{1,\ldots,m}\\\abscard{S}=j}}
K_{\sphilbfrak{f},S}}
{\sphilbfrak{f}^{\otimes_{\txtsym}m}},
\end{aligned}
\label{eq:paper-finite-mode-local-coefficient-embedding}
\end{equation} where \(K_{\sphilbfrak{f},S}\) is the specialization to \(\sphilbfrak{f}\) of the tensor-factor insertion in \eqref{eq:paper-main-tensor-factor-insertion}. Using the maps \(E_{\sphilbfrak{f},j,m}\) in \eqref{eq:paper-finite-mode-local-coefficient-embedding}, define the particle-degree coefficients on \(\sphilbfrak{f}\) recursively by \[\begin{aligned}
c_0
&=D_0,
&
c_m
&=D_m-\sum_{j=0}^{m-1}\fun{E_{\sphilbfrak{f},j,m}}{c_j}
&&\rbk{m\geq1}.
\end{aligned}\] Because \(\sphilbfrak{f}^{\otimes_{\txtsym}m}\) is finite-dimensional, every \(c_m\) is compact. Denote its zero extension to \(\sphilbfrak{h}^{\otimes_{\txtsym}m}\) also by \(c_m\). Under \(W_{\sphilbfrak{f},n}\), let \(T_n\) denote the \(n\)-particle restriction of \(T=\faadj{U_{\sphilbfrak{f}}}\rbk{D\otimes1}U_{\sphilbfrak{f}}\). On the block with exactly \(k\) particles in \(\sphilbfrak{f}\), the maps in \eqref{eq:paper-finite-mode-local-coefficient-embedding} satisfy \[\begin{aligned}
\fnrestr{\fun{\widetilde{E}_{m,n}}{c_m}}
{\sphilbfrak{f}^{\otimes_{\txtsym}k}\otimes
\rbk{\sphilbfrak{f}^\perp}^{\otimes_{\txtsym}\rbk{n-k}}}
&=
\begin{cases}
\fun{E_{\sphilbfrak{f},m,k}}{c_m}\otimes1,&m\leq k,\\
0,&m>k.
\end{cases}
\end{aligned}\] The recursion and the map \(\varpi_n\) in \eqref{eq:paper-main-fock-cutoff-identification} give \[\begin{aligned}
T_n
&=
\bigoplus_{k=0}^{n}\rbk{D_k\otimes1}
=
\sum_{m=0}^{n}\fun{\widetilde{E}_{m,n}}{c_m}
\in\fun{\varpi_n}{\oa{B}_n(\sphilbfrak{h})},
&
\fun{\kappa_n}{T_n}
&=T_{n-1}.
\end{aligned}\] The block norms satisfy \[\begin{aligned}
\sup_{n\in\monnat}\norm{T_n}
&=
\sup_{n\in\monnat}\max_{0\leq k\leq n}\norm{D_k}
=
\sup_{k\in\monnat}\norm{D_k}.
\end{aligned}\] Hence \(\fun{\jmath}{T}\in\fun{\oa{B}}{\sphilbfrak{h}}\). The conjugated algebra satisfies \[\begin{aligned}
\faadj{U_{\sphilbfrak{f}}}
\rbk{\opspbddlin{\fun{\spfockgamma_{\txtbsn}}{\sphilbfrak{f}}}^{\gamma}\otimes1}
U_{\sphilbfrak{f}}
&\subset
\oa{D}_{\txtfock}.
\end{aligned}\] The restriction of \(\jmath\) to this algebra is an isometric \(\ast\)-homomorphism into \(\fun{\oa{B}}{\sphilbfrak{h}}\).

\subsection{Continuum and graph one-particle spaces}\label{sec:paper-continuum-graph-buchholz-realizations}

In this subsection, the finite-mode construction applies to concrete subspaces of both continuum and graph one-particle spaces. We specify the subspace \(\sphilbfrak{k}\) in each case before writing its block algebra and its image in the full Buchholz algebra.

For the continuum case, fix \(d\in\semigrposint\), set \(\sphilbfrak{h}=\fun{\lp^{2}}{\fldreal^{d}}\), and choose a nonempty finite set \(A\subset\ringratint^{d}\). For \(a\in A\), let \(C_a=a+\rightopeninterval{0}{1}^{d}\) and use its indicator function as a one-particle vector. The finite cell subspace and symplectic form are \begin{equation}
\begin{aligned}
\sphilbfrak{k}
&=
\splinspan_{\fldcmp}\set{\fndef{C_a}}{a\in A}
\subset\fun{\lp^{2}}{\fldreal^{d}},
&
\fun{\sigma}{f,g}
&=
\opimag\int_{\fldreal^{d}}\cmpconj{\fun{f}{x}}\fun{g}{x}\opdmsr{x}.
\end{aligned}
\label{eq:paper-continuum-cell-one-particle-space}
\end{equation} The vectors \(\fndef{C_a}\) are orthonormal because the half-open cubes are disjoint and each has unit volume. Thus \(\dim_{\fldcmp}\sphilbfrak{k}=\abscard{A}\), and \(\sphilbfrak{k}^{\perp}\) consists exactly of the \(\lp^{2}\) functions whose integral over each \(C_a\) vanishes. These cell indicators give concrete finite-mode examples, but their span is not dense in \(\fun{\lp^{2}}{\fldreal^{d}}\). The dense finite-mode system used in Appendix \ref{sec:paper-spatial-quasilocal-appendix} instead chooses a full orthonormal basis on each cell.

For the graph case, take the one-particle space of \eqref{eq:paper-main-local-one-particle-subspace} and choose a nonempty finite \(\Lambda\subset\Gamma\). Then the finite-region subspace and symplectic form are \begin{equation}
\begin{aligned}
\sphilbfrak{h}
&=
\fun{\lpseq^{2}}{\Gamma},
&
\sphilbfrak{k}
&=
\sphilbfrak{h}_{\Lambda}
=
\fun{\lpseq^{2}}{\Lambda}
=
\splinspan_{\fldcmp}\set{\delta_x}{x\in\Lambda},
\\
\sphilbfrak{k}^{\perp}
&=
\fun{\lpseq^{2}}{\Gamma\setminus\Lambda},
&
\fun{\sigma}{f,g}
&=
\opimag\sum_{x\in\Gamma}\cmpconj{\fun{f}{x}}\fun{g}{x}.
\end{aligned}
\label{eq:paper-graph-finite-region-one-particle-space}
\end{equation} Here \(\dim_{\fldcmp}\sphilbfrak{k}=\abscard{\Lambda}\).

For either concrete choice of \(\sphilbfrak{k}\), use the regular Fock representation in \eqref{eq:paper-main-regular-fock-representation} on the Fock space of \eqref{eq:paper-main-fock-space-definition}. The finite-mode compact ideal and its multiplier block product have the concrete Fock images \begin{equation}
\begin{aligned}
\fun{\oarepn_{\txtfock,\sphilbfrak{k}}}{\oa{K}(\sphilbfrak{k})^{\gamma}}
&=
\set{\bigoplus_{m\in\monnat}K_m}
{K_m\in\opspbddlin{\sphilbfrak{k}^{\otimes_{\txtsym}m}},
\lim_{m\to\infty}\norm{K_m}=0},
\\
\fun{\diag_{\sphilbfrak{k}}}{\oa{M}(\oa{K}(\sphilbfrak{k})^{\gamma})}
&=
\set{\bigoplus_{m\in\monnat}D_m}
{D_m\in\opspbddlin{\sphilbfrak{k}^{\otimes_{\txtsym}m}},
\sup_{m\in\monnat}\norm{D_m}<\infty}.
\end{aligned}
\label{eq:paper-two-model-finite-mode-block-algebras}
\end{equation} The \(m\)th matrix block acts on \(\sphilbfrak{k}^{\otimes_{\txtsym}m}\) and has matrix size \(\binom{m+\abscard{A}-1}{m}\) in the continuum case or \(\binom{m+\abscard{\Lambda}-1}{m}\) in the graph case. The block norms in the first algebra of \eqref{eq:paper-two-model-finite-mode-block-algebras} tend to zero, while the second algebra contains every bounded block family, as in \eqref{eq:paper-main-fock-intrinsic-block-identification} and \eqref{eq:paper-main-finite-mode-product-isomorphism}.

For either specified one-particle space \(\sphilbfrak{h}\), the Fock image of the full Buchholz algebra consists of operators \(\bigoplus_{n\in\monnat}T_n\) on \(\fun{\spfockgamma_{\txtbsn}}{\sphilbfrak{h}}\), where each \(T_n\in\opspbddlin{\sphilbfrak{h}^{\otimes_{\txtsym}n}}\) satisfies \begin{equation}
\begin{aligned}
T_n
\in
\sum_{m=0}^{n}\fun{\widetilde{E}_{m,n}}{
\fun{\schattencls}{\sphilbfrak{h}^{\otimes_{\txtsym}m}}}
\,\rbk{n\in\monnat},
\quad
\fun{\kappa_n}{T_n}
=
T_{n-1}
\,\rbk{n\in\semigrposint},
\quad
\sup_{n\in\monnat}\norm{T_n}
<
\infty.
\end{aligned}
\label{eq:paper-two-model-full-buchholz-conditions}
\end{equation} These are the exact-particle algebras and connecting maps of \eqref{eq:paper-literature-fock-input}, assembled as in \eqref{eq:paper-main-fock-bounded-inverse-limit-identification}.

For a bounded block family \(\seq{D_m}{m\in\monnat}\) from the second algebra in \eqref{eq:paper-two-model-finite-mode-block-algebras}, the corresponding element of the full abstract Buchholz algebra is \begin{equation}
\begin{aligned}
\fun{\jmath}{
\faadj{U_{\sphilbfrak{k}}}
\rbk{\rbk{\bigoplus_{m\in\monnat}D_m}\otimes
1_{\fun{\spfockgamma_{\txtbsn}}{\sphilbfrak{k}^{\perp}}}}
U_{\sphilbfrak{k}}}
&\in
\fun{\oa{B}}{\sphilbfrak{h}}.
\end{aligned}
\label{eq:paper-two-model-global-finite-mode-element}
\end{equation} On \(\sphilbfrak{h}^{\otimes_{\txtsym}n}\), its Fock representative is, under the unitary \(W_{\sphilbfrak{k},n}\) of \eqref{eq:paper-main-particle-fock-decomposition-unitary}, the direct sum \(\bigoplus_{m=0}^{n}D_m\otimes1_{\rbk{\sphilbfrak{k}^{\perp}}^{\otimes_{\txtsym}\rbk{n-m}}}\). In the graph case, the global local map \eqref{eq:concrete-global-local-inclusion} in Proposition \ref{prop:concrete-spatial-localization-system} below shows that \eqref{eq:paper-two-model-global-finite-mode-element} belongs to the local algebra \(\oa{B}_{\Lambda}\) of \eqref{eq:paper-main-local-buchholz-algebra}.

\subsection{Inclusions of finite graph regions}\label{inclusions-of-finite-graph-regions}

In this subsection, the finite-mode embeddings are specialized to nested finite graph regions and shown to satisfy the spatial compatibility, commutativity, and density conditions.

The underlying concrete inclusions \(\oa{B}_X\subset\oa{B}_Y\) for \(X\subset Y\) and the distinction between local and quasilocal norm approximation are already described in the discussion following \cite[Proposition 2.4 and Example 2.5]{DeuchertLampartLemm001}. The proposition below records the points needed for the abstract construction but not supplied there: the formula in intrinsic block coordinates, functoriality of those coordinates, disjoint-range commutativity, and fixed-cutoff coefficient density.

Proposition \ref{prop:concrete-spatial-localization-system} is the only proposition in this subsection that assumes the finite-region data in \eqref{eq:paper-main-finite-region-directed-set}--\eqref{eq:paper-finite-region-one-particle-projection}.

\begin{prop}[Concrete local subalgebras along finite regions]\label{prop:concrete-spatial-localization-system}
Let $\Lambda,\Lambda''\in\setpowfin{\Gamma}$ be nonempty and satisfy $\Lambda\subset \Lambda''$.
Let
$P_{\sphilbfrak{h}_{\Lambda''},m}\colon \sphilbfrak{h}_{\Lambda''}^{\otimes m}\to
\sphilbfrak{h}_{\Lambda''}^{\otimes_{\txtsym}m}$ be the restriction to
$\sphilbfrak{h}_{\Lambda''}^{\otimes m}$ of
the symmetrizing projection $P_m$ in \eqref{eq:paper-main-symmetrizing-projection}.
For
$0\leq k\leq m$, $\xi\in \sphilbfrak{h}_{\Lambda}^{\otimes_{\txtsym}k}$, and
$\eta\in \sphilbfrak{h}_{\Lambda''\setminus \Lambda}^{\otimes_{\txtsym}\rbk{m-k}}$, define
$$\begin{aligned}
W_{\Lambda,\Lambda'',m}
\colon
\sphilbfrak{h}_{\Lambda''}^{\otimes_{\txtsym}m}
\to
\bigoplus_{k=0}^{m}
\sphilbfrak{h}_{\Lambda}^{\otimes_{\txtsym}k}\otimes
\sphilbfrak{h}_{\Lambda''\setminus \Lambda}^{\otimes_{\txtsym}\rbk{m-k}}.
\end{aligned}$$
Its inverse acts by
$$\begin{aligned}
\fun{\faadj{W_{\Lambda,\Lambda'',m}}}{\xi\otimes\eta}
&=
\binom{m}{k}^{1/2}P_{\sphilbfrak{h}_{\Lambda''},m}\rbk{\xi\otimes\eta}.
\end{aligned}$$
For $D=\seq{D_k}{k\in\monnat}\in\oa{M}(\oa{K}(\sphilbfrak{h}_{\Lambda})^{\gamma})$, the $m$th component of the local inclusion $i_{\Lambda,\Lambda''}$ is
\begin{equation}
\begin{aligned}
\rbk{\fun{i_{\Lambda,\Lambda''}}{D}}_m
&=
\fun{\pi_{\sphilbfrak{h}_{\Lambda''},m}^{-1}}{
\faadj{W_{\Lambda,\Lambda'',m}}
\rbk{
\bigoplus_{k=0}^{m}
\fun{\pi_{\sphilbfrak{h}_{\Lambda},k}}{D_k}\otimes1_{\sphilbfrak{h}_{\Lambda''\setminus \Lambda}^{\otimes_{\txtsym}\rbk{m-k}}}
}
W_{\Lambda,\Lambda'',m}}.
\end{aligned}
\label{eq:concrete-local-inclusion}
\end{equation}
Using the map $\diag_{\sphilbfrak{h}_{\Lambda}}$ in
\eqref{eq:paper-main-finite-mode-product-isomorphism}, the corresponding global map is
\begin{equation}
\begin{aligned}
i_{\Lambda,\Gamma}
\colon
\oa{M}(\oa{K}(\sphilbfrak{h}_{\Lambda})^{\gamma})
\to
\fun{\oa{B}}{\sphilbfrak{h}};\quad
\fun{i_{\Lambda,\Gamma}}{D}
&=
\fun{\jmath}{\faadj{U_{\sphilbfrak{h}_{\Lambda}}}\rbk{\fun{\diag_{\sphilbfrak{h}_{\Lambda}}}{D}\otimes1}U_{\sphilbfrak{h}_{\Lambda}}}.
\end{aligned}
\label{eq:concrete-global-local-inclusion}
\end{equation}
The maps in \eqref{eq:concrete-local-inclusion} and \eqref{eq:concrete-global-local-inclusion} are injective
unital $\ast$-homomorphisms and satisfy
$i_{\Lambda',\Lambda''}\circ i_{\Lambda,\Lambda'}=i_{\Lambda,\Lambda''}$ and
$i_{\Lambda',\Gamma}\circ i_{\Lambda,\Lambda'}=i_{\Lambda,\Gamma}$ whenever
$\Lambda\subset \Lambda'\subset \Lambda''$.
For disjoint finite regions, the global images of these maps commute.
For every $n\in\monnat$, their finite-region ranges are norm dense in $\oa{B}_n(\sphilbfrak{h})$ after applying $p_n$.
If $\sphilbfrak{f}\subset\sphilbfrak{h}_{\Lambda}$, the range of
$E_{m,n}\circ j_{\sphilbfrak{f},m}$ is contained in the range of $p_n\circ i_{\Lambda,\Gamma}$ for $0\leq m\leq n$.
\end{prop}

The tensor-factor identities and the proof of Proposition \ref{prop:concrete-spatial-localization-system} are given in Appendix \ref{sec:paper-spatial-localization-system}.

\begin{cor}[Fock images of the local and quasilocal algebras]\label{cor:paper-fock-local-algebras}
For every nonempty $\Lambda\in\setpowfin{\Gamma}$, the local algebra \eqref{eq:paper-main-local-buchholz-algebra}
and the quasilocal algebra \eqref{eq:paper-main-quasilocal-algebra} satisfy
\begin{equation}
\begin{aligned}
\fun{\oarepn_{\txtfock}}{\oa{B}_{\Lambda}}
=
\faadj{U_{\sphilbfrak{h}_{\Lambda}}}
\rbk{
\opspbddlin{\fun{\spfockgamma_{\txtbsn}}{\sphilbfrak{h}_{\Lambda}}}^{\gamma}
\otimes1}
U_{\sphilbfrak{h}_{\Lambda}},
\quad
\fun{\oarepn_{\txtfock}}{\oa{B}_{\mathrm{ql}}}
=
\gtclos{
\bigcup_{\Lambda\in\setpowfin{\Gamma}}
\fun{\oarepn_{\txtfock}}{\oa{B}_{\Lambda}}
}^{\norm{\mathord{\cdot}}},
\end{aligned}
\label{eq:paper-main-fock-local-algebra-identification}
\end{equation}
where $U_{\sphilbfrak{h}_{\Lambda}}$ is the Fock decomposition unitary in
\eqref{eq:paper-main-fock-decomposition-unitary}.
\end{cor}

\begin{proof}
The identity $\oarepn_{\txtfock}\circ\jmath=\id$ on $\oa{D}_{\txtfock}$ in Section
\ref{sec:concrete-fock-comparison} and \eqref{eq:concrete-global-local-inclusion} give, for
$D\in\oa{M}(\oa{K}(\sphilbfrak{h}_{\Lambda})^{\gamma})$,
$$\begin{aligned}
\fun{\oarepn_{\txtfock}\circ i_{\Lambda,\Gamma}}{D}
&=
\faadj{U_{\sphilbfrak{h}_{\Lambda}}}
\rbk{\fun{\diag_{\sphilbfrak{h}_{\Lambda}}}{D}\otimes1}
U_{\sphilbfrak{h}_{\Lambda}}.
\end{aligned}$$
Surjectivity of $\diag_{\sphilbfrak{h}_{\Lambda}}$ in Proposition \ref{prop:finite-mode-product} gives the first
identity.
Corollary \ref{cor:paper-fock-cutoff-realization} shows that $\oarepn_{\txtfock}$ is an isometric
$\ast$-homomorphism.
An isometric $\ast$-homomorphism maps the norm closure of
$\bigcup_{\Lambda\in\setpowfin{\Gamma}}\oa{B}_{\Lambda}$ onto the norm closure of the union of the images, which
is the second identity.
\end{proof}

\section{Conclusion}\label{conclusion}

In this section, we summarize the representation-independent algebra, its finite-mode trace construction, and the scope of the concrete Fock and spatial results.

The abstract Buchholz algebra is the bounded inverse limit of particle-cutoff quotients constructed from the universal resolvent algebra, the canonical minimal ideals of resolvent subalgebras associated with finite-dimensional one-particle spaces, and the intrinsic particle-block system. No preferred representation enters this definition. Theorem \ref{thm:paper-particle-coefficient-buchholz-isomorphism} identifies this inverse limit isometrically with its particle-degree coefficient coordinates and gives their product. The cutoff quotients have the finite ideal chains and degree-coefficient subquotients of Theorem \ref{thm:paper-finite-stage-ideal-chain}.

The finite-dimensional matrix corners provide canonical traces, Gibbs states under trace summability, and closed-path expansions. Under the heat-kernel hypotheses of Corollary \ref{cor:paper-trace-summable-heat-kernel}, the Gibbs state restricts to a KMS state for the block dynamics by Proposition \ref{prop:paper-intrinsic-gibbs-kms-state}. Under potential superstability, a summable Gaussian-weighted kinetic heat trace, and sectorwise trace-norm Trotter convergence, Proposition \ref{prop:paper-gibbs-sector-product-formula} gives an all-sector product formula that permits unbounded kinetic generators. Together with potential superstability, polynomial corner growth and a linear kinetic norm bound give the finite-dimensional lattice criterion in Lemma \ref{lem:paper-summable-sector-majorant}. The closed-path identity of Proposition \ref{prop:paper-abstract-path-integral} is a particle-number cutoff limit; an absolutely convergent sum of individual path weights needs a separate absolute majorant. The probability-measure and reflection-positivity conclusions of Proposition \ref{prop:path-weight-probability-criterion} hold under its stated nonnegativity, time-reversal, and positivity assumptions on the path weights.

The comparison with bosonic Fock space uses the established exact-particle algebras, connecting maps, and bounded inverse limit from \cite[Lemmas 3.3 and 3.4 and Theorem 3.5]{DetlevBuchholz002}. It identifies the intrinsic cutoff quotients and particle coefficients with those concrete objects. Section \ref{sec:paper-continuum-graph-buchholz-realizations} gives finite cell subspaces of \(\fun{\lp^{2}}{\fldreal^{d}}\) and finite graph-region subspaces of \(\fun{\lpseq^{2}}{\Gamma}\) for the corresponding finite-mode block algebras and their embeddings in the full Buchholz algebra.

For a countable graph, Proposition \ref{prop:concrete-spatial-localization-system} verifies compatible local embeddings, commutation for disjoint regions, and density at every fixed particle cutoff. Their norm-closed union defines the quasilocal algebra within the full bounded inverse limit.

The closed-path expansions of Section \ref{sec:intrinsic-finite-volume} provide the finite-mode starting point for path integrals on the abstract Buchholz algebra. Whether automorphism groups and equilibrium states of interacting bosons can actually be constructed by path integrals without passing to the Fock representation is the subject of our next study.

\section{AI statement}\label{ai-statement}

Large language models were used to translate the author's draft into English and to proofread the English text. They were also used to adjust the overall organization of the paper, to check cross-references and bibliographic citations, and to check the details of the proofs for errors and inconsistencies. The author reviewed and verified all suggested changes and remains fully responsible for the mathematical content and the text of this paper.

\appendix

\section{Appendix}\label{appendix}

In this appendix, supplementary coefficient and Poisson calculations, cutoff-extension and spatial results, and representation, \(K\)-theory, and Busby-map computations are collected.

\subsection{Supplementary particle-coefficient and Poisson structures}\label{supplementary-particle-coefficient-and-poisson-structures}

In this appendix, unitary-symmetry refinements, the Rees deformation, the algebraic Poisson bracket, and the concrete Poisson nonclosure example are collected.

\subsubsection{Particle-coefficient reconstruction data compatible with unitary symmetries}\label{sec:paper-symmetric-particle-coefficient-reconstruction-data}

In this subsubsection, we extend the particle-coefficient reconstruction data to a unitary symmetry of \(\sphilbfrak{h}\) and identify the induced compatible actions.

Let \(\mathcal{G}\) be a group and let \[\begin{aligned}
u
\colon
\mathcal{G}
\to
\opspbddlin{\sphilbfrak{h}};\quad
\fun{u}{g}
&=
u_g
\end{aligned}\] be a unitary representation. Since \(u_g\) preserves \(\sigma\), universality gives \[\begin{aligned}
\alpha_g
\colon
\oaresolventalgebra(\sphilbfrak{h},\sigma)
\to
\oaresolventalgebra(\sphilbfrak{h},\sigma);\quad
\fun{\alpha_g}{\oaresolvent(\lambda,f)}
&=
\oaresolvent(\lambda,\fun{u_g}{f}).
\end{aligned}\] For every generator, \[\begin{aligned}
\fun{\alpha_g\circ\gamma_\theta}{\oaresolvent(\lambda,f)}
=
\fun{\oaresolvent}{\lambda,\fun{u_g}{\napiernum^{\imunit\theta}f}}
=
\fun{\oaresolvent}{\lambda,\napiernum^{\imunit\theta}\fun{u_g}{f}}
=
\fun{\gamma_\theta\circ\alpha_g}{\oaresolvent(\lambda,f)}.
\end{aligned}\] Thus \(\alpha_g\) preserves the gauge-fixed algebra. Assume the block covariance \begin{equation}
\begin{aligned}
\fun{\alpha_g}{e_{\sphilbfrak{f},m}}
&=
e_{\fun{u_g}{\sphilbfrak{f}},m}.
\end{aligned}
\label{eq:paper-group-block-covariance}
\end{equation} It gives \[\begin{aligned}
\fun{\alpha_g}{Q_{\sphilbfrak{f},>n}}
&=
1-\sum_{m=0}^{n}\fun{\alpha_g}{e_{\sphilbfrak{f},m}}
=
Q_{\fun{u_g}{\sphilbfrak{f}},>n}.
\end{aligned}\] For \(A,C\in\oaresolventalgebra(\sphilbfrak{h},\sigma)^{\gamma}\), \[\begin{aligned}
\fun{\alpha_g}{AQ_{\sphilbfrak{f},>n}C}
&=
\fun{\alpha_g}{A}Q_{\fun{u_g}{\sphilbfrak{f}},>n}\fun{\alpha_g}{C}
\in
\oaideal{I}_{>n}(\sphilbfrak{h}).
\end{aligned}\] Consequently \(\fun{\alpha_g}{\oaideal{I}_{>n}(\sphilbfrak{h})}\subset\oaideal{I}_{>n}(\sphilbfrak{h})\). Applying \(\alpha_g\) to the same inclusion for \(g^{-1}\) gives the reverse inclusion, so \[\begin{aligned}
\fun{\alpha_g}{\oaideal{I}_{>n}(\sphilbfrak{h})}
&=
\oaideal{I}_{>n}(\sphilbfrak{h}).
\end{aligned}\] The induced quotient automorphism is \begin{equation}
\begin{aligned}
\alpha_{g,n}
\colon
\oa{B}_n(\sphilbfrak{h})
\to
\oa{B}_n(\sphilbfrak{h});\quad
\fun{\alpha_{g,n}}{A+\oaideal{I}_{>n}(\sphilbfrak{h})}
&=
\fun{\alpha_g}{A}+\oaideal{I}_{>n}(\sphilbfrak{h}).
\end{aligned}
\label{eq:paper-group-particle-cutoff-automorphism}
\end{equation} For \(A,C\in\oaresolventalgebra(\sphilbfrak{h},\sigma)^{\gamma}\), the following identities prove that this map preserves products and adjoints: \[\begin{aligned}
\fun{\alpha_{g,n}}{
\rbk{A+\oaideal{I}_{>n}(\sphilbfrak{h})}
\rbk{C+\oaideal{I}_{>n}(\sphilbfrak{h})}}
&=
\fun{\alpha_g}{A}\fun{\alpha_g}{C}+\oaideal{I}_{>n}(\sphilbfrak{h}),
\\
\fun{\alpha_{g,n}}{\faadjrbk{A+\oaideal{I}_{>n}(\sphilbfrak{h})}}
&=
\faadj{\fun{\alpha_g}{A}}+\oaideal{I}_{>n}(\sphilbfrak{h}).
\end{aligned}\] For every \(A+\oaideal{I}_{>n}(\sphilbfrak{h})\in\oa{B}_n(\sphilbfrak{h})\), \[\begin{aligned}
\fun{\alpha_{g^{-1},n}\circ\alpha_{g,n}}{A+\oaideal{I}_{>n}(\sphilbfrak{h})}
=
\fun{\alpha_{g^{-1}}\circ\alpha_g}{A}+\oaideal{I}_{>n}(\sphilbfrak{h})
=
A+\oaideal{I}_{>n}(\sphilbfrak{h}).
\end{aligned}\] The reverse composition is the identity by the same calculation. Use the canonical injection \(j_{\sphilbfrak{f},m}\) into the inductive limit from \eqref{eq:paper-main-degree-coefficient-canonical-injection}. For every \(m\in\monnat\), define \[\begin{aligned}
\alpha_{\mathrm{coef},m,g}
\colon
\oa{C}_m(\sphilbfrak{h})
\to
\oa{C}_m(\sphilbfrak{h}).
\end{aligned}\] Its action on the canonical images is \[\begin{aligned}
\fun{\alpha_{\mathrm{coef},m,g} \circ j_{\sphilbfrak{f},m}}{K}
=
\fun{j_{\fun{u_g}{\sphilbfrak{f}},m} \circ \alpha_g}{K}.
\end{aligned}\] For \(\sphilbfrak{f}\subset \sphilbfrak{g}\), covariance of the connecting maps is calculated by \[\begin{aligned}
\fun{j_{\fun{u_g}{\sphilbfrak{f}},\fun{u_g}{\sphilbfrak{g}},m}\circ\alpha_g}{K}
&=
e_{\fun{u_g}{\sphilbfrak{g}},m}
\cdot
\fun{\iota_{\fun{u_g}{\sphilbfrak{f}},\fun{u_g}{\sphilbfrak{g}}}\circ\alpha_g}{K}
\cdot
e_{\fun{u_g}{\sphilbfrak{g}},m}
\\
&=
\fun{\alpha_g}{e_{\sphilbfrak{g},m}\fun{\iota_{\sphilbfrak{f},\sphilbfrak{g}}}{K}e_{\sphilbfrak{g},m}}
=
\fun{\alpha_g\circ j_{\sphilbfrak{f},\sphilbfrak{g},m}}{K}.
\end{aligned}\] Hence \(\alpha_{\mathrm{coef},m,g}\) is well-defined. For \(K\in e_{\sphilbfrak{f},m}\oa{K}(\sphilbfrak{f})^\gamma e_{\sphilbfrak{f},m}\) and \(L\in e_{\sphilbfrak{g},m}\oa{K}(\sphilbfrak{g})^\gamma e_{\sphilbfrak{g},m}\), preservation of products and adjoints by \(\alpha_g\) gives \[\begin{aligned}
\fun{\alpha_{\mathrm{coef},m,g}}{
\fun{j_{\sphilbfrak{f},m}}{K}\fun{j_{\sphilbfrak{g},m}}{L}}
&=
\fun{j_{\fun{u_g}{\sphilbfrak{f}+\sphilbfrak{g}},m} \circ \alpha_{g}}
{\fun{j_{\sphilbfrak{f},\sphilbfrak{f}+\sphilbfrak{g},m}}{K}
\cdot
\fun{j_{\sphilbfrak{g},\sphilbfrak{f}+\sphilbfrak{g},m}}{L}}
\\ 
&=
\fun{j_{\fun{u_g}{\sphilbfrak{f}+\sphilbfrak{g}},m}
\circ \alpha_g
\circ j_{\sphilbfrak{f},\sphilbfrak{f}+\sphilbfrak{g},m}}
{K}
\cdot
\fun{j_{\fun{u_g}{\sphilbfrak{f}+\sphilbfrak{g}},m}
\circ \alpha_g
\circ j_{\sphilbfrak{g},\sphilbfrak{f}+\sphilbfrak{g},m}}
{L},
\\ 
\fun{\alpha_{\mathrm{coef},m,g}}{\faadj{\fun{j_{\sphilbfrak{f},m}}{K}}}
&=
\fun{j_{\fun{u_g}{\sphilbfrak{f}},m}
\circ \alpha_g}
{\faadj{K}}
=
\faadjrbk{\fun{j_{\fun{u_g}{\sphilbfrak{f}},m} \circ \alpha_g}
{K}}.
\end{aligned}\] For every \(K\in e_{\sphilbfrak{f},m}\oa{K}(\sphilbfrak{f})^\gamma e_{\sphilbfrak{f},m}\), \[\begin{aligned}
\fun{\alpha_{\mathrm{coef},m,g^{-1}}\circ\alpha_{\mathrm{coef},m,g}
\circ
j_{\sphilbfrak{f},m}}
{K}
=
\fun{j_{\sphilbfrak{f},m} \circ \alpha_{g^{-1}} \circ \alpha_g}{K}
=
\fun{j_{\sphilbfrak{f},m}}{K}.
\end{aligned}\] The reverse composition is the identity by the same calculation. Consequently, the map \(\alpha_{\mathrm{coef},m,g}\) is a \(\ast\)-automorphism.

\begin{defn}[Particle-coefficient reconstruction data compatible with $\mathcal{G}$]\label{def:group-compatible-particle-coefficient-reconstruction-data}
The abstract particle-coefficient reconstruction data of Definition \ref{def:paper-abstract-particle-coefficient-reconstruction-data} are called compatible with $\mathcal{G}$ if
the intrinsic particle-block system satisfies \eqref{eq:paper-group-block-covariance} and, for every
$g\in\mathcal{G}$ and $0\leq m\leq n$, the quotient automorphism $\alpha_{g,n}$ defined in
\eqref{eq:paper-group-particle-cutoff-automorphism} satisfies
$$\begin{aligned}
\alpha_{g,n}\circ E_{m,n}
&=
E_{m,n}\circ\alpha_{\mathrm{coef},m,g}.
\end{aligned}$$
Here both sides are maps from
$\oa{C}_m(\sphilbfrak{h})$ to $\oa{B}_n(\sphilbfrak{h})$.
\end{defn}

\begin{prop}[Fixed particle-degree coefficients]\label{prop:paper-group-fixed-particle-coefficients}
Assume the compatibility condition in Definition \ref{def:group-compatible-particle-coefficient-reconstruction-data}.
If $A_n=\sum_{m=0}^{n}\fun{E_{m,n}}{a_m}$, then $\fun{\alpha_{g,n}}{A_n}=A_n$ for every $g\in\mathcal{G}$ if
and only if $\fun{\alpha_{\mathrm{coef},m,g}}{a_m}=a_m$ for every $g\in\mathcal{G}$ and $0\leq m\leq n$.
The overlap products are equivariant:
\begin{equation}
\begin{aligned}
\fun{\alpha_{\mathrm{coef},r,g}}{\fun{\mu_{p,q}^{r}}{a_p,b_q}}
&=
\fun{\mu_{p,q}^{r}}
{\fun{\alpha_{\mathrm{coef},p,g}}{a_p},\fun{\alpha_{\mathrm{coef},q,g}}{b_q}}.
\end{aligned}
\label{eq:paper-group-overlap-equivariance}
\end{equation}
Thus the bounded particle-coefficient families with $\mathcal{G}$-fixed coefficients form a unital $\ast$-subalgebra.
\end{prop}

\begin{proof}
Equivariance gives
$$\begin{aligned}
\fun{\alpha_{g,n}}{A_n}
&=
\sum_{m=0}^{n}\fun{E_{m,n}}{\fun{\alpha_{\mathrm{coef},m,g}}{a_m}}.
\end{aligned}$$
Lemma \ref{lem:paper-finite-particle-coefficient-uniqueness} proves the equivalence coefficient by coefficient.
Apply $\alpha_{g,n}$ to the overlap-product identity and take $n\geq p+q$.
The first expansion is
$$\begin{aligned}
\sum_{r=\max\setone{p,q}}^{p+q}
\fun{E_{r,n}}{\fun{\alpha_{\mathrm{coef},r,g}}{\fun{\mu_{p,q}^{r}}{a_p,b_q}}}.
\end{aligned}$$
The second expansion is
$$\begin{aligned}
\sum_{r=\max\setone{p,q}}^{p+q}
\fun{E_{r,n}}{
\fun{\mu_{p,q}^{r}}{
\fun{\alpha_{\mathrm{coef},p,g}}{a_p},
\fun{\alpha_{\mathrm{coef},q,g}}{b_q}
}}.
\end{aligned}$$
Finite-coefficient uniqueness gives \eqref{eq:paper-group-overlap-equivariance} for every output degree.
If $a,b\in\oa{C}_{\mathrm{b}}(\sphilbfrak{h})$ have $\mathcal{G}$-fixed coefficients, then, for every
$r\in\monnat$ and $g\in\mathcal{G}$,
$$\begin{aligned}
&\fun{\alpha_{\mathrm{coef},r,g}}{\rbk{a\star b}_r}
=
\sum_{\substack{0\leq p,q\leq r\\p+q\geq r}}
\fun{\alpha_{\mathrm{coef},r,g}}{\fun{\mu_{p,q}^{r}}{a_p,b_q}}
\\
&=
\sum_{\substack{0\leq p,q\leq r\\p+q\geq r}}
\fun{\mu_{p,q}^{r}}
{\fun{\alpha_{\mathrm{coef},p,g}}{a_p},
\fun{\alpha_{\mathrm{coef},q,g}}{b_q}}
=
\rbk{a\star b}_r.
\end{aligned}$$
For adjoints, the coefficientwise adjoint identity in
\eqref{eq:paper-main-particle-coefficient-operations} and the fact that
$\alpha_{\mathrm{coef},r,g}$ is a $\ast$-automorphism give
$$\begin{aligned}
\fun{\alpha_{\mathrm{coef},r,g}}{\rbk{\faadj{a}}_r}
&=
\fun{\alpha_{\mathrm{coef},r,g}}{\faadj{a_r}}
=
\faadj{\fun{\alpha_{\mathrm{coef},r,g}}{a_r}}
=
\faadj{a_r}
=
\rbk{\faadj{a}}_r.
\end{aligned}$$
The unit sequence has degree-zero coefficient $1$ and every positive-degree coefficient equal to $0$.
The degree-zero automorphism fixes $1$, and linearity gives
$\fun{\alpha_{\mathrm{coef},r,g}}{0}=0$ for every $r\geq1$; hence the unit sequence is fixed.
\end{proof}

\subsubsection{Rees deformation and the algebraic Poisson bracket}\label{sec:paper-rees-poisson}

In this subsubsection, we use the output-degree range of the abstract particle-coefficient reconstruction data to construct the Rees deformation and its algebraic Poisson bracket; neither a lattice nor a graph is required.

For input degrees \(p\) and \(q\) and output degree \(r\), define the overlap number by \[\begin{aligned}
k
&=
p+q-r.
\end{aligned}\] The output-degree range \eqref{eq:paper-abstract-particle-coefficient-output-degree-range} is equivalent to \(0\leq k\leq\min\setone{p,q}\). Let \(\oa{C}_{\txtfin}(\sphilbfrak{h})\) be the algebraic direct-sum subspace of the particle-coefficient space \(\oa{C}(\sphilbfrak{h})\) in \eqref{eq:paper-main-particle-coefficient-space}, and define its increasing filtration by \[\begin{aligned}
\oa{C}_{\leq p}(\sphilbfrak{h})
&=
\set{a=\seq{a_m}{m\in\monnat}\in\oa{C}_{\txtfin}(\sphilbfrak{h})}
{a_m=0\text{ for every }m>p}
&&\rbk{p\geq0},
\\
\oa{C}_{\leq p}(\sphilbfrak{h})
&=
0
&&\rbk{p<0}.
\end{aligned}\] Take \(a\in\oa{C}_{\leq p}(\sphilbfrak{h})\) and \(b\in\oa{C}_{\leq q}(\sphilbfrak{h})\). If \(r>p+q\), no pair \(s\leq p\) and \(t\leq q\) satisfies \(s+t\geq r\), and hence \[\begin{aligned}
\rbk{a\star b}_r
&=
\sum_{\substack{0\leq s\leq\min\setone{p,r}\\0\leq t\leq\min\setone{q,r}\\s+t\geq r}}
\fun{\mu_{s,t}^{r}}{a_s,b_t}
=
0.
\end{aligned}\] Thus \(\oa{C}_{\leq p}(\sphilbfrak{h})\oa{C}_{\leq q}(\sphilbfrak{h})\subset
\oa{C}_{\leq p+q}(\sphilbfrak{h})\). For \(r=p+q\), the same index restrictions force \((s,t)=(p,q)\). The symmetry of the highest-degree overlap product gives \[\begin{aligned}
\rbk{a\star b-b\star a}_{p+q}
&=
\fun{\mu_{p,q}^{p+q}}{a_p,b_q}
-
\fun{\mu_{q,p}^{p+q}}{b_q,a_p}
=
0.
\end{aligned}\] The degree filtration therefore satisfies \[\begin{aligned}
\commutator{\oa{C}_{\leq p}(\sphilbfrak{h})}{\oa{C}_{\leq q}(\sphilbfrak{h})}
&\subset
\oa{C}_{\leq p+q-1}(\sphilbfrak{h}).
\end{aligned}\]

Let \(z\) be a central formal variable. The Rees algebra and its fiber at \(z=0\) are \[\begin{aligned}
\fun{\operatorname{Rees}}{\oa{C}_{\txtfin}(\sphilbfrak{h})}
&=
\bigoplus_{p\in\monnat}z^p\oa{C}_{\leq p}(\sphilbfrak{h}),
\\
\setquot{\fun{\operatorname{Rees}}{\oa{C}_{\txtfin}(\sphilbfrak{h})}}
{z\fun{\operatorname{Rees}}{\oa{C}_{\txtfin}(\sphilbfrak{h})}}
&\simeq
\bigoplus_{p\in\monnat}
\setquot{\oa{C}_{\leq p}(\sphilbfrak{h})}{\oa{C}_{\leq p-1}(\sphilbfrak{h})}.
\end{aligned}\] Write \(\fun{\operatorname{gr}}{\oa{C}_{\txtfin}(\sphilbfrak{h})}\) for the graded direct sum representing the fiber at \(z=0\). For \(p\geq0\), define the principal-symbol map by \[\begin{aligned}
\sigma_p
\colon
\oa{C}_{\leq p}(\sphilbfrak{h})
\to
\setquot{\oa{C}_{\leq p}(\sphilbfrak{h})}{\oa{C}_{\leq p-1}(\sphilbfrak{h})};\quad
\fun{\sigma_p}{a}
&=
a+\oa{C}_{\leq p-1}(\sphilbfrak{h}),
\end{aligned}\] and set \(\fun{\sigma_p}{a}=0\) for \(p<0\). Define \[\begin{aligned}
\physpoissonbkt{\mathord{\cdot},\mathord{\cdot}}
\colon
\fun{\operatorname{gr}}{\oa{C}_{\txtfin}(\sphilbfrak{h})}
\times
\fun{\operatorname{gr}}{\oa{C}_{\txtfin}(\sphilbfrak{h})}
\to
\fun{\operatorname{gr}}{\oa{C}_{\txtfin}(\sphilbfrak{h})};\quad
\physpoissonbkt{\fun{\sigma_p}{a},\fun{\sigma_q}{b}}
&=
\frac{1}{\imunit}
\fun{\sigma_{p+q-1}}{\commutator{a}{b}}.
\end{aligned}\]

\begin{prop}[Poisson identities for finitely supported coefficients]\label{prop:paper-graded-poisson}
The graded bracket $\physpoissonbkt{\mathord{\cdot},\mathord{\cdot}}$ is independent of the chosen representatives.
It makes $\fun{\operatorname{gr}}{\oa{C}_{\txtfin}(\sphilbfrak{h})}$ a Poisson algebra whose product is commutative and
whose bracket has degree $-1$.
\end{prop}

\begin{proof}
Let $a'=a+u$ and $b'=b+v$, where $u\in\oa{C}_{\leq p-1}(\sphilbfrak{h})$ and
$v\in\oa{C}_{\leq q-1}(\sphilbfrak{h})$.
The filtration estimates give
$$\begin{aligned}
\commutator{u}{b},\ \commutator{a}{v}
\in
\oa{C}_{\leq p+q-2}(\sphilbfrak{h}),
\quad
\commutator{u}{v}
\in
\oa{C}_{\leq p+q-3}(\sphilbfrak{h})
\subset
\oa{C}_{\leq p+q-2}(\sphilbfrak{h}).
\end{aligned}$$
Consequently, the two candidate principal symbols differ by
$$\begin{aligned}
&
\frac{1}{\imunit}\fun{\sigma_{p+q-1}}{\commutator{a'}{b'}}
-
\frac{1}{\imunit}\fun{\sigma_{p+q-1}}{\commutator{a}{b}}
=
\frac{1}{\imunit}\fun{\sigma_{p+q-1}}{
\commutator{a'}{b'}-\commutator{a}{b}}
\\ 
&=
\frac{1}{\imunit}
\fun{\sigma_{p+q-1}}{
\commutator{u}{b}
+
\commutator{a}{v}
+
\commutator{u}{v}}
=
\frac{1}{\imunit}
\rbk{
\commutator{u}{b}
+
\commutator{a}{v}
+
\commutator{u}{v}
+
\oa{C}_{\leq p+q-2}(\sphilbfrak{h})}
=
0.
\end{aligned}$$
Thus the bracket is independent of the chosen representatives.
The product of two homogeneous principal symbols satisfies
$$\begin{aligned}
\fun{\sigma_p}{a}\fun{\sigma_q}{b}
-
\fun{\sigma_q}{b}\fun{\sigma_p}{a}
&=
\fun{\sigma_{p+q}}{ab-ba}
=
0,
\end{aligned}$$
so the graded product is commutative.
Antisymmetry follows by the explicit calculation
$$\begin{aligned}
\physpoissonbkt{\fun{\sigma_p}{a},\fun{\sigma_q}{b}}
&=
\frac{1}{\imunit}\fun{\sigma_{p+q-1}}{\commutator{a}{b}}
=
-\physpoissonbkt{\fun{\sigma_q}{b},\fun{\sigma_p}{a}}.
\end{aligned}$$
The commutator identity $\commutator{a}{bc}=\commutator{a}{b}c+b\commutator{a}{c}$ gives
$$\begin{aligned}
\physpoissonbkt{\fun{\sigma_p}{a},\fun{\sigma_{q+t}}{bc}}
&=
\frac{1}{\imunit}
\fun{\sigma_{p+q+t-1}}{\commutator{a}{b}c+b\commutator{a}{c}}
\\
&=
\physpoissonbkt{\fun{\sigma_p}{a},\fun{\sigma_q}{b}}\fun{\sigma_t}{c}
+
\fun{\sigma_q}{b}\physpoissonbkt{\fun{\sigma_p}{a},\fun{\sigma_t}{c}},
\end{aligned}$$
which is the Leibniz rule in the second variable.
Commutativity and antisymmetry give the rule in the first variable.
Finally, applying $\imunit^{-2}\sigma_{p+q+t-2}$ to
$$\begin{aligned}
\commutator{a}{\commutator{b}{c}}
+
\commutator{b}{\commutator{c}{a}}
+
\commutator{c}{\commutator{a}{b}}
&=
0
\end{aligned}$$
gives the Jacobi identity for the Poisson bracket.
\end{proof}

The full coefficient space carries the product topology defined by \[\begin{aligned}
\norm{a}_{[N]}
&=
\max_{0\leq m\leq N}\norm{a_m}
&&\rbk{N\in\monnat}.
\end{aligned}\] For \(z\in\fldcmp\), define the Rees product and the first-overlap bracket by \begin{equation}
\begin{aligned}
\star_z
\colon
\oa{C}(\sphilbfrak{h})\times\oa{C}(\sphilbfrak{h})
\to
\oa{C}(\sphilbfrak{h});\quad
\rbk{a\star_zb}_r
&=
\sum_{\substack{0\leq p,q\leq r\\p+q\geq r}}
z^{p+q-r}\fun{\mu_{p,q}^{r}}{a_p,b_q},
\end{aligned}
\label{eq:paper-completed-rees-product}
\end{equation} where \(0^0=1\), and \begin{equation}
\begin{aligned}
\physpoissonbkt{\mathord{\cdot},\mathord{\cdot}}_{\infty}
\colon
\oa{C}(\sphilbfrak{h})\times\oa{C}(\sphilbfrak{h})
\to
\oa{C}(\sphilbfrak{h});\quad
\rbk{\physpoissonbkt{a,b}_{\infty}}_r
&=
\frac{1}{\imunit}
\sum_{\substack{1\leq p,q\leq r+1\\p+q=r+1}}
\rbk{
\fun{\mu_{p,q}^{r}}{a_p,b_q}
-
\fun{\mu_{q,p}^{r}}{b_q,a_p}}.
\end{aligned}
\label{eq:paper-completed-poisson-bracket}
\end{equation} The sum in each output degree is finite, so both maps are well-defined on the full product.

\begin{prop}[Completed Rees product and Poisson limit]\label{prop:paper-completed-rees-poisson}
The product $\star_z$ is associative for every $z\in\fldcmp$.
For every $z\in\fldreal$ and $a,b\in\oa{C}(\sphilbfrak{h})$, it satisfies
\begin{equation}
\begin{aligned}
\faadjrbk{a\star_z b}
&=
\faadj{b}\star_z\faadj{a}.
\end{aligned}
\label{eq:paper-completed-rees-adjoint}
\end{equation}
The product $\star_0$ is commutative.
The bracket in \eqref{eq:paper-completed-poisson-bracket} is antisymmetric and satisfies the Jacobi and Leibniz
identities for $\star_0$.
For every $N\in\monnat$ and $R\geq0$, there are finite constants $C_{N,R}$ and $C_N^{\mathrm{P}}$ such that
\begin{equation}
\begin{aligned}
\norm{a\star_zb}_{[N]}
&\leq
C_{N,R}\norm{a}_{[N]}\norm{b}_{[N]}
&&\rbk{\abs{z}\leq R},
\\
\norm{\physpoissonbkt{a,b}_{\infty}}_{[N]}
&\leq
C_N^{\mathrm{P}}\norm{a}_{[N]}\norm{b}_{[N]}.
\end{aligned}
\label{eq:paper-completed-rees-seminorm-bounds}
\end{equation}
For real $z$ tending to zero,
\begin{equation}
\begin{aligned}
\lim_{z \to 0}
\frac{a\star_zb-b\star_za}{\imunit z}
=
\physpoissonbkt{a,b}_{\infty}
\end{aligned}
\label{eq:paper-rees-commutator-poisson-limit}
\end{equation}
in the product topology.
\end{prop}

\begin{proof}
For $z\neq0$, define
$$\begin{aligned}
D_z
\colon
\oa{C}(\sphilbfrak{h})
\to
\oa{C}(\sphilbfrak{h});\quad
\rbk{\fun{D_z}{a}}_m
&=
z^ma_m.
\end{aligned}$$
For the same fixed $z\in\fldcmp\setminus\setone{0}$, define the inverse map by
$$\begin{aligned}
D_z^{-1}
\colon
\oa{C}(\sphilbfrak{h})
\to
\oa{C}(\sphilbfrak{h});\quad
\rbk{\fun{D_z^{-1}}{a}}_m
&=
z^{-m}a_m.
\end{aligned}$$
Coefficientwise substitution in \eqref{eq:paper-completed-rees-product} gives
$$\begin{aligned}
a\star_zb
&=
\fun{D_z^{-1}}{\fun{D_z}{a}\star\fun{D_z}{b}}.
\end{aligned}$$
Theorem \ref{thm:paper-particle-coefficient-buchholz-isomorphism} gives associativity for $z\neq0$.
For fixed output degree $r$, substituting \eqref{eq:paper-completed-rees-product} twice gives
$$\begin{aligned}
\rbk{\rbk{a\star_zb}\star_zc-a\star_z\rbk{b\star_zc}}_r
&=
\sum_{\substack{0\leq s,t\leq r\\s+t\geq r}}
z^{s+t-r}
\rbk{
\fun{\mu_{s,t}^{r}}{\rbk{a\star_zb}_s,c_t}
-
\fun{\mu_{s,t}^{r}}{a_s,\rbk{b\star_zc}_t}
}.
\end{aligned}$$
Each inner coefficient is a finite polynomial in $z$.
The associator vanishes for $z\neq0$, and norm continuity of this polynomial gives
$$\begin{aligned}
\rbk{\rbk{a\star_0b}\star_0c-a\star_0\rbk{b\star_0c}}_r
&=
\lim_{z\to0}
\rbk{\rbk{a\star_zb}\star_zc-a\star_z\rbk{b\star_zc}}_r
=
0.
\end{aligned}$$
The symmetry of the zero-overlap product gives, for every $r\in\monnat$,
$$\begin{aligned}
\rbk{a\star_0b}_r
&=
\sum_{\substack{0\leq p,q\leq r\\p+q=r}}
\fun{\mu_{p,q}^{r}}{a_p,b_q}
=
\sum_{\substack{0\leq p,q\leq r\\p+q=r}}
\fun{\mu_{q,p}^{r}}{b_q,a_p}
=
\rbk{b\star_0a}_r.
\end{aligned}$$
For nonzero real $z$, the map $D_z$ preserves adjoints and hence
$$\begin{aligned}
\faadjrbk{a\star_zb}
&=
\fun{D_z^{-1}}{\faadjrbk{\fun{D_z}{a}\star\fun{D_z}{b}}}
=
\faadj{b}\star_z\faadj{a}.
\end{aligned}$$
The two sides are coefficientwise polynomials in the real variable $z$, so the identity extends to $z=0$.
This proves \eqref{eq:paper-completed-rees-adjoint}.

Write $\norm{\mu_{p,q}^{r}}$ for the bilinear operator norm and set
$$\begin{aligned}
C_{N,R}
&=
\max_{0\leq r\leq N}
\sum_{\substack{0\leq p,q\leq r\\p+q\geq r}}
R^{p+q-r}\norm{\mu_{p,q}^{r}},
\\
C_N^{\mathrm{P}}
&=
\max_{0\leq r\leq N}
\sum_{\substack{1\leq p,q\leq r+1\\p+q=r+1}}
\rbk{\norm{\mu_{p,q}^{r}}+\norm{\mu_{q,p}^{r}}}.
\end{aligned}$$
For $0\leq r\leq N$ and $\abs{z}\leq R$, the triangle inequality gives
$$\begin{aligned}
\norm{\rbk{a\star_zb}_r}
\leq
\sum_{\substack{0\leq p,q\leq r\\p+q\geq r}}
R^{p+q-r}\norm{\mu_{p,q}^{r}}\norm{a_p}\norm{b_q}
\leq
C_{N,R}\norm{a}_{[N]}\norm{b}_{[N]}.
\end{aligned}$$
Taking the maximum over $0\leq r\leq N$ proves
the first estimate in \eqref{eq:paper-completed-rees-seminorm-bounds}.
The bracket is estimated in the same way:
$$\begin{aligned}
\norm{\rbk{\physpoissonbkt{a,b}_{\infty}}_r}
\leq
\sum_{\substack{1\leq p,q\leq r+1\\p+q=r+1}}
\rbk{\norm{\mu_{p,q}^{r}}+\norm{\mu_{q,p}^{r}}}
\norm{a_p}\norm{b_q}
\leq
C_N^{\mathrm{P}}\norm{a}_{[N]}\norm{b}_{[N]}.
\end{aligned}$$
Taking the maximum over $0\leq r\leq N$ proves the second estimate in
\eqref{eq:paper-completed-rees-seminorm-bounds}.

For real $z\neq0$, define
$$\begin{aligned}
P_z
\colon
\oa{C}(\sphilbfrak{h})\times\oa{C}(\sphilbfrak{h})
\to
\oa{C}(\sphilbfrak{h});\quad
\fun{P_z}{a,b}
&=
\frac{a\star_zb-b\star_za}{\imunit z}.
\end{aligned}$$
The zero-overlap terms cancel by symmetry, and the first-overlap terms give
\eqref{eq:paper-completed-poisson-bracket}.
Thus the remainder in degree $r$ is
$$\begin{aligned}
\rbk{\fun{P_z}{a,b}-\physpoissonbkt{a,b}_{\infty}}_r
&=
\frac{1}{\imunit}
\sum_{\substack{0\leq p,q\leq r\\p+q\geq r+2}}
z^{p+q-r-1}
\rbk{
\fun{\mu_{p,q}^{r}}{a_p,b_q}
-
\fun{\mu_{q,p}^{r}}{b_q,a_p}}.
\end{aligned}$$
If $0<\abs{z}\leq1$, every exponent is at least one, and hence
$$\begin{aligned}
\norm{\rbk{\fun{P_z}{a,b}-\physpoissonbkt{a,b}_{\infty}}_r}
&\leq
\abs{z}
\sum_{\substack{0\leq p,q\leq r\\p+q\geq r+2}}
\rbk{\norm{\mu_{p,q}^{r}}+\norm{\mu_{q,p}^{r}}}
\norm{a_p}\norm{b_q}.
\end{aligned}$$
The maximum over $0\leq r\leq N$ tends to zero with $z$, which proves
\eqref{eq:paper-rees-commutator-poisson-limit}.

For $z\neq0$, associativity of $\star_z$ gives
$$\begin{aligned}
&\fun{P_z}{a,b\star_zc}
=
\fun{P_z}{a,b}\star_zc
+
b\star_z\fun{P_z}{a,c},
\\ 
&\fun{P_z}{a,\fun{P_z}{b,c}}
+
\fun{P_z}{b,\fun{P_z}{c,a}}
+
\fun{P_z}{c,\fun{P_z}{a,b}}
=
0.
\end{aligned}$$
Define the constant required for nested limits by
$$\begin{aligned}
C_N^{\mathrm{E}}
&=
\max_{0\leq r\leq N}
\sum_{\substack{0\leq p,q\leq r\\p+q\geq r+2}}
\rbk{\norm{\mu_{p,q}^{r}}+\norm{\mu_{q,p}^{r}}}.
\end{aligned}$$
The coefficient remainder estimate gives, for $0<\abs{z}\leq1$,
$$\begin{aligned}
\norm{\fun{P_z}{a,b}-\physpoissonbkt{a,b}_{\infty}}_{[N]}
&\leq
\abs{z}C_N^{\mathrm{E}}\norm{a}_{[N]}\norm{b}_{[N]},
\\
\norm{\fun{P_z}{a,b}}_{[N]}
&\leq
\rbk{C_N^{\mathrm{P}}+C_N^{\mathrm{E}}}
\norm{a}_{[N]}\norm{b}_{[N]}.
\end{aligned}$$
For every term of positive overlap in \eqref{eq:paper-completed-rees-product},
$\abs{z}^{p+q-r}\leq\abs{z}$, and therefore
$$\begin{aligned}
\norm{b\star_zc-b\star_0c}_{[N]}
&\leq
\abs{z}C_{N,1}\norm{b}_{[N]}\norm{c}_{[N]}.
\end{aligned}$$
The left-hand side of the Leibniz identity converges because
$$\begin{aligned}
&\norm{
\fun{P_z}{a,b\star_zc}
-
\physpoissonbkt{a,b\star_0c}_{\infty}
}_{[N]}
\\
&\leq
\rbk{C_N^{\mathrm{P}}+C_N^{\mathrm{E}}}
\norm{a}_{[N]}
\norm{b\star_zc-b\star_0c}_{[N]}
+
\abs{z}C_N^{\mathrm{E}}
\norm{a}_{[N]}
\norm{b\star_0c}_{[N]}
\xrightarrow{z\to0}
0.
\end{aligned}$$
The first term on the right-hand side converges because
$$\begin{aligned}
&\norm{\fun{P_z}{a,b}\star_zc
-
\physpoissonbkt{a,b}_{\infty}\star_0c
}_{[N]}
\\ 
&\leq
C_{N,1}
\norm{\fun{P_z}{a,b}-\physpoissonbkt{a,b}_{\infty}}_{[N]}
\norm{c}_{[N]}
+
\abs{z}C_{N,1}
\norm{\physpoissonbkt{a,b}_{\infty}}_{[N]}
\norm{c}_{[N]}
\xrightarrow{z \to 0}
0.
\end{aligned}$$
The second term on the right-hand side satisfies
$$\begin{aligned}
&\norm{b\star_z\fun{P_z}{a,c}
-
b\star_0\physpoissonbkt{a,c}_{\infty}}_{[N]}
\\ 
&\leq
C_{N,1}\norm{b}_{[N]}
\norm{\fun{P_z}{a,c}-\physpoissonbkt{a,c}_{\infty}}_{[N]}
+
\abs{z}C_{N,1}\norm{b}_{[N]}
\norm{\physpoissonbkt{a,c}_{\infty}}_{[N]}
\xrightarrow{z \to 0}
0.
\end{aligned}$$
Substitution of these three limits into the $P_z$ Leibniz identity proves the Leibniz rule.

The first cyclic term of the Jacobi identity satisfies
$$\begin{aligned}
&\norm{
\fun{P_z}{a,\fun{P_z}{b,c}}
-
\physpoissonbkt{a,\physpoissonbkt{b,c}_{\infty}}_{\infty}
}_{[N]}
\\ 
&\leq
\rbk{C_N^{\mathrm{P}}+C_N^{\mathrm{E}}}
\norm{a}_{[N]}
\norm{\fun{P_z}{b,c}-\physpoissonbkt{b,c}_{\infty}}_{[N]}
+
\abs{z}C_N^{\mathrm{E}}\norm{a}_{[N]}
\norm{\physpoissonbkt{b,c}_{\infty}}_{[N]}
\xrightarrow{z \to 0}
0.
\end{aligned}$$
The second cyclic term satisfies
$$\begin{aligned}
&\norm{
\fun{P_z}{b,\fun{P_z}{c,a}}
-
\physpoissonbkt{b,\physpoissonbkt{c,a}_{\infty}}_{\infty}
}_{[N]}
\\ 
&\leq
\rbk{C_N^{\mathrm{P}}+C_N^{\mathrm{E}}}
\norm{b}_{[N]}
\norm{\fun{P_z}{c,a}-\physpoissonbkt{c,a}_{\infty}}_{[N]}
+
\abs{z}C_N^{\mathrm{E}}\norm{b}_{[N]}
\norm{\physpoissonbkt{c,a}_{\infty}}_{[N]}
\xrightarrow{z \to 0}
0.
\end{aligned}$$
The third cyclic term satisfies
$$\begin{aligned}
&\norm{
\fun{P_z}{c,\fun{P_z}{a,b}}
-
\physpoissonbkt{c,\physpoissonbkt{a,b}_{\infty}}_{\infty}
}_{[N]}
\\ 
&\leq
\rbk{C_N^{\mathrm{P}}+C_N^{\mathrm{E}}}
\norm{c}_{[N]}
\norm{\fun{P_z}{a,b}-\physpoissonbkt{a,b}_{\infty}}_{[N]}
+
\abs{z}C_N^{\mathrm{E}}\norm{c}_{[N]}
\norm{\physpoissonbkt{a,b}_{\infty}}_{[N]}
\xrightarrow{z \to 0}
0.
\end{aligned}$$
Substitution into the $P_z$ Jacobi identity gives
$$\begin{aligned}
\physpoissonbkt{a,\physpoissonbkt{b,c}_{\infty}}_{\infty}
+
\physpoissonbkt{b,\physpoissonbkt{c,a}_{\infty}}_{\infty}
+
\physpoissonbkt{c,\physpoissonbkt{a,b}_{\infty}}_{\infty}
&=
0.
\end{aligned}$$
Finally, direct substitution in \eqref{eq:paper-completed-poisson-bracket} gives, for every $r\in\monnat$,
$$\begin{aligned}
\rbk{\physpoissonbkt{a,b}_{\infty}}_r
&=
\frac{1}{\imunit}
\sum_{\substack{1\leq p,q\leq r+1\\p+q=r+1}}
\rbk{
\fun{\mu_{p,q}^{r}}{a_p,b_q}
-
\fun{\mu_{q,p}^{r}}{b_q,a_p}
}
=
-\rbk{\physpoissonbkt{b,a}_{\infty}}_r.
\end{aligned}$$
\end{proof}

\begin{prop}[Closedness and boundedness of the Poisson bracket on bounded particle-degree coefficients]\label{prop:paper-bounded-particle-coefficient-poisson-closed-graph}
For every $\ell\in\monnat$, the coordinate evaluation $a\mapsto a_\ell$ on
$\oa{C}_{\mathrm{b}}(\sphilbfrak{h})$ is bounded.

\begin{enumerate}
\item For all $a,b\in\oa{C}_{\mathrm{b}}(\sphilbfrak{h})$, one has
$\physpoissonbkt{a,b}_{\infty}\in\oa{C}_{\mathrm{b}}(\sphilbfrak{h})$.

\item There is $C>0$ such that
$$\begin{aligned}
\norm{\physpoissonbkt{a,b}_{\infty}}_{\mathrm{b}}
&\leq
C\norm{a}_{\mathrm{b}}\norm{b}_{\mathrm{b}}
\end{aligned}$$
for all $a,b\in\oa{C}_{\mathrm{b}}(\sphilbfrak{h})$.
\end{enumerate}
These two statements are equivalent.
\end{prop}

\begin{proof}
Fix $\ell\in\monnat$.
Taking every input, output, and cutoff degree in
\eqref{eq:paper-abstract-particle-coefficient-overlap-products} to be $\ell$ gives
$$\begin{aligned}
\fun{E_{\ell,\ell}}{a_\ell}\fun{E_{\ell,\ell}}{b_\ell}
&=
\fun{E_{\ell,\ell}}{a_\ell b_\ell}.
\end{aligned}$$
The map $E_{\ell,\ell}$ is injective and preserves adjoints.
It is an injective $\ast$-homomorphism from the coefficient algebra onto $\Ker r_\ell$, and therefore it is
isometric.
For $\ell=0$, the coordinate evaluation satisfies
$$\begin{aligned}
a_0
&=
\fun{E_{0,0}^{-1}\circ\Pi_0}{a},
\end{aligned}$$
which proves its boundedness.
Assume inductively that $a\mapsto a_k$ is bounded for every $0\leq k<\ell$.
The degree-$\ell$ coefficient is recovered by
$$\begin{aligned}
a_\ell
&=
\fun{E_{\ell,\ell}^{-1}}{
\fun{\Pi_\ell}{a}
-
\sum_{m=0}^{\ell-1}\fun{E_{m,\ell}}{a_m}
}.
\end{aligned}$$
The estimate $\norm{\fun{\Pi_\ell}{a}}\leq\norm{a}_{\mathrm{b}}$ and finiteness of the sum prove that
$a\mapsto a_\ell$ is bounded.
Induction proves the assertion for every $\ell$.

The second condition implies the first.
Assume the first condition.
Suppose that $a_j\to a$, $b_j\to b$, and
$\physpoissonbkt{a_j,b_j}_{\infty}\to c$ in
$\norm{\mathord{\cdot}}_{\mathrm{b}}$ as $j\to\infty$.
For each $\ell\in\monnat$, continuity of the coordinate evaluation and the finite sum in
\eqref{eq:paper-completed-poisson-bracket} give
$$\begin{aligned}
c_\ell
&=
\lim_{j\to\infty}\rbk{\physpoissonbkt{a_j,b_j}_{\infty}}_\ell
=
\rbk{\physpoissonbkt{a,b}_{\infty}}_\ell.
\end{aligned}$$
Equality of every coefficient gives $c=\physpoissonbkt{a,b}_{\infty}$.
Thus the bilinear bracket has a closed graph.

Fix $a\in\oa{C}_{\mathrm{b}}(\sphilbfrak{h})$.
If $b_j\to b$ and $\physpoissonbkt{a,b_j}_{\infty}\to c$, apply the closed-graph conclusion for the bilinear
bracket $\physpoissonbkt{\mathord{\cdot},\mathord{\cdot}}_{\infty}$ to the constant sequence $a_j=a$.
It gives $c=\physpoissonbkt{a,b}_{\infty}$.
Thus $b\mapsto\physpoissonbkt{a,b}_{\infty}$ has a closed graph.
The bounded particle-coefficient algebra is Banach by Theorem
\ref{thm:paper-particle-coefficient-buchholz-isomorphism}.
The closed graph theorem proves that this linear map is bounded.

Consider the linear map
$$\begin{aligned}
\oa{C}_{\mathrm{b}}(\sphilbfrak{h})
\to
\opspbddlin{\oa{C}_{\mathrm{b}}(\sphilbfrak{h}),\oa{C}_{\mathrm{b}}(\sphilbfrak{h})};
\quad
a
\mapsto
\rbk{b\mapsto\physpoissonbkt{a,b}_{\infty}}.
\end{aligned}$$
Since $\oa{C}_{\mathrm{b}}(\sphilbfrak{h})$ is Banach,
$\opspbddlin{\oa{C}_{\mathrm{b}}(\sphilbfrak{h}),\oa{C}_{\mathrm{b}}(\sphilbfrak{h})}$
is Banach in the operator norm.

Suppose that $a_j\to a$ and the operators
$b\mapsto\physpoissonbkt{a_j,b}_{\infty}$ converge to $T$ in operator norm.
For every $b\in\oa{C}_{\mathrm{b}}(\sphilbfrak{h})$,
$$\begin{aligned}
\lim_{j \to \infty}
\physpoissonbkt{a_j,b}_{\infty}
=
\fun{T}{b}.
\end{aligned}$$
Closedness of the graph of the bracket gives
$\fun{T}{b}=\physpoissonbkt{a,b}_{\infty}$.
The operator-valued map therefore has a closed graph and hence is bounded by the closed graph theorem.
There is consequently $C>0$ such that
$$\begin{aligned}
\norm{\physpoissonbkt{a,b}_{\infty}}_{\mathrm{b}}
&\leq
C\norm{a}_{\mathrm{b}}\norm{b}_{\mathrm{b}}.
\end{aligned}$$
\end{proof}

\begin{rem}[Formal semiclassical interpretation and analytic limitations]\label{rem:paper-rees-poisson-topological-scope}
The bracket $\physpoissonbkt{\mathord{\cdot},\mathord{\cdot}}_{\infty}$ defined in
\eqref{eq:paper-completed-poisson-bracket} by the first-overlap terms $p+q-r=1$ defines a Poisson structure
for $\star_0$ on $\oa{C}(\sphilbfrak{h})$ and is continuous in the product topology.
The parameter $z$ records the overlap order: $\star_0$ is the zero-overlap product, while the rescaled
commutator in \eqref{eq:paper-rees-commutator-poisson-limit} extracts the first-overlap term.
In this precise filtered-algebra sense, the particle-overlap filtration carries a formal algebraic semiclassical
Poisson structure.
The word ``semiclassical'' is used here only for this formal interpretation.
No definition of strict deformation quantization is adopted here.
The present construction does not introduce a continuous field of $\oacstar$-algebras, does not equip its fibers
with $z$-dependent $\oacstar$-norms, and does not verify the analytic axioms of such a theory.
In particular, product-topological continuity does not imply that the bracket preserves the bounded
particle-coefficient algebra
$\oa{C}_{\mathrm{b}}(\sphilbfrak{h})$, which realizes the abstract Buchholz algebra with its $\oacstar$-norm.
Proposition \ref{prop:paper-bounded-particle-coefficient-poisson-closed-graph} shows that closure of this subalgebra
under the bracket is equivalent to a bounded bilinear estimate.
After the concrete coefficient maps have been constructed, Theorem
\ref{thm:paper-bounded-particle-coefficient-poisson-nonclosure} gives an explicit counterexample and proves that
$\oa{C}_{\mathrm{b}}(\sphilbfrak{h})$ is not closed under the bracket.
\end{rem}

\subsubsection{Concrete counterexample to bounded Poisson closure}\label{concrete-counterexample-to-bounded-poisson-closure}

In this subsubsection, a concrete coefficient sequence shows that the bounded particle-coefficient algebra is not closed under the first-overlap Poisson bracket.

The counterexample uses the maps \(\Theta_m\) and \(\varpi_n\) in \eqref{eq:paper-main-fock-cutoff-identification} only through the concrete realization in Theorem \ref{thm:paper-concrete-particle-coefficient-reconstruction-data}. None of the representation-independent definitions or proofs depends on this counterexample.

\begin{thm}[Failure of Poisson closure for bounded particle-degree coefficients]\label{thm:paper-bounded-particle-coefficient-poisson-nonclosure}
For every separable infinite-dimensional complex Hilbert space $\sphilbfrak{h}$, there are
$a,b\in\oa{C}_{\mathrm{b}}(\sphilbfrak{h})$ such that
$$\begin{aligned}
\physpoissonbkt{a,b}_{\infty}
&\notin
\oa{C}_{\mathrm{b}}(\sphilbfrak{h}).
\end{aligned}$$
\end{thm}

\begin{proof}
Choose orthonormal vectors $e_1,e_2\in\sphilbfrak{h}$ and define
$\sphilbfrak{f}=\splinspan\setone{e_1,e_2}$.
Define self-adjoint unitaries by
$$\begin{aligned}
v,w
\colon
\sphilbfrak{h}
\to
\sphilbfrak{h};\quad
\fnrestr{v}{\sphilbfrak{f}}
&=
\begin{pmatrix}
0&1\\
1&0
\end{pmatrix},
&
\fnrestr{w}{\sphilbfrak{f}}
&=
\begin{pmatrix}
1&0\\
0&-1
\end{pmatrix},
&
\fnrestr{v}{\sphilbfrak{f}^{\perp}}
&=
\fnrestr{w}{\sphilbfrak{f}^{\perp}}
=
1.
\end{aligned}$$
The operators $K=v-1$ and $L=w-1$ have finite rank.
For $m\in\monnat$, define
$$\begin{aligned}
a_m
=
\fun{\Theta_m^{-1}}{K^{\otimes_{\txtsym}m}},
\quad
b_m
=
\fun{\Theta_m^{-1}}{L^{\otimes_{\txtsym}m}}.
\end{aligned}$$
Assemble these coefficients into the sequences
$a = \seq{a_m}{m\in\monnat}$ and $b = \seq{b_m}{m\in\monnat}$,
where $K^{\otimes_{\txtsym}0}=L^{\otimes_{\txtsym}0}=1$.
For each subset $A\subset\setone{1,\ldots,n}$, let $K^{[A]}$ denote $K$ on the tensor factors indexed by
$A$ and the identity on the remaining factors.
The map $\varpi_n$ in
\eqref{eq:paper-main-fock-cutoff-identification} and the subset expansion give
$$\begin{aligned}
\fun{\varpi_n\circ\Pi_n}{a}
&=
\sum_{A\subset\setone{1,\ldots,n}}K^{[A]}
=
\rbk{1+K}^{\otimes_{\txtsym}n}
=
v^{\otimes_{\txtsym}n},
\\
\fun{\varpi_n\circ\Pi_n}{b}
&=
\sum_{B\subset\setone{1,\ldots,n}}L^{[B]}
=
\rbk{1+L}^{\otimes_{\txtsym}n}
=
w^{\otimes_{\txtsym}n}.
\end{aligned}$$
Thus $a,b\in\oa{C}_{\mathrm{b}}(\sphilbfrak{h})$ and
$\norm{a}_{\mathrm{b}}=\norm{b}_{\mathrm{b}}=1$.

Define
$$\begin{aligned}
T
&=
v+w-1,
&
C
&=
\frac{\commutator{v}{w}}{\imunit}.
\end{aligned}$$
The terms in \eqref{eq:paper-completed-poisson-bracket} indexed by subsets
$A,B\subset\setone{1,\ldots,n}$ with $\abscard{A\cap B}=1$ are evaluated on the $n$-particle sector.
Every such pair has a unique representation
$$\begin{aligned}
A
&=
\setone{\ell}\cup I,
&
B
&=
\setone{\ell}\cup J,
\end{aligned}$$
where $\ell\in\setone{1,\ldots,n}$ and
$I,J\subset\setone{1,\ldots,n}\setminus\setone{\ell}$ are disjoint.
The complete subset sum is
$$\begin{aligned}
&\sum_{\substack{A,B\subset\setone{1,\ldots,n}\\\abscard{A\cap B}=1}}
C^{[A\cap B]}K^{[A\setminus B]}L^{[B\setminus A]}
=
\sum_{\ell=1}^{n}
\sum_{\substack{I,J\subset\setone{1,\ldots,n}\setminus\setone{\ell}\\I\cap J=\emptyset}}
C^{[\setone{\ell}]}K^{[I]}L^{[J]}
\\ 
&=
\sum_{\ell=1}^{n}
T^{\otimes\rbk{\ell-1}}
\otimes C\otimes
T^{\otimes\rbk{n-\ell}}.
\end{aligned}$$
For every particle label other than $\ell$, the three choices contribute respectively $1$, $K$, and $L$;
their sum is $1+K+L=T$.
The map $\varpi_n$ in \eqref{eq:paper-main-fock-cutoff-identification} identifies the reconstructed bracket as
\begin{equation}
\begin{aligned}
\fun{\varpi_n\circ\Pi_n}{\physpoissonbkt{a,b}_{\infty}}
&=
\fnrestr{
\sum_{\ell=1}^{n}
T^{\otimes\rbk{\ell-1}}
\otimes C\otimes
T^{\otimes\rbk{n-\ell}}
}{\sphilbfrak{h}^{\otimes_{\txtsym}n}}.
\end{aligned}
\label{eq:paper-pauli-poisson-coordinate-growth}
\end{equation}
On $\sphilbfrak{f}$, direct matrix multiplication gives
$$\begin{aligned}
\fnrestr{T}{\sphilbfrak{f}}
&=
\begin{pmatrix}
0&1\\
1&-2
\end{pmatrix},
&
\fnrestr{C}{\sphilbfrak{f}}
&=
\begin{pmatrix}
0&2\imunit\\
-2\imunit&0
\end{pmatrix}.
\end{aligned}$$
Let $e\in \sphilbfrak{f}$ be a real unit eigenvector of $T$ for $t=-1-\sqrt{2}$.
Direct multiplication of the matrices defining $T$ and $C$ gives
$$\begin{aligned}
Te
&=
te,
&
\bkt{e}{Ce}_{\sphilbfrak{f}}
&=
0,
&
\norm{Ce}_{\sphilbfrak{f}}
&=
2.
\end{aligned}$$
Applying \eqref{eq:paper-pauli-poisson-coordinate-growth}, whose left side uses $\varpi_n$ from
\eqref{eq:paper-main-fock-cutoff-identification}, to $e^{\otimes n}$ gives
$$\begin{aligned}
\fun{\varpi_n\circ\Pi_n}{\physpoissonbkt{a,b}_{\infty}}e^{\otimes n}
&=
t^{n-1}\sum_{\ell=1}^{n}
e^{\otimes\rbk{\ell-1}}\otimes Ce\otimes e^{\otimes\rbk{n-\ell}}.
\end{aligned}$$
The $n$ summands are pairwise orthogonal because $Ce\perp e$.
The isometry of $\varpi_n$ in \eqref{eq:paper-main-fock-cutoff-identification} gives
$$\begin{aligned}
&\norm{\fun{\Pi_n}{\physpoissonbkt{a,b}_{\infty}}}
=
\norm{\fun{\varpi_n\circ\Pi_n}{\physpoissonbkt{a,b}_{\infty}}}
\\
&\geq
\norm{\fun{\varpi_n\circ\Pi_n}{\physpoissonbkt{a,b}_{\infty}}e^{\otimes n}}
=
\abs{t}^{n-1}\sqrt{\sum_{\ell=1}^{n}\norm{Ce}^{2}}
=
2\sqrt{n}\rbk{1+\sqrt{2}}^{n-1}.
\end{aligned}$$
The right-hand side diverges as $n\to\infty$, which proves
$\physpoissonbkt{a,b}_{\infty}\notin\oa{C}_{\mathrm{b}}(\sphilbfrak{h})$.
\end{proof}

\subsection{Particle-cutoff extensions}\label{sec:paper-particle-cutoff-ideal-extension-appendix}

In this appendix, Proposition \ref{prop:busby-reconstruction} reconstructs the essential particle-cutoff extensions in Theorem \ref{thm:paper-finite-stage-ideal-chain} as the pullbacks \eqref{eq:busby-pullback} determined by the Busby maps \eqref{eq:busby-map}.

For an arbitrary essential extension \eqref{eq:paper-main-essential-extension} with \(\oa{J}\neq0\), the multiplier algebra in Definition \ref{def:multiplier-algebra} receives the double-centralizer action \eqref{eq:busby-double-centralizer-action}, and the corona quotient \eqref{eq:busby-corona-quotient-map} converts this action into the Busby map \eqref{eq:busby-map}. In the concrete Fock realization, Theorem \ref{thm:paper-concrete-particle-coefficient-reconstruction-data} identifies the kernel \(\oa{J}_n^{(n)}\) with the compact-operator algebra \(\fun{\schattencls}{\sphilbfrak{h}^{\otimes_{\txtsym}n}}\). Each particle-cutoff extension is therefore an extension by compact operators in that realization. The compact-operator identification is not an assumption of the abstract reconstruction, which requires only the essential extension \eqref{eq:paper-main-essential-extension} and uses no graph or lattice data.

\begin{defn}[Multiplier algebra]\label{def:multiplier-algebra}
Let $\oa{J}$ be a nonzero $\oacstar$-algebra.
A double centralizer of $\oa{J}$ is a pair $\rbk{L,R}$ of bounded linear maps from $\oa{J}$ to $\oa{J}$ such
that
$$\begin{aligned}
\fun{L}{xy}
=
\fun{L}{x}y,
\quad
\fun{R}{xy}
=
x\fun{R}{y},
\quad
x\fun{L}{y}
=
\fun{R}{x}y
\end{aligned}$$
for all $x,y\in\oa{J}$.
The multiplier algebra $\oa{M}(\oa{J})$ is the set of these double centralizers.
Its product is
$$\begin{aligned}
\rbk{L_1,R_1}\rbk{L_2,R_2}
&=
\rbk{L_1\circ L_2,R_2\circ R_1}.
\end{aligned}$$
The adjoint of $\rbk{L,R}$ is $\rbk{L^{\sharp},R^{\sharp}}$, where
$$\begin{aligned}
L^{\sharp}
\colon
\oa{J}
\to
\oa{J};\quad
\fun{L^{\sharp}}{x}
&=
\faadjrbk{\fun{R}{\faadj{x}}}.
\end{aligned}$$
The adjoint right multiplier is
$$\begin{aligned}
R^{\sharp}
\colon
\oa{J}
\to
\oa{J};\quad
\fun{R^{\sharp}}{x}
&=
\faadjrbk{\fun{L}{\faadj{x}}}.
\end{aligned}$$
Addition and scalar multiplication are componentwise, and the norm is
$$\begin{aligned}
\norm{\rbk{L,R}}_{\oa{M}(\oa{J})}
&=
\max\setone{\norm{L},\norm{R}}.
\end{aligned}$$
For $j\in\oa{J}$, define
$$\begin{aligned}
L_j
\colon
\oa{J}
\to
\oa{J};\quad
\fun{L_j}{x}
&=
jx.
\end{aligned}$$
The right multiplication map is
$$\begin{aligned}
R_j
\colon
\oa{J}
\to
\oa{J};\quad
\fun{R_j}{x}
&=
xj.
\end{aligned}$$
The canonical embedding is
\begin{equation}
\begin{aligned}
\lambda_{\oa{J}}
\colon
\oa{J}
\to
\oa{M}(\oa{J});\quad
\fun{\lambda_{\oa{J}}}{j}
&=
\rbk{L_j,R_j}.
\end{aligned}
\label{eq:multiplier-canonical-embedding}
\end{equation}
\end{defn}

\begin{lem}[Multiplier algebra of the block direct sum]\label{lem:paper-intrinsic-block-multiplier-algebra}
Fix a nonzero finite-dimensional complex subspace $\sphilbfrak{k}\subset\sphilbfrak{h}$.
For the block algebra $\oa{K}(\sphilbfrak{k})^{\gamma}$ in
\eqref{eq:paper-main-finite-mode-block-algebra}, the double-centralizer multiplier algebra of Definition
\ref{def:multiplier-algebra} is canonically isometrically $\ast$-isomorphic to the bounded block product in
\eqref{eq:paper-intrinsic-block-multiplier-algebra}.
Under this identification, multiplication and adjoints are coordinatewise.
\end{lem}

\begin{proof}
Let $A=\seq{A_n}{n\in\monnat}$ be a bounded block family and let
$B=\seq{B_n}{n\in\monnat}\in\oa{K}(\sphilbfrak{k})^{\gamma}$.
The block actions associated with $A$ are
$$\begin{aligned}
L_A
\colon
\oa{K}(\sphilbfrak{k})^{\gamma}
\to
\oa{K}(\sphilbfrak{k})^{\gamma};\quad
\fun{L_A}{B}
&=
\seq{A_nB_n}{n\in\monnat},
\\
R_A
\colon
\oa{K}(\sphilbfrak{k})^{\gamma}
\to
\oa{K}(\sphilbfrak{k})^{\gamma};\quad
\fun{R_A}{B}
&=
\seq{B_nA_n}{n\in\monnat}.
\end{aligned}$$
The component limits as $n\to\infty$ and the operator-norm bounds are
$$\begin{aligned}
\norm{A_nB_n}
&\leq
\rbk{\sup_{k\in\monnat}\norm{A_k}}\norm{B_n}
\to
0,
\\
\norm{B_nA_n}
&\leq
\rbk{\sup_{k\in\monnat}\norm{A_k}}\norm{B_n}
\to
0,
\\
\norm{\fun{L_A}{B}}
&\leq
\rbk{\sup_{n\in\monnat}\norm{A_n}}\norm{B},
\\
\norm{\fun{R_A}{B}}
&\leq
\rbk{\sup_{n\in\monnat}\norm{A_n}}\norm{B}.
\end{aligned}$$
These estimates show that both maps are bounded and preserve the norm-null condition in
\eqref{eq:paper-main-finite-mode-block-algebra}.
Since $\fun{L_A}{e_{\sphilbfrak{k},n}}=\fun{R_A}{e_{\sphilbfrak{k},n}}=A_n$ and
$\norm{e_{\sphilbfrak{k},n}}=1$, evaluation on the corner units gives
$$\begin{aligned}
\norm{L_A}
&=
\norm{R_A}
=
\sup_{n\in\monnat}\norm{A_n}.
\end{aligned}$$
For $B,C\in\oa{K}(\sphilbfrak{k})^{\gamma}$, coordinatewise multiplication gives
$$\begin{aligned}
\rbk{\fun{L_A}{BC}}_n
&=
A_nB_nC_n
=
\rbk{\fun{L_A}{B}C}_n,
\\
\rbk{\fun{R_A}{BC}}_n
&=
B_nC_nA_n
=
\rbk{B\fun{R_A}{C}}_n,
\\
\rbk{B\fun{L_A}{C}}_n
&=
B_nA_nC_n
=
\rbk{\fun{R_A}{B}C}_n.
\end{aligned}$$
These identities show that $\rbk{L_A,R_A}$ is a double centralizer in the sense of Definition
\ref{def:multiplier-algebra}.

Conversely, let $\rbk{L,R}\in\oa{M}(\oa{K}(\sphilbfrak{k})^{\gamma})$ be a double centralizer.
For $B=\seq{B_n}{n\in\monnat}\in\oa{K}(\sphilbfrak{k})^{\gamma}$, the block decomposition
\eqref{eq:paper-main-finite-mode-block-algebra} gives
$e_{\sphilbfrak{k},n}B=B_n=Be_{\sphilbfrak{k},n}$.
This identity makes every $e_{\sphilbfrak{k},n}$ central in $\oa{K}(\sphilbfrak{k})^{\gamma}$.
The double-centralizer identities give
$$\begin{aligned}
\fun{L}{e_{\sphilbfrak{k},n}}
=
\fun{L}{e_{\sphilbfrak{k},n}}e_{\sphilbfrak{k},n}
=
e_{\sphilbfrak{k},n}\fun{L}{e_{\sphilbfrak{k},n}}
=
\fun{R}{e_{\sphilbfrak{k},n}}e_{\sphilbfrak{k},n}
=
\fun{R}{e_{\sphilbfrak{k},n}}
\in
e_{\sphilbfrak{k},n}\oa{K}(\sphilbfrak{k})^{\gamma}e_{\sphilbfrak{k},n}.
\end{aligned}$$
Set $A_n=\fun{L}{e_{\sphilbfrak{k},n}}=\fun{R}{e_{\sphilbfrak{k},n}}$.
Since $\norm{e_{\sphilbfrak{k},n}}=1$, this family satisfies
$$\begin{aligned}
\sup_{n\in\monnat}\norm{A_n}
&\leq
\norm{L}
<
\infty.
\end{aligned}$$
For $B_n\in e_{\sphilbfrak{k},n}\oa{K}(\sphilbfrak{k})^{\gamma}e_{\sphilbfrak{k},n}$, the centralizer identities yield
$$\begin{aligned}
\fun{L}{B_n}
=
\fun{L}{e_{\sphilbfrak{k},n}B_n}
=
A_nB_n,
\quad
\fun{R}{B_n}
=
\fun{R}{B_ne_{\sphilbfrak{k},n}}
=
B_nA_n.
\end{aligned}$$
Continuity and the norm-convergent decomposition in
\eqref{eq:paper-main-finite-mode-block-algebra} extend these identities to every
$B\in\oa{K}(\sphilbfrak{k})^{\gamma}$.
The resulting equalities $L=L_A$ and $R=R_A$ determine the bounded family $A$ uniquely.
The product, adjoint, and norm in Definition \ref{def:multiplier-algebra} now become the coordinatewise
operations and the supremum norm in \eqref{eq:paper-intrinsic-block-multiplier-algebra}.
\end{proof}

\begin{prop}[Universal fixed-point algebra as multiplier operators]
\label{prop:paper-universal-fixed-point-multiplier-action}
Fix a nonzero finite-dimensional complex subspace
$\sphilbfrak{f}\subset\sphilbfrak{h}$.
The fixed-point algebra $\oa{K}(\sphilbfrak{f})^{\gamma}$ is an essential ideal of
$\oaresolventalgebra(\sphilbfrak{f})^{\gamma}$.
For $A\in\oaresolventalgebra(\sphilbfrak{f})^{\gamma}$ and
$K\in\oa{K}(\sphilbfrak{f})^{\gamma}$, set
\begin{equation}
\begin{aligned}
\fun{L_A}{K}
&=
AK,
&
\fun{R_A}{K}
&=
KA,
\\
\lambda_{\sphilbfrak{f}}
\colon
\oaresolventalgebra(\sphilbfrak{f})^{\gamma}
&\to
\oa{M}(\oa{K}(\sphilbfrak{f})^{\gamma});
&
\fun{\lambda_{\sphilbfrak{f}}}{A}
&=
\rbk{L_A,R_A}.
\end{aligned}
\label{eq:paper-universal-fixed-point-multiplier-action}
\end{equation}
Then the map $\lambda_{\sphilbfrak{f}}$ in
\eqref{eq:paper-universal-fixed-point-multiplier-action} is an injective unital $\ast$-homomorphism.
Under the bounded-block-product identification in Lemma
\ref{lem:paper-intrinsic-block-multiplier-algebra}, one has
\begin{equation}
\begin{aligned}
\fun{\lambda_{\sphilbfrak{f}}}{A}
&=
\seq{e_{\sphilbfrak{f},m}Ae_{\sphilbfrak{f},m}}{m\in\monnat}.
\end{aligned}
\label{eq:paper-universal-fixed-point-multiplier-blocks}
\end{equation}
Under the Fock isomorphism $\diag_{\sphilbfrak{f}}$ in
\eqref{eq:paper-main-finite-mode-product-isomorphism}, this family satisfies
\begin{equation}
\begin{aligned}
\fun{\diag_{\sphilbfrak{f}}}{\fun{\lambda_{\sphilbfrak{f}}}{A}}
&=
\bigoplus_{m\in\monnat}
\fnrestr{\fun{\oarepn_{\txtfock,\sphilbfrak{f}}}{A}}
{\sphilbfrak{f}^{\otimes_{\txtsym}m}}
=
\fun{\oarepn_{\txtfock,\sphilbfrak{f}}}{A}.
\end{aligned}
\label{eq:paper-universal-fixed-point-multiplier-fock-action}
\end{equation}
\end{prop}

\begin{proof}
Theorem \ref{thm:abstract-minimal-compact-ideal} shows that
$\oa{K}(\sphilbfrak{f})$ is an essential ideal of
$\oaresolventalgebra(\sphilbfrak{f})$.
Let $0\neq A\in\oaresolventalgebra(\sphilbfrak{f})^{\gamma}$.
If $A\oa{K}(\sphilbfrak{f})=0$, set $D=\faadj{A}A$ and let $\oa{L}_D$ be the closed two-sided ideal generated
by $D$ in $\oaresolventalgebra(\sphilbfrak{f})$.
Then $\oa{L}_D\neq0$ and
$\oa{L}_D\oa{K}(\sphilbfrak{f})=0$.
Essentiality gives a nonzero
$Y\in\oa{L}_D\cap\oa{K}(\sphilbfrak{f})$.
Since $Y\in\oa{L}_D$ and $\faadj{Y}\in\oa{K}(\sphilbfrak{f})$, one has $Y\faadj{Y}=0$.
The $\oacstar$-identity $\norm{Y}^{2}=\norm{Y\faadj{Y}}=0$ contradicts $Y\neq0$.
This contradiction shows that there is $X\in\oa{K}(\sphilbfrak{f})$ such that $AX\neq0$.
The uniqueness of the minimal ideal in Theorem
\ref{thm:abstract-minimal-compact-ideal} makes
$\oa{K}(\sphilbfrak{f})$ gauge invariant.
In a faithful regular irreducible representation,
\eqref{eq:paper-minimal-ideal-representation} realizes this ideal as the compact operators, and the gauge
implementation used in Proposition \ref{prop:paper-particle-block-independence} is strongly continuous.
Conjugation by a strongly continuous unitary group is pointwise norm continuous on the compact operators.
Faithfulness therefore shows that
$\theta\mapsto\fun{\gamma_{\theta}}{B}$ is norm continuous for every
$B\in\oa{K}(\sphilbfrak{f})$.
Set $B=X\faadj{X}$ and define its norm-convergent gauge average by
$$\begin{aligned}
B_{\gamma}
&=
\frac{1}{2\pi}\int_{0}^{2\pi}\fun{\gamma_{\theta}}{B}\opdmsr{\theta}.
\end{aligned}$$
Consequently,
$B_{\gamma}\in\oa{K}(\sphilbfrak{f})^{\gamma}$, and gauge invariance of $A$ gives
$$\begin{aligned}
AB_{\gamma}\faadj{A}
&=
\frac{1}{2\pi}\int_{0}^{2\pi}
\fun{\gamma_{\theta}}{AB\faadj{A}}\opdmsr{\theta}.
\end{aligned}$$
The positive element $AB\faadj{A}=\rbk{AX}\faadjrbk{AX}$ is nonzero.
Choose a state $\oastate[\psi]$ with $\fun{\oastate[\psi]}{AB\faadj{A}}>0$.
Continuity of
$\theta\mapsto\fun{\oastate[\psi]\circ\gamma_{\theta}}{AB\faadj{A}}$ and its nonnegativity give
$$\begin{aligned}
\fun{\oastate[\psi]}{AB_{\gamma}\faadj{A}}
&=
\frac{1}{2\pi}\int_{0}^{2\pi}
\fun{\oastate[\psi]\circ\gamma_{\theta}}{AB\faadj{A}}\opdmsr{\theta}
>
0.
\end{aligned}$$
Thus $AB_{\gamma}\neq0$.
If $\oa{L}$ is a nonzero closed two-sided ideal of
$\oaresolventalgebra(\sphilbfrak{f})^{\gamma}$, choose $0\neq A\in\oa{L}$ and apply the preceding construction.
Then
$$\begin{aligned}
0
\neq
AB_{\gamma}
\in
\oa{L}\cap\oa{K}(\sphilbfrak{f})^{\gamma}.
\end{aligned}$$
Definition \ref{def:paper-essential-ideal-extension} therefore shows that
$\oa{K}(\sphilbfrak{f})^{\gamma}$ is essential in
$\oaresolventalgebra(\sphilbfrak{f})^{\gamma}$.

Because $\oa{K}(\sphilbfrak{f})^{\gamma}$ is an ideal, the two multiplication maps in
\eqref{eq:paper-universal-fixed-point-multiplier-action} take values in that ideal.
For $K_1,K_2\in\oa{K}(\sphilbfrak{f})^{\gamma}$, associativity gives
$$\begin{aligned}
\fun{L_A}{K_1K_2}
&=
AK_1K_2
=
\fun{L_A}{K_1}K_2,
\\
\fun{R_A}{K_1K_2}
&=
K_1K_2A
=
K_1\fun{R_A}{K_2},
\\
K_1\fun{L_A}{K_2}
&=
K_1AK_2
=
\fun{R_A}{K_1}K_2.
\end{aligned}$$
Thus $\rbk{L_A,R_A}$ is a double centralizer in the sense of Definition
\ref{def:multiplier-algebra}.
For $A,C\in\oaresolventalgebra(\sphilbfrak{f})^{\gamma}$ and
$K\in\oa{K}(\sphilbfrak{f})^{\gamma}$, one has
$$\begin{aligned}
\fun{L_A\circ L_C}{K}
&=
A(CK)
=
\fun{L_{AC}}{K},
\\
\fun{R_C\circ R_A}{K}
&=
(KA)C
=
\fun{R_{AC}}{K},
\\
\fun{L_A^{\sharp}}{K}
&=
\faadjrbk{\fun{R_A}{\faadj{K}}}
=
\faadj{A}K
=
\fun{L_{\faadj{A}}}{K},
\\
\fun{R_A^{\sharp}}{K}
&=
\faadjrbk{\fun{L_A}{\faadj{K}}}
=
K\faadj{A}
=
\fun{R_{\faadj{A}}}{K}.
\end{aligned}$$
Together with $L_1=R_1=\id$, these formulas show that the map
$\lambda_{\sphilbfrak{f}}$ defined in
\eqref{eq:paper-universal-fixed-point-multiplier-action} is a unital $\ast$-homomorphism.
If $\fun{\lambda_{\sphilbfrak{f}}}{A}=0$, then
$A\oa{K}(\sphilbfrak{f})^{\gamma}=0$.
The essentiality just proved forces $A=0$, so the map
$\lambda_{\sphilbfrak{f}}$ in \eqref{eq:paper-universal-fixed-point-multiplier-action} is injective.

The projection $e_{\sphilbfrak{f},m}$ is central in
$\oa{K}(\sphilbfrak{f})^{\gamma}$ by the block decomposition
\eqref{eq:paper-main-intrinsic-particle-block-decomposition}.
Indeed, if $K=\sum_{n\in\monnat}K_n$ is that decomposition, then
$$\begin{aligned}
e_{\sphilbfrak{f},m}K
&=
K_m
=
Ke_{\sphilbfrak{f},m}.
\end{aligned}$$
Both $Ae_{\sphilbfrak{f},m}$ and $e_{\sphilbfrak{f},m}A$ belong to this ideal, whence
$$\begin{aligned}
Ae_{\sphilbfrak{f},m}
&=
e_{\sphilbfrak{f},m}Ae_{\sphilbfrak{f},m}
=
e_{\sphilbfrak{f},m}A.
\end{aligned}$$
The construction in Lemma \ref{lem:paper-intrinsic-block-multiplier-algebra} now gives
\eqref{eq:paper-universal-fixed-point-multiplier-blocks}.
Finally,
Theorem \ref{thm:paper-fock-intrinsic-block-realization} sends
$e_{\sphilbfrak{f},m}$ to the orthogonal projection onto
$\sphilbfrak{f}^{\otimes_{\txtsym}m}$.
Since $A$ is gauge invariant, its Fock representative preserves every particle sector.
Applying the map $\diag_{\sphilbfrak{f}}$ defined in
\eqref{eq:paper-main-finite-mode-product-isomorphism} to
\eqref{eq:paper-universal-fixed-point-multiplier-blocks} proves
\eqref{eq:paper-universal-fixed-point-multiplier-fock-action}.
\end{proof}

\begin{rem}[Which multiplier algebra is involved]
The fixed-point resolvent algebra
$\oaresolventalgebra(\sphilbfrak{f})^{\gamma}$ is unital, so its own multiplier algebra is canonically equal
to itself.
Proposition \ref{prop:paper-universal-fixed-point-multiplier-action} instead concerns the multiplier algebra of
its essential ideal $\oa{K}(\sphilbfrak{f})^{\gamma}$.
The map $\lambda_{\sphilbfrak{f}}$ in
\eqref{eq:paper-universal-fixed-point-multiplier-action} identifies the universal fixed-point algebra with a
unital $\oacstar$-subalgebra of the full bounded block product
\eqref{eq:paper-intrinsic-block-multiplier-algebra}; it does not assert that every bounded block family comes
from an element of the universal fixed-point algebra.
\end{rem}

\begin{lem}[Canonical essential ideal of the multiplier algebra]\label{lem:multiplier-canonical-essential-ideal}
The map $\lambda_{\oa{J}}$ in \eqref{eq:multiplier-canonical-embedding} is an isometric
$\ast$-homomorphism, and $\fun{\lambda_{\oa{J}}}{\oa{J}}$ is a closed essential ideal of
$\oa{M}(\oa{J})$.
\end{lem}

\begin{proof}
For $j,k,x\in\oa{J}$, associativity gives
$$\begin{aligned}
\fun{L_j\circ L_k}{x}
&=
j\rbk{kx}
=
\rbk{jk}x
=
\fun{L_{jk}}{x},
\\
\fun{R_k\circ R_j}{x}
&=
\rbk{xj}k
=
x\rbk{jk}
=
\fun{R_{jk}}{x}.
\end{aligned}$$
The multiplier embedding preserves products:
$$\begin{aligned}
\fun{\lambda_{\oa{J}}}{j}\fun{\lambda_{\oa{J}}}{k}
&=
\rbk{L_j\circ L_k,R_k\circ R_j}
=
\rbk{L_{jk},R_{jk}}
=
\fun{\lambda_{\oa{J}}}{jk},
\\
\faadjrbk{\fun{\lambda_{\oa{J}}}{j}}
&=
\fun{\lambda_{\oa{J}}}{\faadj{j}}.
\end{aligned}$$
If $\fun{\lambda_{\oa{J}}}{j}=0$, then $L_j=0$ and
$$\begin{aligned}
0
&=
\fun{L_j}{\faadj{j}}
=
j\faadj{j},
&
\norm{j}^{2}
&=
\norm{j\faadj{j}}
=
0.
\end{aligned}$$
Thus $\lambda_{\oa{J}}$ is injective.
For $j\neq0$, the element $\faadj{j}/\norm{j}$ has norm one, and therefore
$$\begin{aligned}
\norm{j}
&=
\frac{\norm{j\faadj{j}}}{\norm{j}}
\leq
\norm{L_j},
&
\norm{L_j}
&\leq
\norm{j}.
\end{aligned}$$
Hence $\norm{\fun{\lambda_{\oa{J}}}{j}}=\norm{j}$, and the range of $\lambda_{\oa{J}}$ is closed.

Let $T=\rbk{L,R}\in\oa{M}(\oa{J})$ and $j,x\in\oa{J}$.
The double-centralizer identities give
$$\begin{aligned}
T\fun{\lambda_{\oa{J}}}{j}
&=
\rbk{L\circ L_j,R_j\circ R}
=
\fun{\lambda_{\oa{J}}}{\fun{L}{j}},
\\
\fun{\lambda_{\oa{J}}}{j}T
&=
\rbk{L_j\circ L,R\circ R_j}
=
\fun{\lambda_{\oa{J}}}{\fun{R}{j}}.
\end{aligned}$$
Therefore $\fun{\lambda_{\oa{J}}}{\oa{J}}$ is a two-sided ideal.
If a nonzero closed two-sided ideal
$\oa{L}\subset\oa{M}(\oa{J})$ contains a nonzero $T=\rbk{L,R}$, then $L\neq0$ or $R\neq0$.
Choose
$j\in\oa{J}$ with $\fun{L}{j}\neq0$ or $\fun{R}{j}\neq0$, respectively.
One of
$$\begin{aligned}
T\fun{\lambda_{\oa{J}}}{j}
&=
\fun{\lambda_{\oa{J}}}{\fun{L}{j}},
&
\fun{\lambda_{\oa{J}}}{j}T
&=
\fun{\lambda_{\oa{J}}}{\fun{R}{j}}
\end{aligned}$$
is nonzero and belongs to
$\oa{L}\cap\fun{\lambda_{\oa{J}}}{\oa{J}}$.
Thus
$\fun{\lambda_{\oa{J}}}{\oa{J}}$ is essential.
\end{proof}

\begin{defn}[Corona algebra]\label{def:corona-algebra}
Let $\oa{J}$ be a nonzero $\oacstar$-algebra.
Identify $\oa{J}$ with the closed essential ideal
$\fun{\lambda_{\oa{J}}}{\oa{J}}\subset\oa{M}(\oa{J})$ from Lemma
\ref{lem:multiplier-canonical-essential-ideal}.
The corona algebra is
\begin{equation}
\begin{aligned}
\oa{Q}(\oa{J})
&=
\setquot{\oa{M}(\oa{J})}{\fun{\lambda_{\oa{J}}}{\oa{J}}}.
\end{aligned}
\label{eq:busby-corona-algebra}
\end{equation}
The corona quotient map is
\begin{equation}
\begin{aligned}
q_{\oa{J}}
\colon
\oa{M}(\oa{J})
\to
\oa{Q}(\oa{J});\quad
\fun{q_{\oa{J}}}{T}
&=
T+\fun{\lambda_{\oa{J}}}{\oa{J}}.
\end{aligned}
\label{eq:busby-corona-quotient-map}
\end{equation}
\end{defn}

\begin{lem}[Multiplier action of an extension]\label{lem:busby-extension-multiplier-action}
Let \eqref{eq:paper-main-essential-extension} be an extension by a nonzero $\oacstar$-algebra $\oa{J}$, and
identify $\oa{J}$ with $\fun{\iota}{\oa{J}}\subset\oa{E}$.
For $e\in\oa{E}$ and $j\in\oa{J}$, set $\fun{L_e}{j}=ej$ and $\fun{R_e}{j}=je$.
Then
\begin{equation}
\begin{aligned}
\sigma_{\oa{E}}
\colon
\oa{E}
\to
\oa{M}(\oa{J});\quad
\fun{\sigma_{\oa{E}}}{e}
&=
\rbk{L_e,R_e}
\end{aligned}
\label{eq:busby-double-centralizer-action}
\end{equation}
is a $\ast$-homomorphism, and
$\fun{\sigma_{\oa{E}}}{j}=\fun{\lambda_{\oa{J}}}{j}$ for every $j\in\oa{J}$.
\end{lem}

\begin{proof}
The ideal property gives $ej,je\in\oa{J}$ and
$\norm{\fun{L_e}{j}},\norm{\fun{R_e}{j}}\leq\norm{e}\norm{j}$.
For $j,k\in\oa{J}$, associativity yields
$$\begin{aligned}
\fun{L_e}{jk}
&=
\fun{L_e}{j}k,
&
\fun{R_e}{jk}
&=
j\fun{R_e}{k},
&
j\fun{L_e}{k}
&=
\fun{R_e}{j}k.
\end{aligned}$$
Thus $\rbk{L_e,R_e}\in\oa{M}(\oa{J})$ by Definition \ref{def:multiplier-algebra}.
For $e,f\in\oa{E}$, the multiplier operations give
$$\begin{aligned}
\fun{\sigma_{\oa{E}}}{ef}
&=
\rbk{L_e\circ L_f,R_f\circ R_e}
=
\fun{\sigma_{\oa{E}}}{e}\fun{\sigma_{\oa{E}}}{f},
\\
\fun{\sigma_{\oa{E}}}{\faadj{e}}
&=
\faadjrbk{\fun{\sigma_{\oa{E}}}{e}}.
\end{aligned}$$
The formulas defining $L_e$ and $R_e$ are linear in $e$, so $\sigma_{\oa{E}}$ is a
$\ast$-homomorphism.
For $j\in\oa{J}$, the two actions are its canonical left and right multipliers in
\eqref{eq:multiplier-canonical-embedding}, giving
$\fun{\sigma_{\oa{E}}}{j}=\fun{\lambda_{\oa{J}}}{j}$.
\end{proof}

\begin{defn}[Busby map]\label{def:busby-map}
Let \eqref{eq:paper-main-essential-extension} be an essential extension by a nonzero
$\oacstar$-algebra $\oa{J}$.
By Lemma \ref{lem:busby-extension-multiplier-action}, the composite
$q_{\oa{J}}\circ\sigma_{\oa{E}}$ vanishes on $\Ker p=\oa{J}$.
The Busby map is the induced $\ast$-homomorphism
\begin{equation}
\begin{aligned}
\tau_{\oa{E}}
\colon
\oa{A}
\to
\oa{Q}(\oa{J});\quad
\fun{\tau_{\oa{E}}}{\fun{p}{e}}
&=
\fun{q_{\oa{J}}\circ\sigma_{\oa{E}}}{e}
\quad\rbk{e\in\oa{E}}.
\end{aligned}
\label{eq:busby-map}
\end{equation}
\end{defn}

For these extension data, define the pullback \(\oacstar\)-algebra by \begin{equation}
\begin{aligned}
\oa{P}(\tau_{\oa{E}})
&=
\set{\rbk{T,a}\in\oa{M}(\oa{J})\oplus\oa{A}}
{\fun{q_{\oa{J}}}{T}=\fun{\tau_{\oa{E}}}{a}}.
\end{aligned}
\label{eq:busby-pullback}
\end{equation}

\begin{prop}[Reconstruction from the Busby map]\label{prop:busby-reconstruction}
Let the short exact sequence \eqref{eq:paper-main-essential-extension} be an essential extension of $\oa{A}$
by the nonzero $\oacstar$-algebra $\oa{J}$ in the sense of Definition
\ref{def:paper-essential-ideal-extension}, and let
$\tau_{\oa{E}} \colon \oa{A} \to \oa{Q}(\oa{J})$ be its Busby map defined by
\eqref{eq:busby-map}.
Then $\tau_{\oa{E}}$ is injective.
The map into the pullback defined in \eqref{eq:busby-pullback},
\begin{equation}
\begin{aligned}
\Xi_{\oa{E}}
\colon
\oa{E}
\to
\oa{P}(\tau_{\oa{E}});\quad
\fun{\Xi_{\oa{E}}}{e}
&=
\rbk{\fun{\sigma_{\oa{E}}}{e},\fun{p}{e}}
\end{aligned}
\label{eq:busby-reconstruction-isomorphism}
\end{equation}
is a $\ast$-isomorphism.
\end{prop}

\begin{proof}
The factorization in Definition \ref{def:busby-map} makes $\tau_{\oa{E}}$ a well-defined
$\ast$-homomorphism.
Define the two-sided annihilator of $\oa{J}$ in $\oa{E}$ by
$$\begin{aligned}
\oa{J}^{\perp}
&=
\set{e\in\oa{E}}{e\oa{J}=0=\oa{J}e}.
\end{aligned}$$
For $e\in\oa{J}^{\perp}$, $f\in\oa{E}$, and $j\in\oa{J}$, the ideal property of $\oa{J}$ gives
$$\begin{aligned}
\rbk{fe}j
&=
f\rbk{ej}
=
0,
&
j\rbk{fe}
&=
\rbk{jf}e
=
0,
\\
\rbk{ef}j
&=
e\rbk{fj}
=
0,
&
j\rbk{ef}
&=
\rbk{je}f
=
0.
\end{aligned}$$
Since $\faadj{j}\in\oa{J}$,
$$\begin{aligned}
\faadj{e}j
&=
\faadjrbk{\faadj{j}e}
=
0,
&
j\faadj{e}
&=
\faadjrbk{e\faadj{j}}
=
0.
\end{aligned}$$
Thus $\oa{J}^{\perp}$ is a two-sided self-adjoint ideal.
If a sequence $\seq{e_{\ell}}{\ell\in\semigrposint}$ in $\oa{J}^{\perp}$ converges in norm to $e\in\oa{E}$,
then for every $j\in\oa{J}$,
$$\begin{aligned}
0
&\leq
\norm{ej}
=
\norm{\rbk{e-e_{\ell}}j}
\leq
\norm{e-e_{\ell}}\norm{j},
\\
0
&\leq
\norm{je}
=
\norm{j\rbk{e-e_{\ell}}}
\leq
\norm{j}\norm{e-e_{\ell}},
\\
\lim_{\ell\to\infty}\norm{e-e_{\ell}}\norm{j}
&=
\lim_{\ell\to\infty}\norm{j}\norm{e-e_{\ell}}
=
0.
\end{aligned}$$
Hence $\oa{J}^{\perp}$ is closed.
If $x\in\oa{J}\cap\oa{J}^{\perp}$, then $\faadj{x}\in\oa{J}$ and
$$\begin{aligned}
x\faadj{x}
&=
0,
&
\norm{x}^2
&=
\norm{x\faadj{x}}
=
0.
\end{aligned}$$
The ideals have trivial intersection:
$$\begin{aligned}
\oa{J}\cap\oa{J}^{\perp}
&=
0.
\end{aligned}$$
Essentiality of $\oa{J}$ and the closed-ideal property of $\oa{J}^{\perp}$ imply
$$\begin{aligned}
\oa{J}^{\perp}
&=
0.
\end{aligned}$$
The double-centralizer action in \eqref{eq:busby-double-centralizer-action} gives
$$\begin{aligned}
\fun{\sigma_{\oa{E}}}{e}
=
0
\iff
e\oa{J}=0=\oa{J}e
\iff
e\in\oa{J}^{\perp}.
\end{aligned}$$
Thus $\Ker\sigma_{\oa{E}}=\oa{J}^{\perp}=0$, and $\sigma_{\oa{E}}$ is injective.

Let $a\in\oa{A}$ satisfy $\fun{\tau_{\oa{E}}}{a}=0$, and choose $e\in\oa{E}$ with
$\fun{p}{e}=a$.
The Busby-map formula \eqref{eq:busby-map} gives
$$\begin{aligned}
\fun{q_{\oa{J}} \circ \sigma_{\oa{E}}}{e}
&=
\fun{\tau_{\oa{E}}}{a}
=
0,
&
\fun{\sigma_{\oa{E}}}{e}
&\in
\oa{J}.
\end{aligned}$$
Thus there is $j\in\oa{J}$ such that
$$\begin{aligned}
\fun{\sigma_{\oa{E}}}{e}
&=
\fun{\sigma_{\oa{E}}}{j}.
\end{aligned}$$
Injectivity of $\sigma_{\oa{E}}$ gives $e=j$, and exactness gives
$$\begin{aligned}
a
&=
\fun{p}{e}
=
\fun{p}{j}
=
0.
\end{aligned}$$
Hence $\tau_{\oa{E}}$ is injective.

For $\rbk{T,a},\rbk{S,b}\in\oa{P}(\tau_{\oa{E}})$, multiplicativity and adjoint preservation give
$$\begin{aligned}
\fun{q_{\oa{J}}}{TS}
&=
\fun{q_{\oa{J}}}{T}\fun{q_{\oa{J}}}{S}
=
\fun{\tau_{\oa{E}}}{a}\fun{\tau_{\oa{E}}}{b}
=
\fun{\tau_{\oa{E}}}{ab},
\\
\fun{q_{\oa{J}}}{\faadj{T}}
&=
\faadjrbk{\fun{q_{\oa{J}}}{T}}
=
\faadjrbk{\fun{\tau_{\oa{E}}}{a}}
=
\fun{\tau_{\oa{E}}}{\faadj{a}}.
\end{aligned}$$
The continuous linear map
$$\begin{aligned}
F_{\oa{E}}
\colon
\oa{M}(\oa{J})\oplus\oa{A}
\to
\oa{Q}(\oa{J});\quad
\fun{F_{\oa{E}}}{T,a}
&=
\fun{q_{\oa{J}}}{T}-\fun{\tau_{\oa{E}}}{a}
\end{aligned}$$
satisfies
$$\begin{aligned}
\oa{P}(\tau_{\oa{E}})
&=
\Ker F_{\oa{E}}.
\end{aligned}$$
Therefore $\oa{P}(\tau_{\oa{E}})$ is a closed $\ast$-subalgebra of
$\oa{M}(\oa{J})\oplus\oa{A}$ and hence a $\oacstar$-algebra.

For $e\in\oa{E}$, the element $e$ itself is a lift of $\fun{p}{e}$, so
$$\begin{aligned}
\fun{q_{\oa{J}} \circ \sigma_{\oa{E}}}{e}
&=
\fun{\tau_{\oa{E}}}{\fun{p}{e}}.
\end{aligned}$$
Thus the map $\Xi_{\oa{E}}$ in \eqref{eq:busby-reconstruction-isomorphism} takes values in
$\oa{P}(\tau_{\oa{E}})$.
For $e,f\in\oa{E}$, coordinatewise operations for the map $\Xi_{\oa{E}}$ in
\eqref{eq:busby-reconstruction-isomorphism} give
$$\begin{aligned}
\fun{\Xi_{\oa{E}}}{e+f}
&=
\fun{\Xi_{\oa{E}}}{e}+\fun{\Xi_{\oa{E}}}{f},
\\
\fun{\Xi_{\oa{E}}}{ce}
&=
c\fun{\Xi_{\oa{E}}}{e}
&&\rbk{c\in\fldcmp},
\\
\fun{\Xi_{\oa{E}}}{ef}
&=
\rbk{\fun{\sigma_{\oa{E}}}{e}\fun{\sigma_{\oa{E}}}{f},
\fun{p}{e}\fun{p}{f}}
=
\fun{\Xi_{\oa{E}}}{e}\fun{\Xi_{\oa{E}}}{f},
\\
\fun{\Xi_{\oa{E}}}{\faadj{e}}
&=
\rbk{\faadjrbk{\fun{\sigma_{\oa{E}}}{e}},\faadjrbk{\fun{p}{e}}}
=
\faadjrbk{\fun{\Xi_{\oa{E}}}{e}}.
\end{aligned}$$

Let $\rbk{T,a}\in\oa{P}(\tau_{\oa{E}})$, and choose $e\in\oa{E}$ with $\fun{p}{e}=a$.
The pullback equality and the Busby-map formula \eqref{eq:busby-map} give
$$\begin{aligned}
\fun{q_{\oa{J}}}{T-\fun{\sigma_{\oa{E}}}{e}}
&=
\fun{q_{\oa{J}}}{T}-\fun{\tau_{\oa{E}}}{a}
=
0.
\end{aligned}$$
Using the identified copy of $\oa{J}$ in $\oa{M}(\oa{J})$, set
$$\begin{aligned}
j
&=
T-\fun{\sigma_{\oa{E}}}{e}
\in
\oa{J}.
\end{aligned}$$
If $e'\in\oa{E}$ is another lift of $a$, then $k=e'-e\in\oa{J}$ and the element obtained from $e'$ is
$$\begin{aligned}
j'
&=
T-\fun{\sigma_{\oa{E}}}{e'}
=
j-k,
&
e'+j'
&=
e+k+j-k
=
e+j.
\end{aligned}$$
Therefore $e+j$ depends only on $\rbk{T,a}$.
The two coordinates of $\fun{\Xi_{\oa{E}}}{e+j}$, with $\Xi_{\oa{E}}$ defined in
\eqref{eq:busby-reconstruction-isomorphism}, are
$$\begin{aligned}
\fun{\sigma_{\oa{E}}}{e+j}
&=
\fun{\sigma_{\oa{E}}}{e}
+j
=
T,
&
\fun{p}{e+j}
&=
a.
\end{aligned}$$
Hence the map $\Xi_{\oa{E}}$ in \eqref{eq:busby-reconstruction-isomorphism} is surjective.
It is already an injective $\ast$-homomorphism, and its inverse is given by
\begin{equation}
\begin{aligned}
\fun{\Xi_{\oa{E}}^{-1}}{T,a}
&=
e+j,
&
j
&=
T-\fun{\sigma_{\oa{E}}}{e},
&
\fun{p}{e}
&=
a.
\end{aligned}
\label{eq:busby-reconstruction-inverse}
\end{equation}
The inverse formula \eqref{eq:busby-reconstruction-inverse} therefore proves that the map $\Xi_{\oa{E}}$ in
\eqref{eq:busby-reconstruction-isomorphism} is a $\ast$-isomorphism.
\end{proof}

\subsection{Supplementary spatial and quasilocal results}\label{sec:paper-spatial-quasilocal-appendix}

In this appendix, the finite-mode embeddings in \eqref{eq:paper-two-model-global-finite-mode-element} form a directed family for an arbitrary separable infinite-dimensional one-particle space. We describe its norm-closed union through finite-particle coordinates and states, and measure the failure of uniform local approximation using explicit operators and vector states. The graph finite-region construction is recovered only in the final subsection.

Fix a separable infinite-dimensional one-particle space \(\sphilbfrak{h}\) and an orthonormal basis \(\seq{u_x}{x\in X}\) indexed by a countably infinite set \(X\). For finite \(\Lambda\subset X\), set \begin{equation}
\begin{aligned}
\sphilbfrak{h}_{\Lambda}
&=
\splinspan_{\fldcmp}\set{u_x}{x\in\Lambda},
&
\sphilbfrak{h}_{X\setminus\Lambda}
&=
\sphilbfrak{h}_{\Lambda}^{\perp},
&
P_{\Lambda}
&\colon
\sphilbfrak{h}
\to
\sphilbfrak{h};
\quad
\fun{P_{\Lambda}}{f}
=
\sum_{x\in\Lambda}\bkt{u_x}{f}_{\sphilbfrak{h}}u_x.
\end{aligned}
\label{eq:paper-basis-localization-data}
\end{equation} The finite-dimensional subspaces are directed by inclusion and have dense union. For \(f\in\sphilbfrak{h}\), its coordinate support is \[\begin{aligned}
\fun{\supp_{X}}{f}
&=
\set{x\in X}{\bkt{u_x}{f}_{\sphilbfrak{h}}\neq0}.
\end{aligned}\] Finite support and infinite support below refer to this set. For \(d\in\semigrposint\) and \(\sphilbfrak{h}=\fun{\lp^{2}}{\fldreal^{d}}\), take \(X=\ringratint^d\times\monnat\). For each \(a\in\ringratint^d\), set \(C_a=a+\rightopeninterval{0}{1}^{d}\) and choose an orthonormal basis \(\seq{u_{(a,m)}}{m\in\monnat}\) of \(\fun{\lp^{2}}{C_a}\) with \(u_{(a,0)}=\fndef{C_a}\). The union of these bases is an orthonormal basis of \(\fun{\lp^{2}}{\fldreal^{d}}\); each \(\sphilbfrak{h}_{\Lambda}\) is finite-dimensional and supported in finitely many unit cubes. The resulting local algebra depends on this basis and is a finite-mode localization; it is not identified with the algebra associated with the full infinite-dimensional space \(\fun{\lp^{2}}{\Omega}\) of a bounded region \(\Omega\).

\subsubsection{Directed systems of local subalgebras}\label{sec:paper-spatial-localization-system}

In this subsubsection, the finite-dimensional basis subspaces in \eqref{eq:paper-basis-localization-data} yield compatible tensor-factor inclusions with dense finite-particle coordinate ranges. The closure of their union is unchanged when finite index sets are taken along an increasing sequence whose union is \(X\).

The set \(\setpowfin{X}\) is directed because \(\Lambda\cup\Lambda'\) is an upper bound for \(\Lambda,\Lambda'\in\setpowfin{X}\). Set \(\oa{B}_{\emptyset}=\fldcmp 1_{\fun{\oa{B}}{\sphilbfrak{h}}}\); formulas involving \(\sphilbfrak{h}_{\Lambda}\) below concern nonempty \(\Lambda\).

Under the canonical particle-number decompositions, define the concrete tensor-factor inclusions by \begin{equation}
\begin{aligned}
i_{\Lambda,\Lambda'}
\colon
\oa{M}(\oa{K}(\sphilbfrak{h}_{\Lambda})^{\gamma})
\to
\oa{M}(\oa{K}(\sphilbfrak{h}_{\Lambda'})^{\gamma}),
\quad
i_{\Lambda,X}
\colon
\oa{M}(\oa{K}(\sphilbfrak{h}_{\Lambda})^{\gamma})
\to
\fun{\oa{B}}{\sphilbfrak{h}}.
\end{aligned}
\label{eq:paper-global-local-inclusion-map}
\end{equation} Use the block representation \(\pi_{\sphilbfrak{h}_{\Lambda},m}\) defined in \eqref{eq:paper-main-finite-mode-block-representation}. Let \(D=\seq{D_r}{r\in\monnat}\in\oa{M}(\oa{K}(\sphilbfrak{h}_{\Lambda})^{\gamma})\). For \(\emptyset\neq \Lambda\subset \Lambda'\in\setpowfin{X}\) and \(m\in\monnat\), the finite-mode action is \begin{equation}
\begin{aligned}
\fun{\pi_{\sphilbfrak{h}_{\Lambda'},m}}{\rbk{\fun{i_{\Lambda,\Lambda'}}{D}}_m}
&=
\bigoplus_{r=0}^{m}
\fun{\pi_{\sphilbfrak{h}_{\Lambda},r}}{D_r}
\otimes
1_{\sphilbfrak{h}_{\Lambda'\setminus \Lambda}^{\otimes_{\txtsym}\rbk{m-r}}}.
\end{aligned}
\label{eq:paper-finite-region-local-inclusion-block-action}
\end{equation} For \(\emptyset\neq \Lambda\in\setpowfin{X}\) and \(m\in\monnat\), use the coordinate map \(p_m\) in \eqref{eq:paper-main-coordinate-map} and the comparison isomorphism \(\varpi_m\) in \eqref{eq:paper-main-fock-cutoff-identification}. The global action is \begin{equation}
\begin{aligned}
\fun{\varpi_m\circ p_m\circ i_{\Lambda,X}}{D}
&=
\bigoplus_{r=0}^{m}
\fun{\pi_{\sphilbfrak{h}_{\Lambda},r}}{D_r}
\otimes
1_{\sphilbfrak{h}_{X\setminus \Lambda}^{\otimes_{\txtsym}\rbk{m-r}}}.
\end{aligned}
\label{eq:paper-global-local-inclusion-block-actions}
\end{equation} The global inclusion in \eqref{eq:paper-global-local-inclusion-map} is the finite-mode embedding of \eqref{eq:paper-two-model-global-finite-mode-element}, explicitly \begin{equation}
\begin{aligned}
\fun{i_{\Lambda,X}}{D}
&=
\fun{\jmath}{
\faadj{U_{\sphilbfrak{h}_{\Lambda}}}
\rbk{\fun{\diag_{\sphilbfrak{h}_{\Lambda}}}{D}\otimes1}
U_{\sphilbfrak{h}_{\Lambda}}}.
\end{aligned}
\label{eq:paper-basis-global-local-inclusion}
\end{equation} Define the basis-local and quasilocal subalgebras by \begin{equation}
\begin{aligned}
\oa{B}_{\Lambda}
&=
\fun{i_{\Lambda,X}}{\oa{M}(\oa{K}(\sphilbfrak{h}_{\Lambda})^{\gamma})}
\quad\rbk{\Lambda\neq\emptyset},
&
\oa{B}_{\mathrm{ql}}
&=
\gtclos{\bigcup_{\Lambda\in\setpowfin{X}}\oa{B}_{\Lambda}}
^{\norm{\mathord{\cdot}}_{\fun{\oa{B}}{\sphilbfrak{h}}}}.
\end{aligned}
\label{eq:paper-basis-local-quasilocal-algebras}
\end{equation} The symbols \(\oa{B}_{\Lambda}\) and \(\oa{B}_{\mathrm{ql}}\) in this appendix refer to the chosen basis \(\seq{u_x}{x\in X}\); their graph specialization is identified at the end of the appendix.

\begin{prop}[Tensor-factor identities for local inclusions]\label{prop:paper-tensor-factor-localization-identities}
For $\emptyset\neq \Lambda\subset \Lambda'\subset \Lambda''\in\setpowfin{X}$,
\begin{equation}
\begin{aligned}
i_{\Lambda',\Lambda''}\circ i_{\Lambda,\Lambda'}
&=
i_{\Lambda,\Lambda''},
&
i_{\Lambda',X}\circ i_{\Lambda,\Lambda'}
&=
i_{\Lambda,X}.
\end{aligned}
\label{eq:paper-spatial-localization-system-compatibility}
\end{equation}
For $\Lambda,\Lambda'\in\setpowfin{X}\setminus\setone{\emptyset}$ with $\Lambda\cap \Lambda'=\emptyset$,
\begin{equation}
\begin{aligned}
\commutator{
\fun{i_{\Lambda,X}}{\oa{M}(\oa{K}(\sphilbfrak{h}_{\Lambda})^{\gamma})}}
{\fun{i_{\Lambda',X}}{\oa{M}(\oa{K}(\sphilbfrak{h}_{\Lambda'})^{\gamma})}}
&=
0.
\end{aligned}
\label{eq:paper-disjoint-local-algebras-commute}
\end{equation}
For the continuum basis, disjoint index sets need not have disjoint supports in $\fldreal^{d}$;
the commutation follows from orthogonality of the one-particle modes.
\end{prop}

\begin{proof}
Fix nonempty $\Lambda,\Lambda',\Lambda''\in\setpowfin{X}$ with $\Lambda\subset \Lambda'\subset \Lambda''$ and
$D\in\oa{M}(\oa{K}(\sphilbfrak{h}_{\Lambda})^{\gamma})$.
With respect to
$$\begin{aligned}
\sphilbfrak{h}_{\Lambda''}
&=
\sphilbfrak{h}_{\Lambda}
\oplus
\sphilbfrak{h}_{\Lambda'\setminus \Lambda}
\oplus
\sphilbfrak{h}_{\Lambda''\setminus \Lambda'},
\end{aligned}$$
both $\fun{i_{\Lambda',\Lambda''}\circ i_{\Lambda,\Lambda'}}{D}$ and $\fun{i_{\Lambda,\Lambda''}}{D}$ act on the $m$-particle space as
$$\begin{aligned}
\bigoplus_{\substack{i,j,k\in\monnat\\i+j+k=m}}
\fun{\pi_{\sphilbfrak{h}_{\Lambda},i}}{D_i}
\otimes1_{\sphilbfrak{h}_{\Lambda'\setminus \Lambda}^{\otimes_{\txtsym}j}}
\otimes1_{\sphilbfrak{h}_{\Lambda''\setminus \Lambda'}^{\otimes_{\txtsym}k}}.
\end{aligned}$$
The two possible iterated decompositions have the same normalization because
$$\begin{aligned}
\binom{m}{i+j}^{1/2}\binom{i+j}{i}^{1/2}
&=
\rbk{\frac{m!}{i!j!k!}}^{1/2}
=
\binom{m}{i}^{1/2}\binom{m-i}{j}^{1/2}.
\end{aligned}$$
This proves $i_{\Lambda',\Lambda''}\circ i_{\Lambda,\Lambda'}=i_{\Lambda,\Lambda''}$.
The same calculation with
$\sphilbfrak{h}=\sphilbfrak{h}_{\Lambda''}\oplus\sphilbfrak{h}_{\Lambda''}^{\perp}$ proves
$i_{\Lambda',X}\circ i_{\Lambda,\Lambda'}=i_{\Lambda,X}$.

Now let $\Lambda,\Lambda'\in\setpowfin{X}$ be disjoint and nonempty, and take
$D\in\oa{M}(\oa{K}(\sphilbfrak{h}_{\Lambda})^{\gamma})$ and
$C\in\oa{M}(\oa{K}(\sphilbfrak{h}_{\Lambda'})^{\gamma})$.
The orthogonal decomposition
$\sphilbfrak{h}=\sphilbfrak{h}_{\Lambda}\oplus\sphilbfrak{h}_{\Lambda'}
\oplus\sphilbfrak{h}_{X\setminus\rbk{\Lambda\cup \Lambda'}}$ and the map $\varpi_m$ in
\eqref{eq:paper-main-fock-cutoff-identification} give, on the $m$-particle block,
\begin{equation}
\begin{aligned}
\fun{\varpi_m\circ p_m\circ i_{\Lambda,X}}{D}
&=
\bigoplus_{\substack{i,j,k\in\monnat\\i+j+k=m}}
\fun{\pi_{\sphilbfrak{h}_{\Lambda},i}}{D_i}
\otimes1_{\sphilbfrak{h}_{\Lambda'}^{\otimes_{\txtsym}j}}
\otimes1_{\sphilbfrak{h}_{X\setminus\rbk{\Lambda\cup \Lambda'}}^{\otimes_{\txtsym}k}},
\\
\fun{\varpi_m\circ p_m\circ i_{\Lambda',X}}{C}
&=
\bigoplus_{\substack{i,j,k\in\monnat\\i+j+k=m}}
1_{\sphilbfrak{h}_{\Lambda}^{\otimes_{\txtsym}i}}
\otimes\fun{\pi_{\sphilbfrak{h}_{\Lambda'},j}}{C_j}
\otimes1_{\sphilbfrak{h}_{X\setminus\rbk{\Lambda\cup \Lambda'}}^{\otimes_{\txtsym}k}}.
\end{aligned}
\label{eq:paper-concrete-disjoint-local-actions}
\end{equation}
The two operators act on different tensor factors in every direct-sum component, so they commute.
This proves \eqref{eq:paper-disjoint-local-algebras-commute}.
\end{proof}

The dense basis span gives coordinate density and compatibility with the particle-degree embeddings.

\begin{prop}[Finite-mode coordinate density]\label{prop:paper-basis-local-coordinate-density}
For every $n\in\monnat$, the finite-mode coordinate ranges are norm dense:
\begin{equation}
\begin{aligned}
\gtclos{
\bigcup_{\substack{\Lambda\in\setpowfin{X}\\\Lambda\neq\emptyset}}
\fun{p_n\circ i_{\Lambda,X}}{\oa{M}(\oa{K}(\sphilbfrak{h}_{\Lambda})^{\gamma})}
}^{\norm{\mathord{\cdot}}_{\oa{B}_n(\sphilbfrak{h})}}
&=
\oa{B}_n(\sphilbfrak{h}).
\end{aligned}
\label{eq:paper-local-coordinate-density}
\end{equation}
The local inclusions also agree with the particle-degree coefficient embeddings.
If $\emptyset\neq\Lambda\in\setpowfin{X}$ and $\sphilbfrak{f}\subset \sphilbfrak{h}_{\Lambda}$ is
finite-dimensional, $0\leq m\leq n$, and
$K\in e_{\sphilbfrak{f},m}\oa{K}(\sphilbfrak{f})^{\gamma}e_{\sphilbfrak{f},m}$, the local particle coefficient
satisfies the embedding compatibility
\begin{equation}
\begin{aligned}
\fun{E_{m,n} \circ j_{\sphilbfrak{f},m}}{K}
&\in
\fun{p_n\circ i_{\Lambda,X}}{\oa{M}(\oa{K}(\sphilbfrak{h}_{\Lambda})^{\gamma})},
\end{aligned}
\label{eq:paper-local-particle-coefficient-embedding-compatibility}
\end{equation}
where $j_{\sphilbfrak{f},m}$ is defined in
\eqref{eq:paper-main-degree-coefficient-canonical-injection}.
\end{prop}

\begin{proof}
Fix $n\in\monnat$, $A_n\in\oa{B}_n(\sphilbfrak{h})$, and $\varepsilon>0$.
Let $\seq{a_m}{0\leq m\leq n}$ be the unique finite particle-degree coefficients of $A_n$ from
Lemma \ref{lem:paper-finite-particle-coefficient-uniqueness}.
Under $\Theta_m$, every $a_m$ is compact on $\sphilbfrak{h}^{\otimes_{\txtsym}m}$.
Finite-rank operators whose range and adjoint range lie in symmetric tensor powers of the algebraic span
of the basis are norm dense among these compact operators.
There is therefore a common nonempty finite $\Lambda\subset X$ and coefficients
$a_{\Lambda,m}$ in the range of $j_{\sphilbfrak{h}_{\Lambda},m}$ such that
\begin{equation}
\begin{aligned}
\norm{A_n-\sum_{m=0}^{n}\fun{E_{m,n}}{a_{\Lambda,m}}}
&<
\varepsilon.
\end{aligned}
\label{eq:paper-basis-local-density-approximation}
\end{equation}
To realize the coefficient sum in \eqref{eq:paper-basis-local-density-approximation} as a coordinate
of a bounded local element, use its finite-mode particle-degree formula to specify the local blocks
through degree $n$ and set all higher local blocks equal to zero.
The resulting finite-mode block family belongs to
$\oa{M}(\oa{K}(\sphilbfrak{h}_{\Lambda})^{\gamma})$.
The tensor-factor action in \eqref{eq:paper-global-local-inclusion-block-actions} then identifies its
$n$th global coordinate with the coefficient sum in \eqref{eq:paper-basis-local-density-approximation}.
The same identification applies to the coefficient $\fun{j_{\sphilbfrak{f},m}}{K}$ whenever
$\sphilbfrak{f}\subset\sphilbfrak{h}_{\Lambda}$, proving
\eqref{eq:paper-local-particle-coefficient-embedding-compatibility}.
The approximation proves
\eqref{eq:paper-local-coordinate-density}.
\end{proof}

The particle-cutoff inverse system alone does not determine the spatial localization maps. Proposition \ref{prop:paper-tensor-factor-localization-identities} proves compatibility and disjoint-range commutativity directly. Proposition \ref{prop:paper-basis-local-coordinate-density} establishes coordinate density and coefficient-reconstruction compatibility. The quasilocal algebra is the algebra \(\oa{B}_{\mathrm{ql}}\) defined in \eqref{eq:paper-basis-local-quasilocal-algebras}.

\begin{prop}[Cofinal invariance of the local closure]\label{prop:paper-cofinal-localization-invariance}
Let $\setfml{D}\subset\setpowfin{X}\setminus\setone{\emptyset}$ be directed by inclusion and assume
that, for every nonempty $\Lambda\in\setpowfin{X}$, there is $\Lambda'\in\setfml{D}$ with $\Lambda\subset \Lambda'$.
The cofinal local closure equals the full local closure:
\begin{equation}
\begin{aligned}
\gtclos{
\bigcup_{\substack{\Lambda\in\setfml{D}\\\Lambda\neq\emptyset}}
\fun{i_{\Lambda,X}}{\oa{M}(\oa{K}(\sphilbfrak{h}_{\Lambda})^{\gamma})}
}^{\norm{\mathord{\cdot}}_{\fun{\oa{B}}{\sphilbfrak{h}}}}
&=
\oa{B}_{\mathrm{ql}}.
\end{aligned}
\label{eq:paper-cofinal-localization-invariance}
\end{equation}
In particular, if $\seq{\Lambda_j}{j\in\monnat}$ is an increasing sequence of nonempty finite subsets of $X$
with $\bigcup_{j\in\monnat}\Lambda_j=X$, then the closure along this sequence agrees with
\eqref{eq:paper-basis-local-quasilocal-algebras}, and the Banach-space quotient of the abstract Buchholz algebra by its
quasilocal subalgebra,
$\setquot{\fun{\oa{B}}{\sphilbfrak{h}}}{\oa{B}_{\mathrm{ql}}}$, does not depend on the exhaustion.
\end{prop}

\begin{proof}
Fix a nonempty $\Lambda\in\setpowfin{X}$ and choose $\Lambda'\in\setfml{D}$ with $\Lambda\subset \Lambda'$.
Equation \eqref{eq:paper-spatial-localization-system-compatibility} gives
$$\begin{aligned}
\fun{i_{\Lambda,X}}{\oa{M}(\oa{K}(\sphilbfrak{h}_{\Lambda})^{\gamma})}
&=
\fun{i_{\Lambda',X}\circ i_{\Lambda,\Lambda'}}{\oa{M}(\oa{K}(\sphilbfrak{h}_{\Lambda})^{\gamma})}
\subset
\fun{i_{\Lambda',X}}{\oa{M}(\oa{K}(\sphilbfrak{h}_{\Lambda'})^{\gamma})}.
\end{aligned}$$
Taking unions and norm closures proves that the right-hand side of
\eqref{eq:paper-cofinal-localization-invariance} is contained in its left-hand side.
The inclusion $\setfml{D}\subset\setpowfin{X}\setminus\setone{\emptyset}$ proves the reverse
containment.
\end{proof}

\subsubsection{Spatially local subalgebras}\label{spatially-local-subalgebras}

In this subsubsection, the gauge-fixed resolvent algebras of finite basis subspaces embed compatibly in the full resolvent algebra, and an increasing sequence of finite index sets whose union is \(X\) gives the same norm closure.

For nonempty finite index sets \(\Lambda\subset \Lambda'\), the universal resolvent inclusions restrict to injective unital \(\ast\)-homomorphisms \[\begin{aligned}
\iota_{\Lambda,\Lambda'}^{\gamma}
\colon
\oaresolventalgebra(\sphilbfrak{h}_{\Lambda})^{\gamma}
\to
\oaresolventalgebra(\sphilbfrak{h}_{\Lambda'})^{\gamma};\quad
\fun{\iota_{\Lambda,\Lambda'}^{\gamma}}{A}
&=
\fun{\iota_{\sphilbfrak{h}_{\Lambda},\sphilbfrak{h}_{\Lambda'}}}{A}
\end{aligned}\] for \(A\in\oaresolventalgebra(\sphilbfrak{h}_{\Lambda})^{\gamma}\), where \(\iota_{\sphilbfrak{h}_{\Lambda},\sphilbfrak{h}_{\Lambda'}}\) is the canonical inclusion in \eqref{eq:paper-canonical-resolvent-inclusion}. Its range is gauge invariant because \[\begin{aligned}
\fun{\gamma_{\theta}\circ\iota_{\sphilbfrak{h}_{\Lambda},\sphilbfrak{h}_{\Lambda'}}}{A}
&=
\fun{\iota_{\sphilbfrak{h}_{\Lambda},\sphilbfrak{h}_{\Lambda'}}\circ\gamma_{\theta}}{A}
=
\fun{\iota_{\sphilbfrak{h}_{\Lambda},\sphilbfrak{h}_{\Lambda'}}}{A}.
\end{aligned}\] They satisfy \begin{equation}
\begin{aligned}
\iota_{\Lambda',\Lambda''}^{\gamma}\circ\iota_{\Lambda,\Lambda'}^{\gamma}
&=
\iota_{\Lambda,\Lambda''}^{\gamma}
\end{aligned}
\label{eq:paper-spatial-localization-compatibility}
\end{equation} whenever \(\Lambda,\Lambda',\Lambda''\in\setpowfin{X}\) are nonempty and \(\Lambda\subset \Lambda'\subset \Lambda''\).

The norm-closed local algebra is \begin{equation}
\begin{aligned}
\oaresolventalgebra_{\mathrm{loc}}(X)^{\gamma}
&=
\gtclos{
\bigcup_{\substack{\Lambda\in\setpowfin{X}\\\Lambda\neq\emptyset}}
\oaresolventalgebra(\sphilbfrak{h}_{\Lambda})^{\gamma}
}^{\norm{\mathord{\cdot}}_{\oaresolventalgebra(\sphilbfrak{h},\sigma)}}.
\end{aligned}
\label{eq:paper-local-gauge-resolvent-algebra}
\end{equation} If \(\setfml{D}\subset\setpowfin{X}\setminus\setone{\emptyset}\) is a cofinal directed family, then \[\begin{aligned}
\gtclos{
\bigcup_{\Lambda\in\setfml{D}}
\oaresolventalgebra(\sphilbfrak{h}_{\Lambda})^{\gamma}
}^{\norm{\mathord{\cdot}}_{\oaresolventalgebra(\sphilbfrak{h},\sigma)}}
&=
\oaresolventalgebra_{\mathrm{loc}}(X)^{\gamma}.
\end{aligned}\] To verify the nontrivial inclusion, fix a nonempty \(\Lambda\in\setpowfin{X}\) and choose \(\Lambda'\in\setfml{D}\) with \(\Lambda\subset \Lambda'\). Equation \eqref{eq:paper-spatial-localization-compatibility} gives \[\begin{aligned}
\oaresolventalgebra(\sphilbfrak{h}_{\Lambda})^{\gamma}
&\subset
\oaresolventalgebra(\sphilbfrak{h}_{\Lambda'})^{\gamma}.
\end{aligned}\] The reverse inclusion follows from \(\setfml{D}\subset\set{\Lambda\subset X}{\Lambda\text{ is finite}}\).

\subsubsection{Quasilocal subalgebra}\label{sec:paper-quasilocal-kernel}

In this subsubsection, the basis-local algebras in \eqref{eq:paper-basis-local-quasilocal-algebras} act as tensor-factor operators in the Fock comparison representation, and their local ranges generate the quasilocal algebra by norm closure.

The construction uses the finite-dimensional subspaces and projections in \eqref{eq:paper-basis-localization-data}. For a nonempty \(\Lambda\in\setpowfin{X}\), the external one-particle space is \(\sphilbfrak{h}_{X\setminus\Lambda}\subset\sphilbfrak{h}\). For \(n\in\monnat\), let \(P_n\) be the specialization \(r=n\) of the symmetrizing projection in \eqref{eq:paper-main-symmetrizing-projection}. For \(0\leq k\leq n\), \(\xi\in\sphilbfrak{h}_{\Lambda}^{\otimes_{\txtsym}k}\), and \(\eta\in\sphilbfrak{h}_{X\setminus\Lambda}^{\otimes_{\txtsym}\rbk{n-k}}\), define the unitary particle decomposition by \[\begin{aligned}
W_{\sphilbfrak{h}_{\Lambda},n}
\colon
\sphilbfrak{h}^{\otimes_{\txtsym}n}
\to
\bigoplus_{k=0}^{n}
\sphilbfrak{h}_{\Lambda}^{\otimes_{\txtsym}k}\otimes \sphilbfrak{h}_{X\setminus\Lambda}^{\otimes_{\txtsym}\rbk{n-k}}.
\end{aligned}\] Its inverse acts by \[\begin{aligned}
\fun{\faadj{W_{\sphilbfrak{h}_{\Lambda},n}}}{\xi\otimes\eta}
&=
\binom{n}{k}^{1/2}P_n\rbk{\xi\otimes\eta}.
\end{aligned}\] The global Fock decomposition is \[\begin{aligned}
U_{\Lambda}
\colon
\fun{\spfockgamma_{\txtbsn}}{\sphilbfrak{h}}
\to
\fun{\spfockgamma_{\txtbsn}}{\sphilbfrak{h}_{\Lambda}}\otimes\fun{\spfockgamma_{\txtbsn}}{\sphilbfrak{h}_{X\setminus\Lambda}}.
\end{aligned}\] It is given on particle blocks by \[\begin{aligned}
\fun{U_{\Lambda}}{\bigoplus_{n\in\monnat}\xi_n}
&=
\bigoplus_{n\in\monnat}\fun{W_{\sphilbfrak{h}_{\Lambda},n}}{\xi_n}.
\end{aligned}\] The isomorphism \(\diag_{\sphilbfrak{h}_{\Lambda}}\) in \eqref{eq:paper-main-finite-mode-product-isomorphism} identifies the intrinsic finite-mode algebra with its concrete gauge-fixed Fock block algebra. The global local map is \[\begin{aligned}
i_{\Lambda,X}
\colon
\oa{M}(\oa{K}(\sphilbfrak{h}_{\Lambda})^{\gamma})
\to
\fun{\oa{B}}{\sphilbfrak{h}}.
\end{aligned}\] Equation \eqref{eq:paper-global-local-inclusion-block-actions} and the definition of \(\diag_{\sphilbfrak{h}_{\Lambda}}\) in \eqref{eq:paper-main-finite-mode-product-isomorphism} show that \(\faadj{U_{\Lambda}}\rbk{\fun{\diag_{\sphilbfrak{h}_{\Lambda}}}{D}\otimes1}U_{\Lambda}\) belongs to the range of \(\oarepn_{\txtfock}\). Thus the map \(\diag_{\sphilbfrak{h}_{\Lambda}}\) in \eqref{eq:paper-main-finite-mode-product-isomorphism} gives the global local map by \[\begin{aligned}
\fun{i_{\Lambda,X}}{D}
&=
\oarepn_{\txtfock}^{-1}
\rbk{\faadj{U_{\Lambda}}\rbk{\fun{\diag_{\sphilbfrak{h}_{\Lambda}}}{D}\otimes1}U_{\Lambda}}.
\end{aligned}\] The range of \(i_{\Lambda,X}\) is the local algebra \(\oa{B}_{\Lambda}\) in \eqref{eq:paper-basis-local-quasilocal-algebras}. Fix nonempty \(\Lambda,\Lambda''\in\setpowfin{X}\) with \(\Lambda\subset \Lambda''\) and an integer \(0\leq k\leq n\). Take \(\xi\in\sphilbfrak{h}_{\Lambda}^{\otimes_{\txtsym}k}\) and \(\eta\in\sphilbfrak{h}_{\Lambda''\setminus \Lambda}^{\otimes_{\txtsym}\rbk{n-k}}\). Define the finite-mode decomposition by \[\begin{aligned}
W_{\sphilbfrak{h}_{\Lambda},n}^{\sphilbfrak{h}_{\Lambda''}}
\colon
\sphilbfrak{h}_{\Lambda''}^{\otimes_{\txtsym}n}
\to
\bigoplus_{k=0}^{n}\sphilbfrak{h}_{\Lambda}^{\otimes_{\txtsym}k}
\otimes \sphilbfrak{h}_{\Lambda''\setminus \Lambda}^{\otimes_{\txtsym}\rbk{n-k}}.
\end{aligned}\] Its inverse acts by \[\begin{aligned}
\fun{\faadjrbk{W_{\sphilbfrak{h}_{\Lambda},n}^{\sphilbfrak{h}_{\Lambda''}}}}{\xi\otimes\eta}
&=
\binom{n}{k}^{1/2}P_n\rbk{\xi\otimes\eta}.
\end{aligned}\] Here \(P_n\) restricts to the symmetrizing projection on \(\sphilbfrak{h}_{\Lambda''}^{\otimes n}\). The finite-mode decomposition unitary is \[\begin{aligned}
U_{\Lambda,\Lambda''}
\colon
\fun{\spfockgamma_{\txtbsn}}{\sphilbfrak{h}_{\Lambda''}}
\to
\fun{\spfockgamma_{\txtbsn}}{\sphilbfrak{h}_{\Lambda}}\otimes\fun{\spfockgamma_{\txtbsn}}{\sphilbfrak{h}_{\Lambda''\setminus \Lambda}}.
\end{aligned}\] It acts by \[\begin{aligned}
\fun{U_{\Lambda,\Lambda''}}{\bigoplus_{n\in\monnat}\xi_n}
&=
\bigoplus_{n\in\monnat}\fun{W_{\sphilbfrak{h}_{\Lambda},n}^{\sphilbfrak{h}_{\Lambda''}}}{\xi_n}.
\end{aligned}\] The corresponding local inclusion is \[\begin{aligned}
i_{\Lambda,\Lambda''}
\colon
\oa{M}(\oa{K}(\sphilbfrak{h}_{\Lambda})^{\gamma})
\to
\oa{M}(\oa{K}(\sphilbfrak{h}_{\Lambda''})^{\gamma}).
\end{aligned}\] The maps \(\diag_{\sphilbfrak{h}_{\Lambda}}\) and \(\diag_{\sphilbfrak{h}_{\Lambda''}}\) in \eqref{eq:paper-main-finite-mode-product-isomorphism} give this inclusion by \[\begin{aligned}
\fun{i_{\Lambda,\Lambda''}}{D}
&=
\fun{\diag_{\sphilbfrak{h}_{\Lambda''}}^{-1}}
{\faadj{U_{\Lambda,\Lambda''}}\rbk{\fun{\diag_{\sphilbfrak{h}_{\Lambda}}}{D}\otimes1}U_{\Lambda,\Lambda''}}.
\end{aligned}\] Proposition \ref{prop:paper-tensor-factor-localization-identities} verifies the tensor-factor identities \[\begin{aligned}
i_{\Lambda',X}\circ i_{\Lambda,\Lambda'}
&=
i_{\Lambda,X},
&
i_{\Lambda',\Lambda''}\circ i_{\Lambda,\Lambda'}
&=
i_{\Lambda,\Lambda''}
\end{aligned}\] for \(\Lambda\subset \Lambda'\subset \Lambda''\). For \(D\in\oa{B}_{\Lambda}\), its faithful comparison representative acts as a scalar on the local-vacuum subspace \(\fldcmp\Psi_{\txtfock,\txtvac}\otimes\fun{\spfockgamma_{\txtbsn}}{\sphilbfrak{h}_{X\setminus\Lambda}}\). Explicitly, \begin{equation}
\begin{aligned}
\fnrestr{\fun{\oarepn_{\txtfock}}{D}}
{\fldcmp\Psi_{\txtfock,\txtvac}\otimes\fun{\spfockgamma_{\txtbsn}}{\sphilbfrak{h}_{X\setminus\Lambda}}}
&=
d_D,
&
d_D
&=
\bkt{\Psi_{\txtfock,\txtvac}}{\fun{\oarepn_{\txtfock}}{D}\Psi_{\txtfock,\txtvac}}.
\end{aligned}
\label{eq:paper-local-element-external-only-block}
\end{equation} The norm closure of these ranges is the quasilocal algebra in \eqref{eq:paper-basis-local-quasilocal-algebras}.

\subsubsection{States on the quasilocal subalgebra}\label{states-on-the-quasilocal-subalgebra}

In this subsubsection, the finite-mode block representations in \eqref{eq:paper-main-finite-mode-block-representation} give a compatibility criterion for restrictions of global states to the quasilocal algebra. The image of the local gauge-fixed resolvent algebra also has a uniform particle-cutoff criterion.

For a subset \(S\subset\opspbddlin{\sphilbfrak{H}}\), define its commutant and double commutant by \[\begin{aligned}
\oacommutant{S}
&=
\set{T\in\opspbddlin{\sphilbfrak{H}}}{TA=AT\text{ for every }A\in S},
&
\oadoublecommutant{S}
&=
\oacommutant{\rbk{\oacommutant{S}}}.
\end{aligned}\] A von Neumann algebra on \(\sphilbfrak{H}\) is a unital \(\ast\)-subalgebra of \(\opspbddlin{\sphilbfrak{H}}\) equal to its double commutant. A state on a von Neumann algebra is normal when it preserves suprema of bounded increasing nets of positive elements.

\begin{defn}[Locally normal state]\label{def:paper-locally-normal-state}
A state $\oastate[\psi]\in\oaspstate_{\fun{\oa{B}}{\sphilbfrak{h}}}$ is locally normal if, for every
nonempty $\Lambda\in\setpowfin{X}$, its restriction to $\oa{B}_{\Lambda}$ extends to a normal state on
$$\begin{aligned}
\oadoublecommutant{
\fun{\fnrestr{\oarepn_{\txtfock}}{\oa{B}_{\Lambda}}}{\oa{B}_{\Lambda}}
}.
\end{aligned}$$
\end{defn}

\begin{thm}[Classification of locally normal states]\label{thm:paper-locally-normal-state-classification}
A state $\oastate[\psi]\in\oaspstate_{\fun{\oa{B}}{\sphilbfrak{h}}}$ is locally normal if and only if,
for every nonempty $\Lambda\in\setpowfin{X}$, there is a unique sequence of positive operators
$R_{\Lambda,k}\in\opspbddlin{\sphilbfrak{h}_{\Lambda}^{\otimes_{\txtsym}k}}$ such that
\begin{equation}
\begin{aligned}
\sum_{k\in\monnat}\sqfun{\trace_{\sphilbfrak{h}_{\Lambda}^{\otimes_{\txtsym}k}}}{R_{\Lambda,k}}
&=
1.
\end{aligned}
\label{eq:paper-locally-normal-density-normalization}
\end{equation}
The state evaluation is
\begin{equation}
\begin{aligned}
\fun{\oastate[\psi]\circ i_{\Lambda,X}}{D}
&=
\sum_{k\in\monnat}
\sqfun{\trace_{\sphilbfrak{h}_{\Lambda}^{\otimes_{\txtsym}k}}}
{R_{\Lambda,k}\fun{\pi_{\sphilbfrak{h}_{\Lambda},k}}{D_k}}
\end{aligned}
\label{eq:paper-locally-normal-density-family}
\end{equation}
for $D=\seq{D_k}{k\in\monnat}\in\oa{M}(\oa{K}(\sphilbfrak{h}_{\Lambda})^{\gamma})$, and such that, whenever
nonempty $\Lambda,\Lambda''\in\setpowfin{X}$ with $\Lambda\subset \Lambda''$,
\begin{equation}
\begin{aligned}
\sum_{k\in\monnat}
\sqfun{\trace_{\sphilbfrak{h}_{\Lambda}^{\otimes_{\txtsym}k}}}
{R_{\Lambda,k}\fun{\pi_{\sphilbfrak{h}_{\Lambda},k}}{D_k}}
&=
\sum_{m\in\monnat}
\sqfun{\trace_{\sphilbfrak{h}_{\Lambda''}^{\otimes_{\txtsym}m}}}
{R_{\Lambda'',m}\fun{\pi_{\sphilbfrak{h}_{\Lambda''},m}}{\rbk{\fun{i_{\Lambda,\Lambda''}}{D}}_m}}.
\end{aligned}
\label{eq:paper-locally-normal-density-compatibility}
\end{equation}
Conversely, every family satisfying
\eqref{eq:paper-locally-normal-density-normalization},
\eqref{eq:paper-locally-normal-density-family}, and
\eqref{eq:paper-locally-normal-density-compatibility} determines a unique state on
$\oa{B}_{\mathrm{ql}}$.
This state extends to $\fun{\oa{B}}{\sphilbfrak{h}}$, every such extension is
locally normal, and all locally normal states arise in this way.
\end{thm}

\begin{proof}
Assume first that $\oastate[\psi]$ is locally normal.
Fix a nonempty finite $\Lambda$.
The finite-mode product representation identifies the local von Neumann algebra with
$$\begin{aligned}
\prod_{k\in\monnat}\opspbddlin{\sphilbfrak{h}_{\Lambda}^{\otimes_{\txtsym}k}}.
\end{aligned}$$
Transport a normal extension of the restriction of $\oastate[\psi]$ through the map
$\diag_{\sphilbfrak{h}_{\Lambda}}$ in \eqref{eq:paper-main-finite-mode-product-isomorphism} and denote it by
$$\begin{aligned}
\oastate[\psi_{\Lambda}]
\colon
\prod_{k\in\monnat}\opspbddlin{\sphilbfrak{h}_{\Lambda}^{\otimes_{\txtsym}k}}
\to
\fldcmp;\quad
\fun{\oastate[\psi_{\Lambda}]}{\fun{\diag_{\sphilbfrak{h}_{\Lambda}}}{D}}
&=
\fun{\oastate[\psi]\circ i_{\Lambda,X}}{D}
&&\rbk{D\in\oa{M}(\oa{K}(\sphilbfrak{h}_{\Lambda})^{\gamma})}.
\end{aligned}$$
Let $z_{\Lambda,k}$ denote the central projection onto the $k$th factor.
Normality gives
$$\begin{aligned}
\sum_{k\in\monnat}\fun{\oastate[\psi_{\Lambda}]}{z_{\Lambda,k}}
&=
\fun{\oastate[\psi_{\Lambda}]}{1}
=
1.
\end{aligned}$$
For $A=\seq{A_j}{j\in\monnat}$, normal monotone convergence and positivity give
$$\begin{aligned}
\fun{\oastate[\psi_{\Lambda}]}{A}
&=
\lim_{N\to\infty}\sum_{k=0}^{N}\fun{\oastate[\psi_{\Lambda}]}{z_{\Lambda,k}A}
=
\sum_{k\in\monnat}\fun{\oastate[\psi_{\Lambda}]}{z_{\Lambda,k}A},
\\
\sum_{k\in\monnat}
\abs{\fun{\oastate[\psi_{\Lambda}]}{z_{\Lambda,k}A}}
&\leq
\norm{A}\sum_{k\in\monnat}\fun{\oastate[\psi_{\Lambda}]}{z_{\Lambda,k}}
=
\norm{A}.
\end{aligned}$$
Trace duality on each finite-dimensional full matrix algebra gives a unique $R_{\Lambda,k}\geq0$ with
$$\begin{aligned}
\fun{\oastate[\psi_{\Lambda}]}{z_{\Lambda,k}A}
&=
\sqfun{\trace_{\sphilbfrak{h}_{\Lambda}^{\otimes_{\txtsym}k}}}{R_{\Lambda,k}A_k},
&
\sqfun{\trace_{\sphilbfrak{h}_{\Lambda}^{\otimes_{\txtsym}k}}}{R_{\Lambda,k}}
&=
\fun{\oastate[\psi_{\Lambda}]}{z_{\Lambda,k}}.
\end{aligned}$$
This proves \eqref{eq:paper-locally-normal-density-normalization} and
\eqref{eq:paper-locally-normal-density-family}; compatibility follows by substituting
$i_{\Lambda'',X}\circ i_{\Lambda,\Lambda''}=i_{\Lambda,X}$ in the two local formulas.

Conversely, fix a family satisfying \eqref{eq:paper-locally-normal-density-compatibility}.
For $A=\seq{A_k}{k\in\monnat}$,
$$\begin{aligned}
\sum_{k\in\monnat}
\abs{\sqfun{\trace_{\sphilbfrak{h}_{\Lambda}^{\otimes_{\txtsym}k}}}{R_{\Lambda,k}A_k}}
&\leq
\sup_{k\in\monnat}\norm{A_k}
\sum_{k\in\monnat}\sqfun{\trace_{\sphilbfrak{h}_{\Lambda}^{\otimes_{\txtsym}k}}}{R_{\Lambda,k}}
=
\norm{A}.
\end{aligned}$$
Let $I$ be a directed set, and let
$\net{A_{\alpha}}{\alpha\in I}$ be a bounded increasing net in the positive cone of the product such that
$A_{\alpha}=\seq{A_{\alpha,k}}{k\in\monnat}$ and
$A=\seq{A_k}{k\in\monnat}$ is its supremum.
The value of the functional at this supremum is
$$\begin{aligned}
\sup_{\alpha\in I}\sum_{k\in\monnat}
\sqfun{\trace_{\sphilbfrak{h}_{\Lambda}^{\otimes_{\txtsym}k}}}{R_{\Lambda,k}A_{\alpha,k}}
&=
\sup_{N\in\monnat}\sup_{\alpha\in I}\sum_{k=0}^{N}
\sqfun{\trace_{\sphilbfrak{h}_{\Lambda}^{\otimes_{\txtsym}k}}}{R_{\Lambda,k}A_{\alpha,k}}
\\
&=
\sup_{N\in\monnat}\sum_{k=0}^{N}
\sqfun{\trace_{\sphilbfrak{h}_{\Lambda}^{\otimes_{\txtsym}k}}}{R_{\Lambda,k}A_k}
=
\sum_{k\in\monnat}
\sqfun{\trace_{\sphilbfrak{h}_{\Lambda}^{\otimes_{\txtsym}k}}}{R_{\Lambda,k}A_k}.
\end{aligned}$$
The resulting normal state is
$$\begin{aligned}
\oastate[\psi_{\Lambda}]
\colon
\prod_{k\in\monnat}\opspbddlin{\sphilbfrak{h}_{\Lambda}^{\otimes_{\txtsym}k}}
\to
\fldcmp;\quad
\fun{\oastate[\psi_{\Lambda}]}{\seq{A_k}{k\in\monnat}}
&=
\sum_{k\in\monnat}
\sqfun{\trace_{\sphilbfrak{h}_{\Lambda}^{\otimes_{\txtsym}k}}}{R_{\Lambda,k}A_k}.
\end{aligned}$$
Using the map $\diag_{\sphilbfrak{h}_{\Lambda}}$ in
\eqref{eq:paper-main-finite-mode-product-isomorphism}, define
$$\begin{aligned}
\oastate[\psi_{\mathrm{loc}}]
\colon
\bigcup_{\substack{\Lambda\in\setpowfin{X}\\\Lambda\neq\emptyset}}\oa{B}_{\Lambda}
\to
\fldcmp;\quad
\fun{\oastate[\psi_{\mathrm{loc}}]}{\fun{i_{\Lambda,X}}{D}}
&=
\fun{\oastate[\psi_{\Lambda}]}{\fun{\diag_{\sphilbfrak{h}_{\Lambda}}}{D}}.
\end{aligned}$$
Equation \eqref{eq:paper-locally-normal-density-compatibility} makes this functional independent of the chosen
region.
It is positive, unital, and has norm one, hence extends uniquely by continuity to
$\oa{B}_{\mathrm{ql}}$.
Denote this extension by
$\oastate[\psi_{\mathrm{ql}}]\colon\oa{B}_{\mathrm{ql}}\to\fldcmp$.
The continuous extension form of the Hahn--Banach theorem gives a norm-preserving extension
$$\begin{aligned}
\oastate[\widetilde{\psi}]
\colon
\fun{\oa{B}}{\sphilbfrak{h}}
\to
\fldcmp;\quad
\fnrestr{\oastate[\widetilde{\psi}]}{\oa{B}_{\mathrm{ql}}}
&=
\oastate[\psi_{\mathrm{ql}}].
\end{aligned}$$
Because $1\in\oa{B}_{\mathrm{ql}}$, one has
$\fun{\oastate[\widetilde{\psi}]}{1}=\norm{\oastate[\widetilde{\psi}]}=1$.
For every self-adjoint $S\in\fun{\oa{B}}{\sphilbfrak{h}}$, the norm-convergent exponential series and
$\norm{\oastate[\widetilde{\psi}]}=1$ give
$$\begin{aligned}
0
&\geq
\lim_{t\downarrow0}
\frac{
\abs{\fun{\oastate[\widetilde{\psi}]}{\napiernum^{\imunit tS}}}^{2}-1
}{t}
=
-2\opimag\fun{\oastate[\widetilde{\psi}]}{S}.
\end{aligned}$$
Replacement of $S$ by $-S$ gives
$$\begin{aligned}
0
&\geq
2\opimag\fun{\oastate[\widetilde{\psi}]}{S},
\end{aligned}$$
which proves $\fun{\oastate[\widetilde{\psi}]}{S}\in\fldreal$.
For a nonzero positive $B\in\fun{\oa{B}}{\sphilbfrak{h}}$, set $C=\norm{B}^{-1}B$.
The spectral inclusion $\opvarspec{C}\subset\closedinterval{0}{1}$ and the continuous functional calculus give
$\norm{1-C}\leq1$.
The reality statement and the norm bound imply
$$\begin{aligned}
\abs{1-\fun{\oastate[\widetilde{\psi}]}{C}}
&=
\abs{\fun{\oastate[\widetilde{\psi}]}{1-C}}
\leq
\norm{1-C}
\leq
1.
\end{aligned}$$
It follows that
$\fun{\oastate[\widetilde{\psi}]}{B}=\norm{B}\fun{\oastate[\widetilde{\psi}]}{C}\geq0$.
The value at $B=0$ is zero, and $\oastate[\widetilde{\psi}]$ is a state.
Every extension has the prescribed normal local restrictions.
Uniqueness of each density and the norm-preserving extension argument prove the classification.
\end{proof}

The union of the finite-mode gauge-fixed resolvent algebras in \eqref{eq:paper-local-gauge-resolvent-algebra} is norm dense in that local algebra by definition. The quotient-coordinate map \(\rho\) of Proposition \ref{prop:bounded-inverse-limit-cstar} sends this union into \(\oa{B}_{\mathrm{ql}}\), since the fixed-point multiplier action in \eqref{eq:paper-universal-fixed-point-multiplier-action} maps each finite-mode algebra into its basis-local algebra in \eqref{eq:paper-basis-local-quasilocal-algebras}. Continuity of \(\rho\) and closedness of \(\oa{B}_{\mathrm{ql}}\) give \(\fun{\rho}{\oaresolventalgebra_{\mathrm{loc}}(X)^\gamma}\subset\oa{B}_{\mathrm{ql}}\).

\begin{prop}[Uniform coefficient criterion for the local resolvent algebra]\label{prop:paper-resolvent-particle-coefficient-uniform-criterion}
For the quotient-coordinate map $\rho$ in Proposition \ref{prop:bounded-inverse-limit-cstar},
the map $\Upsilon$ in \eqref{eq:paper-main-particle-coefficient-isomorphism}, and
$a\in\oa{C}_{\mathrm{b}}(\sphilbfrak{h})$, one has
$\fun{\Upsilon}{a}\in\fun{\rho}{\oaresolventalgebra_{\mathrm{loc}}(X)^\gamma}$ if and only if, for every
$\varepsilon>0$, there is
$P\in\bigcup_{\emptyset\neq\Lambda\in\setpowfin{X}}
\oaresolventalgebra(\sphilbfrak{h}_{\Lambda})^\gamma$ such that
$$\begin{aligned}
\sup_{n\in\monnat}
\norm{\fun{\Pi_n}{a}-\rbk{P+\oaideal{I}_{>n}(\sphilbfrak{h})}}_{\oa{B}_n(\sphilbfrak{h})}
&<
\varepsilon.
\end{aligned}$$
\end{prop}

\begin{proof}
If $\fun{\Upsilon}{a}=\fun{\rho}{A}$, density gives $P$ with $\norm{A-P}<\varepsilon$.
Contractivity of every
quotient map yields
$$\begin{aligned}
\sup_{n\in\monnat}
\norm{\fun{\Pi_n}{a}-\rbk{P+\oaideal{I}_{>n}(\sphilbfrak{h})}}
&=
\sup_{n\in\monnat}\norm{\fun{\rho_n}{A-P}}
\leq
\norm{A-P}
<
\varepsilon.
\end{aligned}$$
Conversely choose $P_j$ with error $2^{-j}$.
The isometries $\rho$ and $\Upsilon$ give
$$\begin{aligned}
\norm{P_j-P_k}
=
\norm{\fun{\rho}{P_j-P_k}}
\leq
\norm{\fun{\rho}{P_j}-\fun{\Upsilon}{a}}
+
\norm{\fun{\Upsilon}{a}-\fun{\rho}{P_k}}
<
2^{-j}+2^{-k}.
\end{aligned}$$
Thus $P_j\to A$ in $\oaresolventalgebra_{\mathrm{loc}}(X)^\gamma$, and continuity gives
$$\begin{aligned}
\fun{\rho}{A}
&=
\lim_{j\to\infty}\fun{\rho}{P_j}
=
\fun{\Upsilon}{a}.
\end{aligned}$$
\end{proof}

\subsubsection{Banach-space quotient}\label{banach-space-quotient}

In this subsubsection, distances from the quasilocal algebra have a dual formula in terms of pairs of states with the same local restrictions. A pure state on the Buchholz algebra also fails to factor through any finite particle cutoff.

The quasilocal subalgebra is unital and norm closed. Its Banach-space quotient is \begin{equation}
\begin{aligned}
\oa{T}_{\mathrm{B}}
&=
\setquot{\fun{\oa{B}}{\sphilbfrak{h}}}{\oa{B}_{\mathrm{ql}}}.
\end{aligned}
\label{eq:paper-main-quasilocal-banach-quotient}
\end{equation} For \(A\in\fun{\oa{B}}{\sphilbfrak{h}}\), its quotient class is \begin{equation}
\begin{aligned}
\sqbk{A}
&=
A+\oa{B}_{\mathrm{ql}}.
\end{aligned}
\label{eq:paper-main-quasilocal-banach-class}
\end{equation} The quotient norm is \begin{equation}
\begin{aligned}
\norm{\sqbk{A}}_{\oa{T}_{\mathrm{B}}}
&=
\inf_{D\in\oa{B}_{\mathrm{ql}}}\norm{A-D}.
\end{aligned}
\label{eq:paper-main-quasilocal-banach-norm}
\end{equation} The quotient is a Banach space; no multiplication on it is used.

Recall that a pure state is an extreme point of the convex state space.

\begin{prop}[A pure state beyond every finite particle cutoff]\label{prop:paper-pure-state-beyond-finite-coordinates}
There is a pure state $\oastate[\psi]$ on $\fun{\oa{B}}{\sphilbfrak{h}}$.
Its Gelfand--Naimark--Segal null space is
$$\begin{aligned}
\oa{N}_{\psi}
&=
\set{B\in\fun{\oa{B}}{\sphilbfrak{h}}}{\fun{\oastate[\psi]}{\faadj{B}B}=0}.
\end{aligned}$$
For $B\in\fun{\oa{B}}{\sphilbfrak{h}}$, define its quotient class by
$$\begin{aligned}
\sqbk{B}_{\psi}
&=
B+\oa{N}_{\psi}
\in
\setquot{\fun{\oa{B}}{\sphilbfrak{h}}}{\oa{N}_{\psi}}.
\end{aligned}$$
Let $\sphilbfrak{H}_{\psi}$ be the Hilbert-space completion of
$\setquot{\fun{\oa{B}}{\sphilbfrak{h}}}{\oa{N}_{\psi}}$ for the inner product
$$\begin{aligned}
\bkt{\sqbk{A}_{\psi}}{\sqbk{B}_{\psi}}_{\sphilbfrak{H}_{\psi}}
&=
\fun{\oastate[\psi]}{\faadj{A}B}.
\end{aligned}$$
Its Gelfand--Naimark--Segal representation is
$$\begin{aligned}
\pi_{\psi}
\colon
\fun{\oa{B}}{\sphilbfrak{h}}
\to
\opspbddlin{\sphilbfrak{H}_{\psi}};\quad
\fun{\pi_{\psi}}{A}\sqbk{B}_{\psi}
&=
\sqbk{AB}_{\psi}.
\end{aligned}$$
The corresponding Gelfand--Naimark--Segal cyclic vector is
$\oagnsvector[\Psi]=\sqbk{1}_{\psi}$.
The representation satisfies $\Ker p_N\not\subset\Ker\pi_{\psi}$ for every $N\in\monnat$,
with $p_N$ as in \eqref{eq:paper-main-coordinate-map}.
\end{prop}

\begin{proof}
Fix an order on $X$.
For $\Lambda=\setone{x_1<\cdots<x_{\abscard{\Lambda}}}$ define
$$\begin{aligned}
\Psi_{\Lambda}
&=
\sqrt{\abscard{\Lambda}!}\fun{P_{\abscard{\Lambda}}}
{u_{x_1}\otimes\cdots\otimes u_{x_{\abscard{\Lambda}}}},
&
\norm{\Psi_{\Lambda}}^2
&=
\abscard{\Lambda}!\frac{\abscard{\Lambda}!}{\rbk{\abscard{\Lambda}!}^2}
=
1.
\end{aligned}$$
Define the finite-mode block representation by
$$\begin{aligned}
\oarepn_{\mathrm{block},\Lambda,\abscard{\Lambda}}
\colon
\oa{M}(\oa{K}(\sphilbfrak{h}_{\Lambda})^{\gamma})
\to
\opspbddlin{\sphilbfrak{h}_{\Lambda}^{\otimes_{\txtsym}\abscard{\Lambda}}};\quad
\fun{\oarepn_{\mathrm{block},\Lambda,\abscard{\Lambda}}}{\seq{D_k}{k\in\monnat}}
&=
\fun{\pi_{\sphilbfrak{h}_{\Lambda},\abscard{\Lambda}}}{D_{\abscard{\Lambda}}}.
\end{aligned}$$
Proposition \ref{prop:finite-mode-product} makes this map surjective and gives
$$\begin{aligned}
\setquot{\oa{M}(\oa{K}(\sphilbfrak{h}_{\Lambda})^{\gamma})}{\Ker\oarepn_{\mathrm{block},\Lambda,\abscard{\Lambda}}}
&\simeq
\opspbddlin{\sphilbfrak{h}_{\Lambda}^{\otimes_{\txtsym}\abscard{\Lambda}}}.
\end{aligned}$$
Now define the local vector state by
$$\begin{aligned}
\oastate[\tau_{\Lambda}]
\colon
\oa{M}(\oa{K}(\sphilbfrak{h}_{\Lambda})^{\gamma})
\to
\fldcmp;\quad
\fun{\oastate[\tau_{\Lambda}]}{D}
&=
\bkt{\Psi_{\Lambda}}{
\fun{\oarepn_{\mathrm{block},\Lambda,\abscard{\Lambda}}}{D}\Psi_{\Lambda}
}_{\sphilbfrak{h}_{\Lambda}^{\otimes_{\txtsym}\abscard{\Lambda}}}.
\end{aligned}$$
The vector state of the full matrix algebra induced by $\Psi_{\Lambda}$ is pure.
Its pullback through the quotient map is therefore the pure state $\oastate[\tau_{\Lambda}]$.
If $\Lambda\subset \Lambda''$, then
$$\begin{aligned}
\Psi_{\Lambda''}
&=
\Psi_{\Lambda}\otimes\Psi_{\Lambda''\setminus \Lambda},
&
\oastate[\tau_{\Lambda''}]\circ i_{\Lambda,\Lambda''}
&=
\oastate[\tau_{\Lambda}].
\end{aligned}$$
Hence the local states satisfying
$\oastate[\tau_{\Lambda''}]\circ i_{\Lambda,\Lambda''}=\oastate[\tau_{\Lambda}]$ define a state $\oastate[\tau]$ on
$\oa{B}_{\mathrm{ql}}$.
If $\oastate[\psi_{1}],\oastate[\psi_{2}]\in\oaspstate_{\oa{B}_{\mathrm{ql}}}$ and
$\oastate[\tau]=t\oastate[\psi_{1}]+\rbk{1-t}\oastate[\psi_{2}]$ with $0<t<1$, purity on every local algebra
gives
$$\begin{aligned}
\oastate[\psi_{1}]\circ i_{\Lambda,X}
&=
\oastate[\psi_{2}]\circ i_{\Lambda,X}
=
\oastate[\tau_{\Lambda}].
\end{aligned}$$
Density of the local union gives $\oastate[\psi_{1}]=\oastate[\psi_{2}]=\oastate[\tau]$.
This equality proves that $\oastate[\tau]$ is pure.
The weak-$\ast$ compact convex extension set
$$\begin{aligned}
\fun{\operatorname{Ext}}{\oastate[\tau]}
&=
\set{\oastate[\psi]\in\oaspstate_{\fun{\oa{B}}{\sphilbfrak{h}}}}
{\fnrestr{\oastate[\psi]}{\oa{B}_{\mathrm{ql}}}=\oastate[\tau]}
\end{aligned}$$
is nonempty by the norm-preserving extension and positivity argument in the proof of
Theorem \ref{thm:paper-locally-normal-state-classification}.
Define the restriction map by
$$\begin{aligned}
\operatorname{res}_{\mathrm{ql}}
\colon
\oaspstate_{\fun{\oa{B}}{\sphilbfrak{h}}}
\to
\oaspstate_{\oa{B}_{\mathrm{ql}}};\quad
\fun{\operatorname{res}_{\mathrm{ql}}}{\oastate[\psi]}
&=
\fnrestr{\oastate[\psi]}{\oa{B}_{\mathrm{ql}}}.
\end{aligned}$$
It is weak-$\ast$ continuous, and hence
$$\begin{aligned}
\fun{\operatorname{Ext}}{\oastate[\tau]}
&=
\fun{\operatorname{res}_{\mathrm{ql}}^{-1}}{\setone{\oastate[\tau]}}
\end{aligned}$$
is a nonempty weak-$\ast$ compact convex set by the Banach--Alaoglu theorem.
By the Krein--Milman theorem, it has an extreme point $\oastate[\psi]$.
Any convex decomposition of $\oastate[\psi]$ restricts to a decomposition of the pure state
$\oastate[\tau]$, so both summands remain in this extension set.
Extremality then makes both summands equal to $\oastate[\psi]$; hence $\oastate[\psi]$ is pure.

For $N\in\monnat$, choose $\Lambda$ with $\abscard{\Lambda}>N$.
Let $D_N\in\oa{B}_{\Lambda}$ be the image under $i_{\Lambda,X}$ of the block family whose
$\abscard{\Lambda}$th component is
$\fun{\pi_{\sphilbfrak{h}_{\Lambda},\abscard{\Lambda}}^{-1}}{
P_{\Psi_{\Lambda},\Psi_{\Lambda}}}$ and whose other components vanish.
The block action in \eqref{eq:paper-global-local-inclusion-block-actions} and the local vector-state formula give
$$\begin{aligned}
\fun{p_N}{D_N}
&=
0,
&
\fun{\oastate[\psi]}{D_N}
&=
\fun{\oastate[\tau]}{D_N}
=
1.
\end{aligned}$$
Thus $D_N\in\Ker p_N$ but $\fun{\pi_{\psi}}{D_N}\neq0$.
\end{proof}

\begin{prop}[State oscillation in the Banach quotient]\label{prop:paper-state-oscillation}
Let $\oa{A}$ be a unital $\oacstar$-algebra and let $\oa{C}\subset\oa{A}$ be a unital
$\oacstar$-subalgebra.
For $x=\faadj{x}\in\oa{A}$ and $c\in\oa{C}$, put
$c_{\mathrm{sa}}=\rbk{c+\faadj{c}}/2$.
Then
$$\begin{aligned}
\norm{x-c_{\mathrm{sa}}}
&\leq
\norm{x-c},
&
\fun{\topdist}{x,\oa{C}}
&=
\fun{\topdist}{x,\oa{C}_{\mathrm{sa}}}.
\end{aligned}$$
The state-oscillation formula is
\begin{equation}
\begin{aligned}
2\fun{\topdist}{x,\oa{C}_{\mathrm{sa}}}
&=
\sup\set{
\abs{\fun{\oastate[\psi_{1}]}{x}-\fun{\oastate[\psi_{2}]}{x}}
}{
\oastate[\psi_{1}],\oastate[\psi_{2}]\in\oaspstate_{\oa{A}},
\fnrestr{\oastate[\psi_{1}]}{\oa{C}}
=
\fnrestr{\oastate[\psi_{2}]}{\oa{C}}}.
\end{aligned}
\label{eq:paper-state-oscillation}
\end{equation}
\end{prop}

\begin{proof}
For a real Banach space $E$, a closed subspace $F\subset E$, and $x\in E$, the real Hahn--Banach theorem gives
$$\begin{aligned}
\fun{\topdist}{x,F}
&=
\sup\set{\abs{\fun{f}{x}}}
{f\in\dual{E},\norm{f}\leq1,\fnrestr{f}{F}=0}.
\end{aligned}$$
If $x\in F$, both sides vanish.
Suppose $x\notin F$.
Define
$$\begin{aligned}
g
\colon
\setquot{\splinspan_{\fldreal}\setone{x}+F}{F}
\to
\fldreal;\quad
\fun{g}{\lambda x+F}
&=
\lambda\fun{\topdist}{x,F}.
\end{aligned}$$
For every $\lambda\in\fldreal$,
$$\begin{aligned}
\norm{\lambda x+F}_{\setquot{E}{F}}
&=
\abs{\lambda}\fun{\topdist}{x,F}
=
\abs{\fun{g}{\lambda x+F}},
\end{aligned}$$
so $\norm{g}=1$.
A norm-preserving Hahn--Banach extension of $g$ to $\setquot{E}{F}$, composed with the quotient map,
annihilates $F$, has norm one, and takes the value $\fun{\topdist}{x,F}$ at $x$.
Conversely, if
$\norm{f}\leq1$ and $\fnrestr{f}{F}=0$, then for every $y\in F$,
$$\begin{aligned}
\abs{\fun{f}{x}}
&=
\abs{\fun{f}{x-y}}
\leq
\norm{x-y}.
\end{aligned}$$
Taking the infimum over $y\in F$ proves the dual formula.

Apply it to $E=\oa{A}_{\mathrm{sa}}$ and $F=\oa{C}_{\mathrm{sa}}$, and extend each real functional to a
self-adjoint complex linear functional.
Then
$$\begin{aligned}
\fun{\topdist}{x,\oa{C}_{\mathrm{sa}}}
&=
\sup_{\substack{\fnrestr{\delta}{\oa{C}_{\mathrm{sa}}}=0\\\norm{\delta}\leq1}}
\abs{\fun{\delta}{x}}.
\end{aligned}$$
Every self-adjoint bounded functional has the unique Jordan decomposition
\cite[Theorem 4.3.6]{KadisonRingrose003}
$$\begin{aligned}
\delta
&=
\delta_+-\delta_-,
&
\norm{\delta}
&=
\norm{\delta_+}+\norm{\delta_-}.
\end{aligned}$$
Since $1\in\oa{C}$ and $\fun{\delta}{1}=0$,
$$\begin{aligned}
\fun{\delta_+}{1}
&=
\fun{\delta_-}{1}
=
\frac{\norm{\delta}}2.
\end{aligned}$$
The assertion is immediate for $\delta=0$.
If $0<\norm{\delta}\leq1$, set
$\widehat{\delta}=\delta/\norm{\delta}$.
Its Jordan parts satisfy
$$\begin{aligned}
\fnrestr{\widehat{\delta}_+}{\oa{C}_{\mathrm{sa}}}
&=
\fnrestr{\widehat{\delta}_-}{\oa{C}_{\mathrm{sa}}}.
\end{aligned}$$
Their complex-linear restrictions agree on $\oa{C}$.
Thus $\oastate[\psi_{1}]=2\widehat{\delta}_+$ and $\oastate[\psi_{2}]=2\widehat{\delta}_-$ are states
with equal restrictions, and
$$\begin{aligned}
\abs{\fun{\oastate[\psi_{1}]}{x}-\fun{\oastate[\psi_{2}]}{x}}
&=
2\abs{\fun{\widehat{\delta}}{x}},
\\
2\abs{\fun{\delta}{x}}
&=
\norm{\delta}
\abs{\fun{\oastate[\psi_{1}]}{x}-\fun{\oastate[\psi_{2}]}{x}}
\leq
\abs{\fun{\oastate[\psi_{1}]}{x}-\fun{\oastate[\psi_{2}]}{x}}.
\end{aligned}$$
This gives one inequality.
For the reverse inequality, any two states with equal restrictions satisfy
$$\begin{aligned}
\abs{\fun{\oastate[\psi_{1}]}{x}-\fun{\oastate[\psi_{2}]}{x}}
\leq
\norm{\oastate[\psi_{1}]-\oastate[\psi_{2}]}
\fun{\topdist}{x,\oa{C}_{\mathrm{sa}}}
\leq
2\fun{\topdist}{x,\oa{C}_{\mathrm{sa}}}.
\end{aligned}$$
\end{proof}

Apply Proposition \ref{prop:paper-state-oscillation} to \(\oa{A}=\fun{\oa{B}}{\sphilbfrak{h}}\) and \(\oa{C}=\oa{B}_{\mathrm{ql}}\). For every \(A=\faadj{A}\in\fun{\oa{B}}{\sphilbfrak{h}}\) this gives \[\begin{aligned}
\fun{\topdist}{A,\oa{B}_{\mathrm{ql}}}
&=
\frac12
\sup\set{
\abs{\fun{\oastate[\psi_{1}]}{A}-\fun{\oastate[\psi_{2}]}{A}}
}{
\oastate[\psi_{1}],\oastate[\psi_{2}]\in\oaspstate_{\fun{\oa{B}}{\sphilbfrak{h}}},
\fnrestr{\oastate[\psi_{1}]}{\oa{B}_{\mathrm{ql}}}
=
\fnrestr{\oastate[\psi_{2}]}{\oa{B}_{\mathrm{ql}}}}.
\end{aligned}\] For the coordinate map \(p_N\) in \eqref{eq:paper-main-coordinate-map}, one has \(\fun{p_N}{\oa{B}_{\mathrm{ql}}}=\oa{B}_N\): finite-support local elements have dense \(N\)th coordinate range, and the range of the restricted \(\ast\)-homomorphism is closed. Hence a nonzero bounded functional \(\delta\colon\fun{\oa{B}}{\sphilbfrak{h}}\to\fldcmp\) annihilating \(\oa{B}_{\mathrm{ql}}\) cannot be continuous for any coordinate seminorm \[\begin{aligned}
\fun{\oa{B}}{\sphilbfrak{h}}
\to
\fldrealnn;\quad
A
&\mapsto
\norm{\fun{p_N}{A}}.
\end{aligned}\] Assume that \(\delta\) is continuous for this coordinate seminorm. Then continuity forces \(\fnrestr{\delta}{\Ker p_N}=0\), and \(\delta\) factors through \(\oa{B}_N\). For each \(A\), choose \(D\in\oa{B}_{\mathrm{ql}}\) with \(\fun{p_N}{D}=\fun{p_N}{A}\); then \(\fun{\delta}{A}=\fun{\delta}{D}=0\), a contradiction. Consequently, for every \(N\in\monnat\), \[\begin{aligned}
\fnrestr{\delta}{\Ker p_N}
&\neq
0.
\end{aligned}\]

\subsubsection{One-mode particle-number resolvent and particle-degree coefficients}\label{sec:paper-one-mode-resolvent}

In this subsubsection, the Fock-space annihilation operators of Section \ref{sec:paper-bosonic-fock-representations} produce, for each unit vector \(f\in\sphilbfrak{h}\), a bounded operator with explicit particle-degree coefficients, spectrum, and norm.

Use the finite-particle subspace and the annihilation operators defined in Section \ref{sec:paper-bosonic-fock-representations}. Their canonical commutation relation on the common finite-particle domain is \[\begin{aligned}
\commutator{\fun{a_{\txtfock}}{f}}{\faadj{\fun{a_{\txtfock}}{g}}}
&=
\bkt{f}{g}_{\sphilbfrak{h}}.
\end{aligned}\] For a unit vector \(f\in\sphilbfrak{h}\), define \[\begin{aligned}
Q_f
&=
P_{f,f}
\in
\opspbddlin{\sphilbfrak{h}},
\\
N_f
\colon
\fun{\spfockgamma_{\txtbsn,\txtfin}}{\sphilbfrak{h}}
\to
\fun{\spfockgamma_{\txtbsn,\txtfin}}{\sphilbfrak{h}};\quad
\fun{N_f}{\xi}
&=
\faadj{\fun{a_{\txtfock}}{f}}\fun{a_{\txtfock}}{f}\xi.
\end{aligned}\] Under the orthogonal decomposition \(\sphilbfrak{h}=\fldcmp f\oplus\setone{f}^{\perp}\), \[\begin{aligned}
\fun{\spfockgamma_{\txtbsn}}{\sphilbfrak{h}}
&\simeq
\bigoplus_{k,\ell\in\monnat}
\rbk{\fldcmp f}^{\otimes_{\txtsym}k}
\otimes
\rbk{\setone{f}^{\perp}}^{\otimes_{\txtsym}\ell},
&
\fnrestr{N_f}{
\rbk{\fldcmp f}^{\otimes_{\txtsym}k}
\otimes
\rbk{\setone{f}^{\perp}}^{\otimes_{\txtsym}\ell}
}
&=
k.
\end{aligned}\] The domain of \(N_f\) is \[\begin{aligned}
\dom N_f
&=
\set{\seq{\xi_{k,\ell}}{k,\ell\in\monnat}}
{\begin{gathered}
\xi_{k,\ell}
\in
\rbk{\fldcmp f}^{\otimes_{\txtsym}k}
\otimes
\rbk{\setone{f}^{\perp}}^{\otimes_{\txtsym}\ell},
\\
\sum_{k,\ell\in\monnat}\rbk{1+k^2}\norm{\xi_{k,\ell}}^2
<
\infty
\end{gathered}}.
\end{aligned}\] For \(\xi=\bigoplus_{k,\ell\in\monnat}\xi_{k,\ell}\in\dom N_f\), define the sequence \(\seq{\xi_M}{M\in\monnat}\) by \[\begin{aligned}
\xi_M
&=
\bigoplus_{\substack{k,\ell\in\monnat\\k+\ell\leq M}}\xi_{k,\ell}.
\end{aligned}\] The approximation satisfies \[\begin{aligned}
\norm{\xi-\xi_M}^2
+
\norm{N_f\xi-N_f\xi_M}^2
&=
\sum_{k+\ell>M}\rbk{1+k^2}\norm{\xi_{k,\ell}}^2
\to
0
&&\rbk{M\to\infty}.
\end{aligned}\] Thus the finite-particle operator is essentially self-adjoint; write \(N_f\) also for its nonnegative closure and put \[\begin{aligned}
A_{\txtfock,f}
&=
\rbk{1+N_f}^{-1}
\in
\opspbddlin{\fun{\spfockgamma_{\txtbsn}}{\sphilbfrak{h}}}.
\end{aligned}\]

\begin{prop}[Particle-degree coefficients of the one-mode resolvent]\label{prop:paper-rank-one-particle-coefficient}
The $m$th coefficient of $A_{\txtfock,f}$ is
\begin{equation}
\begin{aligned}
a_m
\colon
\set{f\in\sphilbfrak{h}}{\norm{f}=1}
\to
\fun{\schattencls}{\sphilbfrak{h}^{\otimes_{\txtsym}m}};\quad
\fun{a_m}{f}
&=
\frac{\rbk{-1}^m}{m+1}
P_{f^{\otimes m},f^{\otimes m}}.
\end{aligned}
\label{eq:paper-rank-one-particle-coefficient}
\end{equation}
The coefficient norms sum to
$$\begin{aligned}
\sum_{m=0}^{n}\norm{\fun{\widetilde{E}_{m,n}\circ a_m}{f}}
&=
\frac{2^{n+1}-1}{n+1}.
\end{aligned}$$
The represented resolvent has norm
$$\begin{aligned}
\norm{A_{\txtfock,f}}
&=
1.
\end{aligned}$$
\end{prop}

\begin{proof}
On the $n$-particle subspace with exactly $k$ particles in the mode $f$,
$$\begin{aligned}
\fun{\widetilde{E}_{m,n}}{Q_f^{\otimes m}}
&=
\sum_{\substack{S\subset\setone{1,\ldots,n}\\\abscard{S}=m}}
\rbk{Q_f^{\otimes m}}^{[S]},
&
\fnrestr{\fun{\widetilde{E}_{m,n}}{Q_f^{\otimes m}}}{N_f=k}
&=
\binom{k}{m}1.
\end{aligned}$$
With $\binom{k}{m}=0$ for $m>k$,
$$\begin{aligned}
\frac{1}{k+1}
&=
\int_0^1\rbk{1-t}^{k}\opdmsr{t}
=
\sum_{m=0}^{k}\frac{\rbk{-1}^m}{m+1}\binom{k}{m}.
\end{aligned}$$
Hence, on every $n$-particle sector,
$$\begin{aligned}
\fnrestr{\rbk{1+N_f}^{-1}}{N_f=k}
&=
\fnrestr{
\sum_{m=0}^{n}\frac{\rbk{-1}^m}{m+1}
\fun{\widetilde{E}_{m,n}}{Q_f^{\otimes m}}
}{N_f=k}.
\end{aligned}$$
The eigenspaces sum to the whole sector, so uniqueness of particle-degree coefficients proves
\eqref{eq:paper-rank-one-particle-coefficient}.
Since
$\norm{\fun{\widetilde{E}_{m,n}}{Q_f^{\otimes m}}}=\binom{n}{m}$,
$$\begin{aligned}
\sum_{m=0}^{n}\frac{1}{m+1}\binom{n}{m}
&=
\frac{1}{n+1}\sum_{m=0}^{n}\binom{n+1}{m+1}
=
\frac{2^{n+1}-1}{n+1}.
\end{aligned}$$
The spectrum and norm of $A_{\txtfock,f}$ are
$$\begin{aligned}
\opvarspec{A_{\txtfock,f}}
&=
\setone{\frac{1}{k+1}:k\in\monnat}\cup\setone{0},
&
\norm{A_{\txtfock,f}}
&=
1.
\end{aligned}$$
\end{proof}

Return the concrete coefficients to the intrinsic blocks by \[\begin{aligned}
\mathfrak{a}
\colon
\set{f\in\sphilbfrak{h}}{\norm{f}=1}
\to
\oa{C}(\sphilbfrak{h});\quad
\fun{\mathfrak{a}}{f}
&=
\seq{\fun{\Theta_m^{-1}\circ a_m}{f}}{m\in\monnat}.
\end{aligned}\] The binomial identity gives \[\begin{aligned}
\fun{\Pi_n\circ\mathfrak{a}}{f}
&=
\fnrestr{A_{\txtfock,f}}{\sphilbfrak{h}^{\otimes_{\txtsym}n}},
&
\sup_{n\in\monnat}\norm{\fun{\Pi_n\circ\mathfrak{a}}{f}}
&=
1.
\end{aligned}\] Thus \(\fun{\mathfrak{a}}{f}\in\oa{C}_{\mathrm{b}}(\sphilbfrak{h})\). The map \(\Upsilon\) in \eqref{eq:paper-main-particle-coefficient-isomorphism} gives \[\begin{aligned}
A_f
&=
\fun{\Upsilon\circ\mathfrak{a}}{f}
\in
\fun{\oa{B}}{\sphilbfrak{h}}.
\end{aligned}\] Its comparison representative is \(A_{\txtfock,f}\).

\subsubsection{Escaping-particle states and local approximation distance}\label{sec:paper-escaping-particle-states}

In this subsubsection, vector states with particles in the orthogonal complements of finite basis subspaces measure the distance of one-mode resolvents from the quasilocal algebra. They also distinguish the two orders of finite-coordinate and quasilocal approximation in Theorem \ref{thm:paper-rank-one-tail-geometry}.

Let the unit vector \(f\in\sphilbfrak{h}\) have infinite coordinate support. Direct the nonempty finite subsets of \(X\) by inclusion, and write \(\Lambda\nearrow X\) for this net direction. For a nonempty \(\Lambda\in\setpowfin{X}\) and \(n\in\monnat\), define \[\begin{aligned}
p_{\Lambda}
&=
\norm{\rbk{1-P_{\Lambda}}f}^2,
&
e_{\Lambda}
&=
\frac{\rbk{1-P_{\Lambda}}f}{\sqrt{p_{\Lambda}}}.
\end{aligned}\] Infinite coordinate support gives \(p_{\Lambda}>0\) for every finite \(\Lambda\), and completeness of the basis gives \(P_{\Lambda}f\to f\) and \(p_{\Lambda}\to0\) as \(\Lambda\nearrow X\). Define the scalar function \[\begin{aligned}
\eta_n
\colon
\set{p\in\fldreal}{0<p\leq1}
\to
\fldreal;\quad
\fun{\eta_n}{p}
&=
\frac{1-\rbk{1-p}^{n+1}}{\rbk{n+1}p}.
\end{aligned}\] Define the vacuum state by \begin{equation}
\begin{aligned}
\oastate[\psi_{\mathrm{vac}}]
\colon
\fun{\oa{B}}{\sphilbfrak{h}}
\to
\fldcmp;\quad
\fun{\oastate[\psi_{\mathrm{vac}}]}{A}
&=
\bkt{\Psi_{\txtfock,\txtvac}}{\fun{\oarepn_{\txtfock}}{A}\Psi_{\txtfock,\txtvac}}_{\fun{\spfockgamma_{\txtbsn}}{\sphilbfrak{h}}}.
\end{aligned}
\label{eq:paper-vacuum-state}
\end{equation} The \(N_f\)-distribution in \(e_{\Lambda}^{\otimes n}\) is binomial, and therefore \begin{equation}
\begin{aligned}
\bkt{e_{\Lambda}^{\otimes n}}{\fun{p_n}{A_f}e_{\Lambda}^{\otimes n}}
&=
\sum_{k=0}^{n}\frac{1}{k+1}\binom{n}{k}p_{\Lambda}^k\rbk{1-p_{\Lambda}}^{n-k}
\\
&=
\int_0^1\rbk{1-p_{\Lambda}+p_{\Lambda} t}^{n}\opdmsr{t}
=
\frac{1-\rbk{1-p_{\Lambda}}^{n+1}}{\rbk{n+1}p_{\Lambda}}
=
\fun{\eta_n}{p_{\Lambda}}.
\end{aligned}
\label{eq:paper-binomial-resolvent-expectation}
\end{equation} Fix \(c>0\) and define \(n_{\Lambda}=\max\setone{0,\lfloor c/p_{\Lambda}\rfloor-1}\) for every nonempty finite \(\Lambda\). Since \(p_{\Lambda}\to0\) as \(\Lambda\nearrow X\), every sufficiently large \(\Lambda\) satisfies \(p_{\Lambda}<\min\setone{c,1}\) and \[\begin{aligned}
0
&\leq
c-\rbk{n_{\Lambda}+1}p_{\Lambda}
<
p_{\Lambda}.
\end{aligned}\] Hence \begin{equation}
\begin{aligned}
\rbk{n_{\Lambda}+1}p_{\Lambda}
&\to
c
&&\rbk{\Lambda\nearrow X}.
\end{aligned}
\label{eq:paper-escaping-particle-scaling}
\end{equation} For these sufficiently large regions, the scalar limit is \[\begin{aligned}
\rbk{n_{\Lambda}+1}\log\rbk{1-p_{\Lambda}}
&=
-\rbk{n_{\Lambda}+1}p_{\Lambda}
\frac{-\log\rbk{1-p_{\Lambda}}}{p_{\Lambda}}
\to
-c,
\end{aligned}\] and the definition of \(\eta_{n_{\Lambda}}\) gives \[\begin{aligned}
\fun{\eta_{n_{\Lambda}}}{p_{\Lambda}}
&\to
\frac{1-\napiernum^{-c}}{c}
&&\rbk{\Lambda\nearrow X}.
\end{aligned}\] Define the associated state net by \begin{equation}
\begin{aligned}
\oastate[\psi_{c,\Lambda}]
\colon
\fun{\oa{B}}{\sphilbfrak{h}}
\to
\fldcmp;\quad
\fun{\oastate[\psi_{c,\Lambda}]}{A}
&=
\bkt{e_{\Lambda}^{\otimes n_{\Lambda}}}{\fun{p_{n_{\Lambda}}}{A}e_{\Lambda}^{\otimes n_{\Lambda}}}.
\end{aligned}
\label{eq:paper-escaping-particle-state}
\end{equation} The Banach--Alaoglu theorem provides a convergent subnet, represented by a directed set \(I\) and a map \(I\ni i\mapsto \Lambda_i\in\setpowfin{X}\setminus\setone{\emptyset}\). The map is eventually cofinal: for every nonempty finite \(\Lambda\subset X\), there is \(i_{\Lambda}\in I\) for which \(\Lambda\subset \Lambda_i\) whenever \(i\geq i_{\Lambda}\). The net \(\net{\oastate[\psi_{c,\Lambda_i}]}{i\in I}\) is weak-\(\ast\) convergent. The limit state is \[\begin{aligned}
\oastate[\psi_c]
\colon
\fun{\oa{B}}{\sphilbfrak{h}}
\to
\fldcmp;\quad
\fun{\oastate[\psi_c]}{A}
&=
\lim_{i\in I}\fun{\oastate[\psi_{c,\Lambda_i}]}{A}.
\end{aligned}\] The state \(\oastate[\psi_c]\) agrees with the vacuum state in \eqref{eq:paper-vacuum-state} on every \(\oa{B}_{\Lambda'}\), because for \(\Lambda\supset \Lambda'\) all particles lie in \(\sphilbfrak{h}_{X\setminus\Lambda'}\) and \eqref{eq:paper-local-element-external-only-block} gives \[\begin{aligned}
\fun{\oastate[\psi_{c,\Lambda}]}{D}
&=
\bkt{\Psi_{\txtfock,\txtvac}}{\fun{\oarepn_{\txtfock}}{D}\Psi_{\txtfock,\txtvac}}
=
\fun{\oastate[\psi_{\mathrm{vac}}]}{D}
\end{aligned}\] for \(D\in\oa{B}_{\Lambda'}\).

For every nonempty \(\Lambda\in\setpowfin{X}\), define \[\begin{aligned}
d_{\Lambda,n}
\colon
\fun{\oa{B}}{\sphilbfrak{h}}
\to
\fldrealnn;\quad
\fun{d_{\Lambda,n}}{A}
&=
\fun{\topdist}{\fun{p_n}{A},\fun{p_n}{\oa{B}_{\Lambda}}}.
\end{aligned}\] If \(P_{\Lambda} f\neq0\), put \(f_{\Lambda}=P_{\Lambda} f/\norm{P_{\Lambda} f}\).

\begin{thm}[Quantitative nonlocality of the one-mode resolvent]\label{thm:paper-rank-one-tail-geometry}
For every infinite-support unit vector $f$,
$$\begin{aligned}
\norm{\sqbk{A_f}}_{\oa{T}_{\mathrm{B}}}
&=
\frac12.
\end{aligned}$$
For every nonempty $\Lambda\in\setpowfin{X}$ with $P_{\Lambda} f\neq0$, the norm of the finite-mode approximation error
is
$$\begin{aligned}
\norm{A_f-A_{f_{\Lambda}}}
&=
1.
\end{aligned}$$
The finite-coordinate minimax value is
$$\begin{aligned}
\sup_{n\in\monnat}
\inf_{\substack{\Lambda\in\setpowfin{X}\\\Lambda\neq\emptyset}}
\fun{d_{\Lambda,n}}{A_f}
&=
0.
\end{aligned}$$
The quasilocal minimax value is
$$\begin{aligned}
\inf_{\substack{\Lambda\in\setpowfin{X}\\\Lambda\neq\emptyset}}
\sup_{n\in\monnat}\fun{d_{\Lambda,n}}{A_f}
&=
\frac12.
\end{aligned}$$
\end{thm}

\begin{proof}
The map $\Upsilon$ in \eqref{eq:paper-main-particle-coefficient-isomorphism} gives
$A_f=\fun{\Upsilon\circ\mathfrak{a}}{f}$, whose comparison representative is $A_{\txtfock,f}$.
Proposition \ref{prop:paper-rank-one-particle-coefficient} and faithfulness of the comparison representation give
$$\begin{aligned}
\opvarspec{A_f}
&=
\setone{\frac{1}{k+1}:k\in\monnat}\cup\setone{0}.
\end{aligned}$$
The diagonal action of the comparison representative therefore gives
$$\begin{aligned}
\norm{A_f-\frac{1}{2}}
&=
\max\setone{
\sup_{k\in\monnat}\abs{\frac{1}{k+1}-\frac12},
\frac12}
=
\frac12.
\end{aligned}$$
This equality gives the upper tail bound.
Choose
$n_{\Lambda}=\lfloor p_{\Lambda}^{-2}\rfloor$ and $\xi_{\Lambda}=e_{\Lambda}^{\otimes n_{\Lambda}}$.
Define the state net indexed by
$\Lambda\nearrow X$ by
$$\begin{aligned}
\oastate[\psi_{\infty,\Lambda}]
\colon
\fun{\oa{B}}{\sphilbfrak{h}}
\to
\fldcmp;\quad
\fun{\oastate[\psi_{\infty,\Lambda}]}{A}
&=
\bkt{\xi_{\Lambda}}{\fun{p_{n_{\Lambda}}}{A}\xi_{\Lambda}}.
\end{aligned}$$
The Banach--Alaoglu theorem gives a convergent subnet, represented by a directed set $I$ and a map
$I\ni i\mapsto \Lambda_i\in\setpowfin{X}\setminus\setone{\emptyset}$ such that
for every nonempty finite $\Lambda\subset X$, there is $i_{\Lambda}\in I$ for
which $\Lambda\subset \Lambda_i$ whenever $i\geq i_{\Lambda}$.
The net $\net{\oastate[\psi_{\infty,\Lambda_i}]}{i\in I}$ converges in the weak-$\ast$ topology.
Its limit state is
$$\begin{aligned}
\oastate[\psi_\infty]
\colon
\fun{\oa{B}}{\sphilbfrak{h}}
\to
\fldcmp;\quad
\fun{\oastate[\psi_\infty]}{A}
&=
\lim_{i\in I}\fun{\oastate[\psi_{\infty,\Lambda_i}]}{A}.
\end{aligned}$$
Then \eqref{eq:paper-binomial-resolvent-expectation} gives
$$\begin{aligned}
\fun{\oastate[\psi_\infty]}{A_f}
&=
0,
&
\fun{\oastate[\psi_{\mathrm{vac}}]}{A_f}
&=
1.
\end{aligned}$$
For a finite index set $\Lambda'$ and $D\in\oa{B}_{\Lambda'}$, sufficiently large $\Lambda\supset \Lambda'$ satisfy
$\xi_{\Lambda}\in \sphilbfrak{h}_{X\setminus\Lambda'}^{\otimes_{\txtsym}n_{\Lambda}}$.
Equation
\eqref{eq:paper-local-element-external-only-block} gives
$$\begin{aligned}
\fun{\oastate[\psi_{\infty,\Lambda}]}{D}
&=
d_D
=
\fun{\oastate[\psi_{\mathrm{vac}}]}{D}.
\end{aligned}$$
Taking the limit along $i\in I$ in this local equality shows that
$\oastate[\psi_\infty]$ and $\oastate[\psi_{\mathrm{vac}}]$ agree on every local algebra and on
$\oa{B}_{\mathrm{ql}}$.
Equation \eqref{eq:paper-state-oscillation} gives the lower tail bound.

Because $e_{\Lambda}\perp f_{\Lambda}$, for every $n\in\monnat$,
$$\begin{aligned}
\bkt{e_{\Lambda}^{\otimes n}}{\fun{p_n}{A_{f_{\Lambda}}}e_{\Lambda}^{\otimes n}}
&=
1,
&
\bkt{e_{\Lambda}^{\otimes n}}{\fun{p_n}{A_f}e_{\Lambda}^{\otimes n}}
&=
\fun{\eta_n}{p_{\Lambda}}.
\end{aligned}$$
The matrix coefficient $\fun{\eta_n}{p_{\Lambda}}$ tends to zero as $n\to\infty$.
Thus $\norm{A_f-A_{f_{\Lambda}}}\geq1$.
Since both operators are
positive contractions,
$$\begin{aligned}
-1
&\leq
A_f-A_{f_{\Lambda}}
\leq
1,
\end{aligned}$$
which gives the reverse bound.
For fixed $n$,
$$\begin{aligned}
\norm{\fun{p_n}{A_{f_{\Lambda}}-A_f}}
&\leq
\norm{\fnrestr{N_{f_{\Lambda}}-N_f}{\sphilbfrak{h}^{\otimes_{\txtsym}n}}}
\leq
n\norm{Q_{f_{\Lambda}}-Q_f}
\leq
2n\norm{f_{\Lambda}-f}
\to
0
&&\rbk{\Lambda\nearrow X},
\end{aligned}$$
which proves the first minimax identity.

For $D\in\oa{B}_{\Lambda}$, equation \eqref{eq:paper-local-element-external-only-block} gives
$$\begin{aligned}
\fnrestr{\fun{p_n}{D}}{\sphilbfrak{h}_{X\setminus\Lambda}^{\otimes_{\txtsym}n}}
&=
d_D.
\end{aligned}$$
Choose a unit vector $v_{\Lambda}\in \sphilbfrak{h}_{X\setminus\Lambda}\cap\setone{e_{\Lambda}}^{\perp}$.
Then
$$\begin{aligned}
\bkt{v_{\Lambda}^{\otimes n}}{\fun{p_n}{A_f}v_{\Lambda}^{\otimes n}}
&=
1,
&
\bkt{e_{\Lambda}^{\otimes n}}{\fun{p_n}{A_f}e_{\Lambda}^{\otimes n}}
&=
\fun{\eta_n}{p_{\Lambda}}.
\end{aligned}$$
The quotient-norm distance satisfies
$$\begin{aligned}
\norm{\fun{p_n}{A_f}-\fun{p_n}{D}}
&\geq
\max\setone{\abs{1-d_D},\abs{\fun{\eta_n}{p_{\Lambda}}-d_D}}
\geq
\frac{1-\fun{\eta_n}{p_{\Lambda}}}2.
\end{aligned}$$
Taking $n\to\infty$ gives the lower bound $1/2$, and approximation by the scalar $1/2$ gives the upper bound.
\end{proof}

\begin{cor}[Coefficient norms do not detect finite-mode locality]\label{cor:paper-coefficient-norms-do-not-detect-locality}
No condition depending only on $\seq{\norm{a_m}}{m\in\monnat}$ characterizes
$\oa{B}_{\mathrm{ql}}$.
\end{cor}

\begin{proof}
For a finite-support unit vector $f$ and an infinite-support unit vector $g$,
$$\begin{aligned}
\norm{\fun{a_m}{f}}
&=
\frac1{m+1}
=
\norm{\fun{a_m}{g}},
&
A_f
&\in
\oa{B}_{\mathrm{ql}},
&
\norm{\sqbk{A_g}}_{\oa{T}_{\mathrm{B}}}
&=
\frac12.
\end{aligned}$$
\end{proof}

For a nonempty \(\Lambda\in\setpowfin{X}\), use the projection \(P_{\Lambda}\) of \eqref{eq:paper-basis-localization-data} and the convention \(P_{\Lambda}^{\otimes_{\txtsym}0}=1_{\fldcmp}\), and put \[\begin{aligned}
\oa{C}_{\mathrm{b}}(\Lambda)
&=
\set{b=\seq{b_m}{m\in\monnat}\in\oa{C}_{\mathrm{b}}(\sphilbfrak{h})}
{\fun{\Theta_m}{b_m}=P_{\Lambda}^{\otimes_{\txtsym}m}\fun{\Theta_m}{b_m}P_{\Lambda}^{\otimes_{\txtsym}m}\text{ for every }m}.
\end{aligned}\]

\begin{thm}[Coefficient criterion for quasilocality]\label{thm:paper-particle-coefficient-quasilocality-criterion}
For the map $\Upsilon$ in \eqref{eq:paper-main-particle-coefficient-isomorphism} and every nonempty
$\Lambda\in\setpowfin{X}$, one has
$\fun{\Upsilon}{\oa{C}_{\mathrm{b}}(\Lambda)}=\oa{B}_{\Lambda}$.
Consequently,
$\fun{\Upsilon}{a}\in\oa{B}_{\mathrm{ql}}$ if and only if
\begin{equation}
\begin{aligned}
\lim_{\Lambda\nearrow X}
\inf_{b\in\oa{C}_{\mathrm{b}}(\Lambda)}
\sup_{n\in\monnat}
\norm{\sum_{m=0}^{n}\fun{E_{m,n}}{a_m-b_m}}_{\oa{B}_n}
&=
0,
\end{aligned}
\label{eq:paper-particle-coefficient-quasilocality-criterion}
\end{equation}
where $\Lambda\nearrow X$ is the directed limit over the nonempty finite subsets of $X$.
\end{thm}

\begin{proof}
Take $D=\seq{D_k}{k\in\monnat}\in\oa{M}(\oa{K}(\sphilbfrak{h}_{\Lambda})^{\gamma})$, and let
$\seq{c_m}{m\in\monnat}$ be its finite-mode particle-degree coefficients.
The global inclusion extends each $c_m$ by
zero to $\sphilbfrak{h}^{\otimes_{\txtsym}m}$, so for every $m\in\monnat$,
$$\begin{aligned}
c_m
&=
P_{\Lambda}^{\otimes_{\txtsym}m}c_mP_{\Lambda}^{\otimes_{\txtsym}m}.
\end{aligned}$$
Uniqueness of the particle-coefficient reconstruction gives coefficients $b_m$ of $\fun{i_{\Lambda,X}}{D}$ satisfying
$\fun{\Theta_m}{b_m}=c_m$.
Therefore,
$$\begin{aligned}
\oa{B}_{\Lambda}
&\subset
\fun{\Upsilon}{\oa{C}_{\mathrm{b}}(\Lambda)}.
\end{aligned}$$
Conversely, fix $b\in\oa{C}_{\mathrm{b}}(\Lambda)$, and define
$$\begin{aligned}
K_m
&=
\fnrestr{\fun{\Theta_m}{b_m}}{\sphilbfrak{h}_{\Lambda}^{\otimes_{\txtsym}m}}.
\end{aligned}$$
Then define the map
$$\begin{aligned}
\widetilde{E}_{\sphilbfrak{h}_{\Lambda},m,k}
\colon
\opspbddlin{\sphilbfrak{h}_{\Lambda}^{\otimes_{\txtsym}m}}
\to
\opspbddlin{\sphilbfrak{h}_{\Lambda}^{\otimes_{\txtsym}k}};\quad
\fun{\widetilde{E}_{\sphilbfrak{h}_{\Lambda},m,k}}{K}
&=
\sum_{\substack{S\subset\setone{1,\ldots,k}\\\abscard{S}=m}}K^{[S]}
\end{aligned}$$
and the elements $$\begin{aligned}
D_k
=
\fun{\pi_{\sphilbfrak{h}_{\Lambda},k}^{-1}}{\sum_{m=0}^{k}\fun{\widetilde{E}_{\sphilbfrak{h}_{\Lambda},m,k}}{K_m}}.
\end{aligned}$$
The comparison identity gives
$$\begin{aligned}
\norm{D_k}
&\leq
\norm{\fun{\Pi_k}{b}}
\leq
\norm{b}_{\mathrm{b}},
&
\fun{p_n\circ i_{\Lambda,X}}{\seq{D_k}{k\in\monnat}}
&=
\fun{p_n\circ\Upsilon}{b}.
\end{aligned}$$
Thus $\seq{D_k}{k\in\monnat}\in\oa{M}(\oa{K}(\sphilbfrak{h}_{\Lambda})^{\gamma})$ and the reverse inclusion follows.
Finally, the
particle-coefficient isometry gives
$$\begin{aligned}
\norm{a-b}_{\mathrm{b}}
&=
\sup_{n\in\monnat}\norm{\sum_{m=0}^{n}\fun{E_{m,n}}{a_m-b_m}}_{\oa{B}_n}.
\end{aligned}$$
Taking the distance from the increasing union of the local coefficient spaces proves
\eqref{eq:paper-particle-coefficient-quasilocality-criterion}.
\end{proof}

For \(d\in\semigrposint\) and a sequence \(\seq{f_j}{j\in\monnat}\) of unit vectors in \(\sphilbfrak{h}\), define the map into the Banach-space quotient and its entrywise matrix maps by \[\begin{aligned}
\Theta_{\mathrm{tail}}
\colon
\lpseq^1
\to
\oa{T}_{\mathrm{B}};\quad
\fun{\Theta_{\mathrm{tail}}}{\seq{\alpha_j}{j\in\monnat}}
&=
\sum_{j\in\monnat}\alpha_j\sqbk{A_{f_j}},
\\
\Theta_{\mathrm{tail},d}
\colon
\fun{M_d}{\lpseq^1}
\to
\fun{M_d}{\oa{T}_{\mathrm{B}}};\quad
\fun{\Theta_{\mathrm{tail},d}}{\sqbk{x_{k\ell}}}
&=
\sqbk{\fun{\Theta_{\mathrm{tail}}}{x_{k\ell}}}.
\end{aligned}\] For \(x_{k\ell}=\seq{x_{k\ell,j}}{j\in\monnat}\), equip the domain with \[\begin{aligned}
\norm{\sqbk{x_{k\ell}}}_{\lpseq^1\rbk{\fun{M_d}{\fldcmp}}}
&=
\sum_{j\in\monnat}\norm{\sqbk{x_{k\ell,j}}}_{\fun{M_d}{\fldcmp}},
\end{aligned}\] and the range with the quotient matrix norm \[\begin{aligned}
\norm{\sqbk{\sqbk{A_{k\ell}}}}_{\fun{M_d}{\oa{T}_{\mathrm{B}}}}
&=
\inf_{\sqbk{D_{k\ell}}\in\fun{M_d}{\oa{B}_{\mathrm{ql}}}}
\norm{\sqbk{A_{k\ell}-D_{k\ell}}}_{\fun{M_d}{\fun{\oa{B}}{\sphilbfrak{h}}}}.
\end{aligned}\]

\begin{prop}[Uniform matrix upper bound for the tail-sequence map]\label{prop:paper-tail-sequence-matrix-upper-bound}
For every $d\geq1$ and
$x=\sqbk{x_{k\ell}}_{1\leq k,\ell\leq d}\in\fun{M_d}{\lpseq^1}$, write
$$\begin{aligned}
x_{k\ell}
&=
\seq{x_{k\ell,j}}{j\in\monnat},
&
\Lambda_j
&=
\sqbk{x_{k\ell,j}}_{1\leq k,\ell\leq d}
\in
\fun{M_d}{\fldcmp}.
\end{aligned}$$
The tail map satisfies the estimate
$$\begin{aligned}
\norm{\fun{\Theta_{\mathrm{tail},d}}{x}}
&=
\norm{\sum_{j\in\monnat}\Lambda_j\otimes\sqbk{A_{f_j}}}
\leq
\sum_{j\in\monnat}\norm{\Lambda_j}\norm{A_{f_j}}
=
\norm{x}_{\lpseq^1\rbk{\fun{M_d}{\fldcmp}}}.
\end{aligned}$$
\end{prop}

\begin{proof}
The definitions of $\Theta_{\mathrm{tail},d}$ and $\Lambda_j$ give
$$\begin{aligned}
\fun{\Theta_{\mathrm{tail},d}}{x}
&=
\sum_{j\in\monnat}\Lambda_j\otimes\sqbk{A_{f_j}}.
\end{aligned}$$
For every $\sqbk{A_{k\ell}}\in\fun{M_d}{\fun{\oa{B}}{\sphilbfrak{h}}}$, choosing the zero representative in
the quotient-matrix norm gives
$$\begin{aligned}
\norm{\sqbk{\sqbk{A_{k\ell}}}}_{\fun{M_d}{\oa{T}_{\mathrm{B}}}}
&\leq
\norm{\sqbk{A_{k\ell}}}_{\fun{M_d}{\fun{\oa{B}}{\sphilbfrak{h}}}}.
\end{aligned}$$
Apply this inequality to every $\Lambda_j\otimes A_{f_j}$.
Since
$\norm{\Lambda_j\otimes A_{f_j}}=\norm{\Lambda_j}\norm{A_{f_j}}$ and $\norm{A_{f_j}}=1$,
$$\begin{aligned}
\norm{\fun{\Theta_{\mathrm{tail},d}}{x}}_{\fun{M_d}{\oa{T}_{\mathrm{B}}}}
\leq
\sum_{j\in\monnat}
\norm{\Lambda_j}_{\fun{M_d}{\fldcmp}}
\norm{\sqbk{A_{f_j}}}_{\oa{T}_{\mathrm{B}}}
\leq
\sum_{j\in\monnat}\norm{\Lambda_j}_{\fun{M_d}{\fldcmp}}
=
\norm{x}_{\lpseq^1\rbk{\fun{M_d}{\fldcmp}}}.
\end{aligned}$$
The series $\sum_{j\in\monnat}\norm{\Lambda_j}_{\fun{M_d}{\fldcmp}}$ converges by the definition of the domain norm.
\end{proof}

\begin{cor}[Failure of a uniform lower bound]\label{cor:paper-tail-sequence-lower-bound-counterexample}
For an arbitrary sequence of unit vectors $\seq{f_j}{j\in\monnat}$, the tail-sequence map need not be
injective and need not have a positive uniform lower bound.
\end{cor}

\begin{proof}
For an infinite-support unit vector $f$ and $0<\theta<2\pi$, choose
$f_0=f$, $f_1=\napiernum^{\imunit\theta}f$, and $u=(1,-1,0,\ldots)$.
Then
$$\begin{aligned}
\norm{u}_{\lpseq^1}
&=
2,
&
\fun{\Theta_{\mathrm{tail}}}{u}
&=
\sqbk{A_{f_0}}-\sqbk{A_{f_1}}
=
0.
\end{aligned}$$
\end{proof}

\subsubsection{Counterexamples for projectively germ-equivalent vectors}\label{sec:paper-projective-germ-counterexamples}

In this subsubsection, pairs of unit vectors whose components outside a finite basis subspace are proportional can define distinct quotient classes of one-mode resolvents. Asymptotic states also distinguish such classes when the external components agree exactly.

For infinite-support unit vectors \(f,g\in\sphilbfrak{h}\), call \(f\) and \(g\) projectively germ equivalent if there are \(\Lambda_0\in\setpowfin{X}\) and \(\lambda\in\fldcmp\setminus\setone{0}\) such that \[\begin{aligned}
\rbk{1-P_{\Lambda_0}}g
&=
\lambda\rbk{1-P_{\Lambda_0}}f.
\end{aligned}\]

\begin{thm}[Projectively germ-equivalent vectors can define different quotient classes]\label{thm:paper-projective-germ-counterexample}
There are projectively germ-equivalent unit vectors $f,g$ for which
$\sqbk{A_f}\neq\sqbk{A_g}$ in $\oa{T}_{\mathrm{B}}$.
\end{thm}

\begin{proof}
Choose a finite-support unit vector $e_0$ and an infinite-support unit vector $u\perp e_0$.
For
$0<b<1$ and $0<\rho<1$, put
$$\begin{aligned}
f
&=
\sqrt{1-b^2}e_0+bu,
&
g
&=
\sqrt{1-\rho^2b^2}e_0+\rho bu.
\end{aligned}$$
Outside a finite index set containing $\fun{\supp_{X}}{e_0}$, one has $g=\rho f$.
If $p_{\Lambda}=\norm{\rbk{1-P_{\Lambda}}f}^2$, then
$\norm{\rbk{1-P_{\Lambda}}g}^2=\rho^2p_{\Lambda}$, and the normalized vectors obtained from
$\rbk{1-P_{\Lambda}}f$ and $\rbk{1-P_{\Lambda}}g$ are equal.
Define the strictly decreasing function
$$\begin{aligned}
\eta
\colon
\set{c\in\fldreal}{c>0}
\to
\fldreal;\quad
\fun{\eta}{c}
=
\frac{1-\napiernum^{-c}}{c}
=
\int_0^1\napiernum^{-cs}\opdmsr{s}.
\end{aligned}$$
The derivative satisfies
$$\begin{aligned}
\fun{\eta'}{c}
=
-\int_0^1s\napiernum^{-cs}\opdmsr{s}
<
0.
\end{aligned}$$
Fix $c>0$.
For every nonempty finite $\Lambda$ containing $\fun{\supp_{X}}{e_0}$, set
$n_{\Lambda}=\max\setone{0,\lfloor c/p_{\Lambda}\rfloor-1}$.
Let $\oastate[\psi_{f,c,\Lambda}]$ be the state defined by the right-hand side of
\eqref{eq:paper-escaping-particle-state} with the present vectors $f$ and $e_{\Lambda}$.
The Banach--Alaoglu theorem provides a convergent subnet indexed by a directed set $I$.
For each $i\in I$, the associated $\Lambda_i$ is a finite subset of $X$ containing
$\fun{\supp_{X}}{e_0}$, and for every finite $\Lambda\supset\fun{\supp_{X}}{e_0}$, there is
$i_{\Lambda}\in I$ for which $\Lambda\subset \Lambda_i$ whenever $i\geq i_{\Lambda}$.
The net $\net{\oastate[\psi_{f,c,\Lambda_i}]}{i\in I}$ converges in the weak-$\ast$ topology.
Define the limit state by
$$\begin{aligned}
\oastate[\psi_{f,c}]
\colon
\fun{\oa{B}}{\sphilbfrak{h}}
\to
\fldcmp;\quad
\fun{\oastate[\psi_{f,c}]}{A}
&=
\lim_{i\in I}\fun{\oastate[\psi_{f,c,\Lambda_i}]}{A}.
\end{aligned}$$
Equation \eqref{eq:paper-escaping-particle-scaling} applies to the present $p_{\Lambda}$.
The common normalized external vector satisfies
$$\begin{aligned}
\abs{\bkt{f}{e_{\Lambda}}_{\sphilbfrak{h}}}^{2}
=
p_{\Lambda},
\quad
\abs{\bkt{g}{e_{\Lambda}}_{\sphilbfrak{h}}}^{2}
=
\rho^{2}p_{\Lambda}.
\end{aligned}$$
Consequently, the scaling assumption gives
$$\begin{aligned}
\rbk{n_{\Lambda}+1}p_{\Lambda}
\to
c,
\quad
\rbk{n_{\Lambda}+1}\rho^{2}p_{\Lambda}
\to
\rho^{2}c
\quad\rbk{\Lambda\nearrow X}.
\end{aligned}$$
The calculation in \eqref{eq:paper-binomial-resolvent-expectation} gives the evaluations
$$\begin{aligned}
\fun{\oastate[\psi_{f,c}]}{A_f}
&=
\fun{\eta}{c},
&
\fun{\oastate[\psi_{f,c}]}{A_g}
&=
\fun{\eta}{\rho^2c}.
\end{aligned}$$
The strict decrease of $\eta$ gives $\fun{\eta}{c}\neq\fun{\eta}{\rho^2c}$.
For every finite index set $\Lambda'$, equation \eqref{eq:paper-local-element-external-only-block} gives
$$\begin{aligned}
\fnrestr{\rbk{\oastate[\psi_{f,c}]-\oastate[\psi_{\mathrm{vac}}]}}{\oa{B}_{\Lambda'}}
&=
0.
\end{aligned}$$
The bounded difference functional therefore annihilates the norm closure $\oa{B}_{\mathrm{ql}}$, but its value
at $A_f-A_g$ is nonzero.
This proves the assertion.
\end{proof}

\begin{prop}[Reduction by a local unitary when $\abs{\lambda}=1$]\label{prop:paper-unit-modulus-germ-local-unitary-reduction}
Suppose $\emptyset\neq\Lambda_0\in\setpowfin{X}$ and
$$\begin{aligned}
\abs{\lambda}
&=
1,
&
\rbk{1-P_{\Lambda_0}}g
&=
\lambda\rbk{1-P_{\Lambda_0}}f.
\end{aligned}$$
There is a unitary $U_{\txtloc,\Lambda_0}$ on $\sphilbfrak{h}_{\Lambda_0}$ such that the unitary
$$\begin{aligned}
\widetilde{U}_{\Lambda_0}
\colon
\sphilbfrak{h}_{\Lambda_0}\oplus \sphilbfrak{h}_{\Lambda_0}^{\perp}
\to
\sphilbfrak{h}_{\Lambda_0}\oplus \sphilbfrak{h}_{\Lambda_0}^{\perp};\quad
\fun{\widetilde{U}_{\Lambda_0}}{\xi\oplus\eta}
&=
\fun{U_{\txtloc,\Lambda_0}}{\xi}\oplus\eta
\end{aligned}$$
satisfies $g=\lambda\widetilde{U}_{\Lambda_0}f$.
Define
$$\begin{aligned}
D_{\Lambda_0,k}
&=
\fun{\pi_{\sphilbfrak{h}_{\Lambda_0},k}^{-1}}{U_{\txtloc,\Lambda_0}^{\otimes_{\txtsym}k}},
&
\widehat{U}_{\Lambda_0}
&=
\fun{i_{\Lambda_0,X}}{\seq{D_{\Lambda_0,k}}{k\in\monnat}}
\in
\oa{B}_{\Lambda_0}.
\end{aligned}$$
The conjugated operator and its quotient class satisfy
\begin{equation}
\begin{aligned}
A_g
&=
\widehat{U}_{\Lambda_0}A_f\faadj{\widehat{U}_{\Lambda_0}},
&
\sqbk{A_f}=\sqbk{A_g}
&\iff
\commutator{\widehat{U}_{\Lambda_0}}{A_f}\in\oa{B}_{\mathrm{ql}}.
\end{aligned}
\label{eq:paper-unit-modulus-germ-commutator-reduction}
\end{equation}
\end{prop}

\begin{proof}
Put $g'=\lambda^{-1}g$.
The external parts of $f$ and $g'$ agree, so
$\norm{P_{\Lambda_0}f}=\norm{P_{\Lambda_0}g'}$.
A unitary $U_{\txtloc,\Lambda_0}$ carrying the first local vector to the second therefore
exists.
If both local vectors vanish, choose $U_{\txtloc,\Lambda_0}=1_{\sphilbfrak{h}_{\Lambda_0}}$; otherwise extend the isometry between their
one-dimensional spans to a unitary of $\sphilbfrak{h}_{\Lambda_0}$.
In both cases $\widetilde{U}_{\Lambda_0}f=g'$ and hence $g=\lambda\widetilde{U}_{\Lambda_0}f$.
On every particle sector,
$$\begin{aligned}
\fnrestr{N_g}{\sphilbfrak{h}^{\otimes_{\txtsym}n}}
=
\widetilde{U}_{\Lambda_0}^{\otimes_{\txtsym}n}
\fnrestr{N_f}{\sphilbfrak{h}^{\otimes_{\txtsym}n}}
\faadjrbk{\widetilde{U}_{\Lambda_0}^{\otimes_{\txtsym}n}},
\quad
A_{\txtfock,g}
=
\fun{\oarepn_{\txtfock}}{\widehat{U}_{\Lambda_0}}A_{\txtfock,f}
\faadjrbk{\fun{\oarepn_{\txtfock}}{\widehat{U}_{\Lambda_0}}}.
\end{aligned}$$
Faithfulness gives the first identity in \eqref{eq:paper-unit-modulus-germ-commutator-reduction}.
Finally,
$$\begin{aligned}
A_g-A_f
&=
\commutator{\widehat{U}_{\Lambda_0}}{A_f}\faadj{\widehat{U}_{\Lambda_0}},
\end{aligned}$$
and $\widehat{U}_{\Lambda_0},\faadj{\widehat{U}_{\Lambda_0}}\in\oa{B}_{\mathrm{ql}}$, which proves the equivalence.
\end{proof}

\begin{cor}[Particle-coefficient criterion when $\abs{\lambda}=1$]\label{cor:paper-unit-modulus-germ-particle-coefficient-criterion}
Under Proposition \ref{prop:paper-unit-modulus-germ-local-unitary-reduction}, define
$$\begin{aligned}
d_m
&=
\fun{\Theta_m^{-1}}{\fun{a_m}{g}-\fun{a_m}{f}},
&
d
&=
\seq{d_m}{m\in\monnat}
\in
\oa{C}_{\mathrm{b}}(\sphilbfrak{h}).
\end{aligned}$$
Membership of the commutator in $\oa{B}_{\mathrm{ql}}$ is equivalent to
$$\begin{aligned}
\commutator{\widehat{U}_{\Lambda_0}}{A_f}\in\oa{B}_{\mathrm{ql}}
&\iff
\lim_{\Lambda\nearrow X}
\inf_{b\in\oa{C}_{\mathrm{b}}(\Lambda)}
\sup_{n\in\monnat}
\norm{\sum_{m=0}^{n}\fun{E_{m,n}}{d_m-b_m}}_{\oa{B}_n}
=
0.
\end{aligned}$$
\end{cor}

\begin{proof}
For the map $\Upsilon$ in \eqref{eq:paper-main-particle-coefficient-isomorphism}, uniqueness of the
particle-coefficient reconstruction gives
$$\begin{aligned}
\fun{\Upsilon}{d}
&=
A_g-A_f
=
\commutator{\widehat{U}_{\Lambda_0}}{A_f}\faadj{\widehat{U}_{\Lambda_0}}.
\end{aligned}$$
Since multiplication by $\widehat{U}_{\Lambda_0}$ preserves $\oa{B}_{\mathrm{ql}}$, Theorem
\ref{thm:paper-particle-coefficient-quasilocality-criterion} applied to $d$ gives the equivalence.
\end{proof}

\begin{thm}[Unit-modulus projective-germ counterexample]\label{thm:paper-unit-modulus-projective-germ-counterexample}
There are infinite-support projectively germ-equivalent unit vectors $f,g$ with $\lambda=1$ but
$\sqbk{A_f}\neq\sqbk{A_g}$.
\end{thm}

\begin{proof}
Choose a finite-support unit vector $e_0$ and an infinite-support unit vector $u\perp e_0$.
Take $a,b>0$
with $a^2+b^2=1$ and $a\neq b$, and put
$$\begin{aligned}
f
&=
ae_0+bu,
&
g
&=
-ae_0+bu.
\end{aligned}$$
They agree outside $\fun{\supp_{X}}{e_0}$.
For regions $\Lambda$ containing this support, set
$$\begin{aligned}
r_{\Lambda}
&=
\norm{\rbk{1-P_{\Lambda}}u}^2,
&
u_{\Lambda}
&=
\frac{\rbk{1-P_{\Lambda}}u}{\sqrt{r_{\Lambda}}}.
\end{aligned}$$
For every finite $\Lambda$ containing $\fun{\supp_{X}}{e_0}$, define
$n_{\Lambda}=\max\setone{0,\lfloor r_{\Lambda}^{-1}\rfloor-1}$.
The estimate
$$\begin{aligned}
0
&\leq
1-\rbk{n_{\Lambda}+1}r_{\Lambda}
<
r_{\Lambda}
\end{aligned}$$
holds for all sufficiently large $\Lambda$.
This estimate proves $\rbk{n_{\Lambda}+1}r_{\Lambda}\to1$ as $\Lambda\nearrow X$.
Define the unit vectors by
$$\begin{aligned}
h_{\Lambda}^{\pm}
&=
\frac{1}{\sqrt{n_{\Lambda}+1}}e_0
\mathbin{\pm}
\sqrt{\frac{n_{\Lambda}}{n_{\Lambda}+1}}u_{\Lambda}.
\end{aligned}$$
The orthogonality $e_0\perp u_{\Lambda}$ gives
$$\begin{aligned}
\norm{h_{\Lambda}^\pm}^{2}
&=
\frac{1}{n_{\Lambda}+1}\norm{e_0}^{2}
+
\frac{n_{\Lambda}}{n_{\Lambda}+1}\norm{u_{\Lambda}}^{2}
=
1.
\end{aligned}$$
Define the two state nets by
$$\begin{aligned}
\oastate[\psi_{\Lambda}^\pm]
\colon
\fun{\oa{B}}{\sphilbfrak{h}}
\to
\fldcmp;\quad
\fun{\oastate[\psi_{\Lambda}^\pm]}{A}
&=
\bkt{\rbk{h_{\Lambda}^\pm}^{\otimes n_{\Lambda}}}
{\fun{p_{n_{\Lambda}}}{A}\rbk{h_{\Lambda}^\pm}^{\otimes n_{\Lambda}}}.
\end{aligned}$$
For a fixed local index set $\Lambda'$, put $\Lambda''=\Lambda'\cup\fun{\supp_{X}}{e_0}$.
For $\Lambda\supset \Lambda''$,
$e_0\in \sphilbfrak{h}_{\Lambda''}$ and $u_{\Lambda}\in \sphilbfrak{h}_{\Lambda''}^\perp$.
The ambient local-particle decomposition is
$$\begin{aligned}
\sphilbfrak{h}^{\otimes_{\txtsym}n_{\Lambda}}
&\simeq
\bigoplus_{j=0}^{n_{\Lambda}}
\sphilbfrak{h}_{\Lambda''}^{\otimes_{\txtsym}j}
\otimes
\rbk{\sphilbfrak{h}_{\Lambda''}^\perp}^{\otimes_{\txtsym}\rbk{n_{\Lambda}-j}}.
\end{aligned}$$
Under this decomposition,
$$\begin{aligned}
\rbk{h_{\Lambda}^\pm}^{\otimes n_{\Lambda}}
&\simeq
\bigoplus_{j=0}^{n_{\Lambda}}
\binom{n_{\Lambda}}{j}^{1/2}
\rbk{n_{\Lambda}+1}^{-j/2}
\rbk{\frac{n_{\Lambda}}{n_{\Lambda}+1}}^{\rbk{n_{\Lambda}-j}/2}
\rbk{\mathbin{\pm}1}^{n_{\Lambda}-j}
e_0^{\otimes j}\otimes u_{\Lambda}^{\otimes\rbk{n_{\Lambda}-j}}.
\end{aligned}$$
Thus, for $D\in\oa{B}_{\Lambda'}\subset\oa{B}_{\Lambda''}$, both signs have the common expectation
\begin{equation}
\begin{aligned}
\fun{\oastate[\psi_{\pm,\Lambda}]}{D}
&=
\sum_{j=0}^{n_{\Lambda}}\binom{n_{\Lambda}}{j}
\rbk{n_{\Lambda}+1}^{-j}
\rbk{\frac{n_{\Lambda}}{n_{\Lambda}+1}}^{n_{\Lambda}-j}
\bkt{e_0^{\otimes j}}{\fun{p_j}{D}e_0^{\otimes j}}_{\sphilbfrak{h}_{\Lambda''}^{\otimes_{\txtsym}j}}.
\end{aligned}
\label{eq:paper-phase-paired-states-local-agreement}
\end{equation}
Apply the Banach--Alaoglu theorem to the product of the two state spaces.
There is a convergent subnet indexed by a directed set $I$.
For each $i\in I$, the associated $\Lambda_i$ is a finite subset of $X$ containing
$\fun{\supp_{X}}{e_0}$, and for every finite $\Lambda\supset\fun{\supp_{X}}{e_0}$, there is
$i_{\Lambda}\in I$ for which $\Lambda\subset \Lambda_i$ whenever $i\geq i_{\Lambda}$.
There are states $\oastate[\psi_{\txtnonneg}],\oastate[\psi_{\txtneg}]\in\oaspstate_{\fun{\oa{B}}{\sphilbfrak{h}}}$ satisfying
$$\begin{aligned}
\fun{\oastate[\psi_{\pm}]}{A}
&=
\lim_{i\in I}\fun{\oastate[\psi_{\pm,\Lambda_i}]}{A}
&&\rbk{A\in\fun{\oa{B}}{\sphilbfrak{h}}}.
\end{aligned}$$
Equation
\eqref{eq:paper-phase-paired-states-local-agreement} and density of the local union give
$$\begin{aligned}
\fnrestr{\oastate[\psi_{\txtnonneg}]}{\oa{B}_{\mathrm{ql}}}
&=
\fnrestr{\oastate[\psi_{\txtneg}]}{\oa{B}_{\mathrm{ql}}}.
\end{aligned}$$

The four binomial parameters are
$$\begin{aligned}
p_{\pm,f,\Lambda}
&=
\abs{\frac{a}{\sqrt{n_{\Lambda}+1}}
\mathbin{\pm}b\sqrt{\frac{n_{\Lambda}}{n_{\Lambda}+1}}\sqrt{r_{\Lambda}}}^2,
&
p_{\pm,g,\Lambda}
&=
\abs{-\frac{a}{\sqrt{n_{\Lambda}+1}}
\mathbin{\pm}b\sqrt{\frac{n_{\Lambda}}{n_{\Lambda}+1}}\sqrt{r_{\Lambda}}}^2.
\end{aligned}$$
With $s_\pm=\rbk{a\pm b}^2$,
$$\begin{aligned}
\rbk{n_{\Lambda}+1}p_{\pm,f,\Lambda}
&\to
s_\pm,
&
\rbk{n_{\Lambda}+1}p_{\pm,g,\Lambda}
&\to
s_\mp.
\end{aligned}$$
Because $a\neq b$, one has $s_+,s_->0$ and $s_+\neq s_-$.
Hence, for every sufficiently large $\Lambda$,
$$\begin{aligned}
p_{\txtnonneg,f,\Lambda},\,
p_{\txtneg,f,\Lambda},\,
p_{\txtnonneg,g,\Lambda},\,
p_{\txtneg,g,\Lambda}
&>
0.
\end{aligned}$$
All four parameters therefore lie in the domain of $\eta_{n_{\Lambda}}$.
Applying
\eqref{eq:paper-binomial-resolvent-expectation} for these regions and taking the chosen weak-$\ast$ limit yields
$$\begin{aligned}
\fun{\oastate[\psi_{\txtnonneg}]}{A_f}
&=
\fun{\eta}{s_+},
&
\fun{\oastate[\psi_{\txtnonneg}]}{A_g}
&=
\fun{\eta}{s_-},
\\
\fun{\oastate[\psi_{\txtneg}]}{A_f}
&=
\fun{\eta}{s_-},
&
\fun{\oastate[\psi_{\txtneg}]}{A_g}
&=
\fun{\eta}{s_+}.
\end{aligned}$$
Since $\eta$ is strictly decreasing and $s_+\neq s_-$,
$$\begin{aligned}
\fun{\rbk{\oastate[\psi_{\txtnonneg}]-\oastate[\psi_{\txtneg}]}}{A_f-A_g}
&=
2\rbk{\fun{\eta}{s_+}-\fun{\eta}{s_-}}
\neq
0.
\end{aligned}$$
The functional annihilates $\oa{B}_{\mathrm{ql}}$, so $A_f-A_g\notin\oa{B}_{\mathrm{ql}}$.
\end{proof}

\subsubsection{Graph specialization and finite-region verification}\label{sec:paper-graph-spatial-specialization}

In this subsubsection, the abstract finite-mode inclusions are specialized to the graph vertex set \(X=\Gamma\) and verified against Proposition \ref{prop:concrete-spatial-localization-system}.

For the vertex set \(\Gamma\) in \eqref{eq:paper-main-local-one-particle-subspace}, choose \(X=\Gamma\) and the standard basis \(u_x=\delta_x\) of \(\sphilbfrak{h}=\fun{\lpseq^{2}}{\Gamma}\). Equation \eqref{eq:paper-basis-localization-data} recovers the graph subspaces and projections in \eqref{eq:paper-main-local-one-particle-subspace} and \eqref{eq:paper-finite-region-one-particle-projection}. The basis-indexed local algebras then agree with \eqref{eq:paper-main-local-buchholz-algebra}, and the quasilocal algebra agrees with \eqref{eq:paper-main-quasilocal-algebra}.

\begin{proof}[Proof of Proposition \ref{prop:concrete-spatial-localization-system}]
For the standard basis, $\sphilbfrak{h}_{\Lambda''\setminus\Lambda}$ in
\eqref{eq:concrete-local-inclusion} is the orthogonal complement of
$\sphilbfrak{h}_{\Lambda}$ inside $\sphilbfrak{h}_{\Lambda''}$.
Thus the block formula \eqref{eq:concrete-local-inclusion} is the tensor-factor action
\eqref{eq:paper-finite-region-local-inclusion-block-action}, and the global map
\eqref{eq:concrete-global-local-inclusion} is \eqref{eq:paper-basis-global-local-inclusion}.
Proposition \ref{prop:paper-tensor-factor-localization-identities} gives composition compatibility and
commutation for disjoint regions.
Proposition \ref{prop:paper-basis-local-coordinate-density} gives the fixed-cutoff density and
coefficient-embedding compatibility.
\end{proof}

\subsection{Particle-cutoff representations, ordered K-theory, and Busby maps}\label{sec:paper-cutoff-representation-appendix}

In this appendix, the concrete cutoff representations, their ordered K-theory, and the finite-mode Busby maps are computed.

\subsubsection{Representations of a particle-cutoff quotient}\label{representations-of-a-particle-cutoff-quotient}

In this subsubsection, we classify the irreducible representations of each particle-cutoff quotient for every separable infinite-dimensional complex Hilbert space.

For a unital \(\oacstar\)-algebra \(\oa{A}\), an irreducible representation is a nonzero \(\ast\)-homomorphism \(\pi\colon\oa{A}\to\opspbddlin{\sphilbfrak{H}_{\pi}}\) whose only closed \(\fun{\pi}{\oa{A}}\)-invariant subspaces are \(\setone{0}\) and \(\sphilbfrak{H}_{\pi}\). Define \[\begin{aligned}
\fun{\operatorname{Prim}}{\oa{A}}
&=
\set{\Ker\pi}{\pi\text{ is an irreducible representation of $\oa{A}$}}.
\end{aligned}\] An ideal \(P\subset\oa{A}\) is called primitive when \(P\in\fun{\operatorname{Prim}}{\oa{A}}\). For the \(K\)-groups we use the conventions in \cite[Sections 4.1 and 8.1, and Theorem 9.3.1]{BruceBlackadar001}. Explicitly, \(\fun{K_0}{\oa{A}}\) is the Grothendieck group of stable Murray--von Neumann equivalence classes of projections in the matrix algebras \(\fun{M_d}{\oa{A}}\), and \(\fun{K_1}{\oa{A}}\) is the group of stable homotopy classes of unitaries in these matrix algebras, with block direct sum as the group operation. Here two projections \(p\) and \(q\) are Murray--von Neumann equivalent when, after adjoining zero matrix blocks if their matrix sizes differ, there is a matrix \(v\) over \(\oa{A}\) such that \(p=v\faadj{v}\) and \(q=\faadj{v}v\). Define \[\begin{aligned}
\fun{K_0}{\oa{A}}_+
&=
\set{\sqbk{p}\in\fun{K_0}{\oa{A}}}
{p\text{ is a projection in $\fun{M_d}{\oa{A}}$ for some $d\geq1$}},
&
\sqbk{1_{\oa{A}}}
&\in
\fun{K_0}{\oa{A}}_+.
\end{aligned}\]

\begin{cor}[Representations of a particle-cutoff quotient]\label{cor:finite-stage-consequences}
Fix $n\in\monnat$.
For $0\leq j\leq n$, use the maps $\varpi_j$ in
\eqref{eq:paper-main-fock-cutoff-identification} to define
$$\begin{aligned}
\pi_{n,j}
\colon
\oa{B}_n
\to
\opspbddlin{\sphilbfrak{h}^{\otimes_{\txtsym}j}};\quad
\fun{\pi_{n,j}}{A}
&=
\begin{cases}
\fun{\varpi_j\circ r_{j+1}\circ\cdots\circ r_n}{A},&j<n,\\
\fun{\varpi_n}{A},&j=n.
\end{cases}
\end{aligned}$$
These $n+1$ representations exhaust the irreducible representations of
$\oa{B}_n$.
Thus $\fun{\operatorname{Prim}}{\oa{B}_n}$ has $n+1$ points.
The center of $\oa{B}_n$ is $\fldcmp1$, and
$$\begin{aligned}
\fun{K_0}{\oa{B}_n}
&\simeq\ringratint^{n+1},
&
\fun{K_1}{\oa{B}_n}
&=0.
\end{aligned}$$
\end{cor}

\begin{proof}
The assertion is immediate for $n=0$ because $\oa{B}_0=\fldcmp$.
Suppose $n\geq1$, assume the result through
$n-1$, and let
$\pi\colon\oa{B}_n\to\opspbddlin{\sphilbfrak{H}_{\pi}}$ be a nonzero irreducible representation.
Transport it through the comparison isomorphism $\varpi_n$ in
\eqref{eq:paper-main-fock-cutoff-identification} by defining
$$\begin{aligned}
\pi_{\mathrm{c}}
\colon
\fun{\varpi_n}{\oa{B}_n}
\to
\opspbddlin{\sphilbfrak{H}_{\pi}};\quad
\fun{\pi_{\mathrm{c}}}{T}
&=
\fun{\pi\circ\varpi_n^{-1}}{T}.
\end{aligned}$$
For the remainder of this proof and every $0\leq m\leq n$, use the maps $\varpi_m$ in
\eqref{eq:paper-main-fock-cutoff-identification} to identify $\oa{B}_m$ with
$\fun{\varpi_m}{\oa{B}_m}$, write $\pi$ for $\pi_{\mathrm{c}}$, write $r_m$ for
the map $\kappa_m$ in \eqref{eq:paper-main-fock-connecting-map} when $m\geq1$, and use the action defining $\pi_{m,n}$ in
Corollary \ref{cor:finite-stage-consequences}.
Under these stated identifications,
$\Ker r_n$ is the compact-operator ideal
$\fun{\schattencls}{\sphilbfrak{h}^{\otimes_{\txtsym}n}}$.

If $\fun{\pi}{\fun{\schattencls}{\sphilbfrak{h}^{\otimes_{\txtsym}n}}}=0$, the equality
$\Ker r_n=\fun{\schattencls}{\sphilbfrak{h}^{\otimes_{\txtsym}n}}$ and surjectivity of $r_n$ define a unique
irreducible representation
$$\begin{aligned}
\widetilde{\pi}
\colon
\oa{B}_{n-1}
\to
\opspbddlin{\sphilbfrak{H}_{\pi}};\quad
\fun{\widetilde{\pi}}{\fun{r_n}{A}}
&=\fun{\pi}{A}.
\end{aligned}$$
The induction hypothesis makes $\pi$ unitarily equivalent to one of $\pi_{n,j}$, $0\leq j<n$.

Suppose instead that the restriction of $\pi$ to the compact ideal is nonzero.
The closed linear span of the
compact ideal acting on $\sphilbfrak{H}_{\pi}$ is
$$\begin{aligned}
\spclosedlinspan\fun{\pi}{\fun{\schattencls}{\sphilbfrak{h}^{\otimes_{\txtsym}n}}}\sphilbfrak{H}_{\pi}
&=\sphilbfrak{H}_{\pi},
\end{aligned}$$
because the left side is a nonzero invariant closed subspace.
Fix a unit vector
$e\in\sphilbfrak{h}^{\otimes_{\txtsym}n}$.
For every
$x\in\sphilbfrak{h}^{\otimes_{\txtsym}n}$, define
$$\begin{aligned}
K_x
&=
P_{x,e}
\in
\fun{\schattencls}{\sphilbfrak{h}^{\otimes_{\txtsym}n}}.
\end{aligned}$$
Simplicity of the $\oacstar$-algebra
$\fun{\schattencls}{\sphilbfrak{h}^{\otimes_{\txtsym}n}}$ of compact operators gives
$\Ker\fnrestr{\pi}{\fun{\schattencls}{\sphilbfrak{h}^{\otimes_{\txtsym}n}}}=0$.
Choose a unit vector $\eta\in\fun{\pi}{K_e}\sphilbfrak{H}_{\pi}$.
For every $A\in\oa{B}_n$ and $x,y\in\sphilbfrak{h}^{\otimes_{\txtsym}n}$, one has
$\faadj{K_x}K_y=\bkt{x}{y}K_e$ and $AK_x=K_{Ax}$, and hence
$$\begin{aligned}
\bkt{\fun{\pi}{K_x}\eta}{\fun{\pi}{K_y}\eta}_{\sphilbfrak{H}_{\pi}}
=\bkt{x}{y}_{\sphilbfrak{h}^{\otimes_{\txtsym}n}},
\quad
\fun{\pi}{A}\fun{\pi}{K_x}\eta
=\fun{\pi}{AK_x}\eta
=\fun{\pi}{K_{Ax}}\eta.
\end{aligned}$$
The same identity with $\faadj{A}$ proves invariance of the closed span of these vectors, so irreducibility makes that
span all of $\sphilbfrak{H}_{\pi}$.
Thus the rule
$$\begin{aligned}
U
\colon
\sphilbfrak{H}_{\pi}
\to
\sphilbfrak{h}^{\otimes_{\txtsym}n};\quad
\fun{U}{\fun{\pi}{K_x}\eta}
&=x
\end{aligned}$$
extends uniquely to a unitary.
For $x,y\in\sphilbfrak{h}^{\otimes_{\txtsym}n}$,
$$\begin{aligned}
\fun{U\circ\pi}{K_y}\faadj{U}x
&=\fun{U\circ\pi}{K_yK_x}\eta
=\bkt{e}{x}\fun{U\circ\pi}{K_y}\eta
=\bkt{e}{x}y
=\fun{K_y}{x}.
\end{aligned}$$
Hence $U\fun{\pi}{K_y}\faadj{U}=K_y$.
Taking products and adjoints gives
$$\begin{aligned}
U\fun{\pi}{K_x\faadj{K_y}}\faadj{U}
&=K_x\faadj{K_y}.
\end{aligned}$$
The linear span of $K_x\faadj{K_y}=P_{x,y}$ is the finite-rank operators, which are norm dense in
the compact operators.
Therefore
$$\begin{aligned}
U\fun{\pi}{K}\faadj{U}
&=K
\end{aligned}$$
for every compact $K$.
For $A\in\oa{B}_n$ and compact $K$,
$$\begin{aligned}
\rbk{U\fun{\pi}{A}\faadj{U}-A}K
&=U\fun{\pi}{AK}\faadj{U}-AK
=0.
\end{aligned}$$
The compact operators form an essential ideal in
$\opspbddlin{\sphilbfrak{h}^{\otimes_{\txtsym}n}}$, hence
$U\fun{\pi}{A}\faadj{U}=A$.
Thus $\pi$ is equivalent to $\pi_{n,n}$.

For $j<n$, $\pi_{n,j}$ factors through $r_n$, whereas the defining representation satisfies
$$\begin{aligned}
\fnrestr{\pi_{n,j}}{\Ker r_n}
&=0,
&
\fnrestr{\pi_{n,n}}{\Ker r_n}
&=\fnrestr{\pi_{n,n}}{\fun{\schattencls}{\sphilbfrak{h}^{\otimes_{\txtsym}n}}}
\neq0.
\end{aligned}$$
For $0\leq i<j<n$,
$$\begin{aligned}
\pi_{i,n}
&=\pi_{i,n-1}\circ r_n,
&
\pi_{n,j}
&=\pi_{j,n-1}\circ r_n.
\end{aligned}$$
Surjectivity of $r_n$ and the induction hypothesis prove pairwise inequivalence and exhaustiveness.

If $z\in \fun{\oacenter}{\oa{B}_n}$, then in the faithful defining representation
$$\begin{aligned}
z
&\in\oacommutant{\fun{\schattencls}{\sphilbfrak{h}^{\otimes_{\txtsym}n}}}
=\fldcmp1.
\end{aligned}$$
Thus $\fun{\oacenter}{\oa{B}_n}=\fldcmp1$.

For $n\geq1$, apply the six-term exact sequence of
\cite[Theorem 9.3.1]{BruceBlackadar001} to
$$\begin{aligned}
0
\to\fun{\schattencls}{\sphilbfrak{h}^{\otimes_{\txtsym}n}}
\to\oa{B}_n
\to\oa{B}_{n-1}
\to0.
\end{aligned}$$
For this extension the cyclic exact diagram is
\begin{equation}
\begin{array}{ccccc}
\fun{K_0}{\fun{\schattencls}{\sphilbfrak{h}^{\otimes_{\txtsym}n}}}
&\to&\fun{K_0}{\oa{B}_n}&\to&\fun{K_0}{\oa{B}_{n-1}}\\
\uparrow\partial_1&&&&\downarrow\partial_0\\
\fun{K_1}{\oa{B}_{n-1}}&\leftarrow&\fun{K_1}{\oa{B}_n}&\leftarrow&
\fun{K_1}{\fun{\schattencls}{\sphilbfrak{h}^{\otimes_{\txtsym}n}}}.
\end{array}
\label{eq:paper-six-term-k-theory-sequence}
\end{equation}
Exactness means that at each of the six groups in \eqref{eq:paper-six-term-k-theory-sequence} the image of the incoming map is the kernel of the
outgoing map.
The compact-operator groups are
$\fun{K_0}{\fun{\schattencls}{\sphilbfrak{h}^{\otimes_{\txtsym}n}}}=\ringratint$, with the rank-one projection
as generator, and
$\fun{K_1}{\fun{\schattencls}{\sphilbfrak{h}^{\otimes_{\txtsym}n}}}=0$.
The induction hypothesis gives $\fun{K_1}{\oa{B}_{n-1}}=0$.
Hence the domain of $\partial_1$ and the group at the
lower right are zero.
Exactness at the remaining four positions gives
$$\begin{aligned}
0
\to\fun{K_1}{\oa{B}_n}
\to\fun{K_1}{\oa{B}_{n-1}}=0,
\quad
0
\to\ringratint
\to\fun{K_0}{\oa{B}_n}
\to\ringratint^n
\to0.
\end{aligned}$$
Exactness yields $\fun{K_1}{\oa{B}_n}=0$.
Since $\ringratint^n$ is free abelian, the second sequence splits, giving
$\fun{K_0}{\oa{B}_n}\simeq\ringratint^{n+1}$.
\end{proof}

A representation \(\pi\) of \(\oa{B}\) is said to factor through \(p_n\) if \(\pi=\pi_{n}\circ p_n\) for a representation \(\pi_{n}\) of \(\oa{B}_n\).

\begin{prop}[Global irreducible representations through finite coordinates]\label{prop:global-finite-coordinate-irreducible-representations}
For $j\in\monnat$, use the map $\varpi_j$ in \eqref{eq:paper-main-fock-cutoff-identification} to define
$$\begin{aligned}
\widehat{\pi}_j
\colon
\oa{B}
\to
\opspbddlin{\sphilbfrak{h}^{\otimes_{\txtsym}j}};\quad
\fun{\widehat{\pi}_j}{\seq{A_m}{m\in\monnat}}
&=\fun{\varpi_j}{A_j}.
\end{aligned}$$
Then $\widehat{\pi}_j$ is irreducible and
$$\begin{aligned}
\oadoublecommutant{\fun{\widehat{\pi}_j}{\oa{B}}}
&=\opspbddlin{\sphilbfrak{h}^{\otimes_{\txtsym}j}}.
\end{aligned}$$
For a primitive ideal $P$ of $\oa{B}$, the following two statements are equivalent:
\begin{enumerate}
\item $\Ker p_n\subset P$ for some $n\in\monnat$;
\item $P=\Ker\widehat{\pi}_j$ for some $j\in\monnat$.
\end{enumerate}
If $n$ is fixed in the assertion $\Ker p_n\subset P$, the integer $j$ is unique in $0\leq j\leq n$.
Hence the primitive
ideals of irreducible representations factoring through finite coordinates are exactly
$\seq{\Ker\widehat{\pi}_j}{j\in\monnat}$, and all generated von Neumann algebras are type $\txtI$ factors.
\end{prop}

\begin{proof}
The coordinate map $p_j$ defined in \eqref{eq:paper-main-coordinate-map} is surjective by Corollary
\ref{cor:paper-particle-cutoff-coordinate-surjectivity}.
Corollary \ref{cor:finite-stage-consequences} gives
$$\begin{aligned}
\widehat{\pi}_j
&=\pi_{j,j}\circ p_j,
&
\oacommutant{\fun{\widehat{\pi}_j}{\oa{B}}}
&=\fldcmp1.
\end{aligned}$$
Taking commutants in the equality $\oacommutant{\fun{\widehat{\pi}_j}{\oa{B}}}=\fldcmp1$ proves
$\oadoublecommutant{\fun{\widehat{\pi}_j}{\oa{B}}}=\opspbddlin{\sphilbfrak{h}^{\otimes_{\txtsym}j}}$.
If $\Ker p_n\subset P$, then
$\setquot{P}{\Ker p_n}$ is primitive in $\oa{B}_n$.
Corollary \ref{cor:finite-stage-consequences} gives a unique
$0\leq j\leq n$ such that
$$\begin{aligned}
P
&=\Ker\rbk{\pi_{n,j}\circ p_n}
=\Ker\widehat{\pi}_j.
\end{aligned}$$
Conversely, $\widehat{\pi}_j=\pi_{j,j}\circ p_j$ gives
$\Ker p_j\subset\Ker\widehat{\pi}_j$.
The double-commutant computation proves the type $\txtI$ assertion.
\end{proof}

\begin{thm}[Ordered $K_0$ of a particle-cutoff quotient]\label{thm:finite-stage-ordered-k0}
For $n\in\monnat$, give $\ringratint^{n+1}$ the lexicographic positive cone
$$\begin{aligned}
\rbk{\ringratint^{n+1}}_+
&=
\setone{0}\cup
\set{\rbk{z_0,\ldots,z_n}\in\ringratint^{n+1}}
{\text{the first nonzero coordinate is positive}}.
\end{aligned}$$
The group isomorphism in Corollary \ref{cor:finite-stage-consequences} can be chosen as
\begin{equation}
\begin{aligned}
\rbk{\fun{K_0}{\oa{B}_n},\fun{K_0}{\oa{B}_n}_+,\sqbk{1_{\oa{B}_n}}}
&\simeq
\rbk{\ringratint^{n+1},\rbk{\ringratint^{n+1}}_+,\rbk{1,0,\ldots,0}}.
\end{aligned}
\label{eq:finite-stage-ordered-k0}
\end{equation}
\end{thm}

\begin{proof}
Choose an orthonormal basis $\seq{\delta_x}{x\in I}$ indexed by a countably infinite set $I$ and identify
$\sphilbfrak{h}$ with $\fun{\lpseq^{2}}{I}$ through this basis.
For a nonempty finite $X\subset I$, put
$$\begin{aligned}
\sphilbfrak{f}_X
&=\splinspan_{\fldcmp}\set{\delta_x}{x\in X},
&
d_{X,k}
&=\dim_{\fldcmp}\sphilbfrak{f}_X^{\otimes_{\txtsym}k}
=\binom{\abscard{X}+k-1}{k}
&&\rbk{0\leq k\leq n}.
\end{aligned}$$
Using the decomposition unitary $W_{\sphilbfrak{f},n}$ defined in
\eqref{eq:paper-main-particle-fock-decomposition-unitary}, specialized to $\sphilbfrak{f}=\sphilbfrak{f}_X$, and the map $\varpi_n$
in \eqref{eq:paper-main-fock-cutoff-identification}, define
$$\begin{aligned}
L_{X,n}
&\colon
\opspbddlin{\fun{\spfockgamma_{\txtbsn}}{\sphilbfrak{f}_X}}^{\gamma}
\to
\oa{B}_n;
\quad
\fun{\varpi_n\circ L_{X,n}}{D}
&=
\faadj{W_{\sphilbfrak{f}_X,n}}
\rbk{\bigoplus_{k=0}^{n}
\fun{\pi_{\sphilbfrak{f}_X,k}}{D_k}\otimes
1_{\rbk{\sphilbfrak{f}_X^{\perp}}^{\otimes_{\txtsym}\rbk{n-k}}}}
W_{\sphilbfrak{f}_X,n}.
\end{aligned}$$
Define the finite-dimensional local cutoff algebra by
$$\begin{aligned}
\oa{A}_{X,n}
&=
\fun{L_{X,n}}{\opspbddlin{\fun{\spfockgamma_{\txtbsn}}{\sphilbfrak{f}_X}}^{\gamma}}
\subset
\oa{B}_n.
\end{aligned}$$
For $0\leq k\leq n$, its block representation is
$$\begin{aligned}
\chi_{X,k}
\colon
\oa{A}_{X,n}
\to
\opspbddlin{\sphilbfrak{f}_X^{\otimes_{\txtsym}k}};\quad
\fun{\chi_{X,k}}{\fun{L_{X,n}}{D}}
&=
\fun{\pi_{\sphilbfrak{f}_X,k}}{D_k}
\end{aligned}$$
for $D=\seq{D_j}{j\in\monnat}\in\opspbddlin{\fun{\spfockgamma_{\txtbsn}}{\sphilbfrak{f}_X}}^{\gamma}$.
This formula is well-defined because equality in $\oa{A}_{X,n}$ implies equality of every represented block
with index $0\leq k\leq n$.
For $s\in\semigrposint$, define the matrix amplification by
$$\begin{aligned}
\chi_{X,k}^{(s)}
\colon
\fun{M_s}{\oa{A}_{X,n}}
\to
\opspbddlin{\rbk{\sphilbfrak{f}_X^{\otimes_{\txtsym}k}}^{\oplus s}};\quad
\fun{\chi_{X,k}^{(s)}}{\sqbk{A_{ab}}_{1\leq a,b\leq s}}
&=
\sqbk{\fun{\chi_{X,k}}{A_{ab}}}_{1\leq a,b\leq s}.
\end{aligned}$$
The corresponding block-rank map is
$$\begin{aligned}
r_{X,k}^{(s)}
\colon
\set{p\in\fun{M_s}{\oa{A}_{X,n}}}{p=p^2=\faadj{p}}
\to
\monnat;\quad
\fun{r_{X,k}^{(s)}}{p}
&=
\fun{\larank}{\fun{\chi_{X,k}^{(s)}}{p}}.
\end{aligned}$$
For a projection $p\in\fun{M_s}{\oa{A}_{X,n}}$, define integers
$z_0(p),\ldots,z_n(p)$ recursively by
\begin{equation}
\begin{aligned}
\fun{r_{X,k}^{(s)}}{p}
&=\sum_{m=0}^{k}z_m(p)d_{X,k-m}
&&\rbk{0\leq k\leq n}.
\end{aligned}
\label{eq:particle-coefficient-rank-coordinates}
\end{equation}
They are unique because $d_{X,0}=1$.

If $X\subset Y$, then $\oa{A}_{X,n}\subset\oa{A}_{Y,n}$.
Define the inclusion and its matrix amplification by
$$\begin{aligned}
\iota_{X,Y,n}
&\colon
\oa{A}_{X,n}
\to
\oa{A}_{Y,n};
&
\fun{\iota_{X,Y,n}}{A}
&=
A,
\\
\iota_{s,X,Y,n}
&\colon
\fun{M_s}{\oa{A}_{X,n}}
\to
\fun{M_s}{\oa{A}_{Y,n}};
&
\fun{\iota_{s,X,Y,n}}{\sqbk{A_{ab}}_{1\leq a,b\leq s}}
&=
\sqbk{\fun{\iota_{X,Y,n}}{A_{ab}}}_{1\leq a,b\leq s}.
\end{aligned}$$
For a projection $p\in\fun{M_s}{\oa{A}_{X,n}}$, define
$$\begin{aligned}
r_{s,Y,k}
\colon
\set{p\in\fun{M_s}{\oa{A}_{X,n}}}{p=p^2=\faadj{p}}
\to
\monnat.
\end{aligned}$$
For $0\leq k\leq n$, it is given by
$$\begin{aligned}
\fun{r_{s,Y,k}}{p}
&=
\fun{\larank}{
\fun{\chi_{s,Y,k}\circ\iota_{s,X,Y,n}}{p}}.
\end{aligned}$$
The symmetric-tensor decomposition
$\sphilbfrak{f}_Y=\sphilbfrak{f}_X\oplus \sphilbfrak{f}_{Y\setminus X}$ gives
$$\begin{aligned}
\fun{r_{s,Y,k}}{p}
&=\sum_{j=0}^{k}\fun{r_{s,X,j}}{p}d_{Y\setminus X,k-j}.
\end{aligned}$$
The dimension generating functions are
$$\begin{aligned}
\sum_{k\in\monnat}d_{X,k}t^k
&=\rbk{1-t}^{-\abscard{X}},
&
\sum_{k\in\monnat}d_{Y\setminus X,k}t^k
&=\rbk{1-t}^{-\abscard{Y\setminus X}}.
\end{aligned}$$
Substitution of \eqref{eq:particle-coefficient-rank-coordinates} yields
$$\begin{aligned}
\sum_{k=0}^{n}\fun{r_{s,Y,k}}{p}t^k
&=
\rbk{\sum_{m=0}^{n}z_m(p)t^m}
\rbk{1-t}^{-\abscard{Y}}
\pmod{t^{n+1}}.
\end{aligned}$$
Thus $z_0(p),\ldots,z_n(p)$ are unchanged when the support is enlarged from $X$ to $Y$.

Each particle-degree coefficient is compact.
Approximating its finitely many matrix entries by finite-rank
operators and then approximating their range and adjoint-range vectors by finite-support vectors gives
$$\begin{aligned}
\oa{B}_n
&=
\gtclos{\bigcup_{\substack{X\subset I\\0<\abscard{X}<\infty}}\oa{A}_{X,n}}^{\norm{\mathord{\cdot}}_{\oa{B}_n}}.
\end{aligned}$$
Thus $\oa{B}_n$ is the inductive limit of the finite-dimensional unital algebras $\oa{A}_{X,n}$.
Ordered $K_0$, including the distinguished order unit, commutes with this inductive limit.

For every finite nonempty $X$, the block representations give a $\ast$-isomorphism
$$\begin{aligned}
\oa{A}_{X,n}
&\simeq
\bigoplus_{k=0}^{n}\opspbddlin{\sphilbfrak{f}_X^{\otimes_{\txtsym}k}}.
\end{aligned}$$
Hence the block-rank vector identifies $\fun{K_0}{\oa{A}_{X,n}}$ with $\ringratint^{n+1}$.
The triangular relation \eqref{eq:particle-coefficient-rank-coordinates} has diagonal entries
$d_{X,0}=1$, so the Newton-rank vector
$$\begin{aligned}
\rbk{z_0,\ldots,z_n}
&\in
\ringratint^{n+1}
\end{aligned}$$
is another integral coordinate system on $\fun{K_0}{\oa{A}_{X,n}}$.
The equality of $z_0(p),\ldots,z_n(p)$ under the inclusion $X\subset Y$ shows that every connecting map
$\fun{K_0}{\iota_{X,Y,n}}$ is the identity in these coordinates.
Therefore the inductive-limit map
$$\begin{aligned}
\operatorname{NewtRank}_n
\colon
\fun{K_0}{\oa{B}_n}
\to
\ringratint^{n+1}
\end{aligned}$$
is a group isomorphism.

It remains to identify the positive cone.
If a nonzero positive class is represented by a projection
$p\in\fun{M_s}{\oa{A}_{X,n}}$ and $z_j(p)$ is its first nonzero Newton-rank coordinate, then
$$\begin{aligned}
z_j(p)
&=
\fun{r_{s,X,j}}{p}
>
0.
\end{aligned}$$
Conversely, suppose $z_j>0$ is the first nonzero coordinate of
$z=(z_0,\ldots,z_n)$.
For every $k\geq j$, the required block rank is
$$\begin{aligned}
\sum_{m=0}^{k}z_m d_{X,k-m}
&=
z_j\binom{\abscard{X}+k-j-1}{k-j}
+
\sum_{m=j+1}^{k}z_m\binom{\abscard{X}+k-m-1}{k-m}.
\end{aligned}$$
This polynomial in $\abscard{X}$ has degree $k-j$ and positive leading coefficient.
Since only finitely many
$k\leq n$ occur, all the required block ranks are nonnegative for sufficiently large $X$.
After a matrix amplification, projections of precisely these ranks exist in every full matrix block of
$\oa{A}_{X,n}$.
Their direct sum represents $z$ by
\eqref{eq:particle-coefficient-rank-coordinates}.
Thus the positive cone is exactly the lexicographic cone.
Finally, the unit has block ranks $d_{X,k}$, so its Newton-rank vector is $(1,0,\ldots,0)$.
This proves \eqref{eq:finite-stage-ordered-k0}.
\end{proof}

\subsubsection{Concrete Busby map of a particle-cutoff extension}\label{concrete-busby-map-of-a-particle-cutoff-extension}

In this subsubsection, we compute the Busby map of the represented particle-cutoff extension on concrete particle blocks.

Fix \(n\in\semigrposint\). For every \(0\leq m\leq n\), use the map \(\varpi_m\) in \eqref{eq:paper-main-fock-cutoff-identification} to identify the abstract quotient with its represented image, and use \(\oa{B}_m\) in this subsection for \[\begin{aligned}
\oa{B}_m
&=
\fun{\varpi_m}{\oa{B}_m(\sphilbfrak{h})}
=
\fun{\rho_m}{\oaresolventalgebra(\sphilbfrak{h},\sigma)^{\gamma}}.
\end{aligned}\] For \(1\leq m\leq n\), the connecting map \(\kappa_m\) defined in \eqref{eq:paper-main-fock-connecting-map} maps \(\oa{B}_m\) to \(\oa{B}_{m-1}\) in these represented algebras. Define \[\begin{aligned}
\oa{Q}_n
&=
\setquot{\opspbddlin{\sphilbfrak{h}^{\otimes_{\txtsym}n}}}
{\fun{\schattencls}{\sphilbfrak{h}^{\otimes_{\txtsym}n}}},
\end{aligned}\] and define the quotient map by \[\begin{aligned}
q_n
\colon
\opspbddlin{\sphilbfrak{h}^{\otimes_{\txtsym}n}}
\to
\oa{Q}_n;\quad
\fun{q_n}{T}
&=T+\fun{\schattencls}{\sphilbfrak{h}^{\otimes_{\txtsym}n}}.
\end{aligned}\] Thus \(\oa{Q}_n\) is, by definition, the Calkin \(\oacstar\)-algebra of \(\sphilbfrak{h}^{\otimes_{\txtsym}n}\). The concrete particle-coefficient decomposition in \eqref{eq:paper-concrete-particle-coefficient-decomposition} gives, for every \(a\in\oa{B}_{n-1}\), unique operators \[\begin{aligned}
a
=
\sum_{m=0}^{n-1}\fun{\widetilde{E}_{m,n-1}}{a_m},
\quad
a_m
\in
\fun{\schattencls}{\sphilbfrak{h}^{\otimes_{\txtsym}m}}
\quad\rbk{0\leq m\leq n-1}.
\end{aligned}\] Define the concrete coordinate map \[\begin{aligned}
\widetilde{\Pi}_{n-1}
\colon
\bigoplus_{m=0}^{n-1}
\fun{\schattencls}{\sphilbfrak{h}^{\otimes_{\txtsym}m}}
\to
\oa{B}_{n-1};\quad
\fun{\widetilde{\Pi}_{n-1}}{\rbk{a_0,\ldots,a_{n-1}}}
&=
\sum_{m=0}^{n-1}\fun{\widetilde{E}_{m,n-1}}{a_m}.
\end{aligned}\] The subset-sum formula and the maximum norm on the finite direct sum give \[\begin{aligned}
\norm{
\fun{\widetilde{\Pi}_{n-1}}{\rbk{a_0,\ldots,a_{n-1}}}}
\leq
\sum_{m=0}^{n-1}
\binom{n-1}{m}\norm{a_m}
\leq
2^{n-1}
\max_{0\leq m\leq n-1}\norm{a_m}.
\end{aligned}\] The direct-sum identity in \eqref{eq:paper-concrete-particle-coefficient-decomposition} therefore makes \(\widetilde{\Pi}_{n-1}\) a bounded bijection. Its inverse has closed graph. Indeed, let \(\seq{a_j}{j\in\semigrposint}\) be a sequence in \(\oa{B}_{n-1}\) and let \(a\in\oa{B}_{n-1}\) and \(c\) belong to the finite direct sum. If \[\begin{aligned}
\lim_{j\to\infty}a_j
&=
a,
&
\lim_{j\to\infty}\fun{\widetilde{\Pi}_{n-1}^{-1}}{a_j}
&=
c,
\end{aligned}\] continuity of \(\widetilde{\Pi}_{n-1}\) gives \[\begin{aligned}
\fun{\widetilde{\Pi}_{n-1}}{c}
&=
\lim_{j\to\infty}
\fun{\widetilde{\Pi}_{n-1}}{
\fun{\widetilde{\Pi}_{n-1}^{-1}}{a_j}}
=
\lim_{j\to\infty}a_j
=
a.
\end{aligned}\] Injectivity gives \(c=\fun{\widetilde{\Pi}_{n-1}^{-1}}{a}\). The closed graph theorem therefore gives \[\begin{aligned}
\norm{\widetilde{\Pi}_{n-1}^{-1}}
&<
\infty.
\end{aligned}\] Define the bounded linear map \[\begin{aligned}
s_n
\colon
\oa{B}_{n-1}
\to
\oa{B}_n;\quad
\fun{s_n}{a}
&=\sum_{m=0}^{n-1}\fun{\widetilde{E}_{m,n}}{a_m}.
\end{aligned}\] For \(0\leq m\leq n\) and \(K\in\fun{\schattencls}{\sphilbfrak{h}^{\otimes_{\txtsym}m}}\), the subset-sum formula gives \[\begin{aligned}
\norm{\fun{\widetilde{E}_{m,n}}{K}}
&\leq
\binom{n}{m}\norm{K}.
\end{aligned}\] The section map obeys the norm estimate \[\begin{aligned}
\norm{\fun{s_n}{a}}
\leq
\sum_{m=0}^{n-1}
\binom{n}{m}\norm{a_m}
\leq
\rbk{2^n-1}
\norm{\fun{\widetilde{\Pi}_{n-1}^{-1}}{a}}
\leq
\rbk{2^n-1}
\norm{\widetilde{\Pi}_{n-1}^{-1}}
\norm{a}.
\end{aligned}\] Thus \(s_n\) is bounded. It is unital and adjoint preserving, but need not be multiplicative. Equations \eqref{eq:paper-concrete-particle-coefficient-decomposition} and \eqref{eq:paper-concrete-deletion-actions} give the short exact sequence \[\begin{aligned}
0
&\to
\fun{\schattencls}{\sphilbfrak{h}^{\otimes_{\txtsym}n}}
\to
\oa{B}_n
\to
\oa{B}_{n-1}
\to
0.
\end{aligned}\] The quotient map in this sequence is \(\kappa_n\), and the inclusion map is \[\begin{aligned}
\iota_n
\colon
\fun{\schattencls}{\sphilbfrak{h}^{\otimes_{\txtsym}n}}
\to
\oa{B}_n;\quad
\fun{\iota_n}{K}
&=
\fun{\widetilde{E}_{n,n}}{K}
=
K.
\end{aligned}\] Definition \ref{def:busby-map} applied to this essential extension gives \[\begin{aligned}
\tau_n
\colon
\oa{B}_{n-1}
\to
\oa{Q}_n;\quad
\fun{\tau_n}{a}
&=\fun{q_n\circ s_n}{a}.
\end{aligned}\] Two injective unital \(\ast\)-homomorphisms \(\phi,\psi\colon\oa{B}_{n-1}\to\oa{Q}_n\) are unitarily equivalent if a unitary \(V\in\oa{Q}_n\) satisfies \(\fun{\phi}{a}=V\fun{\psi}{a}\faadj{V}\) for every \(a\in\oa{B}_{n-1}\). In this paper, the extension invariant associated with the short exact sequence means the class of \(\tau_n\) in \[\begin{aligned}
\setquot{
\set{\phi\colon\oa{B}_{n-1}\to\oa{Q}_n}
{\phi\text{ is an injective unital $\ast$-homomorphism}}}
{\text{unitary equivalence}}.
\end{aligned}\]

\begin{cor}[Concrete pullback of the particle-cutoff extension]\label{cor:particle-cutoff-busby-pullback}
The map $\tau_n$ is an injective unital $\ast$-homomorphism, and
\begin{equation}
\begin{aligned}
\oa{B}_n
&\simeq
\set{\rbk{T,a}\in
\opspbddlin{\sphilbfrak{h}^{\otimes_{\txtsym}n}}\oplus\oa{B}_{n-1}}
{\fun{q_n}{T}=\fun{\tau_n}{a}}.
\end{aligned}
\label{eq:particle-cutoff-busby-pullback}
\end{equation}
\end{cor}

\begin{proof}
For $a\in\oa{B}_{n-1}$ with the concrete particle-degree coefficients from
\eqref{eq:paper-concrete-particle-coefficient-decomposition},
equation \eqref{eq:paper-concrete-deletion-actions} gives
$$\begin{aligned}
\fun{\kappa_n\circ s_n}{a}
=
\sum_{m=0}^{n-1}
\fun{\kappa_n\circ\widetilde{E}_{m,n}}{a_m}
=
\sum_{m=0}^{n-1}
\fun{\widetilde{E}_{m,n-1}}{a_m}
=
a.
\end{aligned}$$
For $a,b\in\oa{B}_{n-1}$, set
$D_{a,b}=\fun{s_n}{a}\fun{s_n}{b}-\fun{s_n}{ab}$.
The multiplicativity defect satisfies
$$\begin{aligned}
\fun{\kappa_n}{D_{a,b}}
&=ab-ab=0.
\end{aligned}$$
The kernel of $\kappa_n$ is the compact ideal, so $D_{a,b}$ is compact and
$$\begin{aligned}
\fun{\tau_n}{ab}
&=\fun{q_n}{D_{a,b}+\fun{s_n}{ab}}
=\fun{\tau_n}{a}\fun{\tau_n}{b}.
\end{aligned}$$
The unit has the coefficient tuple $\rbk{1,0,\ldots,0}$, while the coefficients of $\faadj{a}$ are
$\rbk{\faadj{a_0},\ldots,\faadj{a_{n-1}}}$.
The section is unital and adjoint preserving:
$$\begin{aligned}
\fun{s_n}{1}
&=
\fun{\widetilde{E}_{0,n}}{1}
=
1,
\\
\fun{s_n}{\faadj{a}}
&=
\sum_{m=0}^{n-1}
\fun{\widetilde{E}_{m,n}}{\faadj{a_m}}
=
\sum_{m=0}^{n-1}
\faadjrbk{\fun{\widetilde{E}_{m,n}}{a_m}}
=
\faadjrbk{\fun{s_n}{a}}.
\end{aligned}$$
If $\fun{\tau_n}{a}=0$, then $\fun{s_n}{a}$ is compact and hence has zero image under $\kappa_n$.
Since $\fun{\kappa_n\circ s_n}{a}=a$, we obtain $a=0$.

Define the comparison map by
\begin{equation}
\begin{aligned}
\Xi_n
\colon
\oa{B}_n
\to
\set{\rbk{T,a}\in
\opspbddlin{\sphilbfrak{h}^{\otimes_{\txtsym}n}}\oplus\oa{B}_{n-1}}
{\fun{q_n}{T}=\fun{\tau_n}{a}};\quad
\fun{\Xi_n}{T}
&=\rbk{T,\fun{\kappa_n}{T}}.
\end{aligned}
\label{eq:particle-cutoff-pullback-comparison-map}
\end{equation}
For $T\in\oa{B}_n$,
$$\begin{aligned}
\fun{\kappa_n}{T-\fun{s_n\circ\kappa_n}{T}}
&=
\fun{\kappa_n}{T}
-\fun{\kappa_n\circ s_n\circ\kappa_n}{T}
=
0,
\\
T-\fun{s_n\circ\kappa_n}{T}
&\in
\Ker\kappa_n
=
\fun{\schattencls}{\sphilbfrak{h}^{\otimes_{\txtsym}n}},
\\
\fun{q_n}{T}
&=
\fun{q_n\circ s_n\circ\kappa_n}{T}
=
\fun{\tau_n\circ\kappa_n}{T}.
\end{aligned}$$
Thus the map $\Xi_n$ in \eqref{eq:particle-cutoff-pullback-comparison-map} takes values in the pullback in
\eqref{eq:particle-cutoff-busby-pullback}.
Conversely, if $\fun{q_n}{T}=\fun{\tau_n}{a}$, then $T-\fun{s_n}{a}$ is compact and
$$\begin{aligned}
T
&=\rbk{T-\fun{s_n}{a}}+\fun{s_n}{a}
\in\fun{\schattencls}{\sphilbfrak{h}^{\otimes_{\txtsym}n}}+\oa{B}_n
=\oa{B}_n.
\end{aligned}$$
The compact-difference calculation shows that the first-coordinate projection of the pullback takes values in
$\oa{B}_n$, and
$$\begin{aligned}
\fun{\kappa_n}{T}
&=\fun{\kappa_n}{T-\fun{s_n}{a}}+\fun{\kappa_n\circ s_n}{a}
=0+a=a.
\end{aligned}$$
Consequently the map $\Xi_n$ in \eqref{eq:particle-cutoff-pullback-comparison-map} is bijective, with inverse
\begin{equation}
\begin{aligned}
\fun{\Xi_n^{-1}}{T,a}
&=
T.
\end{aligned}
\label{eq:particle-cutoff-pullback-comparison-inverse}
\end{equation}
The map $\Xi_n^{-1}$ in \eqref{eq:particle-cutoff-pullback-comparison-inverse} preserves the algebraic
operations coordinatewise, so \eqref{eq:particle-cutoff-pullback-comparison-map} gives the isomorphism
\eqref{eq:particle-cutoff-busby-pullback}.
\end{proof}

Define the bilinear multiplicative defect by \[\begin{aligned}
\Omega_n
\colon
\oa{B}_{n-1}\times\oa{B}_{n-1}
\to
\fun{\schattencls}{\sphilbfrak{h}^{\otimes_{\txtsym}n}};\quad
\fun{\Omega_n}{a,b}
&=\fun{s_n}{a}\fun{s_n}{b}-\fun{s_n}{ab}.
\end{aligned}\] For \(a,b,c\in\oa{B}_{n-1}\), associativity has the two expansions \[\begin{aligned}
\rbk{\fun{s_n}{a}\fun{s_n}{b}}\fun{s_n}{c}
&=\rbk{\fun{s_n}{ab}+\fun{\Omega_n}{a,b}}\fun{s_n}{c},\\
\fun{s_n}{a}\rbk{\fun{s_n}{b}\fun{s_n}{c}}
&=\fun{s_n}{a}\rbk{\fun{s_n}{bc}+\fun{\Omega_n}{b,c}}.
\end{aligned}\] Subtracting \(\fun{s_n}{abc}\) gives \[\begin{aligned}
\fun{\Omega_n}{ab,c}+\fun{\Omega_n}{a,b}\fun{s_n}{c}
&=\fun{\Omega_n}{a,bc}+\fun{s_n}{a}\fun{\Omega_n}{b,c}.
\end{aligned}\] Fix \(0\leq p,q\leq n-1\), \(K\in\fun{\schattencls}{\sphilbfrak{h}^{\otimes_{\txtsym}p}}\), and \(L\in\fun{\schattencls}{\sphilbfrak{h}^{\otimes_{\txtsym}q}}\). Define \[\begin{aligned}
a
=
\fun{\widetilde{E}_{p,n-1}}{K},
\quad
b
=
\fun{\widetilde{E}_{q,n-1}}{L}.
\end{aligned}\] The definition of \(s_n\) and the overlap formula \eqref{eq:paper-exact-overlap-product} give \[\begin{aligned}
\fun{s_n}{a}
=
\fun{\widetilde{E}_{p,n}}{K},
\quad
\fun{s_n}{b}
=
\fun{\widetilde{E}_{q,n}}{L},
\quad
\fun{s_n}{a}\fun{s_n}{b}
=
\sum_{r=\max\setone{p,q}}^{\min\setone{n,p+q}}
\fun{\widetilde{E}_{r,n}\circ\widetilde{\mu}_{p,q}^{r}}{K,L}.
\end{aligned}\] The product in \(\oa{B}_{n-1}\) and the uniqueness of its concrete particle-degree coefficients give \[\begin{aligned}
ab
&=
\sum_{r=\max\setone{p,q}}^{\min\setone{n-1,p+q}}
\fun{\widetilde{E}_{r,n-1}\circ\widetilde{\mu}_{p,q}^{r}}{K,L},
\\
\fun{s_n}{ab}
&=
\sum_{r=\max\setone{p,q}}^{\min\setone{n-1,p+q}}
\fun{\widetilde{E}_{r,n}\circ\widetilde{\mu}_{p,q}^{r}}{K,L}.
\end{aligned}\] Subtracting the two finite sums yields \[\begin{aligned}
\fun{\Omega_n}{a,b}
&=
\begin{cases}
\fun{\widetilde{E}_{n,n}}{
\fun{\widetilde{\mu}_{p,q}^{n}}{K,L}},
&
p+q\geq n,
\\
0,
&
p+q<n.
\end{cases}
\end{aligned}\]

\subsubsection{Finite-mode inclusions and compatibility with the Busby map}\label{finite-mode-inclusions-and-compatibility-with-the-busby-map}

In this subsubsection, finite-mode inclusions intertwine the represented quotient maps and Busby maps; the basic realization requires no graph.

Only Proposition \ref{prop:graph-local-busby-naturality} uses the finite-region data in \eqref{eq:paper-main-finite-region-directed-set}--\eqref{eq:paper-finite-region-one-particle-projection}.

\begin{prop}[Finite-mode inclusion and Busby compatibility]\label{prop:concrete-busby-realization}
Fix a nonzero finite-dimensional complex subspace $\sphilbfrak{f}\subset\sphilbfrak{h}$.
Define
$$\begin{aligned}
L_{\sphilbfrak{f},n}
\colon
\opspbddlin{\fun{\spfockgamma_{\txtbsn}}{\sphilbfrak{f}}}^{\gamma}
\to
\oa{B}_n;\quad
\fun{L_{\sphilbfrak{f},n}}{D}
&=
\faadj{W_{\sphilbfrak{f},n}}
\rbk{\bigoplus_{k=0}^{n}
\fun{\pi_{\sphilbfrak{f},k}}{D_k}\otimes1_{\rbk{\sphilbfrak{f}^\perp}^{\otimes_{\txtsym}\rbk{n-k}}}}
W_{\sphilbfrak{f},n}.
\end{aligned}$$
For $D=\seq{D_k}{k\in\monnat}\in\opspbddlin{\fun{\spfockgamma_{\txtbsn}}{\sphilbfrak{f}}}^{\gamma}$, let
$\seq{c_m}{m\in\monnat}$ be its particle-degree coefficients from Section
\ref{sec:paper-concrete-finite-mode-embeddings}, extended by zero to
$\sphilbfrak{h}^{\otimes_{\txtsym}m}$.
Then
\begin{equation}
\begin{aligned}
\fun{L_{\sphilbfrak{f},n}}{D}-\fun{s_n\circ L_{\sphilbfrak{f},n-1}}{D}
&=\fun{\widetilde{E}_{n,n}}{c_n}
\in\fun{\schattencls}{\sphilbfrak{h}^{\otimes_{\txtsym}n}}.
\end{aligned}
\label{eq:local-busby-compact-difference}
\end{equation}
The Busby representatives consequently satisfy
\begin{equation}
\begin{aligned}
\fun{q_n\circ L_{\sphilbfrak{f},n}}{D}
&=\fun{\tau_n\circ L_{\sphilbfrak{f},n-1}}{D}.
\end{aligned}
\label{eq:local-busby-realization}
\end{equation}
\end{prop}

\begin{proof}
The particle-coefficient expansion established in Section \ref{sec:paper-concrete-finite-mode-embeddings} reads
$$\begin{aligned}
\fun{L_{\sphilbfrak{f},n}}{D}
&=\sum_{m=0}^{n}\fun{\widetilde{E}_{m,n}}{c_m},
&
\fun{L_{\sphilbfrak{f},n-1}}{D}
&=\sum_{m=0}^{n-1}\fun{\widetilde{E}_{m,n-1}}{c_m}.
\end{aligned}$$
Applying the definition of $s_n$ to the particle-coefficient expansion of
$\fun{L_{\sphilbfrak{f},n-1}}{D}$ gives
$$\begin{aligned}
\fun{s_n\circ L_{\sphilbfrak{f},n-1}}{D}
&=\sum_{m=0}^{n-1}\fun{\widetilde{E}_{m,n}}{c_m}.
\end{aligned}$$
Subtraction proves \eqref{eq:local-busby-compact-difference}.
Since $c_n$ is the zero extension of an operator
on the finite-dimensional space $\sphilbfrak{f}^{\otimes_{\txtsym}n}$, it has finite rank and is compact.
Applying $q_n$
and using $\tau_n=q_n\circ s_n$ proves \eqref{eq:local-busby-realization}.
\end{proof}

\begin{prop}[Compatibility with inclusions of finite graph regions]\label{prop:graph-local-busby-naturality}
Assume the finite-region data in
\eqref{eq:paper-main-finite-region-directed-set}--\eqref{eq:paper-finite-region-one-particle-projection}, and fix
nonempty finite regions
$\Lambda\subset \Lambda''\subset\Gamma$.
For the map $i_{\Lambda,\Lambda''}$ in \eqref{eq:concrete-local-inclusion},
\begin{equation}
\begin{aligned}
L_{\sphilbfrak{h}_{\Lambda''},n}\circ i_{\Lambda,\Lambda''}
&=L_{\sphilbfrak{h}_{\Lambda},n}.
\end{aligned}
\label{eq:local-busby-naturality}
\end{equation}
\end{prop}

\begin{proof}
Use the orthogonal decomposition
$$\begin{aligned}
\sphilbfrak{h}
&=\sphilbfrak{h}_{\Lambda}\oplus \sphilbfrak{h}_{\Lambda''\setminus \Lambda}\oplus \sphilbfrak{h}_{\Lambda''}^\perp.
\end{aligned}$$
For $i,j,k\in\monnat$ satisfying $i+j+k=n$, fix
$\xi\in\sphilbfrak{h}_{\Lambda}^{\otimes_{\txtsym}i}$,
$\eta\in\sphilbfrak{h}_{\Lambda''\setminus \Lambda}^{\otimes_{\txtsym}j}$, and
$\zeta\in\rbk{\sphilbfrak{h}_{\Lambda''}^\perp}^{\otimes_{\txtsym}k}$.
The decomposition unitary and its action on the algebraic direct sum are
$$\begin{aligned}
W_{\Lambda,\Lambda'',n}^{(3)}
\colon
\sphilbfrak{h}^{\otimes_{\txtsym}n}
\to
\bigoplus_{\substack{i,j,k\in\monnat\\i+j+k=n}}
\sphilbfrak{h}_{\Lambda}^{\otimes_{\txtsym}i}\otimes \sphilbfrak{h}_{\Lambda''\setminus \Lambda}^{\otimes_{\txtsym}j}
\otimes\rbk{\sphilbfrak{h}_{\Lambda''}^\perp}^{\otimes_{\txtsym}k}.
\end{aligned}$$
Its inverse acts by
$$\begin{aligned}
\fun{\faadjrbk{W_{\Lambda,\Lambda'',n}^{(3)}}}{\xi\otimes\eta\otimes\zeta}
&=\sqrt{\frac{n!}{i!j!k!}}P_n\rbk{\xi\otimes\eta\otimes\zeta}.
\end{aligned}$$
Equation \eqref{eq:concrete-local-inclusion} says that
$\fun{i_{\Lambda,\Lambda''}}{D}$ acts as the identity on $\sphilbfrak{h}_{\Lambda''\setminus \Lambda}$.
Therefore
$$\begin{aligned}
W_{\Lambda,\Lambda'',n}^{(3)}\fun{L_{\sphilbfrak{h}_{\Lambda''},n}\circ i_{\Lambda,\Lambda''}}{D}\faadjrbk{W_{\Lambda,\Lambda'',n}^{(3)}}
&=
\bigoplus_{\substack{i,j,k\in\monnat\\i+j+k=n}}
\fun{\pi_{\sphilbfrak{h}_{\Lambda},i}}{D_i}\otimes1_{\sphilbfrak{h}_{\Lambda''\setminus \Lambda}^{\otimes_{\txtsym}j}}
\otimes1_{\rbk{\sphilbfrak{h}_{\Lambda''}^\perp}^{\otimes_{\txtsym}k}},\\
W_{\Lambda,\Lambda'',n}^{(3)}\fun{L_{\sphilbfrak{h}_{\Lambda},n}}{D}\faadjrbk{W_{\Lambda,\Lambda'',n}^{(3)}}
&=
\bigoplus_{\substack{i,j,k\in\monnat\\i+j+k=n}}
\fun{\pi_{\sphilbfrak{h}_{\Lambda},i}}{D_i}\otimes1_{\sphilbfrak{h}_{\Lambda''\setminus \Lambda}^{\otimes_{\txtsym}j}}
\otimes1_{\rbk{\sphilbfrak{h}_{\Lambda''}^\perp}^{\otimes_{\txtsym}k}}.
\end{aligned}$$
The right sides coincide, and $W_{\Lambda,\Lambda'',n}^{(3)}$ is unitary.
This proves
\eqref{eq:local-busby-naturality}.
\end{proof}

Equation \eqref{eq:local-busby-realization} determines the finite-mode restriction for a general \(\sphilbfrak{h}\), while \eqref{eq:local-busby-naturality} adds compatibility with inclusions of finite graph regions.

\bibliography{myref.bib}

\end{document}